\documentclass{article}
\usepackage{authblk}
\usepackage[english]{babel}
\usepackage[top=1in,bottom=1in,left=1in,right=1in]{geometry}
\usepackage{indentfirst}
\usepackage{mathtools}
\usepackage{amssymb}
\usepackage{amsthm}
\usepackage{bm}
\usepackage[table]{xcolor}
\usepackage{graphicx}
\usepackage[percent]{overpic}
\usepackage{tikz}
\usepackage{caption}
\usepackage{subfigure} 
\usepackage{float}
\usepackage{enumerate}
\usepackage[title]{appendix}
\usepackage{todonotes}
\usepackage{hyperref}

\presetkeys{todonotes}{inline}{}
\presetkeys{todonotes}{
  color=white,      
  bordercolor=black,
  textcolor=black,
  tickmarkheight=0pt  
}{}
\let\oldtodo\todo
\renewcommand{\todo}[1]{\oldtodo{\textbf{\color{red} ToDo:}~#1}}

\newcommand{\swatch}[1]{%
  \begingroup
  \definecolor{swatchcol}{HTML}{#1}%
  \tikz[baseline=(sw.base)]{
    \node[
      inner sep=0pt,
      outer sep=0pt,
      minimum width=1.9em,
      minimum height=0.9em,
      rounded corners=1.5pt,
      draw=black!40,
      fill=swatchcol
    ] (sw) {};
  }%
  \endgroup
}
\newcommand{\legenditem}[2]{\swatch{#1} & \raisebox{-0.75ex}{\texttt{#2}}}

\numberwithin{equation}{section}
\newtheorem{theorem}{Theorem}[section]
\newtheorem{proposition}[theorem]{Proposition}
\newtheorem{conjecture}[theorem]{Conjecture}
\newtheorem{corollary}[theorem]{Corollary}

\newtheorem{remark}[theorem]{Remark}

\newtheorem{lemma}[theorem]{Lemma}

\newcommand{\calS}{\mathcal{S}}
\newcommand{\R}{\mathbb{R}}
\newcommand{\Z}{\mathbb{Z}}

\newcommand{\muT}{\mu_{\Tt}}
\newcommand{\Ps}[1]{\mathcal{P}_{\!#1}}

\DeclareMathOperator{\Ll}{L}
\DeclareMathOperator{\Aa}{A}
\DeclareMathOperator{\Ee}{E}
\DeclareMathOperator{\Tt}{T}
\DeclareMathOperator{\Xx}{X}
\DeclareMathOperator{\Zz}{Z}

\title{Lagrange points of the restricted three-body problem in spaces of constant curvature}

\author{
Miguel Ayala\thanks{Department of Mathematics and Statistics, McGill University, 805 Sherbrooke Street West, Montreal, QC H3A 0B9, Canada. E-mail: \url{miguel.ayala@mail.mcgill.ca}},
Carlos Barrera-Anzaldo\thanks{Departamento Académico de Matemáticas, Instituto Tecnológico Autónomo de México, Calle Río Hondo 1, 01080 Ciudad de México, Mexico. E-mail: \url{carlos.barrera.anzaldo@itam.mx}},
Luis C. García-Naranjo\thanks{Dipartimento di Matematica ``Tullio Levi-Civita'', Università degli Studi di Padova, Via Trieste 63, 35121 Padova, Italy. E-mail: \url{luis.garcianaranjo@math.unipd.it}}
}

\date{\today}

\definecolor{graycover}{RGB}{80,80,80}
\definecolor{azulciencias}{RGB}{11,61,98}
\definecolor{tealink}{RGB}{21, 154, 137}
\definecolor{BrickRed}{rgb}{0.6, 0.2, 0.2}
\hypersetup{
   pdftitle={},
   pdfauthor={},
   pdfsubject={},
   pdfkeywords={},
   bookmarksnumbered=true,     
   bookmarksopen=true,         
   bookmarksopenlevel=1,       
   colorlinks=true,         
    linkcolor=black, 
   citecolor=BrickRed,
   pdfstartview=Fit,           
   pdfpagemode=UseOutlines,    
   pdfpagelayout=TwoPageRight
}

\begin{document}

\maketitle

\begin{abstract}
We continue the study initiated by Kilin ({\em Reg. Chaot. Dyn.} 4, (1999)) and by Mart\'inez and Sim\'o  ({\em Celest. Mech. Dynam. Astronom.} 128, (2017)) 
on the classification and stability of the relative equilibria of the restricted three-body problem in two-dimensional spaces of constant curvature, which generalize the
 classical Lagrange points of the planar problem. After formulating the problem as an autonomous Lagrangian system with two degrees of freedom, whose only parameters are the curvature 
$\kappa$ and the mass ratio $\mu$ of the primaries, we establish several classification results for the case $\kappa>0$ by combining analytical methods with computer-assisted proofs. 
These results provide rigorous confirmation of phenomena for which previously only numerical evidence was available. We also provide topological explanations for the qualitative 
differences between the behavior of relative equilibria in positive curvature and that observed in the planar and negative-curvature cases. 
Our analysis focuses on the regime of small $\mu$ and indicates that positive curvature has a stabilizing effect on the triangular equilibria 
$\Ll_4$ and $\Ll_5$, whereas negative curvature appears to have the opposite effect.
\end{abstract}

\medskip


\medskip


\tableofcontents

\section{Introduction}

This paper continues the research program started by Kilin \cite{Ki99}
and Mart\'inez and Sim\'o \cite{MS17} on the
classification and stability of relative
equilibria (RE) of the (circular) restricted
3-body problem 
(R3BP) in two-dimensional spaces
of constant curvature. The goal is to understand the behavior of these solutions, which generalize
the classical Lagrangian libration points
from the planar case, as functions of the
mass ratio between the primaries, $\mu:=\mu_2/\mu_1$, and the spatial curvature, $\kappa$.
As already pointed out in the abstract of
\cite{MS17},  the situation for negative
$\kappa$ is very similar to the 
planar one. Instead, for $\kappa>0$ new RE
arise and several bifurcations occur 
as the parameters $(\mu, \kappa)$ are varied.
Hence, this work considers mainly 
$\kappa>0$. We also  focus on the 
regime of small  $\mu \in (0,1)$,
in which one of the primaries is
much larger than the other, since it
is in this situation that the classical equilateral
Lagrange points $\Ll_4,\Ll_5$, become
linearly stable for $\kappa=0$. Our study is therefore complementary to that of 
 Andrade, P\'erez-Chavela, Vidal \cite{APV18}, who study the same problem  in the
case of equal masses ($\mu=1$).

One of the driving motivations for this work was to understand mechanisms leading to stability of 
relative equilibria in the full three-body problem in spaces of constant curvature,
especially in the positive case. A number
of recent papers focus on their existence,
\cite{FuPC24, Fuj24,  Zhu24},
but there are few stability results.
On the other hand, it is known that (at least in the planar case) the presence
of one large mass in a RE has a stabilizing 
effect \cite{R60}. As was mentioned
above, it is because of this reason that our study focuses on the regime of small  $\mu$.

The  motivation described above to study the R3BP is unrelated to 
Poincar\'e's  discovery of the homoclinic tangling that  leads to chaotic behavior  in the planar $N$-body problem as the number of
$N$ bodies increases from two to three, and which 
continues to motivate research these days (e.g. \cite{GuarMarSe16, BaGiGuar22, BaGiGuar23, LaGuarSe25}). 
In the case of
non-zero curvature, the underlying two-body problem is
already non-integrable, so the appropriate system to investigate the transition from
integrability to chaos as the number of bodies increases is the restricted
2-body problem considered in~\cite{Prz03, BoMa06}. We refer the reader to~\cite{Shche06, BorMaKi16, Bol25}
for the nonintegrability of the 2-body problem for non-zero curvature (see also~\cite{Jack23} for discussion of secular dynamics). 

Even though the two-body problem is non-integrable for non-zero curvature, 
the R3BP, in its circular version, is nevertheless well defined,
 since the circular orbits of the
planar two-body problem persist analytically 
 for non-vanishing curvature
\cite{GNM21}. Moreover, these circular orbits appear to be nonlinearly stable 
on an open dense subset of the parameter space \cite{GNBMM}, providing  a
natural and robust framework  for formulating the restricted three-body problem.

Although our interest in the problem is mainly mathematical,
our results may be useful to understand the influence
of spatial curvature on the stability of the Lagrange libration 
points  when relativistic effects are
not completely negligible, as for the Sun-Mercury R3BP. For instance, our results in Sec.~\ref{sec:stabilitytriangular}
 prove that positive curvature
has a stabilizing effect for $\Ll_4$ and   $\Ll_5$, whereas  numerical evidence indicates that the opposite is
 true for  negative curvature.

\subsection{Contributions}

A first contribution of this work is to 
formulate the R3BP in spaces of constant curvature
as an autonomous two-degree of freedom Lagrangian system
in which the mass ratio $\mu$ and the curvature
$\kappa$ appear as the only parameters. In our approach, which
borrows ideas from \cite{MT13, GNM21}, the
Riemannian distance between the primaries is normalized to $1$
regardless of the curvature $\kappa$. This is essential
to appropriately distinguish the role played by the curvature, which constitutes a considerable improvement with respect to the normalization 
approaches followed in \cite{MS17}. Our approach also exploits the full classification of circular orbits in spaces of nonzero constant curvature, which correspond to certain relative equilibria, obtained in \cite{GNBMM}. This classification was not available when \cite{Ki99, MS17} were written.  Our set-up
allows us to translate the problem of 
existence of RE into the problem of finding critical
points of a potential function 
  $ \mathcal{V}_{\kappa,\mu}$ (defined in \ref{eq:scaledpotential}) that depends on the parameters $\kappa$ and $\mu$, and is also useful to
study stability.

A second contribution is to give a topological
explanation for the appearance of new RE for positive curvature. As  observed in \cite{MS17}, the 
classical Lagrange points $\Ll_1,\dots, \Ll_5$, of the planar case are non-degenerate and they therefore must persist as RE, at least for small non-zero curvature. On the other hand,
the numerical results of \cite{Ki99,MS17}, suggest
existence of  new RE  for small positive $\kappa$.  The mechanism responsible for this phenomenon is clarified in 
the light of Theorem~\ref{prop:eqptsSpos} (see Corollary \ref{cor:extraRE}) and may be illustrated with an even simpler example: consider  
the Kepler (or, more generally, the central
force problem) on the sphere. It is clear that, in addition to the singularity at the force field's center,
there must be a singularity or an equilibrium at its antipodal point. Moreover, the distance from
such an antipodal point to the field's center tends to infinity as the radius of the sphere grows, explaining   how this additional singularity or equilibrium disappears when $\kappa= 0$.
 In a similar manner, the compactness
of the sphere restricts the number and type of
RE  for the R3BP for $\kappa>0$ in a way that is quantifiable
by the Poincar\'e-Hopf Index Theorem. This topological
restriction implies that 
`new' RE must necessarily exist for small $\kappa>0$. The
distance from these new RE to the primaries
tends to infinity as $\kappa\to 0^+$ and only the
Lagrange points $\Ll_1,\dots, \Ll_5$ remain for $\kappa=0$. These observations represent a significant contribution with respect to \cite{MS17}, where the analysis for $\kappa>0$ is restricted to the hemisphere containing the primaries and  no
global aspects are taken into consideration. Moreover, Theorem~\ref{prop:eqptsSpos} is also useful 
in the analysis of the bifurcations of RE for $\kappa>0$, since it leads to a restriction on the types of bifurcations
that may occur.

A third contribution of this work is the derivation of necessary and sufficient conditions for the 
existence and stability of RE that are valid for all parameter values $(\kappa,\mu)$. 
Our work relies on these conditions to obtain results in certain regions of the parameter space, but the
criteria   are equally applicable to the investigation of regions not considered in this paper.

A fourth contribution is to provide rigorous classification and stability results for the RE of the problem in certain regions of the parameter space $(\kappa,\mu)$ with $\kappa>0$. Such regions often require $\mu$ to be small. Our approach to investigating existence is analogous to the one usually employed in the planar case (see, e.g., \cite{MO17}): we divide the analysis into collinear and triangular RE and make appropriate coordinate choices. Our analytical proofs are complemented by computer-assisted proofs (CAPs), which rely on the Intermediate Value Theorem from elementary calculus combined with interval arithmetic. These rigorous results are all original since, for $\kappa>0$, and with the exception of the limit cases $\mu=0$ and $\mu=1$, the investigation in \cite{Ki99, MS17} (including classification and stability results) is only numerical. Our combination of 
analytical and computer-assisted  proofs is particularly powerful for the following reason. On the one hand, 
CAPs yields rigorous results on large regions of the parameter space. On the other hand, analytical
proofs are necessary to formulate rigorous continuation arguments as the curvature $\kappa$ approaches $0$, since CAPs become inconclusive for small values of $\kappa$.

Our rigorous existence and stability results are complemented by numerical investigations in wider regions of the parameter space $(\kappa,\mu)$, as well as by bifurcation diagrams. Some of these reproduce earlier results from \cite{Ki99, MS17}, whereas others, such as Fig.~\ref{fig:asympcollinear}, are presented in terms of the Riemannian distance from the Lagrange points to the center of mass of the primaries. This approach clarifies their behavior as the curvature $\kappa \to 0$. We also perform asymptotic expansions for the locations of the RE as $\kappa \to 0$, showing that some of them converge to the classical Lagrange points $\Ll_1,\dots,\Ll_5$, while others tend to infinity.

\subsection{The parameter region}

As mentioned above, our work considers the case in which the primaries have different masses and their mutual Riemannian distance is normalized to 1. Furthermore, in the case of positive curvature $\kappa$, we restrict our attention to circular motions of the primaries that subtend an acute angle, since only this family admits a continuation as a function of $\kappa$ through $\kappa=0$ \cite{GNM21}. As explained in more
detail in Section~\ref{ss:PrimaryMotion}, these considerations imply that the parameters $\kappa$, $\mu$
take values on the domain
\begin{equation*}
(\kappa, \mu)\in \left (-\infty, \frac{\pi^2}{4}\right ) \times (0,1).
\end{equation*}

In our treatment of the problem for positive curvature, given $0<\mu_0\leq 1$, we will often denote
by $\Ps{\mu_0}$ the parameter region consisting of all allowed positive values of the curvature and mass ratios 
less than $\mu_0$. Namely, 
\begin{equation}
\label{eq:defPs}
\Ps{\mu_0}:= \left (0, \frac{\pi^2}{4}\right ) \times (0,\mu_0).
\end{equation}
In particular, in order to provide a complete description of the structure of  the
RE, our CAPs and numerical investigations were,
for the most part,  implemented only for $(\kappa,\mu)\in \Ps{\mu_{\Tt}}$
where $\mu_{\Tt}\approx 0.089$ is defined in Sec.~\ref{subsec:massratio}. Our code for the 
existence and stability CAPs
\cite{github_codes} also works  outside this parameter region.
However, a complete description of the behavior of the RE as a function of the parameter $\kappa$, valid for a broader range of values of $\mu$, becomes considerably more intricate, and we therefore 
did not pursue more generality.

\subsection{Organization of the paper}

We begin by recalling a number of preliminaries in  Sec.~\ref{sec:pre}. In particular,
we review a geometric construction of surfaces of constant 
curvature, originally due to Montaldi, and Tokieda~\cite{MT13}, which allows the curvature to vary continuously
 from negative, through zero, to positive values whlie 
keeping the distance between two points (in our case, the primaries)
fixed. We also recall the results of \cite{GNBMM,GNM21} on circular solutions of the two-body problem in spaces
of constant curvature, which model the motion of the primaries, and adapt it to our setting.

In Sec.~\ref{sec:R3BP} we formulate the R3BP in surfaces of constant curvature as an autonomous $2$-degree of
freedom mechanical Lagrangian system. The key point of our construction is that 
the curvature $\kappa$ and the mass ratio $\mu$ are the only parameters appearing in the Lagrangian.
In particular, we give explicit expressions for the (augmented) potential, $\mathcal{V}_{\kappa,\mu}$,
 whose critical points are in one-to-one correspondence with the RE of the system. 

In Sec.~\ref{sec:collinearclass} we formulate Theorem~\ref{prop:eqptsSpos} about the topological restrictions on the
number and type of RE for $\kappa>0$. We also present necessary and sufficient conditions for existence of 
RE and convenient stability criteria that are used in the following sections. The discussion is naturally
divided into collinear and triangular RE, treating separately the cases of positive and negative $\kappa$.

Our existence results for collinear and triangular RE for positive $\kappa$ are respectively given in 
Secs.~\ref{sec:collinear} and \ref{sec:triangular}. These include Theorems~\ref{thm:numbercollinearRE} and~\ref{thm:numbertriangularRE}, which are proved analytically, complemented
with existence CAPs. Summaries of the bifurcation behavior of RE in the parameter region $\Ps{\muT}$, 
based on our rigorous and 
numerical results, are presented in Conjectures~\ref{thm:classcollinear} and \ref{thm:classtriangular}.

Sections~\ref{s:stability-collinear} and \ref{sec:stabilitytriangular} are devoted to the stability analysis of the collinear and triangular RE, respectively.
In both cases, for $\kappa>0$,  we present  CAPs of stability, complemented by numerical investigations. For the
collinear  RE, we also provide the analytically proved Theorem~\ref{thm:stabilitycollinear}. Based on our results, the stability behavior
 in the parameter region $\Ps{\muT}$, 
is summarized in  Conjectures~\ref{thm:class-stability-collinear} and \ref{thm:class-stability-triangular}.
Additionally, in Sec.~\ref{ss:stabilityL4L5}, we present numerical evidence that positive curvature promotes the gyroscopic stabilization of $\Ll_4$ and $\Ll_5$,
whereas  negative curvature has the opposite effect.

In Sec.~\ref{sec:asymp}, we present  asymptotic expansions that determine the
locations of both collinear and triangular RE in the limit as $\kappa\to 0$.

We discuss several directions of future research in Sec.~\ref{sec:future-work} and conclude the paper with three appendices. The first appendix 
reviews standard results on the stability analysis of equilibria in Lagrangian mechanical systems with two degrees of freedom. The
second appendix contains analytical proofs of the main existence and stability theorems and derives the asymptotic expansions
presented in Sec.~\ref{sec:asymp}. We have chosen to defer these proofs to an appendix in order to improve the 
readability of the main text and to highlight the principal contributions of the paper. Finally, the third appendix summarizes
 the methodology underlying the  CAPs developed in this work.
 
The support code for all CAPs is available in \cite{github_codes}.

\section{Preliminaries}
\label{sec:pre}

In this section, we review the aspects of the classical planar restricted three-body problem relevant to our analysis, the family of surfaces of constant curvature introduced in~\cite{MT13}, and the circular solutions (relative equilibria) of the two-body problem in spaces of constant curvature, which determine the motion of the primaries in the curved restricted three-body problem.

\subsection{Lagrange libration points}
\label{subsec:libpoints}
 The  classical  planar restricted
3-body problem (R3BP) 
is a limit case of the planar 3-body problem in which one of the bodies (called \textbf{\textit{satellite}}) has a negligible mass, and the other two bodies, of masses $\mu_{1}$ and $\mu_{2}$ (called \textbf{\textit{primaries}}), move on a  Keplerian orbit. In this work
we will only consider the case in which the Keplerian orbit is circular, so the terminology R3BP will
always mean \textit{circular}, planar, restricted
3-body problem. 
The R3BP concerns the 
determination of the motion of the satellite and
is conveniently studied on a rotating coordinate
system in which the primaries are at rest (and whose origin is their common center of rotation which coincides with their center of mass). The resulting mechanical system for the motion of the
satellite is autonomous and has two degrees
of freedom. After an appropriate scaling of 
units, one may assume
that the distance between the primaries is $1$, and
the only external parameter left in the problem is the mass ratio $\mu = \mu_2/\mu_1\in (0,1]$.

The equilibrium points of the equations of motion for the satellite are the famous  \textbf{\textit{Lagrange  points}}
(or Lagrange libration points) and  are denoted by $\Ll_{1}\dots, \Ll_{5}$. The  points $\Ll_{1}$, $\Ll_{2}$, and $\Ll_{3}$ are called \textbf{\textit{collinear}} since they lie on the line between the primaries and their precise position is a function of  $\mu$.  The Lagrange points $\Ll_{4}$ and $\Ll_{5}$ 
are called \textbf{\textit{triangular}} and lie at the   vertices of the two equilateral triangles whose common base is the line segment between the primaries. 
We refer the reader to \cite{Sz67} for details and
historical context.

\vspace{0.3cm}

\begin{figure}[ht]
\centering
\begin{overpic}[width=.5\linewidth]{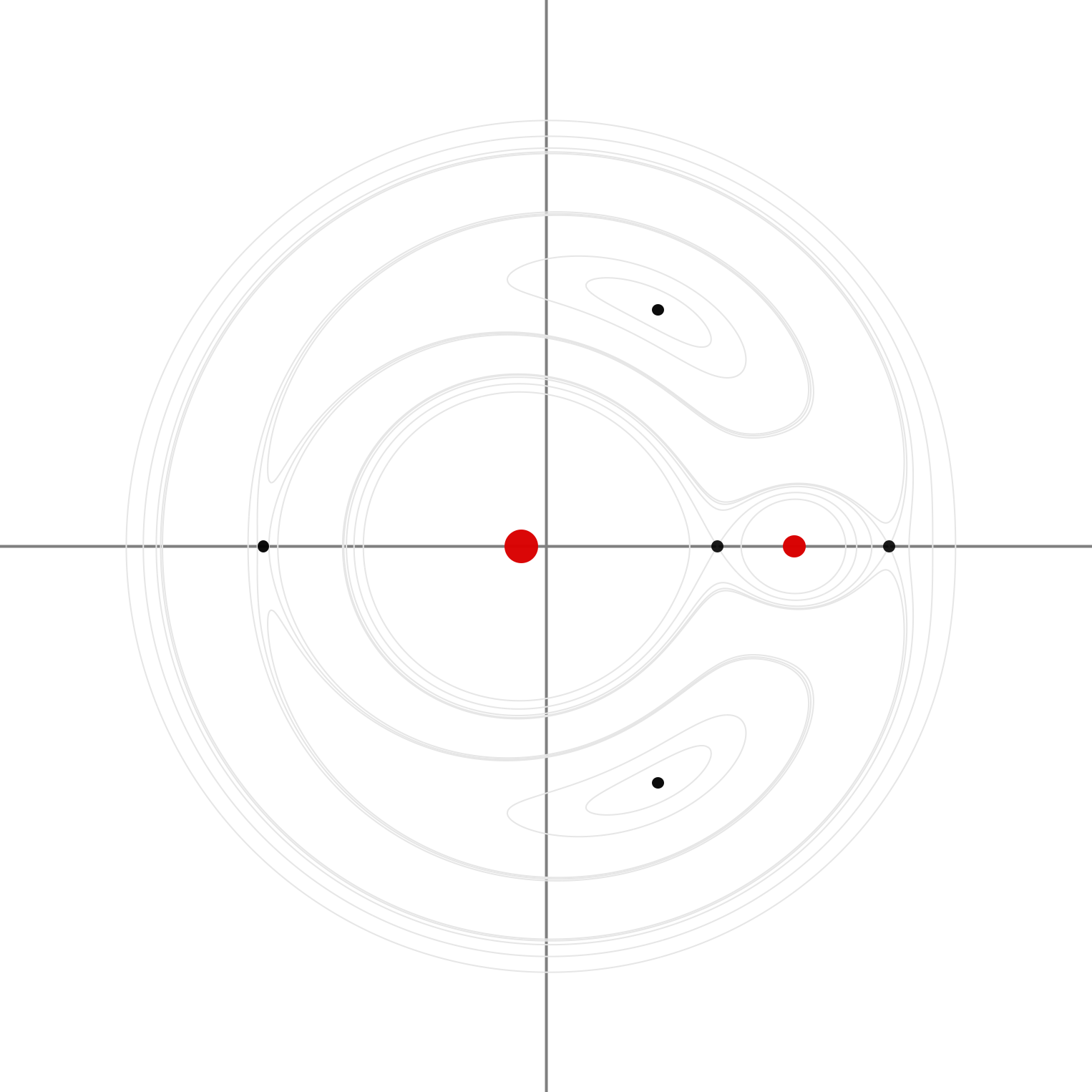}
    \put(46,53){$\mu_{1}$}
    \put(71,53){$\mu_{2}$}
    \put(64,46){$\Ll_{1}$}
    \put(80,46){$\Ll_{2}$}
    \put(22.5,46){$\Ll_{3}$}
    \put(58.75,67.5){$\Ll_{4}$}
    \put(58.75,24.5){$\Ll_{5}$}
\end{overpic}
\captionsetup{width=.75\textwidth}
\caption{Lagrange  points $\Ll_1, \dots, \Ll_5$ for the R3BP. These points correspond to critical points of the augmented potential $V_{0,\mu}$ given in Eq.~\eqref{eq:amended0}, projected onto the plane whose origin coincides with the center of mass of the primaries $\mu_1$, $\mu_2$. Some level curves of $V_{0,\mu}$ are plotted in gray.}
\label{fig:lagrangeplane}
\end{figure}
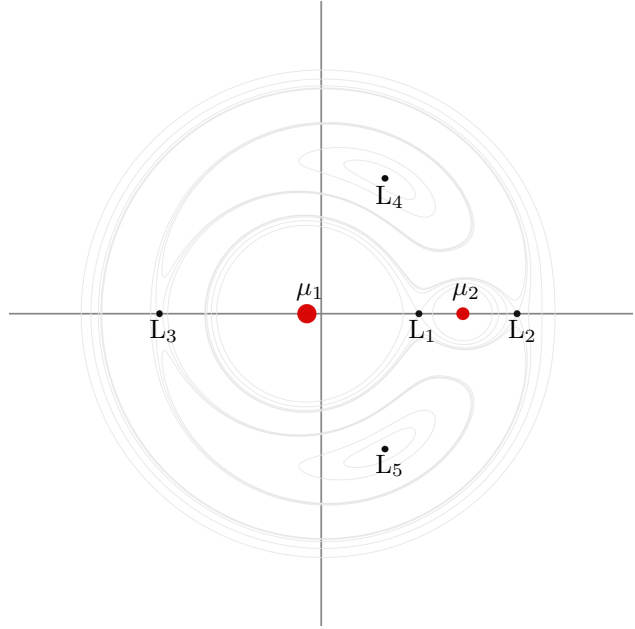

In Sec.~\ref{subsec:libpointsS0} below, we will give explicitly the Lagrangian system associated with the R3BP on a useful Riemannian model of the plane $\R^{2}$.

\subsection{Family of surfaces of constant curvature}
\label{subsec:surface}

We  follow the approach of \cite{GNM21} for the treatment of  the $2$-body problem in spaces with constant curvature (see also \cite{MT13}, where this approach was first introduced to study vortex rings dynamics). This setting
will allow us to vary the curvature maintaining 
the distance between the primaries normalized to $1$.
Consider the family of surfaces in $\R^3$ parametrized by 
$\kappa \in \R$,
\begin{equation*}
    \mathcal{S}_{\kappa}=\lbrace (x,y,z)\in\mathbb{R}^{3} : x^{2}+y^{2}+\kappa z^{2}-2z=0 \rbrace.
\end{equation*}
For every $p\in\calS_{\kappa}$, we  endow
$T_p\calS_\kappa$
with the scalar product, 
\begin{equation}
\label{eq:innerproduct}
    \langle u_p,v_p\rangle_\kappa = u_p^T K_{\kappa} v_p, \qquad u_p,v_p\in T_p\calS_\kappa \subset \R^3,
\end{equation}
where $K_\kappa=\mbox{diag}[1,1,\kappa]$. We denote by $\Vert\cdot\Vert_{\kappa}$ the norm induced by the inner product \eqref{eq:innerproduct}, namely
\begin{equation}
\label{eq:inducednorm}
    \Vert u_{p} \Vert_{\kappa}=\langle u_{p},u_{p}\rangle_{\kappa}^{1/2}.
\end{equation}

The inner product \eqref{eq:innerproduct}
defines a  Riemannian 
 metric on $\calS_\kappa$ that coincides with
 the restriction to $\calS_\kappa$ of the ambient 
 (pseudo) metric in $\R^3$ given by
\begin{equation}
\label{eq:metric}
    \mbox{d}s_{\kappa}^2=\mbox{d}x^2+\mbox{d}y^2+ \kappa\mbox{ d}z^2.
\end{equation}
The key point is that, equipped with this Riemannian 
metric, the surface $\calS_\kappa$ has constant
curvature $\kappa$. A  description of 
$\calS_\kappa$, together with an isometry to 
the standard models of constant curvature spaces, is
given below according to the sign of $\kappa$:

\begin{itemize}
    \item $\kappa>0$. The surface $\calS_\kappa$ is an ellipsoid with centre at $(0,0,1/\kappa)$.
    The mapping  
    \begin{equation}
    \label{eq:isopositive}
       \Phi^{+}:\calS_{\kappa}\to\mathbb{S}^{2}_{1/\sqrt{\kappa}}, \qquad  \Phi^{+}(x,y,z)=\left(x,y,\sqrt{\kappa}\ z-\frac{1}{\sqrt{\kappa}}\right),
    \end{equation}
    is an isometry where, $\mathbb{S}^{2}_{R}$ denotes the sphere of radius $R$ equipped with the standard Euclidean metric. In particular,  this implies that the maximum Riemannian distance between two points on $\calS_{\kappa}$ for $\kappa>0$ is $\pi/\sqrt{\kappa}$.
    \item $\kappa=0$. The surface $\calS_0$ is a paraboloid. In this case the metric \eqref{eq:metric} on $\R^3$ is degenerate, but
    restricted to $\calS_0$ is a Riemannian metric. The vertical projection, 
    \begin{equation*}
       \Phi^{0}:\calS_{0}\to\R^{2}\times\lbrace 0 \rbrace, \qquad \Phi^{0}(x,y,z)=(x,y,0),
    \end{equation*}
    is an isometry with the Euclidean plane. 
    \item $\kappa<0$. The surface $\calS_\kappa$ is a hyperboloid of two sheets and  
    we restrict to the upper sheet, $z\geq0$. In this case, we give
    an isometry of $\calS_\kappa$ with the classic pseudo-sphere model for
    the constant negative curvature space. Recall that this is given by the 
    surface 
    $$\mathbb{L}^{2}_{R}:=\left \lbrace (x,y,z)\in\R^{3}:x^{2}+y^{2}-z^{2}=-R^2, \ z>0 \right \rbrace,$$
    equipped with the restriction of the Minkowski pseudo-metric $\mbox{d}s_{M}^{2}=\mbox{d}x^{2}+\mbox{d}y^{2}-\mbox{d}z^{2}$. As is well-known, the resulting Riemannian surface
    has curvature $-1/R^2$. 
    It can be checked that
    \begin{equation}
    \label{eq:isonegative}
       \Phi^{-}:\calS_{\kappa}\to \mathbb{L}^{2}_{1/\sqrt{-\kappa}} \qquad  \Phi^{-}(x,y,z)=\left(x,y,\sqrt{-\kappa}\ z+\frac{1}{\sqrt{-\kappa}}\right),
    \end{equation}
    is an isometry.
\end{itemize}

\subsubsection{Distance formulas and the generalized gravitational
potential}

Fix $\kappa \in \R$ and consider the  functions $\sin_{\kappa}, \cos_{\kappa}:\R \to \R$ defined by 
\begin{equation*}
\begin{gathered}
    \sin_{\kappa}(q)=\left\lbrace \begin{matrix} \frac{1}{\sqrt{\kappa}}\sin(\sqrt{\kappa} \ q) && \mbox{if} && \kappa>0,  \\ q && \mbox{if} && \kappa=0, \\  \frac{1}{\sqrt{-\kappa}}\sinh(\sqrt{-\kappa} \ q) && \mbox{if} && \kappa<0, \end{matrix} \right. \qquad 
    \cos_{\kappa}(q)=\left\lbrace \begin{matrix} \cos(\sqrt{\kappa} \ q) && \mbox{if} && \kappa>0,  \\ 1 && \mbox{if} && \kappa=0, \\  \cosh(\sqrt{-\kappa} \ q) && \mbox{if} && \kappa<0. \end{matrix} \right. 
\end{gathered}
\end{equation*}
These functions provide an  interpolation between the circular and the hyperbolic trigonometric functions 
which is analytic in $\kappa$. 
They often appear in the study of mechanical problems
in spaces of constant curvature, e.g. \cite{CarRaSan08,MT13,GNM21}. 

With this notation, it is a simple exercise to verify that  the Riemannian distance $d_{\calS_{\kappa}}$ between two points $\mathbf{x}_{1}, \mathbf{x}_{2}\in\calS_{\kappa}$ satisfies, in the case $\kappa\neq 0$,
\begin{equation}
\label{eq:distance}
    \cos_{\kappa}(d_{\cal S_{\kappa}}(\mathbf{x}_{1},\mathbf{x}_{2}))= \vert \kappa \vert \left\langle \Phi^{\pm}(\mathbf{x}_{1}), \Phi^{\pm}(\mathbf{x}_{2}) \right\rangle_{\pm},
\end{equation}
where $\langle \cdot, \cdot \rangle_{\pm} = \langle \cdot, \cdot \rangle_{\pm 1}$. If instead,  $\kappa=0$,
\begin{equation*}
    d_{\calS_{0}}(\mathbf{x}_{1},\mathbf{x}_{2})=\left[\left\langle \Phi^{0}(\mathbf{x}_{1}-\mathbf{x}_{2}), \Phi^{0}(\mathbf{x}_{1}-\mathbf{x}_{2}) \right\rangle_{+}\right]^{1/2}.
\end{equation*}

The function $\cot_\kappa$ defined by 
$$
\cot_\kappa(q)=\frac{\cos_\kappa(q)}{\sin_\kappa(q)},
$$
is fundamental in the generalization of the
$N$-body problem to spaces of non-zero constant
curvature since it is the widely accepted generalization of the classical gravitational potential 
$\frac{1}{q}$, see e.g. \cite{KozHar92,DiaPCSan12}. Note that, apart from the `collision
singularity' at $q=0$, for $\kappa>0$, $\cot_\kappa$
has an additional singularity at $q=\frac{\pi}{\sqrt{\kappa}}$ corresponding to antipodal configurations.

\subsection{The motion of the primaries: relative equilibria of the 2-body problem on \texorpdfstring{$\calS_{\kappa}$}{Sₖ}}
\label{ss:PrimaryMotion}

We now recall the classification of the relative equilibria (RE) of the curved two-body problem (C2BP).
This is fundamental for our purposes since the 
circular motion of the primaries in our setting of
the curved R3BP are special types of RE of the C2BP. 
Our presentation follows mainly \cite{GNBMM, GNM21}. Previous results on RE for the C2BP were given in \cite{DiaPCSan12, GNMarPCRO16, BorMaKi16}

The C2BP concerns the motion of two particles with positive masses $\mu_{1}$ and $\mu_{2}$  on a surface with constant curvature $\kappa\in\R$ under their mutual gravitational attaraction. In our treatment of the curved R3BP, these point masses represent the
primary bodies. We assume that $\mu_1>\mu_2$ and denote the mass ratio
$\mu=\mu_2/\mu_1\in (0,1)$.
We will use the surface $\calS_{\kappa}$ described in Section \ref{subsec:surface} as our model of a surface with constant curvature $\kappa$. The   attractive gravitational
potential is
\begin{equation}
\label{eq:potential}
    V_{\kappa}(\mathbf{x}_{1},\mathbf{x}_{2})=-G\mu_{1}\mu_{2}\cot_{\kappa}(q).
\end{equation}
where $\mathbf{x}_{j}\in \calS_\kappa$ is the position 
of the primary of mass $\mu_j$, $q$ is the distance between the masses (i.e. $q=d_{\calS_{\kappa}}(\mathbf{x}_{1},\mathbf{x}_{2})$) and $G$ is a gravitational constant.

The configuration space $Q_{\kappa}$ of the C2BP is
\begin{equation*}
    Q_{\kappa}=\calS_{\kappa}\times\calS_{\kappa} \setminus \Delta_\kappa,
\end{equation*}
where $\Delta_\kappa\subseteq \calS_{\kappa}\times\calS_{\kappa}$ consists of the collision points on the diagonal (for all $\kappa$) together with pairs of antipodal points if $\kappa>0$. The Lagrangian $L_{C2BP}:TQ_{\kappa}\to\R$ is given by 
\begin{equation}
\label{eq:lagrangianC2BP}
    L_{C2BP} = \frac{1}{2}\mu_{1}\Vert \dot{\mathbf{x}}_{1} \Vert_{\kappa}^{2}+\frac{1}{2}\mu_{2}\Vert \dot{\mathbf{x}}_{2} \Vert_{\kappa}^{2}-V_{\kappa}(\mathbf{x}_{1},\mathbf{x}_{2}).
\end{equation}
This Lagrangian is invariant under the action of
the group $\operatorname{G}_\kappa$\footnote{$\operatorname{G}_\kappa$ is isomorphic to $\operatorname{SO}(3)$ if $\kappa>0$, $\operatorname{SO}(2,1)$ if $\kappa<0$ and $\operatorname{SE}(2)$ if $\kappa=0$.} of isometries of $\mathcal{S}_\kappa$ that 
simultaneously translates the masses. The RE are
solutions of the equations of motion which at the same time  are orbits of one-parameter subgroups of
the $\operatorname{G}_\kappa$ action. In particular, along these  
solutions, the distance between the particles remains constant.  The classification  and stability results  in the case $\mu_{1}\neq\mu_{2}$ are summarized below.
\begin{itemize}
    \item If $\kappa<0$, then for any $q>0$ there are exactly two types of RE, depending on whether the one-parameter subgroup of the isometries is compact (\textbf{\textit{elliptic RE}}) or not (\textbf{\textit{hyperbolic RE}}). In the elliptic RE case, the primaries execute a periodic motion,  rotating uniformly about a fixed point on the geodesic connecting the primaries (center of rotation).
    \item If $\kappa>0$, then for any $0<q<\pi/\sqrt{\kappa}$, $q\neq \pi/(2\sqrt{\kappa})$, there is exactly one relative equilibria. If $0<q<\pi/(2\sqrt{\kappa})$, it is called \textbf{\textit{acute RE}}. Otherwise, it is called \textbf{\textit{obtuse RE}}. In both cases, the primaries
    undergo a periodic motion, rotating about a center of rotation that again lies on the geodesic between them.
    \item If $\kappa=0$, then for any $q>0$ there is exactly one RE, which corresponds with the usual Keplerian RE: the primaries rotate uniformly about their center of mass. 
\end{itemize}

\begin{remark}
    For $\kappa>0$, the classification of RE is different when $\mu_{1}=\mu_{2}$. See Theorem 4.1 of \cite{GNBMM} for details. 
\end{remark}

According to Theorem 4.1 of \cite{GNM21}, given a fixed value of the distance between the primaries $q>0$ and a mass ratio $\mu=\mu_{2}/\mu_{1}\in (0,1)$, there is an analytic transition in $\kappa$ from the elliptic RE for $\kappa<0$ to the acute RE for $\kappa>0$, for $\kappa\in (-\infty,(\pi/2q)^2)$. These RE are interpolated at $\kappa=0$ by the Keplerian RE with distance $q$ between the primaries.
\textit{These RE constitute the circular motions of the primaries in our setting
of the curved R3BP.} As mentioned in the introduction
we 
normalize the Riemannian distance between the primaries to 
$$q=1.$$
In particular, in order to remain within the range of acute RE for positive $\kappa$,  the value of 
the curvature must satisfy 
$$\kappa<\frac{\pi^2}{4}.$$

We now give explicit
expressions for the RE discussed above which will be needed in the treatment of the curved R3BP in Section \ref{sec:R3BP}. We may assume that the center of rotation, $C$, of these
circular motions of the primaries is located at the minimum of $\calS_\kappa$.   Picking a convenient initial configuration, these circular motions are explicitly given, for $\kappa\neq 0$, by 
 \begin{equation}
 \label{eq:REC2BP}
    \mathbf{x}_{1}(t)=R_{\omega t}\begin{pmatrix} -\sin_{\kappa}(q_{1}) \\ 0 \\ \frac{1}{\kappa}\left(1-\cos_{\kappa}(q_{1})\right) \end{pmatrix}, \qquad  \mathbf{x}_{2}(t)=R_{\omega t}\begin{pmatrix} \sin_{\kappa}(q_{2})\\ 0 \\ \frac{1}{\kappa}\left(1-\cos_{\kappa}(q_{2})\right) \end{pmatrix},  
\end{equation}
where $R_{\omega t}$ denotes the rotation matrix
$$
R_{\omega t}=\begin{pmatrix} \cos \omega t & -\sin  \omega t &  0 \\\sin  \omega t & \cos  \omega t & 0\\ 0 & 0 & 1 \end{pmatrix}.
$$
As can be verified by direct calculation using the  equations of motion associated with the Lagrangian \eqref{eq:lagrangianC2BP}, the angular velocity $\omega$ satisfies
\begin{equation}
\label{eq:angularvelocity}
     \omega^{2}=\frac{G\mu_{1}\sqrt{\mu^{2}+2\mu\cos_{\kappa}(2)+1}}{\sin_{\kappa}^{3}(1)\cos_{\kappa}(1)}.
\end{equation}
In Eq.~\eqref{eq:REC2BP},  $q_j$, $j=1,2$, denotes the Riemannian
distance between the  primary of mass $\mu_j$  
and the center of rotation, and is determined implicitly by
\begin{equation}
\label{eq:rotationRE}
	\sin_{\kappa}(2q_{1}) = \mu \sin_{\kappa}(2q_{2}), \qquad q_{2} = 1 - q_{1}.
\end{equation}
Moreover, we can manipulate relation \eqref{eq:rotationRE} to obtain an expression depending only  on $q_{1}$, namely
\begin{equation}
\label{eq:positionprimariesaux}
\begin{gathered}
    \cot_{\kappa}(2q_{1})=\frac{1+\mu\cos_{\kappa}(2)}{\mu\sin_{\kappa}(2)}.
\end{gathered}
\end{equation}

For $\kappa =0$ we may rewrite formulas \eqref{eq:REC2BP} as
\begin{equation}
\label{eq:RE2BP}
    \mathbf{x}_{1}(t)=R_{\omega t}\begin{pmatrix} -q_{1}\\  0 \\ \frac{1}{2}q_{1}^2 \end{pmatrix}, \quad \mathbf{x}_{2}(t)=R_{\omega t}\begin{pmatrix}  q_{2}\\ 0 \\  \frac{1}{2}q_{2}^2\end{pmatrix}
\end{equation}
with 
\begin{equation*}
    \omega^2=G(\mu_1+\mu_2), \qquad q_{1}=\frac{\mu_2}{\mu_1+\mu_2}, \qquad q_{2}=\frac{\mu_1}{\mu_1+\mu_2}.
\end{equation*}
We recognize from the first two components in formulas \eqref{eq:RE2BP} the standard circular solutions of the 2-body problem on the plane. 

One may directly check that the family given in the formulas \eqref{eq:REC2BP} and \eqref{eq:RE2BP} depends analytically on $\kappa$.

\begin{remark}
    The results from \cite{GNMarPCRO16} (see also \cite{GNBMM,GNM21}) imply that the periodic
    solution \eqref{eq:REC2BP} is (nonlinearly) 
    stable on the
    range $\kappa^*(\mu)<\kappa<0$ and unstable for $\kappa<\kappa^*(\mu)$ for a certain $\kappa^*$ that depends on $\mu$. On the other hand, the results 
    from  \cite{GNBMM,GNM21} imply that it is linearly
    stable for $0<\kappa<(\pi/2)^2$, with numerical
    evidence for nonlinear stability given
    in \cite{GNBMM} for almost all $(\kappa,\mu)\in \Ps{1}$.
\end{remark}

\begin{remark}
    Some of the numerical results from \cite{Ki99,MS17} 
    concern the R3BP on the sphere when the
    primaries rotate on an obtuse RE. We do not
    consider them in this work since this branch of RE
    of the C2BP does not connect with the Keplerian
    RE for $\kappa=0$. We refer the reader to \cite{GNM21} for geometric explanations on the
    existence of this branch of RE which involve the 
    motion of two masses on the sphere under the effect of a repelling potential. 
\end{remark}

\subsection{The Lagrange  points in \texorpdfstring{$\calS_0$}{S0}}
\label{subsec:libpointsS0}

 We now give a few details
of the classical (planar) R3BP when defined
in the $\calS_{0}$ model. Our purpose is to give expressions
for the Lagrange points $\Ll_1,\dots, \Ll_5$ for
later comparison with the case $\kappa\neq 0$.
For the standard treatment of the problem on 
$\R^2$ we refer the reader to  Chapter 4 of \cite{MO17}.

 The position of the primaries on the rotating coordinate system is given by $\mathbf{y}_{1}, \mathbf{y}_{2}\in \calS_0$, with
\begin{equation}
\label{eq:primariespositions}
    \mathbf{y}_{1}= \left ( -\frac{\mu}{1+\mu} , 0 ,  \frac{\mu^{2}}{2(1+\mu)^{2}}\right ), \qquad \mathbf{y}_{2}=\left (\frac{1}{1+\mu} , 0 ,  \frac{1}{2(1+\mu)^{2}}\right ).
\end{equation}
After an appropriate rescaling in space and time (see Section \ref{ss:time-independent} below), the Lagrangian for the motion of the satellite in this rotating system is
\begin{equation}
\label{eq:lagS0}
    \tilde L_{0,\mu}:=T(\calS_{0}\setminus\lbrace \mathbf{y}_{1},\mathbf{y}_{2} \rbrace)\to\R, \qquad \tilde L_{0,\mu}(\mathbf{y},\mathbf{y}')=\frac{1}{2}\Vert \mathbf{y}' \Vert_{0}^{2}+\langle \mathbf{e}_{3}\times \mathbf{y},\mathbf{y}'\rangle_{0}-V_{0,\mu}(\mathbf{y}),
\end{equation}
where $\mathbf{e}_{3}=(0,0,1)$, the symbol $\times$ denotes the classical vector product on $\R^{3}$, and the so-called \textit{amended potential}, $V_{0,\mu}$, depends parametrically on $\mu$ and is given by 
\begin{equation}
\label{eq:amended0}
    V_{0,\mu}(\mathbf{y})=-\frac{1}{2}(x^{2}+y^{2})-\frac{1}{1+\mu}\left (\frac{1}{\sqrt{(x+\frac{\mu}{1+\mu})^{2}+y^{2}}}+\frac{\mu}{\sqrt{(x-\frac{1}{1+\mu})^{2}+y^{2}}}\right ),
\end{equation}
where $\mathbf{y}=(x,y,z)\in \calS_{0}\setminus\lbrace \mathbf{y}_{1},\mathbf{y}_{2} \rbrace$.

The Lagrange  points  correspond to critical points of the amended potential $V_{0,\mu}$. The 
collinear points $\Ll_1, \Ll_2, \Ll_3$, are saddle points of $V_{0,\mu}$ and their coordinates are determined in terms of  the roots of certain fifth-degree polynomials
(see~\cite[Chapter 4.4]{Sz67}). In the limit $\mu\to 0^+$, they are approximated by
\begin{equation}
\label{eq:planarcollinear}
\begin{aligned}
    \Ll_{1}&=\left ( 1-\left( \frac{\mu}{3} \right)^{1/3} \,  , \,  0 \, , \, \frac{1}{2}-\mu^{1/3}   \right )+\mathcal{O}(\mu^{2/3}), \qquad \Ll_{2}=\left ( 1+\left( \frac{\mu}{3} \right)^{1/3}  \,  , \,  0 \, , \, \frac{1}{2}+\mu^{1/3}  \right ) +\mathcal{O}(\mu^{2/3}), \\ &\hspace{3.5cm}
    \Ll_{3}=\left ( -1-\frac{5}{12}\mu  \,  , \,  0 \, , \, \frac{1}{2}+\frac{5}{12}\mu  \right )+\mathcal{O}(\mu^{2}).
\end{aligned}
\end{equation}
 Meanwhile, the triangular  points,  $\Ll_4, \Ll_5$,  are local
maxima of $V_{0,\mu}$, and are given by
\begin{equation}
\label{eq:planartriangular}
    \Ll_{4}=\left ( \frac{1-\mu}{2(1+\mu)}  \,  , \,\frac{\sqrt{3}}{2}  \,  , \,  \frac{1+\mu+\mu^{2}}{2(1+\mu)^{2}} \right ), \qquad \Ll_{5}=\left ( \frac{1-\mu}{2(1+\mu)} \,  , \, -\frac{\sqrt{3}}{2}  \,  , \,  \frac{1+\mu+\mu^{2}}{2(1+\mu)^{2}}\right ).
\end{equation}

\section{The restricted 3-body problem in spaces of constant curvature 
}
\label{sec:R3BP}

The restricted $3$-body problem (R3BP) in a space of constant curvature $\kappa$ concerns the motion of a particle with negligible  mass (called  \textbf{\textit{satellite}}) under the gravitational attraction  of two bodies (called  \textbf{\textit{primaries}}) that rotate uniformly along a RE of the curved  2-body problem as explained in Subsection \ref{ss:PrimaryMotion} above. 
In this section, we derive
a two-degree of freedom, time-independent,
mechanical Lagrangian for the problem
depending parametrically on the curvature $\kappa$ and the mass ratio 
$\mu$ (Eqs. \eqref{eq:Lagmukappa} and \eqref{eq:scaledpotential}).
This will allow us to systematically
study the RE of the curved R3BP  as functions of $\kappa$ and $\mu$
in the setting of
mechanical systems of this type (which are reviewed in Appendix \ref{ap:linear}).

\subsection{Setting of the problem in 
\texorpdfstring{$\calS_{\kappa}$}{S	extsubscript{κ}}.}

As before, our model for a surface with constant curvature $\kappa$ will be $\calS_{\kappa}$.  We will concentrate in the case in which the primaries have different masses,  assuming $0<\mu_2<\mu_1$, and we normalize the distance
between the primaries to $q=1$. We also assume that the center of rotation, $C$, of the
circular motions of the primaries is located at the minimum of $\calS_\kappa$.   Under these assumptions, the position  of the
primaries at time $t$ is  $\mathbf{x}_{j}(t)\in \calS_\kappa$, $j=1,2$, given by Eq.~\eqref{eq:REC2BP} for $\kappa\neq0$ and by Eq.~\eqref{eq:RE2BP} for $\kappa=0$. 

The configuration space $\mathcal{Q}_\kappa$ of the problem is time-dependent and consists of the points $\mathbf{x}\in\calS_{\kappa}$ excluding the points where the satellite collides with the primaries and their antipodal points when $\kappa>0$. More precisely,
\begin{equation*}
    \mathcal{Q}_\kappa=\mathcal{S}_\kappa\setminus \Delta_\kappa(t),
\end{equation*}
where 
\begin{equation*}
    \Delta_\kappa(t)=\begin{cases} \bigcup_{j=1,2} \left\lbrace  \mathbf{x}\in\calS_{\kappa} : d_{\calS_{\kappa}}(\mathbf{x},\mathbf{x}_{j}(t))=0\right\rbrace \quad \mbox{if} \quad \kappa \leq 0, \\ \bigcup_{j=1,2} \left\lbrace  \mathbf{x}\in\calS_{\kappa} : d_{\calS_{\kappa}}(\mathbf{x},\mathbf{x}_{j}(t))=0  \quad \mbox{or}
    \quad d_{\calS_{\kappa}}(\mathbf{x},\mathbf{x}_{j}(t))=\frac{\pi}{\sqrt{\kappa}} \right \rbrace  \quad \mbox{if} \quad \kappa > 0. \end{cases}
\end{equation*}

Considering that the interaction of the satellite
with the primaries is given by a  potential analogous to $V_{\kappa}$ given in Eq.~\eqref{eq:potential},
the (time-dependent)   Lagrangian of the system  $L:T\mathcal{Q}_\kappa\times\R\to \R$  is given by
\begin{equation}
\label{eq:LagTimeDep}
    L(\mathbf{x},\dot{\mathbf{x}},t)=\frac{1}{2}\Vert \dot{\mathbf{x}} \Vert_k^{2} + G\mu_{1}\cot_{\kappa}\big(d_{\mathcal{\calS_{\kappa}}}(\mathbf{x},\mathbf{x}_{1}(t))\big)+ G\mu_{2}\cot_{\kappa}\big(d_{\mathcal{\calS_{\kappa}}}(\mathbf{x},\mathbf{x}_{2}(t))\big).
\end{equation}

We seek solutions of the Lagrangian system associated with the Lagrangian \eqref{eq:LagTimeDep} in which the satellite rotates about the center of rotation $C$
 with the same angular velocity as the primaries. We refer to such solutions as \textbf{\textit{Relative Equilibria (RE) for the R3BP in spaces of constant curvature}}.\footnote{This terminology is motivated by the fact that these solutions correspond to relative equilibria of the curved three-body problem in the limit case where one of the bodies has negligible mass. This nomenclature is classical in the planar case.} Our goal is to classify these solutions and analyze their stability as functions of the curvature $\kappa\in (-\infty , \pi^2/4)$ and the mass ratio $\mu = \mu_2 / \mu_1\in (0,1)$.

\subsection{A rotating coordinate frame and elimination of parameters.}
\label{ss:time-independent}

Similar to the planar R3BP, 
we  introduce a rotating coordinate
system on $\mathcal{Q}_{\kappa}$ which renders
the Lagrangian \eqref{eq:LagTimeDep}  autonomous. Specifically, we define
$\mathbf{y}\in\calS_{\kappa}$ via  the  time-dependent change of variables:
\begin{equation}
\label{eq:changevar}
    \mathbf{x}=R_{\omega t}\mathbf{y}.
\end{equation}
Additionally,  we introduce the dimensionless time variable $\tau:=\omega t$. 

In terms of $\mathbf{y}$ and $\mathbf{y}':=\frac{d}{d\tau}\mathbf{y}$, the Lagrangian \eqref{eq:LagTimeDep} is autonomous and only depends on the parameters $\mu$ and $\kappa$. Up to a negligible overall constant factor, we obtain the new Lagrangian
\begin{equation}
\label{eq:LagAut}
  \tilde{L}_{\kappa,\mu}=T\tilde{\mathcal{Q}}_{\kappa}\to\R, \qquad  \tilde{L}_{\kappa,\mu}(\mathbf{y},\mathbf{y}')=\frac{1}{2}\Vert \mathbf{y}' \Vert_{\kappa}^{2} + \langle \mathbf{e}_{3}\times \mathbf{y},\mathbf{y}' \rangle_{\kappa}-V_{\kappa,\mu}(\mathbf{y}),
\end{equation}
where  the potential $V_{\kappa,\mu}$ is given by
\begin{equation}
\label{eq:augpotential}
    V_{\kappa,\mu}(\mathbf{y})=-\frac{1}{2}\Vert \mathbf{e}_{3}\times\mathbf{y} \Vert_{\kappa}^{2} - \frac{G \mu_{1}}{\omega^{2}}\left[\cot_{\kappa}\big(d_{\calS_{\kappa}}(\mathbf{y},\mathbf{y}_{1})\big)+\mu \cot_{\kappa}\big(d_{\calS_{\kappa}}(\mathbf{y},\mathbf{y}_{2})\big)\right].
\end{equation}  
Note that the coefficient $G\mu_{1}/\omega^{2}$ appearing above can be expressed in terms of  $\kappa$ and $\mu$ (see formulas \eqref{eq:angularvelocity} and \eqref{eq:constantgamma} below). 

The position of the primaries in the rotating frame,
given by $\mathbf{y}_{j}:=R_{-\omega t}\mathbf{x}_{j}(t)\in \calS_\kappa$, $j=1,2$,
is time-independent and, in view of formulas \eqref{eq:REC2BP}, depends on $\kappa,\mu$ as follows:
 \begin{equation}
 \label{eq:primariesSkappa}
 \mathbf{y}_{1}= \mathbf{y}_{1}(\kappa, \mu)=\begin{pmatrix} -\sin_{\kappa}(q_{1}(\kappa, \mu)) \\ 0 \\  \frac{1}{\kappa}\left(1-\cos_{\kappa}(q_{1}(\kappa, \mu))\right) \end{pmatrix}, \qquad  \mathbf{y}_{2}=\mathbf{y}_{2}( \kappa,\mu)=\begin{pmatrix} \sin_{\kappa}(q_{2}(\kappa,\mu)) \\ 0 \\  \frac{1}{\kappa}\left(1-\cos_{\kappa}(q_{2}(\kappa,\mu)\right) \end{pmatrix}.
\end{equation}
Here, $q_j(\kappa,\mu)$, $j=1,2$, is the Riemannian distance between
$\mathbf{y}_{j}(\kappa,\mu)$ and the center of rotation, and, in accordance with formula \eqref{eq:positionprimariesaux}, are implicitly determined
 by
\begin{equation}
\label{eq:positionprimaries}
\begin{gathered}
    \cot_{\kappa}(2q_{1}(\kappa, \mu)) = \frac{1+\mu\cos_{\kappa}(2)}{\mu\sin_{\kappa}(2)} ,\qquad q_{2}(\kappa, \mu)=1-q_{1}(\kappa, \mu).
\end{gathered}
\end{equation}

The new configuration space $\tilde{\mathcal{Q}}_{\kappa}$ is time-independent and consists of all points  $\mathbf{y}\in \calS_\kappa$ for which the potential $V_{\kappa,\mu}$ is defined. The singularities of $V_{\kappa,\mu}$ occur at the primaries, $\mathbf{y}_{j}(\kappa,\mu)$, and, in the case of positive curvature, also at their antipodal points.

The RE of the curved R3BP correspond to equilibrium solutions of the Lagrangian system defined by the Lagrangian \eqref{eq:LagAut}. Due to the form of the change of variables \eqref{eq:changevar} and the mechanical structure of the Lagrangian \eqref{eq:LagAut}, the following holds (see Appendix \ref{ap:linear}).
\begin{proposition}
    Relative Equilibria for the curved R3BP are in one-to-one correspondence with critical points of the potential $V_{\kappa,\mu}:\tilde{Q}_{\kappa}\to\R$ given by Eq.~\eqref{eq:augpotential}.
\end{proposition}

\begin{remark}
    The time independence of the 
    Lagrangian \eqref{eq:LagAut} implies that the corresponding equations of
    motion possess the Jacobi integral 
    \begin{equation*}
\label{eq:JacobiIntegral}
  E_{\kappa,\mu}=T\tilde{\mathcal{Q}}_{\kappa}\to\R, \qquad  E_{\kappa,\mu}(\mathbf{y},\mathbf{y}')=\frac{1}{2}\Vert \mathbf{y}' \Vert_{\kappa}^{2} + V_{\kappa,\mu}(\mathbf{y}).
\end{equation*}
\end{remark}

\subsection{A convenient rescaling}
\label{subsec:rescaling}

The critical points of $V_{0,\mu}$ correspond to the classical Lagrange points $\Ll_{1},\dots, \Ll_{5}$,
 reviewed in sections \ref{subsec:libpoints} and \ref{subsec:libpointsS0}. To study their continuation to positive and negative $\kappa$, we must analyze the critical points of $V_{\kappa,\mu}$. A difficulty arises from the fact that the domain of $V_{\kappa,\mu}$, namely $\tilde {\mathcal{Q}}_\kappa$, 
  also depends on  $\kappa$. To address this issue, we consider separately the cases $\kappa>0$ and $\kappa<0$, and introduce a suitable scaling that reformulates the problem on spaces of constant curvature $\pm 1$. Under this rescaling, the dependence on $\kappa$ is transferred to the Lagrangian 
  as a parameter. 

Let $\mathcal{M}^\pm$ denote  the standard sphere or pseudo-sphere of constant curvature $\pm 1$. That is, 
\begin{equation*}
 \mathcal{M}^+= \mathbb{S}_{1}^{2}, \qquad  \mathcal{M}^-= \mathbb{L}_{1}^{2}.
\end{equation*}
According to the sign of  $\kappa\neq 0$, we  consider the diffeomorphism
\begin{equation}
\label{eq:diffeo}
\Psi_{\kappa}: \mathcal{M}^\pm \to \calS_\kappa, \qquad 
    \Psi_{\kappa}(r_{1},r_{2},r_{3})=\left( \frac{r_{1}}{\sqrt{\vert\kappa\vert}}, \frac{r_{2}}{\sqrt{\vert\kappa\vert}} , \frac{r_{3}\pm1}{\vert\kappa\vert}\right).
\end{equation}
The inverse diffeomorphism $\Psi_{\kappa}^{-1}:\calS_\kappa\to \mathcal{M}^\pm$
is obtained by composing the scaling by $\sqrt{|\kappa|}$, which maps $\mathbb{S}_{1}^{2}\to \mathbb{S}_{\sqrt{\kappa}}^{2}$
and $\mathbb{L}_{1}^{2}\to \mathbb{L}_{\sqrt{-\kappa}}^{2}$, according to the sign of $\kappa$, with the isometries $\Phi^{\pm}$ given by Eqs. \eqref{eq:isopositive} and \eqref{eq:isonegative}.
 It is therefore clear that  $\Psi_{\kappa}$ is not an isometry
 unless $\kappa=\pm 1$.  The following proposition is a direct consequence
 of the definitions, and will be used ahead. 
 \begin{proposition}
 \label{prop:scaling}
 Let $\mathbf{r_1}, \mathbf{r_2}\in \mathcal{M}^\pm$ and suppose that their mutual Riemannian distance is $d$. Then,
 the Riemannian distance between  $\Psi_{\kappa}(\mathbf{r_1}), \Psi_{\kappa}(\mathbf{r_2})\in \calS_\kappa$ is $\frac{d}{\sqrt{|\kappa|}}$.
 \end{proposition}

 Let $\mathbf{r}\in \mathcal{M}^\pm$ be determined by $\mathbf{y}=\Psi_\kappa(\mathbf{r})$. Then the Lagrangian system \eqref{eq:LagAut} can be reformulated, up to a negligible overall constant factor of $|\kappa|$, as follows\footnote{We 
abbreviate  $\|\cdot \|_{\pm}=\|\cdot \|_{\pm 1}$ and $\langle \cdot, \cdot \rangle_{\pm} =\langle \cdot, \cdot \rangle_{\pm 1} $ to simplify the notation.} 
\begin{equation}
\label{eq:Lagmukappa}
  \mathcal{L}_{\kappa,\mu}:T(\mathcal{M}^\pm\setminus\Delta)\to \R, \qquad \mathcal{L}_{\kappa,\mu}(\mathbf{r},\mathbf{r}') = \frac{1}{2}\Vert \mathbf{r}' \Vert_{\pm}^{2}+\langle \mathbf{e}_{3}\times \mathbf{r},\mathbf{r}'\rangle_{\pm}-\mathcal{V}_{\kappa,\mu}(\mathbf{r}),
\end{equation}
where the potential
\begin{equation}
    \label{eq:scaledpotential}
  \mathcal{V}_{\kappa,\mu}: \mathcal{M}^\pm \setminus \Delta \to \R, \qquad        \mathcal{V}_{\kappa,\mu}(\mathbf{r})=-\frac{1}{2}\Vert \mathbf{e}_{3}\times \mathbf{r} \Vert_{\pm}^{2}-\Gamma_{\kappa,\mu}[\phi(\langle \mathbf{r},\mathbf{p}_{1} \rangle_{\pm})+\mu\phi(\langle \mathbf{r},\mathbf{p}_{2}\rangle_{\pm})].
    \end{equation}
Our notation in the potential \eqref{eq:scaledpotential} is as follows:
\begin{itemize}
\item $\mathbf{p}_{1}=\mathbf{p}_{1}(\kappa,\mu)$ and $\mathbf{p}_{2}=\mathbf{p}_{2}(\kappa,\mu)$ denote  the location of the primaries on $\mathcal{M}^\pm$. So, by their definition, they satisfy $\Psi_\kappa(\mathbf{p}_{j}(\kappa,\mu))=\mathbf{y}_{j}(\kappa,\mu)$, $j=1,2$. Their explicit positions
as functions of  $\kappa$ and $\mu$ are:
\begin{equation}
\label{eq:primariesMpm}
    \mathbf{p}_{1}=\mathbf{p}_{1}(\kappa,\mu)=\begin{pmatrix} -\sqrt{\vert\kappa\vert}\sin_{\kappa}(q_{1}(\kappa,\mu)) \\ 0 \\  \mp\cos_{\kappa}(q_{1}(\kappa,\mu)) \end{pmatrix}, \qquad 
    \mathbf{p}_{2}=\mathbf{p}_{2}(\kappa,\mu)=\begin{pmatrix} \sqrt{\vert\kappa\vert}\sin_{\kappa}(q_{2}(\kappa,\mu)) \\ 0 \\  \mp\cos_{\kappa}(q_{2}(\kappa,\mu)) \end{pmatrix},
    \end{equation}
    with $q_j(\kappa,\mu)$ determined, as before,
    by formulas \eqref{eq:positionprimaries} for $j=1,2$.
\item $\Delta$ is the set of points on $\mathcal{M}^\pm$ at which $\mathcal{V}_{\kappa,\mu}$
is undefined. It consists of the primaries $\lbrace \mathbf{p}_{1},\mathbf{p}_{2} \rbrace$ for $\kappa<0$ and the primaries and their antipodal
points $\lbrace \pm\mathbf{p}_{1},\pm\mathbf{p}_{2} \rbrace$ for $\kappa>0$.

\item $\phi$ denotes the real function $\phi(\lambda)=\frac{\lambda}{\sqrt{\vert 1-\lambda^2\vert}}$.
    
    \item $\Gamma_{\kappa,\mu}$ denotes the constant
\begin{equation}
\label{eq:constantgamma}
    \Gamma_{\kappa,\mu}=\vert\kappa\vert^{3/2}\frac{G\mu_{1}}{\omega^{2}}=\frac{\vert\kappa\vert^{3/2}\sin_{\kappa}^{3}(1)\cos_{\kappa}(1)}{\sqrt{\mu^{2}+2\mu\cos_{\kappa}(2)+1}}.
\end{equation}
\end{itemize}

Note that the expression \eqref{eq:scaledpotential} for  $\mathcal{V}_{\kappa,\mu}$ no longer
involves the generalized cotangent
function $\cot_\kappa$ appearing in the formula
\eqref{eq:augpotential}. Such 
simplification is possible 
in view of the relation 
\begin{equation}
\label{eq:distancerelation}
    \cos_{\kappa}(d_{\calS_{\kappa}}(\mathbf{y},\mathbf{y}_{i}))=\kappa\langle \mathbf{r},\mathbf{p}_{i}\rangle_{\pm},
\end{equation}
that follows from the relation \eqref{eq:distance}.

This is the setting that we will use to study the existence and stability
of the RE of the R3BP for nonzero curvature. Our formulation is convenient since the RE correspond to critical points of the potential $\mathcal{V}_{\kappa,\mu}$ and the curvature $\kappa$ appears as a parameter. Once a branch of equilibria $\textbf{R}_{\kappa,\mu}\in \mathcal{M}^\pm$ has been determined for non-zero $\kappa$, we can study its behavior as $\kappa$ passes through $0$ by considering the corresponding branch $\mathbf{Y}_{\kappa,\mu}\coloneqq\Psi_\kappa(\mathbf{R}_{\kappa,\mu})\in \calS_\kappa$. In this way, we can compare  $\lim_{\kappa \to 0}\mathbf{Y}_{\kappa,\mu}$ with the classical 
Lagrange libration points $\Ll_{1}, \dots, \Ll_{5}\in \calS_0$ given in Eqs. \eqref{eq:planarcollinear} and \eqref{eq:planartriangular} in Sec.~\ref{subsec:libpointsS0}. A good part of our work below consists of the classification of 
the critical points of the potential $\mathcal{V}_{\kappa,\mu}$ as a function of the parameters $\kappa, \mu$.

\section{Classification and stability of RE: general considerations}
\label{sec:collinearclass}

In view of Sec.~\ref{sec:R3BP}, the study of the curved R3BP is reduced to the analysis of a mechanical system on the configuration manifold $\mathcal{M}\setminus\Delta^{\pm}$ defined by the autonomous Lagrangian $\mathcal{L}_{\kappa,\mu}$ given in Eq.~\eqref{eq:Lagmukappa}.
In particular, the analysis of the RE is reduced to the classification of the critical points of the smooth function $\mathcal{V}_{\kappa,\mu}$ in Eq.~\eqref{eq:scaledpotential} (see Appendix \ref{ap:linear}). The complication is the dependence on the two parameters $\kappa,\mu$.

Moreover, using  standard results recalled in Appendix~\ref{ap:linear}, the stability of the equilibrium points can be established by determining the type of critical point of $\mathcal{V}_{\kappa,\mu}$ and, in the case of local maxima, by exploiting the structure of the linearization.

\subsection{Topological restrictions for RE for positive curvature. }

A fundamental difference between the positively curved R3BP and the 
planar and negatively curved versions of the problem is that, aside from 
collisions and antipodal configurations, the 
surface $\calS_\kappa$ is compact for $\kappa>0$ and non-compact otherwise.
This leads to a topological constraint on the kinds of critical points
that a smooth function may have. It turns out that this constraint
persists even in the presence of the collision and antipodal configurations. 

\begin{theorem}
\label{prop:eqptsSpos}
Fix $\mu\in (0,1)$ and $0<\kappa <\frac{\pi^2}{4}$, and suppose that all 
critical points of $\mathcal{V}_{\kappa,\mu}$ defined by the expression \eqref{eq:scaledpotential} are non-degenerate. Then,
\begin{equation}
\label{eq:formulaindex}
    \# \mbox{local max} +   \# \mbox{local min} -  \# \mbox{saddle points} = -2.
\end{equation}
\end{theorem}

\begin{proof}
To simplify notation, we will write $\mathcal{V}=\mathcal{V}_{\kappa,\mu}$ omitting the subscripts $\kappa, \mu$, 
 throughout the proof. 
 
Recall from the discussion following Eq.~\eqref{eq:scaledpotential}, that $\mathcal{V}$ is a smooth function on $\mathbb{S}_{1}^{2}\setminus\Delta$ where $\Delta=\{\pm \mathbf{p}_{1}, \pm \mathbf{p}_{2}\}$. It is not difficult to deduce the following behavior of $\mathcal{V}$ from its explicit expression in Eq.~\eqref{eq:scaledpotential}:
\begin{equation}
\label{eq:auxtopological}
    \lim_{\mathbf{r}\to \pm \mathbf{p}_i}\mathcal{V}( \mathbf{r}) = \mp \infty, \qquad i=1,2. 
\end{equation}

Consider the  Riemannian gradient  field $\nabla\mathcal{V}$, which is 
a smooth vector field on $\mathbb{S}_{1}^{2}\setminus\Delta$, and define
\begin{equation}
\label{eq:regvecfield}
    X := \frac{1}{1 + \Vert \nabla \mathcal{V} \Vert_{+}^{2}} \, \nabla \mathcal{V}.
\end{equation}
Then $X$ is a smooth vector field on $\mathbb{S}_{1}^{2} \setminus \Delta$, which in view of property \eqref{eq:auxtopological}, may be smoothly extended
to all of $\mathbb{S}_{1}^{2}$ defining $X(\mathbf{r}) := 0$ for $\mathbf{r} \in \Delta$.

By construction, the set $Z(X)$ of equilibrium points
of $X$, and the set $Z(\nabla \mathcal{V})$ of equilibrium points of $\nabla \mathcal{V}$, satisfy
\begin{equation}
\label{eq:zeros}
    Z({X}) = Z(\nabla \mathcal{V}) \cup \Delta.
\end{equation}

Considering that $X$ is a smooth vector field on $\mathbb{S}^{2}_{1}$, we may apply the 
Poincar\'e--Hopf Index Theorem (see e.g.~\cite[p. 134]{GP74}) to obtain
\begin{equation*}
    \sum_{\mathbf{r} \in Z({X})} \operatorname{ind}_{\mathbf{r}}(X) = \chi(\mathbb{S}^{2}_{1})=2,
\end{equation*}
where $\chi(\mathbb{S}^{2})$ denotes the Euler characteristic of $\mathbb{S}^{2}$ (which equals $2$), and $\operatorname{ind}_{\mathbf{r}}(X)$ denotes the index of the equilibrium point $\mathbf{r}$ of $X$.

Since the vector
fields $X$ and  $\nabla\mathcal{V}$ only differ
by a positive scalar factor on $\mathbb{S}_{1}^{2}\setminus \Delta$, their equilibrium points and their indices coincide on this domain. Hence,
\begin{equation*}
    \operatorname{ind}_{\mathbf{r}}({X})
    = \operatorname{ind}_{\mathbf{r}}(\nabla\mathcal{V}),
    \qquad 
    \mathbf{r} \in Z(\nabla\mathcal{V}).
\end{equation*}
Therefore, in view of formula \eqref{eq:zeros}, we have
\begin{equation}
\label{eq:PHsing}
    \sum_{\mathbf{r} \in Z(\nabla \mathcal{V})} 
        \operatorname{ind}_{\mathbf{r}}(\nabla \mathcal{V})
    +
    \sum_{\mathbf{r} \in \Delta} 
        \operatorname{ind}_{\mathbf{r}}(X)
    = 2.
\end{equation}

Now, it is well-known that the non-degenerate critical points of $\mathcal{V}$ satisfy
\begin{equation*}
\operatorname{ind}_{\mathbf{r}}(\nabla \mathcal{V}) = 
\begin{cases}
+1 & \text{if } \mathbf{r} \text{ is a local minimum or a local maximum},\\
-1 & \text{if } \mathbf{r} \text{ is a saddle point}.
\end{cases}
\end{equation*}
Thus, the first sum in formula \eqref{eq:PHsing} represents the total number of maxima and minima minus the number of saddle points, and it only 
remains to determine the contribution of the term $\sum_{\mathbf{r} \in \Delta} 
        \operatorname{ind}_{\mathbf{r}}(X)$.

In view of property \eqref{eq:auxtopological}, we conclude that
the gradient field
$\nabla \mathcal{V}$ points away from  $\mathbf{p}_{i}$ in a neighborhood of $\mathbf{p}_{i}$. Since $X$ differs from $\nabla \mathcal{V}$ only by a positive scalar factor, it has the same direction, and therefore $\mathbf{p}_{i}$, $i=1,2,$
are sources of $X$. A similar argument, shows that
$-\mathbf{p}_{i}$, $i=1,2,$ are sinks of $X$.
In any case, the index of each of these points is
$+1$ and, therefore, $\sum_{\mathbf{r} \in \Delta} 
        \operatorname{ind}_{\mathbf{r}}(X)=4$.
Substituting into formula \eqref{eq:PHsing} gives
\begin{equation*}
    \# \text{local maxima} + \# \text{local minima} - \# \text{saddle points} = 2 - 4 = -2.
\end{equation*}
\end{proof}

As a consequence of the theorem, we conclude the
existence of `new' RE for small positive 
curvature. 

\begin{corollary}
\label{cor:extraRE}
Fix a mass ratio $\mu\in (0,1)$. For small  $\kappa>0$
the curved R3BP possesses other RE
in addition to the continuation
of the Lagrange points $\Ll_{1},\dots,\Ll_{5}$.
\end{corollary}
\begin{proof}
    Recall that the collinear Lagrange points $\Ll_{1},\Ll_2,\Ll_3$ are saddle points of  $\mathcal{V}_{0,\mu}$,
    whereas  $\Ll_{4},\Ll_5$ are local maxima. By their
    non-degeneracy, these points persist as
    critical points of $\mathcal{V}_{\kappa,\mu}$ of
    the same type for small $\kappa>0$. These critical
    points on their own do not satisfy formula \eqref{eq:formulaindex},
    so there must exist other critical points
    of $\mathcal{V}_{\kappa,\mu}$.
\end{proof}

\begin{remark}
\label{rmk:top-balance}
Our existence results show that the  additional RE that exist for small $\kappa>0$ are the three collinear points $\Ee_2$, $\Ee_3$ and $\Aa_1$ in
 Theorem~\ref{thm:numbercollinearRE}. As shown in Theorem~\ref{thm:stabilitycollinear},  in the regime where the mass ratio $\mu\in (0,1/10)$, 
 the points $\Ee_3$ and 
$\Aa_1$ are saddle-points of $\mathcal{V}_{\kappa,\mu}$ while $\Ee_2$ is a local minimum. The presence of these additional critical points 
of $\mathcal{V}_{\kappa,\mu}$ balances
Eq.~\eqref{eq:formulaindex}.
\end{remark}

\subsection{Conditions for existence of relative equilibria}
\label{ss:existenceCollinear}

We now derive conditions for RE that will be
investigated analytically and using CAPs in Sections \ref{sec:collinear} and \ref{sec:triangular}. Recall from the discussion in Sec. \ref{subsec:rescaling} that  RE are in one-to-one correspondence with  
critical points of the augmented potential 
$\mathcal{V}_{\kappa,\mu}:\mathcal{M}^\pm\setminus \Delta\to \R$, defined in formula \eqref{eq:scaledpotential}. We will introduce convenient sets of coordinates
in $\mathcal{M}^\pm$ to find these critical points.

Let $\mathcal{G}^\pm \subset \mathcal{M}^\pm$ be the geodesic containing the primaries and the rotation center $C$. This curve generalizes the straight line connecting the primaries and the rotation center
from the planar R3BP which is usually taken as the
horizontal axis. In the curved case, $\mathcal{G}^+$ is a great circle in $\mathbb{S}_1^2$ and $\mathcal{G}^-$ a hyperbola passing through the vertex of $C\in \mathbb{L}_1^2$ (see Figure \ref{fig:coords}).
According to our conventions, $\mathcal{G}^\pm$  is contained in the plane $r_2=0$ (see Eq.~\eqref{eq:primariesMpm}). 
As one may expect, $\mathcal{V}_{\kappa,\mu}$ is  invariant under reflections about $\mathcal{G}^\pm$, and indeed, one checks that  $\mathcal{V}_{\kappa,\mu}$, as defined in formula \eqref{eq:scaledpotential}, satisfies 
\begin{equation}
\label{eq:sym}
    \mathcal{V}_{\kappa,\mu}(r_{1},-r_{2},r_{3})=\mathcal{V}_{\kappa,\mu}(r_{1},r_{2},r_{3}),
\end{equation}
independently of the sign of $\kappa$.
We will exploit this symmetry  to analyze the critical points of $\mathcal{V}_{\kappa,\mu}$.

As in the planar R3BP, it is convenient to treat separately the cases corresponding to points on $\mathcal{G}^{\pm}$ and 
points outside $\mathcal{G}^{\pm}$.
 Following the classical terminology,  critical points of $\mathcal{V}_{\kappa,\mu}$ lying on $\mathcal{G}^\pm$ will be called  \textbf{\textit{collinear RE}}. 
All other critical points of  $\mathcal{V}_{\kappa,\mu}$ will be called \textbf{\textit{triangular RE}}.
Note that, due to the symmetry \eqref{eq:sym}, triangular
RE arise in pairs. This is consistent with the well-known
symmetry between $\Ll_{4}$ and $\Ll_{5}$ in the planar R3BP.
Therefore, for the study of triangular RE, it suffices to restrict our
attention to   the open `half-space'
\begin{equation*}
    \mathcal{M}^\pm_{>0}\coloneqq \lbrace (r_{1},r_{2},r_{3})\in\mathcal{M}^\pm : r_{2} > 0 \rbrace.
\end{equation*}
With this convention, $\mathcal{M}^+_{>0}$  is a hemisphere and 
$\mathcal{M}^-_{>0}$ is half of the one-sheeted hyperboloid, see
 Figure \ref{fig:coords}. 

\begin{figure}[ht]
    \centering
    \captionsetup{width=.75\linewidth}
    \subfigure[\large $\mathcal{M}^+=\mathbb{S}_{1}^{2}$]{\begin{overpic}[width=0.3\textwidth,tics=20]{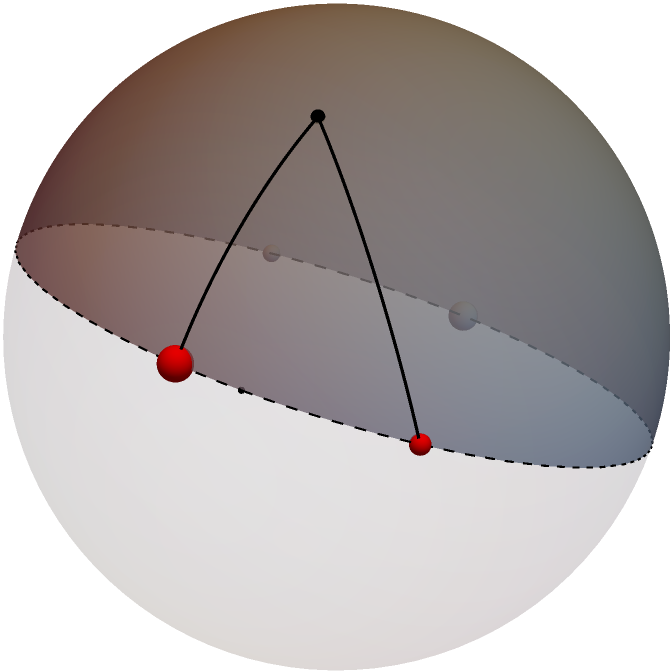}
    \put(36,57){\textcolor{graycover}{$\mathbf{a}_{2}$}}
    \put(70,55){\textcolor{graycover}{$\mathbf{a}_{1}$}}
    \put(19,40){$\mathbf{p}_{1}$}
    \put(59,27){$\mathbf{p}_{2}$}
    \put(58,55){$d_{2}$}
    \put(28,65){$d_{1}$}
    \put(45,85){$\mathbf{r}$}
    \put(37,35){\footnotesize $C$}
    \put(-5,68){$\mathcal{G}^+$}
    \put(83,90){$\mathcal{M}^+_{>0}$}
    \end{overpic}}
    \hspace{2cm}
    \subfigure[\large $\mathcal{M}^-=\mathbb{L}_{1}^{2}$]{\begin{overpic}[width=0.4\textwidth,tics=20]{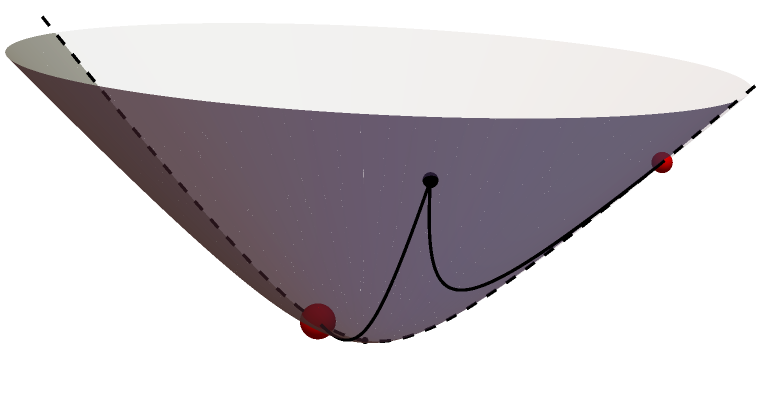}
    \put(89,29){$\mathbf{p}_{2}$}
    \put(38,5){$\mathbf{p}_{1}$}
    \put(45,15){$d_{1}$}
    \put(60,18){$d_{2}$}
    \put(56,32){$\mathbf{r}$}
    \put(46,2){$C$}
    \put(1,50){$\mathcal{G}^-$}
    \put(38,41){$\mathcal{M}^-_{>0}$}
    \end{overpic}} 
    \caption{  The space $\mathcal{M}^\pm$ for positive and negative curvature. The geodesic $\mathcal{G}^\pm$ connecting the primaries 
    $\mathbf{p}_{1}$ and $\mathbf{p}_{2}$ is represented as a dashed curve which contains the center of rotation $C$, and, in the case of positive curvature, the antipodal points $\mathbf{a}_{j}:=-\mathbf{p}_{j}$, $j=1,2$. The curve $\mathcal{G}^\pm$ divides $\mathcal{M}^\pm$ 
     into two regions and for the analysis of triangular
     RE we restrict to $\mathcal{M}^\pm_{>0}$, which is
     shaded in a darker color and contains the 
     satellite, positioned at $\mathbf{r}$. 
     The geodesics connecting $\mathbf{r}$ to the
     primaries, of lengths $d_1,d_2$, are also illustrated. }
    \label{fig:coords}
\end{figure}
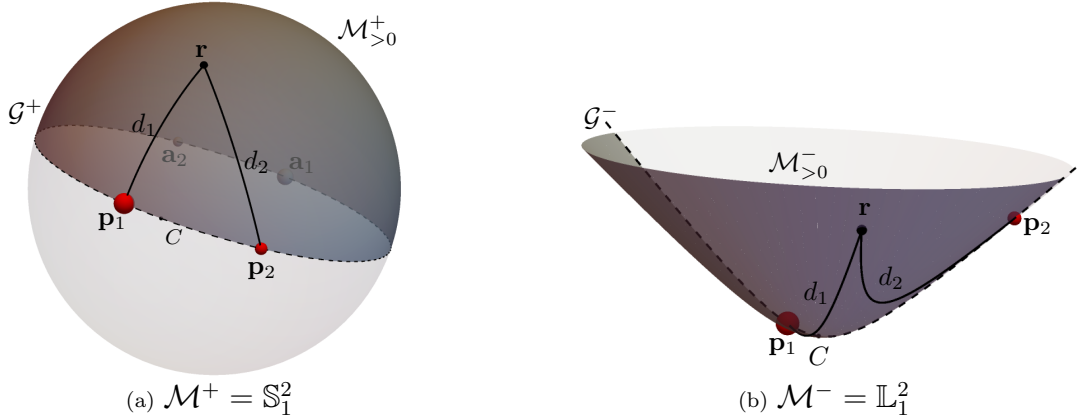

\subsubsection{Conditions for existence of collinear RE}

\paragraph{Case $\kappa>0$.}

We  parametrize $\mathbf{r}\in\mathbb{S}_{1}^{2}$ with spherical coordinates $\phi, \theta$ as
\begin{equation}
\label{eq:positionM+}
     \mathbf{r}=\left ( \sin\theta\cos\phi , \sin\phi , -\cos\theta\cos\phi\right ), 
\end{equation}
with $\theta\in\mathbb{R}/2\pi\mathbb{Z}$ and  $\phi \in (-\frac{\pi}{2},\frac{\pi}{2})$. In these coordinates, the great
circle $\mathcal{G}^+$, which contains the primaries and the center of rotation $C$, is given by $\phi=0$  and is parametrized by
$\theta$. Moreover, the center of rotation $C$ corresponds to $\theta=0$.

The symmetry \eqref{eq:sym} implies that $\mathcal{V}_{\kappa,\mu}$ is
an even function of $\phi$ and hence, 
\begin{equation*}
    \left.\dfrac{\partial\mathcal{V}_{\kappa,\mu}}{\partial\phi}\right\vert_{\phi=0}=0.
\end{equation*}
Therefore, a necessary and sufficient condition for collinear RE with coordinates $\phi=0$, $\theta=\theta_{0}$, is that
\begin{equation*}
    \left.\dfrac{\partial\mathcal{V}_{\kappa,\mu}}{\partial\theta}\right\vert_{\phi=0,\theta=\theta_{0}}=0.
\end{equation*}
The explicit form of this condition is the content of the following proposition. 
 \begin{proposition}
    \label{prop:necessarycondpos}
     The point $\mathbf{r}_{0}=(\sin\theta_{0},0,-\cos\theta_{0})\in\mathcal{G}^+$ is a collinear RE for $\kappa>0$ if and only if
    \begin{equation}
    \label{eq:collinearpositive}
        \frac{1}{2}\sin(2\theta_{0}) -\Gamma\left[\frac{\sin(\theta_{0}+\sqrt{\kappa}q_{1})}{\left\vert\sin(\theta_{0}+\sqrt{\kappa}q_{1})\right\vert^{3}} + \mu\frac{\sin(\theta_{0}-\sqrt{\kappa}q_{2})}{\left\vert\sin(\theta_{0}-\sqrt{\kappa}q_{2})\right\vert^{3}}\right]=0.
    \end{equation}
    \end{proposition}
    Note that, in the statement of the proposition, we have
    omitted the  dependence of $q_{1}$ and $q_{2}$ on the parameters $\kappa$ and $\mu$ given by Eq.~\eqref{eq:positionprimaries}. This convention will be used  throughout the paper. The same  applies
    to the constant 
     $\Gamma$ defined in Eq.~\eqref{eq:constantgamma}.
\begin{proof}
    In view of Eqs. \eqref{eq:scaledpotential} and \eqref{eq:primariesMpm},  the expression for $\mathcal{V}_{\kappa,\mu}$ in the coordinates $(\phi,\theta)$ is
    \begin{equation}
    \label{eq:potentialpolarpos}
    \begin{aligned}
        \mathcal{V}_{\kappa,\mu}(\phi,\theta) = & \frac{1}{2}\left(\cos^{2}\!\theta\cos^{2}\!\phi -1\right) - \\
        & \Gamma \cos\!\phi \left[ \frac{\cos(\theta+\sqrt{\kappa}q_{1}\!)}{\sqrt{1-\cos^{2}\!\phi\cos^{2}(\theta+\sqrt{\kappa}q_{1}\!)}} + \frac{\mu\cos(\theta-\sqrt{\kappa}q_{2}\!)}{\sqrt{1-\cos^{2}\!\phi\cos^{2}(\theta-\sqrt{\kappa}q_{2}\!)}}  \right].
        \end{aligned}
        \end{equation}
     Differentiating with respect to $\theta$ and evaluating at
    $\theta=\theta_0$ gives Eq. \eqref{eq:collinearpositive}.
\end{proof}

\paragraph{Case $\kappa<0$.} The classification of collinear RE for negative curvature has already been established by Mart\'inez and Sim\'o~\cite{MS17} and  
is recalled in Proposition~\ref{prop:ClassCollinearNeg} in Sec.~\ref{sec:collinear}.
Nevertheless, we derive existence for collinear RE when $\kappa<0$, as  they
are needed to obtain the asymptotic expansions in Theorem~\ref{prop:asymLpositive}. Such expansions
 explicitly show that the RE found by 
Mart\'inez and Sim\'o converge to the classical collinear Lagrangian points of the planar R3BP.

We  parametrize $\mathbf{r}\in\mathbb{L}_{1}^{2}$ with  coordinates $\phi, \theta$ as
\begin{equation}
\label{eq:positionM-}
     \mathbf{r}=\left ( \sinh\theta\cosh\phi , \sinh\phi , \cosh\theta\cosh\phi \right), 
\end{equation}
where now $\phi,\theta\in \R$. The geodesic $\mathcal{G}^-$ containing the primaries and the center of rotation $C$, is
again given by $\phi=0$, and is parametrized by
$\theta\in \R$, with the center of rotation corresponding to $\theta=0$. Equations \eqref{eq:scaledpotential} and \eqref{eq:primariesMpm} yield the following expression for $\mathcal{V}_{\kappa,\mu}$ in these coordinates  
\begin{equation}
\label{eq:potentialpolarneg}
	\begin{aligned}
 		\mathcal{V}_{\kappa,\mu}(\phi,\theta) =  & \frac{1}{2}  \left( 1 - \cosh^{2}\!\theta\cosh^{2}\!\phi \right) + \\
		& \Gamma_{\kappa,\mu}\cosh\phi  \left[ \frac{\cosh(\theta+\sqrt{-\kappa}q_{1})}{\sqrt{\cosh^{2}\!\phi\cosh^{2}(\theta+\sqrt{-\kappa}q_{1})-1}} + \frac{\mu\cosh(\theta-\sqrt{-\kappa}q_{2})}{\sqrt{\cosh^{2}\!\phi\cosh^{2}(\theta-\sqrt{-\kappa}q_{2})-1}} \right].
	\end{aligned}
\end{equation}
   Once again, we see that $\mathcal{V}_{\kappa,\mu}$ is even in $\phi$. In analogy with Proposition \ref{prop:necessarycondpos}, evaluating at $\phi=0$ and
differentiating with respect to $\theta$, we obtain
    \begin{proposition}
    \label{prop:necessarycondneg}
     The point $\mathbf{r}_{0}=(\sinh\theta_{0},0,\cosh\theta_{0})\in\mathcal{G}^-$ is a collinear RE for $\kappa<0$ if and only if
    \begin{equation}
    \label{eq:collinearnegative}
        \frac{1}{2}\sinh(2\theta_{0}) + \Gamma_{\kappa,\mu}\left[\frac{\sinh(\theta_{0}+\sqrt{-\kappa}q_{1})}{\left\vert\sinh(\theta_{0}+\sqrt{-\kappa}q_{1})\right\vert^{3}} + \mu\frac{\sinh(\theta_{0}-\sqrt{-\kappa}q_{2})}{\left\vert\sinh(\theta_{0}-\sqrt{-\kappa}q_{2})\right\vert^{3}}\right]=0.
    \end{equation}
    \end{proposition}

\subsubsection{Conditions for existence of triangular RE}
\label{sec:triangularbalanced}
Since all the singularities of the potential $\mathcal{V}_{\kappa,\mu}$ lie on the geodesic $\mathcal{G}^\pm$, its restriction to $\mathcal{M}^\pm_{>0}$ is smooth. As coordinates for $\mathbf{r}\in \mathcal{M}^\pm_{>0}$ we will use its Riemannian distances $d_1, d_2>0$ to the primaries $\mathbf{p}_{1}, \mathbf{p}_{2}\in\mathcal{G}$, as illustrated in Fig. \ref{fig:coords}. These provide a global chart for $\mathcal{M}^\pm_{>0}$, and will be referred  to  as \textbf{\textit{distance coordinates}}. Our choice 
to use these coordinates was motivated by the analysis for $\Ll_{4}$,  $\Ll_{5}$  in Meyer and Offin's book \cite[Sec. 4.3]{MO17} in  the case of zero curvature. 

\begin{remark}
\label{rmk:distances}
    Note that $d_1, d_2$ are Riemannian distances on our scaled
    manifold $\mathcal{M}^\pm$. The `true' Riemannian distances from
    the satellite to the primaries on $\calS_\kappa$ are $\delta_1=d_1/\sqrt{|\kappa|}, \delta_2=d_2/\sqrt{|\kappa|}$ (see Proposition \ref{prop:scaling}).
\end{remark}

The precise domain $W_\kappa\subset \R^2$ of the distance coordinates $(d_1,d_2)$ 
will be given below according to the value of $\kappa$. We now
introduce an important concept.
We will say that the point $\mathbf r\in \mathcal{M}^\pm$ with coordinates $(d_1,d_2)\in W_\kappa$
defines a \textbf{\textit{triangular-balanced configuration}} if
 \begin{equation}
    \label{eq:mubalancecondition}
    \begin{matrix}
        \sin(d_1) = \Lambda_{\kappa,\mu}^{+} \sin(d_2) & \mbox{if} & \kappa>0, \\ 
        \hspace{-0.74cm} d_1 = d_2 & \mbox{if} & \kappa=0, \\ \sinh(d_1) = \Lambda_{\kappa,\mu}^{-} \sinh(d_2) & \mbox{if} & \kappa<0. 
        \end{matrix}
    \end{equation}
 In the above conditions, the constants $\Lambda_{\kappa,\mu}^{\pm}$
   depend parametrically on $\kappa$ and $\mu$ by 
\begin{equation}
\label{eq:constantLambda}
    \Lambda_{\kappa,\mu}^{+}=\left(\frac{\sin(\sqrt{\kappa}q_{1})}{\mu\sin(\sqrt{\kappa}q_{2})}\right)^{1/3}, \qquad \Lambda_{\kappa,\mu}^{-}=\left(\frac{\sinh(\sqrt{-\kappa}q_{1})}{\mu\sinh(\sqrt{-\kappa}q_{2})}\right)^{1/3}.    
\end{equation}
Note that $\Lambda^\pm_{\kappa,\mu}$ converges to 1 as $\kappa\to 0$ uniformly in $\mu$ providing a continuous passage of conditions \eqref{eq:mubalancecondition} through $\kappa=0$.

\begin{remark} The constant  $\Lambda_{\kappa,\mu}^{\pm}$ enters many of our calculations in the sections below.
We will often omit its dependence on $\kappa$ and $\mu$, and denote it  by $\Lambda^\pm$ ,or simply even by
$\Lambda$   when working with positive curvature. 
In view of the relation between $q_{1}$ and $q_{2}$ given by Eq.~\eqref{eq:rotationRE}, we may write the following alternative
 expressions for $\Lambda^\pm$
\begin{equation}
\label{eq:alternativeLambda}
	\Lambda^{+}=\left(\frac{\cos(\sqrt{\kappa}q_{2})}{\cos(\sqrt{\kappa}q_{1})}\right)^{1/3}, \qquad \Lambda^{-}=\left(\frac{\cosh(\sqrt{-\kappa}q_{2})}{\cosh(\sqrt{-\kappa}q_{1})}\right)^{1/3}.
\end{equation}
 In particular, given that $q_{1}<q_{2}$, we have
\begin{equation}
\label{eq:inequalityLambda}
	0<\Lambda^{+}<1, \qquad 1<\Lambda^{-}.
\end{equation}
\end{remark}

Below, we will show that all triangular RE define triangular-balanced
configurations. Therefore, equation~\eqref{eq:mubalancecondition} provides {\em necessary conditions} for the existence of RE. In the case
$\kappa =0 $, the condition $d_1=d_2$  is well-known (see e.g. \cite{MO17}), and is consistent with the equilateral-triangle form of the Lagrange points $\Ll_{4}, \Ll_{5}$. For non-vanishing curvature, the geometric interpretation of the triangular-balanced configurations is more involved. The coordinate curves
on $W_k\subset \R^2$ leading to triangular balanced configurations are illustrated in Figs. \ref{fig:graphpos}(a) and \ref{fig:graphsneg}(a). The
corresponding curves on $\mathcal{M}^\pm$ are illustrated in  Figs. \ref{fig:graphpos}(b) and \ref{fig:graphsneg}(b). For $\kappa>0$, these curves
have two connected components.

In distance coordinates, the condition for triangular RE becomes,
\begin{equation}
\label{eq:derdistance}
    \dfrac{\partial\mathcal{V}_{\kappa,\mu}}{\partial d_{i}}=0, \qquad i=1,2.
\end{equation}
We  will obtain explicit expressions
for these conditions treating separately the cases of $\kappa$ positive and negative, and make use of the following formulas, which are obtained from relation \eqref{eq:distancerelation} by putting $\kappa=\pm 1$:
\begin{subequations}
\begin{equation}
\label{eq:formuladistancepositive}
     \cos(d_{i})=\langle \mathbf{r},\mathbf{p}_{i} \rangle_{+} \quad \mbox{for} \quad \kappa>0,
\end{equation}
\begin{equation}
\label{eq:formuladistancenegative}
     \cosh(d_{i})=-\langle \mathbf{r},\mathbf{p}_{i} \rangle_{-} \quad \mbox{for} \quad \kappa<0.
\end{equation}
\end{subequations}

\paragraph{Case $\kappa>0$.}

Let $\mathbf{r}=\big(r_{1}, \sqrt{1-r_{1}^{2}-r_{3}^{2}} ,r_{3}\big)\in(\mathbb{S}_{1}^{2})_{>0}$. In view of formulas \eqref{eq:formuladistancepositive} and \eqref{eq:primariesMpm}, the distances $d_{1},d_{2}$ to the primaries satisfy:
    \begin{equation}
    \label{eq:kappatrigpositive}
    \begin{gathered}
        \cos(d_{1})=-\sin(\!\sqrt{\kappa}q_{1})r_{1}-\cos(\!\sqrt{\kappa}q_{1})r_{3}, \\
        \cos(d_{2})=\sin(\!\sqrt{\kappa}q_{2})r_{1}-\cos(\!\sqrt{\kappa}q_{2})r_{3}.
    \end{gathered}
    \end{equation}
    It can be checked that Eq.~\eqref{eq:kappatrigpositive} defines a smooth correspondence between the open disk $r_{1}^{2}+r_{3}^{2}<1$ and the bounded domain
    \begin{equation}
    \label{eq:domaindistancecoordpos}
        W_{\kappa}=\left\lbrace (d_{1},d_{2})\in\R^{2} : \vert d_{1}-d_{2} \vert < \sqrt{\kappa}, \ \vert d_{1}+d_{2}-\pi \vert < \pi-\sqrt{\kappa} \right\rbrace.
    \end{equation}
    Thus, our coordinates $(d_{1},d_{2})$ take values in $W_{\kappa}$ as defined above. Note that the boundary of the domain $W_{\kappa}$ corresponds to the geodesic $\mathcal{G}^+$. The domain $W_{\kappa}$ is illustrated in Figure \ref{fig:graphpos}(a), where we also indicate the corresponding position
    of the primaries and antipodal points in this chart.

    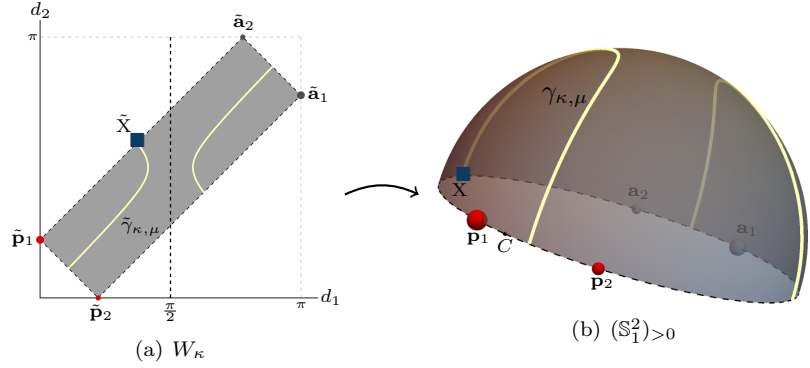
\begin{figure}[ht]
    \centering
    \captionsetup{width=.7\linewidth}
    \begin{tikzpicture}
        \node (A) at (0,0) {
            \subfigure[$W_{\kappa}$]{%
                \begin{overpic}[width=0.25\textwidth]{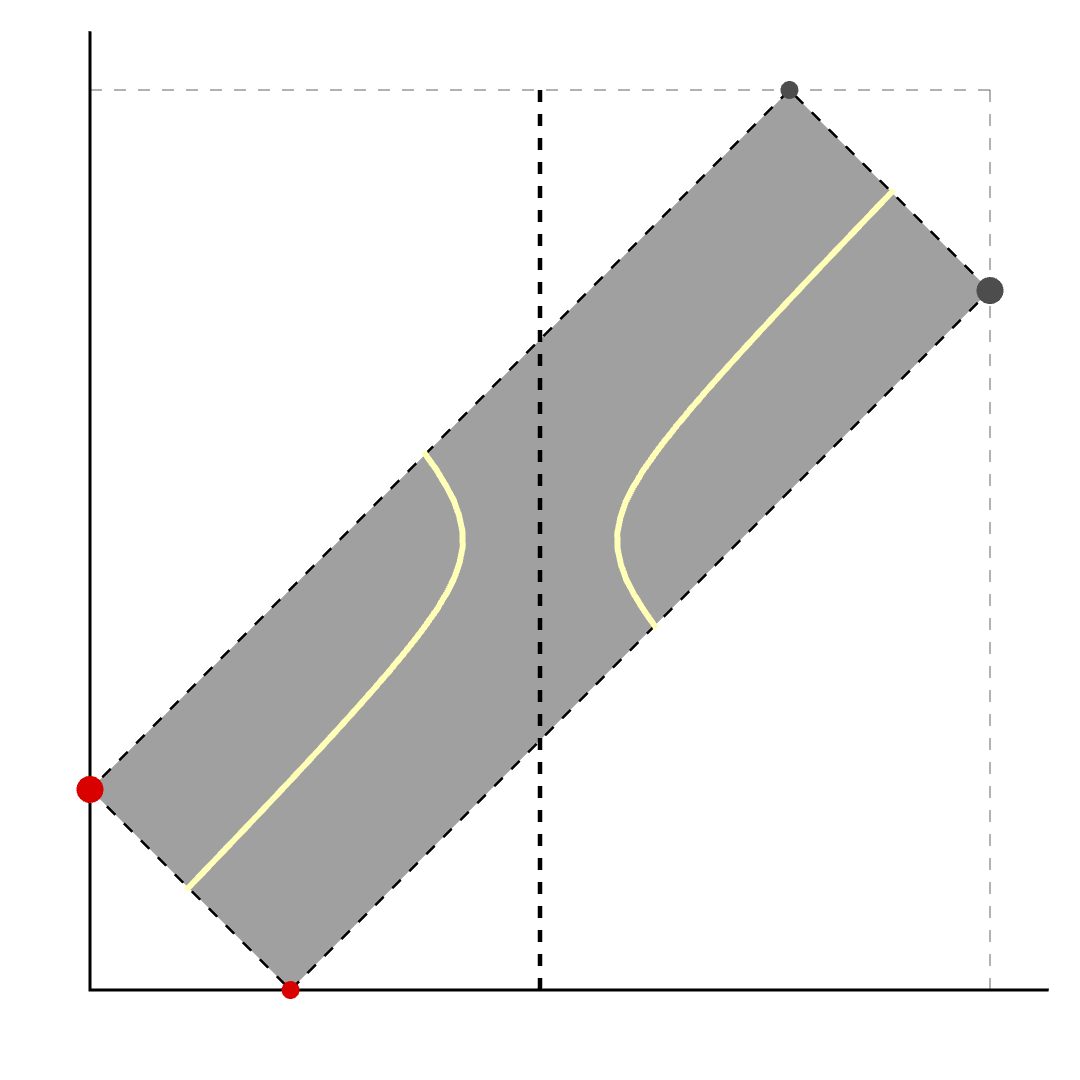}
                    \put(70,95){\scriptsize $\tilde{\mathbf{a}}_{2}$}
                    \put(94,72){\scriptsize $\tilde{\mathbf{a}}_{1}$}
                    \put(-1,26){\scriptsize $\tilde{\mathbf{p}}_{1}$}
                    \put(24,2){\scriptsize $\tilde{\mathbf{p}}_{2}$}
                    \put(98,7){\scriptsize $d_{1}$}
                    \put(5,99){\scriptsize $d_{2}$}
                    \put(90,5){\tiny $\pi$}
                    \put(47,3){\tiny $\frac{\pi}{2}$}
                    \put(4,91){\tiny $\pi$}
                    \put(34,30){\scriptsize $\tilde{\gamma}_{\kappa,\mu}$}
                    \put(32.5,61.5){\scriptsize $\tilde{\Xx}$}
                    \put(36.75,56.5){\scriptsize \textcolor{azulciencias}{$\blacksquare$}}
                \end{overpic}
            }
        };
        \node (B) at (6,0) {
             \subfigure[$(\mathbb{S}_{1}^{2})_{>0}$]{%
                \begin{overpic}[width=0.3\textwidth]{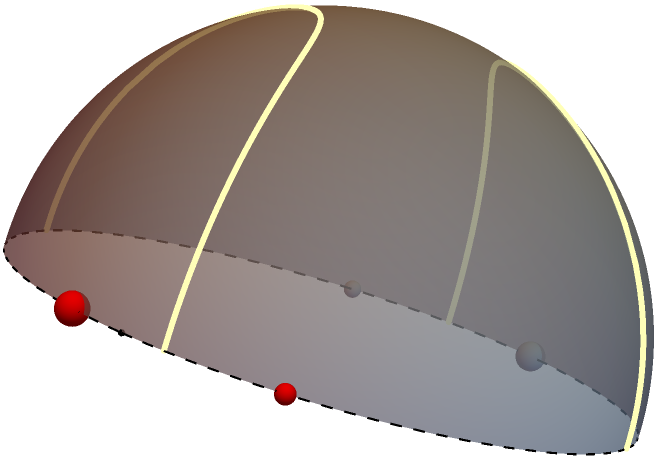}
                    \put(51,29){\scriptsize \textcolor{graycover}{$\mathbf{a}_{2}$}}
                    \put(80,19){\scriptsize \textcolor{graycover}{$\mathbf{a}_{1}$}}
                    \put(8,17){\scriptsize $\mathbf{p}_{1}$}
                    \put(41,5){\scriptsize $\mathbf{p}_{2}$}
                    \put(16,14){\scriptsize $C$}
                    \put(28,55){$\gamma_{\kappa,\mu}$}
                     \put(4,29){\scriptsize $\Xx$}
                     \put(4.9,33.1){\scriptsize \textcolor{azulciencias}{$\blacksquare$}}
                \end{overpic}
            }
            
        };
        \draw[->, thick, bend left=25] (A.east) to node[above]{} (B.west);
    \end{tikzpicture}
    \caption{(a) Bounded domain $W_{\kappa}$ for the distance coordinates, and (b) Half-space $(\mathbb{S}_{1}^{2})_{>0}$. 
    The points $\mathbf{p}_{j}$, $j=1,2$ are the primaries and $\mathbf{a}_{j}:=-\mathbf{p}_{j}$, $j=1,2$ are their antipodal
    points. Their avatars in the domain $W_{\kappa}$ are indicated as
    $\tilde{\mathbf{p}}_{j}$ and $\tilde{\mathbf{a}}_{j}$.
     ${\gamma}_{\kappa,\mu}$ and $\tilde{\gamma}_{\kappa,\mu}$ are the curves of triangle-balanced configurations in the respective domain.}
    \label{fig:graphpos}
    \end{figure}

\begin{proposition}
\label{prop:triangularconditionpositive}
The point $\mathbf{r}\in(\mathbb{S}_{1}^{2})_{>0}$ with distance coordinates $(d_{1},d_{2})\in W_k$ is a triangular RE for $\kappa>0$ if and only if
\begin{equation}
\label{eq:partialdistancescoordspositive}
\begin{gathered}
	\sin^{3}(d_{1}) \sin(\!\sqrt{\kappa}q_{2}) \big(\sin(\!\sqrt{\kappa}q_{2})\cos(d_{1})+\sin(\!\sqrt{\kappa}q_{1})\cos(d_{2})\big) = \Gamma_{\kappa,\mu}\sin^{2}(\sqrt{\kappa}), \\
	\sin^{3}(d_{2}) \sin(\!\sqrt{\kappa}q_{1}) \big(\sin(\!\sqrt{\kappa}q_{2})\cos(d_{1})+\sin(\!\sqrt{\kappa}q_{1})\cos(d_{2})\big) = \mu \, \Gamma_{\kappa,\mu}\sin^{2}(\sqrt{\kappa}).
\end{gathered}
\end{equation}
In particular, all triangular RE with $\kappa>0$ form a triangular-balanced configuration.
\end{proposition}

\begin{proof}
The amended potential $\mathcal{V}_{\kappa,\mu}$ is expressed in distance coordinates as
    \begin{equation}
    \label{eq:potentialdistancespos}
         \mathcal{V}_{\kappa,\mu}(d_{1},d_{2})=-\frac{1}{2}\left[1-\frac{\big(\sin(\!\sqrt{\kappa}q_{2})\cos(d_{1})+\sin(\!\sqrt{\kappa}q_{1})\cos(d_{2})\big)^{2}}{\sin^{2}(\sqrt{\kappa})} \right]-\Gamma_{\kappa,\mu}\big[ \cot(d_{1})+\mu\cot(d_{2}) \big].
     \end{equation}
    Equations \eqref{eq:partialdistancescoordspositive} are rearrangements
    of the conditions \eqref{eq:derdistance}.  For the claim about triangular-balanced configurations, we note that
    we can deduce relation \eqref{eq:mubalancecondition} for $\kappa>0$  through simple manipulations of equations \eqref{eq:partialdistancescoordspositive}.
\end{proof}

The curve of triangular-balanced configurations for $\kappa>0$ has two connected components, both of which are depicted in $\kappa>0$, both in the chart  $W_\kappa$ and in the  hemisphere $(\mathbb{S}_{1}^{2})_{>0}$. As will be shown in Lemma~\ref{lemma:one_component}, all triangular RE satisfy $d_1<\pi/2$, so the 
the triangular-balanced configurations that give rise to triangular RE  lie on the left component in both
diagrams of Fig~\ref{fig:graphpos}.

\begin{remark}
\label{rmk:defX}
    The intersection of the curve of triangular balanced configurations with 
     the boundary of the domain, marked by the square labeled $\tilde \Xx$ in Fig.~\ref{fig:graphpos}(a) (and by $\Xx$ in Fig.~\ref{fig:graphpos}(b)), will play an important role, since at this point collinear RE bifurcate into triangular RE (see
     Sec.~\ref{sec:muT}).
\end{remark}

\paragraph{Case $\kappa<0$.}
We now put $\mathbf{r}=\big(r_{1}, \sqrt{r_{3}^{2}-1-r_{1}^{2}} ,r_{3}\big)\in(\mathbb{L}_{1}^{2})_{>0}$ and note that, in view of relations \eqref{eq:formuladistancenegative} and \eqref{eq:primariesMpm}, the distances $d_{1},d_{2}$ to the primaries satisfy:
    \begin{equation}
    \label{eq:kappatrignegative}
    \begin{gathered}
        \cosh(d_{1})=\sinh(\sqrt{-\kappa}q_{1})r_{1}+\cosh(\sqrt{-\kappa}q_{1})r_{3}, \\
        \cosh(d_{2})=-\sinh(\sqrt{-\kappa}q_{2})r_{1}+\cosh(\sqrt{-\kappa}q_{2})r_{3}.
    \end{gathered}
    \end{equation}
    This time, Eq.~\eqref{eq:kappatrignegative} defines a smooth correspondence between the region $r_{3}^{2}-r_{1}^{2}>1$, $r_{3}>0$ and the unbounded domain
    \begin{equation}
    \label{eq:domaindistancecoordneg}
        W_{\kappa}=\left\lbrace (d_{1},d_{2})\in\R^{2} : \vert d_{1}-d_{2} \vert < \sqrt{\kappa}, \ d_{1}+d_{2} > \sqrt{\kappa} \right\rbrace,
    \end{equation}
    which is 
     illustrated in Figure \ref{fig:graphsneg}(a). Thus, the distance coordinates $(d_{1},d_{2})$ are defined  in  the set $W_{\kappa}$
     given above. The potential in these coordinates becomes
    \begin{equation}
    \label{eq:potentialdistancesneg}
    \begin{aligned}
        \mathcal{V}_{\kappa,\mu}(d_{1},d_{2})=-\frac{1}{2} & \left[\frac{\big(\sinh(\sqrt{-\kappa}q_{2})\cosh(d_{1})+\sinh(\sqrt{-\kappa}q_{1})\cosh(d_{2})\big)^{2}}{\sinh^{2}(\sqrt{-\kappa})}-1\right] \\
        & \qquad\qquad\qquad\qquad\qquad\qquad\qquad\qquad\qquad\qquad-\Gamma_{\kappa,\mu}\big[ \coth(d_{1})+\mu\coth(d_{2}) \big].
    \end{aligned}
    \end{equation}

    In analogy with Proposition \ref{prop:triangularconditionpositive}, by taking the derivatives of the potential \eqref{eq:potentialdistancesneg} and after simple manipulations, we obtain the following.

    \begin{proposition}
    \label{prop:triangularconditionnegative}
    Let $\kappa<0$. If $(d_1,d_2)\in W_\kappa$ are the coordinates
    of a triangular RE, then they define a triangular-balanced configuration.
    \end{proposition}

The curve of triangular-balanced configurations determined by Eq.~\eqref{eq:mubalancecondition} in the case $\kappa<0$
is illustrated in Fig.~\ref{fig:graphsneg}, both in the domain $W_\kappa$, and in $(\mathbb{L}_{1}^{2})_{>0}$.
   
   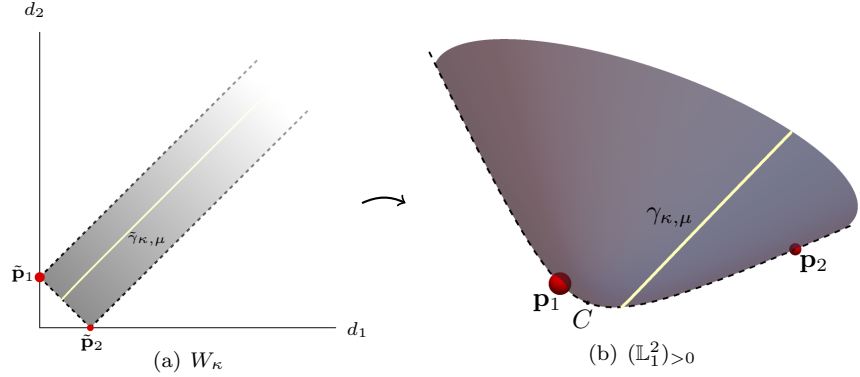
\begin{figure}[ht]
    \centering
    \captionsetup{width=.7\linewidth}
    \begin{tikzpicture}
        \node (A) at (0,0) {
            \subfigure[$W_{\kappa}$]{%
                \begin{overpic}[width=0.25\textwidth]{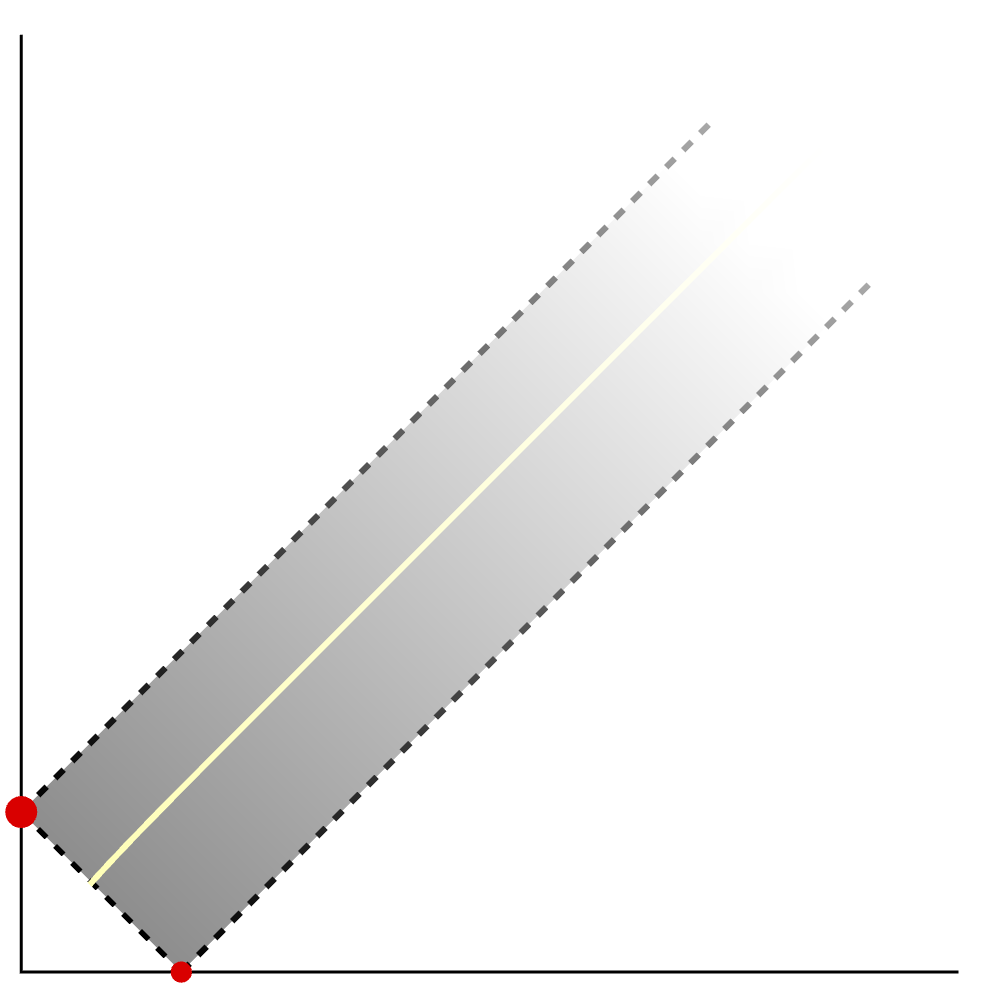}
                \put(-7,18){\scriptsize $\tilde{\mathbf{p}}_{1}$}
                \put(15,-4){\scriptsize $\tilde{\mathbf{p}}_{2}$}
                \put(100,0){\scriptsize $d_{1}$}
                \put(-3,102){\scriptsize $d_{2}$}
                \put(30,30){\tiny $\tilde{\gamma}_{\kappa,\mu}$}
                \end{overpic}
            }
        };
        \node (B) at (6,0) {
             \subfigure[$(\mathbb{L}_{1}^{2})_{>0}$]{%
                \begin{overpic}[width=0.35\textwidth]{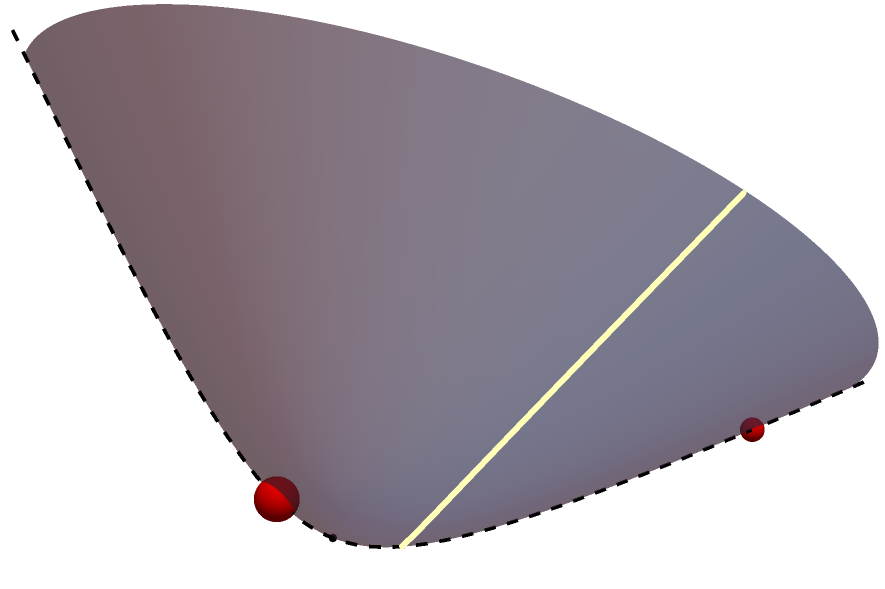}
                    \put(86,14){$\mathbf{p}_{2}$}
                    \put(25,5){$\mathbf{p}_{1}$}
                    \put(34,0){$C$}
                    \put(51,25){$\gamma_{\kappa,\mu}$}
                    
                \end{overpic}
            }
        };
        \draw[->, thick, bend left=25] (A.east) to node[above]{} (B.west);
    \end{tikzpicture}
    \caption{(a) Unbounded domain $W_{\kappa}$ for the distance coordinates, and (b) Half-space $(\mathbb{L}_{1}^{2})_{>0}$.
    The points $\mathbf{p}_{j}$, $j=1,2$ are the primaries. Their avatars in the domain $W_{\kappa}$ are indicated as
    $\tilde{\mathbf{p}}_{j}$.
    The curve of  triangle-balanced configurations is respectively 
    indicated by ${\gamma}_{\kappa,\mu}$ and $\tilde{\gamma}_{\kappa,\mu}$.}
    \label{fig:graphsneg}
    \end{figure}

\subsection{Conditions for stability of relative equilibria}
\label{subsec:stability}

We now consider the stability of RE. 
So far, we have established that  $\mathbf{r}_0\in \mathcal{M}^\pm\setminus \Delta$ is a RE of the curved  R3BP  if and only if it is  a critical point of the potential $\mathcal{V}_{\kappa,\mu}$ defined in Eq.~\eqref{eq:scaledpotential}. This 
 means that $\mathbf{u}_0:=(\mathbf{r}_0, \mathbf{0})\in T(\mathcal{M}^\pm\setminus \Delta)$ is an equilibrium of the Euler--Lagrange equations associated with the Lagrangian 
\begin{equation*}
    \mathcal{L}_{\kappa,\mu}(\mathbf{r}, \mathbf{r}') = \frac{1}{2} \Vert \mathbf{r}' \Vert_{\pm}^{2} + \langle \mathbf{e}_3 \times \mathbf{r}, \mathbf{r}' \rangle_{\pm} - \mathcal{V}_{\kappa,\mu}(\mathbf{r}).
\end{equation*}
This is a (2 degrees of freedom) Lagrangian of mechanical type, containing a gyroscopic (linear) term in the velocities. In Appendix~\ref{ap:linear} we review the
stability  properties of the equilibrium points of this type of Lagrangians. 

All the stability criteria
needed for our analysis are given in Theorem~\ref{thm:criticalpointsV}, which is formulated  in terms of the following  
standard terminology, based on the type of eigenvalues of the
linearization matrix: we say that an equilibrium is \textbf{\textit{elliptic}} 
if all eigenvalues are purely imaginary,  \textbf{\textit{center--saddle}} 
if it has one nonzero-real pair and  
one imaginary pair of eigenvalues,  
\textbf{\textit{complex--saddle}}  if its eigenvalues form a quadruplet $\pm \alpha+\pm i\beta$ with $\alpha, \beta \neq 0$,
and  \textbf{\textit{saddle--saddle}}  when  all eigenvalues are non-zero real.

 In particular,  Theorem~\ref{thm:criticalpointsV} states that
$\mathbf{u}_0$ is Lyapunov stable if $\mathbf{r}_0$ is a minimum of $\mathcal{V}_{\kappa,\mu}$, and is 
a center--saddle if  and only if $\mathbf{r}_0$ is a saddle-point of $\mathcal{V}_{\kappa,\mu}$. Determining the 
stability properties  of local maxima of $\mathcal{V}_{\kappa,\mu}$ is more delicate since they
may be stabilized by the presence of the gyroscopic term.

In our discussion, we will often make a slight abuse of notation by referring to $\mathbf{r}_0$ as an equilibrium point instead of the corresponding equilibrium state $\mathbf{u}_0=(\mathbf{r}_0,\mathbf{0})$.

\subsubsection{Conditions for stability of collinear RE}
\label{sss:stability-conds-collinear}

\paragraph{Case $\kappa>0$.} The stability of a collinear RE in the case of positive curvature can be analyzed using the following result. Item (v) about the non-existence of saddle--saddle points 
is due to Mart\'inez and Sim\'o \cite{MS17}.

\begin{proposition}
\label{prop:stabilitycollpos}
Let $\mathbf{r}_{0}=(\sin\theta_{0},0,-\cos\theta_{0})\in\mathcal{G}^+$ be a collinear RE for $\kappa>0$. Define the quantities
\begin{equation*}
\begin{split}
    \lambda_{1} &= \sin^{2}\!\theta_{0} - \cos^{2}\!\theta_{0} - 2\Gamma\left( \frac{\cos(\theta_{0}+\sqrt{\kappa}q_{1})}{\lvert \sin(\theta_{0}+\sqrt{\kappa}q_{1}) \rvert^{3}} +\mu\,\frac{\cos(\theta_{0}-\sqrt{\kappa}q_{2})}{\lvert \sin(\theta_{0}-\sqrt{\kappa}q_{2}) \rvert^{3}} \right) ,\\
    \lambda_{2} &= -\cos^{2}\!\theta_{0}+\Gamma\left( \frac{\cos(\theta_{0}+\sqrt{\kappa}q_{1})}{\lvert \sin(\theta_{0}+\sqrt{\kappa}q_{1}) \rvert^{3}} +\mu\,\frac{\cos(\theta_{0}-\sqrt{\kappa}q_{2})}{\lvert \sin(\theta_{0}-\sqrt{\kappa}q_{2}) \rvert^{3}}  \right),\\
    b &= 1 + \cos^2\!\theta_{0} - \Gamma\left( \frac{\cos(\theta_{0}+\sqrt{\kappa}q_{1})}{\vert \sin(\theta_{0}+\sqrt{\kappa}q_{1})\vert^{3}}  + \mu \frac{\cos(\theta_{0}-\sqrt{\kappa}q_{2})}{\vert \sin(\theta_{0}-\sqrt{\kappa}q_{2})\vert^{3}} \right).
\end{split}
\end{equation*}
Then, the following statements hold:
\begin{enumerate}
    \item[(i)] If $\lambda_{1} > 0$ and $\lambda_{2} > 0$, then $\mathbf{r}_{0}$ is Lyapunov stable.
    \item[(ii)]  $\mathbf{r}_{0}$ is elliptic  if and only if  $\lambda_{1} \lambda_{2} > 0$, $b>0$ and 
    $b^{2}\geq 4\lambda_{1}\lambda_{2}$.
    \item[(iii)] $\mathbf{r}_{0}$  is a center--saddle if and only if  $\lambda_{1}\lambda_{2} < 0$.
     \item[(iv)] $\mathbf{r}_{0}$  is a complex--saddle if and only if $b^{2} < 4\lambda_{1}\lambda_{2}$.
      \item[(v)]  \cite[Proposition 6.1]{MS17} $\mathbf{r}_{0}$  is not a saddle--saddle.
\end{enumerate}
Furthermore, we can rewrite
\begin{equation}
\label{eq:lambda2alt}
	\lambda_{2} = -\frac{ \Gamma}{\sin\theta_{0}}\left( \frac{\sin(\sqrt{\kappa}q_{1})}{\vert\sin(\theta_{0}+\sqrt{\kappa}q_{1})\vert^{3}} - \mu\frac{\sin(\sqrt{\kappa}q_{2})}{\vert\sin(\theta_{0}-\sqrt{\kappa}q_{2})\vert^{3}} \right).
\end{equation}
\end{proposition}

\begin{proof}
Recall that the potential $\mathcal{V}_{\kappa,\mu}$ in polar coordinates is given by Eq.~\eqref{eq:potentialpolarpos}. Since $\mathcal{V}_{\kappa,\mu}(\phi,\theta)$ is an even function of $\phi$, its Hessian matrix at the collinear critical point $(0,\theta_0) \in \mathcal{G}^\pm$, denoted by $\mathcal{V}''_{\kappa,\mu}(0,\theta_0)$, is diagonal. Furthermore, by differentiating the expression \eqref{eq:potentialpolarpos}, one obtains
\begin{equation*}
    \left. \frac{\partial^2 \mathcal{V}_{\kappa,\mu}}{\partial \theta^2} \right|_{\phi=0, \, \theta=\theta_0} = \lambda_1, \qquad
    \left. \frac{\partial^2 \mathcal{V}_{\kappa,\mu}}{\partial \phi^2} \right|_{\phi=0, \, \theta=\theta_0} = \lambda_2,
\end{equation*}
with $\lambda_1$ and $\lambda_2$ given in the statement of the proposition. In other words, $\lambda_1$ and $\lambda_2$ are the eigenvalues of $\mathcal{V}''_{\kappa,\mu}(0,\theta_0)$.

If $\lambda_1$ and $\lambda_2$ are positive, then $\mathbf{r}_{0}$ is a local minimum of $\mathcal{V}_{\kappa,\mu}$  and hence is Lyapunov stable by Theorem~\ref{thm:criticalpointsV}. 

 Expressing the Lagrangian  \eqref{eq:Lagmukappa} in terms of 
generalized coordinates and velocities $\phi, \theta, \dot \phi, \dot \theta$ (using  the parametrization
  \eqref{eq:positionM+}), one finds that the kinetic energy matrix $\mathbb{M}(\mathbf{r}_{0})$, and the term 
  $\sigma(\mathbf{r}_{0})$ in Eq.~\eqref{eq:formskew}, are given in these coordinates by
\begin{equation*}
    \mathbb{M}(\mathbf{r}_{0})=\begin{pmatrix} 1 & 0 \\ 0 & 1 \end{pmatrix}, \qquad \sigma(\mathbf{r}_{0}) = 2\cos\theta_{0}.
\end{equation*}
Therefore, applying  Lemma \ref{l:charpolmaxima}, one concludes, after some algebraic manipulations,
 that the characteristic polynomial of the linearization matrix at $\mathbf{r}_{0}$ is 
\begin{equation*}
    p(x) = x^{4} + b x^{2} + \lambda_{1}\lambda_{2},
\end{equation*}
with $b=\lambda_{1}+\lambda_{2}+\sigma_{0}^2$ given as in the statement of the proposition. The conclusions $(ii)$-$(iv)$ follow immediately from Theorem \ref{thm:criticalpointsV}. 

Furthermore, 
 from the expressions for $\lambda_{1}$, $\lambda_{2}$, and $b$, one obtains
\begin{equation*}
\lambda_{2}+b=1, \qquad \lambda_{1}-2b=-1-4\cos^{2}\!\theta_{0}.
\end{equation*}
Therefore, if $b<0$, then $\lambda_{2}>0$ and $\lambda_{1}<0$. Hence, $\lambda_{1}\lambda_{2} < 0$, and item $(iv)$ of Theorem \ref{thm:criticalpointsV} implies that 
$\mathbf{r}_{0}$ cannot be a saddle-saddle.

Finally, we prove that the alternative expression \eqref{eq:lambda2alt} for $\lambda_2$ holds. Abbreviating 
\begin{equation*}
A= \frac{\Gamma}{\lvert \sin(\theta_{0}+\sqrt{\kappa}q_{1}) \rvert^{3}}, \qquad B= \frac{\mu\,\Gamma}{\lvert \sin(\theta_{0}-\sqrt{\kappa}q_{2}) \rvert^{3}}, 
\end{equation*}
we may write,
\begin{equation}
\label{eq:auxlambda2.1}
\lambda_2=-\cos^2\theta_0 +A \cos(\theta_{0}+\sqrt{\kappa}q_{1}) + B\cos (\theta_{0}-\sqrt{\kappa}q_{2}).
\end{equation}
On the other hand, by Eq.~\eqref{eq:collinearpositive} we have
\begin{equation}
\label{eq:auxlambda2.2}
0=\sin\theta_0\cos\theta_0 -A \sin(\theta_{0}+\sqrt{\kappa}q_{1}) - B\sin(\theta_{0}-\sqrt{\kappa}q_{2}).
\end{equation}
Multiplying Eq.~\eqref{eq:auxlambda2.1} by $\sin \theta_0$, Eq.~\eqref{eq:auxlambda2.2} by $\cos \theta_0$, and adding them yields,
\begin{equation*}
\sin\theta_0 \lambda_2= -A \sin(\sqrt{\kappa}q_{1})  +B \sin(\sqrt{\kappa}q_{2}).
\end{equation*}
Finally, it is easy to verify that Eq.~\eqref{eq:collinearpositive} is not satisfied for $\theta_0=0,\pi$ for any $\mu\in(0,1)$. Therefore, the above equation may be divided by $\sin\theta_0\neq 0$, yielding an expression for $\lambda_2$ equivalent to Eq.~\eqref{eq:lambda2alt}.
\end{proof}

\paragraph{Case $\kappa<0$.} 
Proceeding in analogy with the positive curvature case, it is possible to formulate a criterion 
similar to Proposition~\ref{prop:stabilitycollpos} valid for negative curvature.
We do not give it explicitly since the work of Mart\'inez and ~Sim\'o~\cite{MS17} establishes
the stability properties of all collinear RE in this case. Their results are recalled in Proposition~\ref{thm:stabilitycollinearneg} in Section~\ref{s:stability-collinear}.

\subsubsection{Conditions for stability of triangular RE}
\label{sss:Stability-conds-triangular}

In contrast with the collinear case, the Hessian matrix $\mathcal{V}_{\kappa,\mu}''(\mathbf{r}_{0})$ evaluated at a triangular RE $\mathbf{r}_{0}$ is not diagonal in distance coordinates,
so the stability analysis is a bit more involved. 

\paragraph{Case $\kappa>0$.} The main result to be used in the stability analysis of Sec.~\ref{sec:stabilitytriangular}  is the following.
\begin{proposition}
\label{prop:stabilitytripos}
Let $\mathbf{r}_{0}=\mathbf{r}_{0}(d_{1},d_{2})$ be a triangular RE for $\kappa>0$ written in distance coordinates \eqref{eq:kappatrigpositive}. Define the quantities
\begin{equation*}
\begin{split}
    \Delta &= \frac{3\sin(\!\sqrt{\kappa}q_{1})\sin(\!\sqrt{\kappa}q_{2})}{\sin^{4}(\sqrt{\kappa})}\Big[ \sin(\!\sqrt{\kappa}q_{2})\cos(d_{1}) + \sin(\!\sqrt{\kappa}q_{1})\cos(d_{2}) \Big] \\ 
    & \qquad\times\Big[ \big(1+2\cos(2d_{2})\big)\sin(\!\sqrt{\kappa}q_{1})\cos(d_{1}) + \big(1+2\cos(2d_{1})\big)\sin(\!\sqrt{\kappa}q_{2})\cos(d_{2}) \Big],\\
    T &= \frac{1}{\sin^{2}(\sqrt{\kappa})} \Big[ (1 - 4\cos^{2}(d_{2})) \sin^{2}(\sqrt{\kappa}q_{1}) - 6 \cos(d_{1}) \cos(d_{2}) \sin(\!\sqrt{\kappa}q_{1}) \sin(\!\sqrt{\kappa}q_{2})\\
    &\qquad + (1-4\cos^{2}(d_{1}))\sin^{2}(\sqrt{\kappa}q_{2})\Big], \\
    M &= \frac{\sin^{2}(d_{1})\sin^{2}(d_{2})}{\big[ \cos(\!\sqrt{\kappa})-\cos(d_{1}+d_{2})\big]\big[ \cos(d_{1}-d_{2}) - \cos(\!\sqrt{\kappa})\big]}.
\end{split}
\end{equation*}  
Then, the following statements hold:
\begin{enumerate}
    \item[(i)] If $\Delta > 0$ and $T>0$, then $\mathbf{r}_{0}$ is Lyapunov stable.
    \item[(ii)] $\mathbf{r}_{0}$ is elliptic if and only if  $\Delta > 0$ and $M \geq 4\Delta$.
     \item[(iii)] $\mathbf{r}_{0}$ is a center--saddle if and only if  $\Delta < 0$.
     \item[(iv)]  $\mathbf{r}_{0}$ is a complex--saddle if and only if  $M < 4\Delta$.
        \item[(v)] $\mathbf{r}_{0}$ is not a saddle--saddle.
\end{enumerate}
\end{proposition}

Our proof will rely on the following auxiliary result.

\begin{lemma}
\label{lemma:characteristic}
Let $\mathbf{r}_{0}$ be a triangular RE in the case $\kappa>0$.
The characteristic polynomial of the linearization matrix at $\mathbf{r}_{0}$ is
\begin{equation}
\label{eq:charpoltriag}
    p(\lambda) = \lambda^{4} + \lambda^{2} + \det\big[\mathbb{M}(\mathbf{r}_{0})^{-1}\mathcal{V}''_{\kappa,\mu}(\mathbf{r}_{0}) \big],
\end{equation}
where $\mathbb{M}(\mathbf{r}_{0})$ and $\mathcal{V}''_{\kappa,\mu}(\mathbf{r}_{0})$ respectively denote the 
kinetic energy matrix and the Hessian matrix of the potential $\mathcal{V}_{\kappa,\mu}$ evaluated at $\mathbf{r}_{0}$
in some set of coordinates around $\mathbf{r}_{0}$.
\end{lemma}
\begin{proof} The independence of the choice of coordinates is a consequence of Lemma~\ref{l:charpolmaxima},
which also implies that we  only need to prove that the relation
\begin{equation}
\label{eq:prooflemma}
    \operatorname{tr}\left(\mathbb{M}(\mathbf{r}_{0})^{-1}\mathcal{V}''_{\kappa,\mu}(\mathbf{r}_{0})\right)
      +\sigma( \mathbf{r}_{0})^{2}\, \det\left(\mathbb{M}(\mathbf{r}_{0})^{-1}\right) = 1,
\end{equation}
holds in  a particular set of coordinates around $\mathbf{r}_{0}$. Here $\sigma( \mathbf{r}_{0})$ is determined by Eq.~\eqref{eq:formskew}.

Expressing the Lagrangian \eqref{eq:Lagmukappa} in the distance coordinates $(d_{1},d_{2},\dot{d}_{1},\dot{d}_{2})$ determined by formula \eqref{eq:kappatrigpositive}, one may derive the following expressions:
\begin{equation*}
	\mathbb{M}(\mathbf{r}_{0}) = \frac{1}{1 - k^{2}} \begin{pmatrix} 1 & -k \\ -k & 1 \end{pmatrix}, \qquad k = \frac{\cos(\sqrt{\kappa}) - \cos d_{1}\cos d_{2}}{\sin d_{1}\sin d_{2}},
\end{equation*}
and
\begin{equation*}
	\sigma( \mathbf{r}_{0}) = \frac{2\,\csc(\sqrt{\kappa})\big(\sin(\sqrt{\kappa} q_{2})\cos d_{1} + \sin(\sqrt{\kappa} q_{1})\cos d_{2}\big)}{\sqrt{1 - k^{2}}}. 
\end{equation*}
Identity \eqref{eq:prooflemma} can be verified by a direct computation, using
the above expressions for $\mathbb{M}( \mathbf{r}_{0})$ and $\sigma( \mathbf{r}_{0})$ and noting that the coordinates $d_{1}, d_{2}$ satisfy relations \eqref{eq:partialdistancescoordspositive}
by our assumption that $\mathbf{r}_{0}$ is a triangular RE.
\end{proof}

\begin{remark}  An equivalent formulation of this lemma, is given in \cite{MS17} (see Lemma 6.1).
 The slight disagreement
in the formulas for the characteristic polynomial 
(in \cite[Lemma 6.1]{MS17} the quadratic term has a factor of  has  $\alpha^2$, which is $\omega^2$  in our notation) is due to 
our time normalization $\tau=\omega t$,  introduced in Section \ref{ss:time-independent}.
\end{remark}

We are now ready to present the proof
of Proposition \ref{prop:stabilitytripos}.

\begin{proof}[Proof of Proposition \ref{prop:stabilitytripos}.]
Using formulas \eqref{eq:partialdistancescoordspositive}, which hold since $\mathbf{r}_{0}(d_{1},d_{2})$ is a triangular RE, a direct computation of the second derivatives of $\mathcal{V}_{\kappa,\mu}$ expressed  in distance coordinates (Eq.~\eqref{eq:potentialdistancespos})
  shows that the quantities $\Delta$, $T$ and $M$
in the statement of the proposition satisfy
\begin{equation*}
    \Delta = \det\big( \mathcal{V}_{\kappa,\mu}''(\mathbf{r}_{0}(d_{1},d_{2})) \big), 
    \qquad
    T = \mbox{Tr}\big( \mathcal{V}_{\kappa,\mu}''(\mathbf{r}_{0}(d_{1},d_{2})) \big), \qquad \mbox{and} \qquad M = \det(\mathbb{M}(\mathbf{r}_{0}(d_{1},d_{2}))).
\end{equation*}

Therefore, if $\Delta>0$ and $T>0$ then $\mathbf{r}_{0}$ is a local minimum of $\mathcal{V}_{\kappa,\mu}$ and hence is Lyapunov stable by Theorem~\ref{thm:criticalpointsV}. 

On the other hand, by Lemma \ref{lemma:characteristic}, the characteristic polynomial
of the linearization at  $\mathbf{r}_{0}$ is
\begin{equation*}
p(\lambda)=\lambda^4+\lambda^2+\frac{\Delta}{M},
\end{equation*}
and items $(ii)$-$(v)$ follow from Theorem~\ref{thm:criticalpointsV}.
\end{proof}

\paragraph{Case $\kappa<0$.} The analysis
for negative curvature is analogous. The corresponding expressions are obtained by replacing the trigonometric functions with their hyperbolic counterparts and considering distance coordinates as in relations \eqref{eq:kappatrignegative}. 
In the statement of the following proposition, we have incorporated some results from~\cite{MS17}
as item $(i)$.

\begin{proposition}
\label{prop:stabilitytrineg}
Let $\mathbf{r}_{0}=\mathbf{r}_{0}(d_{1},d_{2})$ be a triangular RE for $\kappa<0$ written in distance coordinates \eqref{eq:kappatrignegative}. 
\begin{enumerate}
\item[(i)]\cite[Proposition 6.2]{MS17} $\mathbf{r}_{0}$ is not a center--saddle nor a  saddle--saddle.
\end{enumerate}
    Define the quantities
\begin{equation*}
\begin{split}
    \Delta &= \frac{3\sinh(\sqrt{-\kappa}q_{1})\sinh(\sqrt{-\kappa}q_{2})}{\sinh^{4}(\sqrt{-\kappa})}\Big[ \sinh(\sqrt{-\kappa}q_{2})\cosh(d_{1}) + \sinh(\sqrt{-\kappa}q_{1})\cosh(d_{2}) \Big] \\ 
    & \qquad\times\Big[ \big(1+2\cosh(2d_{2})\big)\sinh(\sqrt{-\kappa}q_{1})\cosh(d_{1}) + \big(1+2\cosh(2d_{1})\big)\sinh(\sqrt{-\kappa}q_{2})\cosh(d_{2}) \Big],\\
    T &= \frac{1}{\sinh^{2}(\sqrt{-\kappa})} \Big[ (1 - 4\cosh^{2}(d_{2})) \sinh^{2}(\sqrt{-\kappa}q_{1}) - 6 \cosh(d_{1}) \cosh(d_{2}) \sinh(\sqrt{-\kappa}q_{1}) \sinh(\sqrt{-\kappa}q_{2})\\
    &\qquad + (1-4\cosh^{2}(d_{1}))\sinh^{2}(\sqrt{-\kappa}q_{2})\Big], \\
    M &= -\frac{\sinh^{2}(d_{1})\sinh^{2}(d_{2})}{\big[ \cosh(\sqrt{-\kappa})-\cosh(d_{1}+d_{2})\big]\big[ \cosh(\sqrt{-\kappa})-\cosh(d_{1}-d_{2})\big]}.
\end{split}
\end{equation*}
Then, the following statements hold:
\begin{enumerate}
    \item[(ii)] If $\Delta > 0$ and $T>0$, then $\mathbf{r}_{0}$ is Lyapunov stable.
    \item[(iii)] $\mathbf{r}_{0}$ is elliptic if and only if  $\Delta > 0$ and $M - 4\Delta\geq 0$.
        \item[(iv)]  $\mathbf{r}_{0}$ is a complex--saddle if and only if  $M - 4\Delta<0$.
     \end{enumerate}
\end{proposition}

\subsection{\texorpdfstring{The mass ratio $\mu_{\Tt}$}{The mass ratio mu_Tt}}
\label{subsec:massratio}

The qualitative behavior of the RE as a function of the curvature $\kappa\in (0,\pi^2/4)$ remains unchanged
as long as the mass ratio $\mu$ is smaller than a particular value, $\mu_{\Tt}\approx 0.089$.
At $\mu=\muT$, there exists a value of the curvature, $\kappa_{\Tt}$, at which  a qualitative change in the 
 bifurcation structure, described in Sec.~\ref{sec:muT}, takes
place. As a consequence, any  description of the behavior of the RE as a function of $\kappa$ in
 a parameter range  $\Ps{\mu_0}$ with $\mu_0>\mu_{\Tt}$  becomes considerably more intricate.  For this reason, the CAPs,  bifurcation diagrams, numerical results and classification
 conjectures
 presented in  the sections below are restricted to  parameter values $(\kappa,\mu)\in \Ps{\mu_{\Tt}}$.

As explained in Sec.~\ref{sec:muT}, the parameter values $(\kappa_{\Tt},\mu_{\Tt})$ are determined  by the  solution, described below, of the 
following nonlinear system for the
unknowns $(\theta, \kappa, \mu)$: 
\begin{equation}
\label{eq:mubalancedsystem}
    \left\lbrace\begin{matrix}
    \sin(2\theta)\sin^{2}(\theta+\sqrt{\kappa}q_{1})+2\Gamma(1+\mu\Lambda^{2})=0,\\
    \cos(2\theta)\sin^{3}(\theta+\sqrt{\kappa}q_{1})\sin(\!\sqrt{\kappa}q_{2})-2\Gamma\cos(\theta)\sin(\!\sqrt{\kappa})=0,\\
     \sin(\theta+\sqrt{\kappa}q_{1})-\Lambda \sin(\theta-\sqrt{\kappa}q_{2})=0.
\end{matrix}\right.
\end{equation}
Numerically, we obtain the approximation
\begin{equation}
\label{eq:numapprox}
       \tilde{\theta}=-1.0471975511965979, \qquad \tilde{\kappa}=0.7444610429129391, \qquad \tilde{\mu}=0.08867648509140362.
\end{equation}
Using a CAP (see \cite{github_codes}),  we rigorously validated the existence of a unique solution of Eqs.~\eqref{eq:mubalancedsystem} in a neighborhood of the approximation \eqref{eq:numapprox}, which defines $(\kappa_{\Tt},\mu_{\Tt})$. Moreover, the validation
 proves that  $\kappa_{\Tt}$ and $\muT$ respectively agree with the numerical approximations  $\tilde \kappa$ and $\tilde \mu$
to at least 14 decimal digits.

\section{Classification of Collinear RE}
\label{sec:collinear}

We begin by recalling the classification obtained by Martínez and Simó \cite{MS17} for negative curvature.

\begin{proposition}\cite[Proposition 4.1]{MS17}[Classification of Collinear Relative Equilibria for Negative Curvature]
\label{prop:ClassCollinearNeg}
For any mass ratio $\mu \in (0,1)$ and curvature $\kappa<0$, the curved R3BP with negative curvature admits exactly three collinear relative equilibria. We denote them by $\Ll_{1}$, $\Ll_{2}$, and $\Ll_{3}$ according to the following convention:
\begin{itemize}
    \item $\Ll_{1}$ corresponds to the unique solution of Eq.~\eqref{eq:collinearnegative} for $\theta_0\in (-\sqrt{-\kappa}\,q_{1},\,\sqrt{-\kappa}\,q_{2})$;
    \item $\Ll_{2}$ corresponds to the unique solution of Eq.~\eqref{eq:collinearnegative}  for $\theta_0\in (\sqrt{-\kappa}\,q_{2},\,\infty)$;
    \item $\Ll_{3}$ corresponds to the unique solution of Eq.~\eqref{eq:collinearnegative}  for $\theta_0\in (-\infty,\,-\sqrt{-\kappa}\,q_{1})$.
\end{itemize}
\end{proposition}

Below we consider the case $\kappa>0$, sometimes under additional assumptions on the mass ratio $\mu\in(0,1)$. To the best of our knowledge, no rigorous classification results are available in this case. The numerical studies in \cite{Ki99,MS17} suggest that the behavior of collinear RE as functions of the parameters $\kappa$ and $\mu$ is intricate, and several bifurcations take place.

We begin by describing the strategy for our 
analysis in Sec.~\ref{sss:defIj}. We will
then present some analytical results on the number of collinear RE for small values of $\kappa>0$ in Sec.~\ref{ss:AnalyticalExistenceCollinear}. These analytical results are complemented by CAPs covering a wider range of curvature values
in Sec.~\ref{sss:CAP-existence-collinear}. Finally,  some  numerical investigations and conjectures are presented in
Sec.~\ref{sss:collinearexistencenumerics}.

\subsection{Analysis setup}
\label{sss:defIj}

According to Proposition \ref{prop:necessarycondpos}, the
collinear RE for $\kappa>0$
 correspond to solutions of Eq.~\eqref{eq:collinearpositive} for the angular variable 
 $\theta\in \R/2\pi\Z$ that parametrizes $\mathcal{G}^+$. This equation may
 be written as $f(\theta;\kappa,\mu)=0$ with
 \begin{equation*}
    f (\theta;\kappa,\mu):= \frac{1}{2}\sin(2\theta) -\Gamma\left[\frac{\sin(\theta+\sqrt{\kappa}q_{1})}{\left\vert\sin(\theta+\sqrt{\kappa}q_{1})\right\vert^{3}} + \mu\frac{\sin(\theta-\sqrt{\kappa}q_{2})}{\left\vert\sin(\theta-\sqrt{\kappa}q_{2})\right\vert^{3}}\right],
 \end{equation*}
 where we have omitted the dependence
 of $q_1, q_2$ and $\Gamma>0$ on $\kappa, \mu$. We recall that 
 $$
 0<q_1,q_2<1, \qquad \mbox{and} \qquad q_1+q_2=1,
 $$ and that $\Gamma$ is given by formula \eqref{eq:constantgamma},  which we rewrite
 here for future reference as:
\begin{equation}
\label{eq:auxGammaq}
    \Gamma = \frac{\sin^{3}(\sqrt{\kappa})\cos(\sqrt{\kappa})}{\sqrt{\mu^{2} + 2\mu\cos(2\sqrt{\kappa}) + 1}}. 
\end{equation}

To perform the analysis, we  remove the absolute values in the above expression for $f(\theta;\kappa,\mu)$ according to the value of $\theta$. 
We do this by dividing 
the great circle $\mathcal{G}^+$ into four disjoint arc segments, $\mathcal{I}_{1}$, $\mathcal{I}_{2}$, $\mathcal{I}_{3}$, $\mathcal{I}_{4}$. These
segments are separated
by the singularities of the function $\theta\mapsto f(
\theta;\kappa,\mu)$ which correspond to collisions of the satellite with the primaries $\mathbf{p}_{j}$  and their
antipodal points $\mathbf{a}_{j}=-\mathbf{p}_{j}$, $j=1,2$. The precise definition of these segments, identified with  a  possible range of the angular variable $\theta$ (which is consistent 
with the parametrization \eqref{eq:positionM+}), is as follows 
(see Fig.~\ref{fig:circlecollinear} for an illustration):
\begin{itemize}
    \item $\mathcal{I}_{1}=(-\sqrt{\kappa}q_1,\sqrt{\kappa}q_2)$ is the segment between the primaries.
      \item $\mathcal{I}_{2}=(\sqrt{\kappa}q_2,\pi-\sqrt{\kappa}q_1)$ is the segment between $\mathbf{p}_{2}$ and 
      $\mathbf{a}_{1}$.
    \item $\mathcal{I}_{3}=(-\pi + \sqrt{\kappa} q_{2},\, - \sqrt{\kappa} q_{1})$ is the segment between $\mathbf{a}_{2}$ and 
      $\mathbf{p}_{1}$.
\item $\mathcal{I}_{4}=(-\pi - \sqrt{\kappa} q_{1}, -\pi + \sqrt{\kappa} q_{2})$ is the segment between $\mathbf{a}_{1}$ and
$\mathbf{a}_{2}$.
\end{itemize}

For $j=1,\dots, 4$, the restriction of  $f(\theta;\kappa,\mu)$ to the
interval $\mathcal{I}_{j}$ will be denoted as $f_j(\theta;\kappa,\mu)$.
Our approach to the classification 
of the collinear RE is to analyze the equations
\begin{equation}
\label{eq:CollinearRECondbyIntervals}
    f_j(\theta;\kappa,\mu)=0, \qquad \theta \in \mathcal{I}_{j}, \qquad j=1,\dots, 4,
\end{equation}
as functions of the parameters $\kappa, \mu$. One of the difficulties of
this task is that the endpoints of the intervals $\mathcal{I}_{j}$ depend
on $\kappa$ and $\mu$. Another one is that $f_j(\theta;\kappa,\mu)\to \pm \infty$ at these endpoints. On the other hand, 
we will exploit the smoothness of the functions $\theta \mapsto f_j(\theta;\mu, \kappa)$, $\theta\in \mathcal{I}_{j}$.

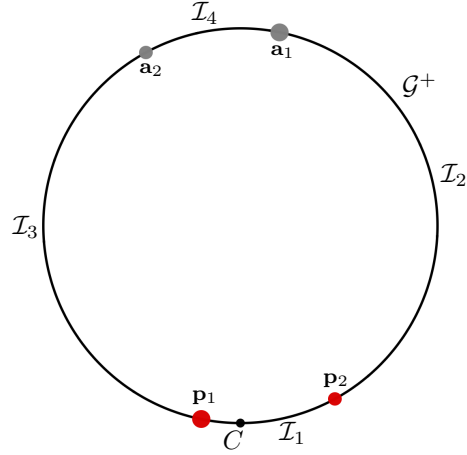
\begin{figure}[ht]
\centering
\captionsetup{width=.8\linewidth}
\begin{overpic}[width=.35\textwidth]{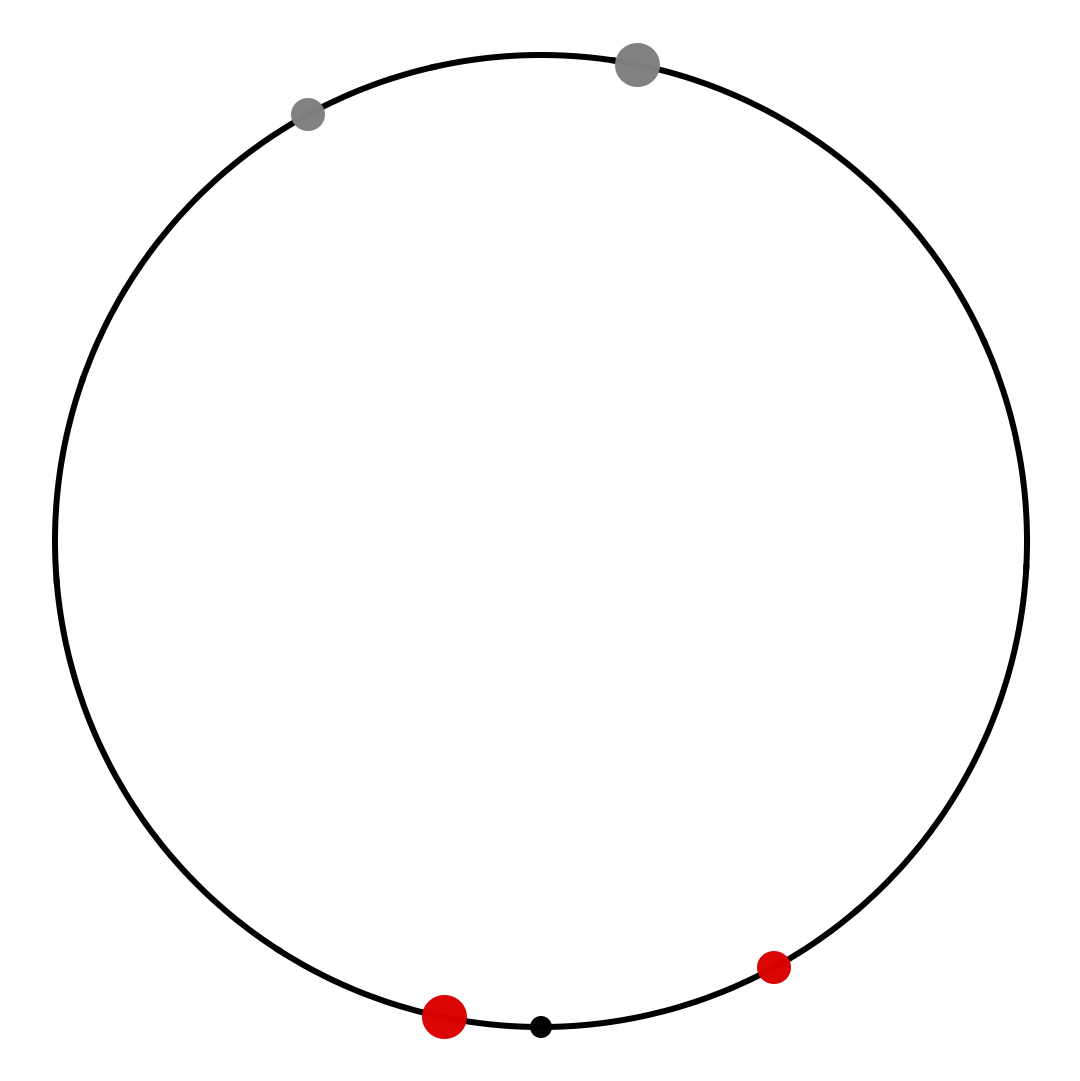}
   \put(39,10){\footnotesize $\mathbf{p}_{1}$}
   \put(69,14){\footnotesize $\mathbf{p}_{2}$}
   \put(57,89){\footnotesize $\mathbf{a}_{1}$}
   \put(27,85){\footnotesize $\mathbf{a}_{2}$}
    \put(85,80){ $\mathcal{G}^+$}
   \put(59,1){$\mathcal{I}_{1}$}
    \put(46,-1){$C$}
   \put(96,60){$\mathcal{I}_{2}$}
   \put(-2,48){$\mathcal{I}_{3}$}
   \put(39,97){$\mathcal{I}_{4}$}
\end{overpic}
\caption{The geodesic $\mathcal{G}^+$ containing the rotation center $C$, the primaries $\mathbf{p}_{i}$ (red points),  and their antipodal points $\mathbf{a}_{i}$ (gray points), $i=1,2$. These latter four points partition $\mathcal{G}^+$ into the four segments  
$\mathcal{I}_{1}$, $\mathcal{I}_{2}$, $\mathcal{I}_{3}$, $\mathcal{I}_{4}$.
According to our conventions  $C\in \mathcal{I}_{1}$ and corresponds to $\theta=0$.}
\label{fig:circlecollinear}
\end{figure}

We conclude this section by computing  the coordinate $\theta$ of the point marked by $\Xx$ in Fig.~\ref{fig:graphpos} (see 
Remark~\ref{rmk:defX}), which lies at the intersection of $\mathcal{G}^+$ and $\gamma_{\kappa,\mu}$. Although this is not
strictly necessary for the existence analysis that follows,  it will prove useful in 
Sec.~\ref{sec:muT}.  

Since $\Xx$ lies in the segment between $\mathbf{a}_{2}$ and 
      $\mathbf{p}_{1}$, we have $\theta \in \mathcal{I}_3$. Since  $\theta$ is negative on this interval, the distance coordinates $d_1$ and $d_2$ of $\Xx$ satisfy:
      \begin{equation*}
d_1=-\theta -\sqrt{\kappa}q_1, \qquad d_2=-\theta +\sqrt{\kappa}q_2.
\end{equation*}
Therefore, in view of condition \eqref{eq:mubalancecondition}, the condition that $\Xx$ is a triangular-balanced configuration implies
that $\theta$ satisfies
\begin{equation}
\label{eq:Xcoord}
\sin \left ( \theta + \sqrt{\kappa}	q_1\right )- \Lambda \sin \left ( \theta -  \sqrt{\kappa}	q_2 \right ) =0.
\end{equation}

\subsection{Analytical results}
\label{ss:AnalyticalExistenceCollinear}

Our main  analytic result on the existence of collinear RE is the following.

{\begin{theorem}[\textit{Classification of collinear RE}]
\label{thm:numbercollinearRE}
The following
statements hold about the collinear RE of the R3BP on a space of curvature
$\kappa>0$:
\begin{enumerate}
    \item For fixed $0 < \kappa < \pi^2/16$ and $\mu\in (0,1)$, there exists exactly one collinear RE on the segment $\mathcal{I}_1$. It will be denoted $\Ll_1$.
    \item For fixed $0 < \kappa < \pi^2/36$ and $\mu\in \left (0,\frac{1}{10} \right )$, there exist exactly two collinear RE on the segment $\mathcal{I}_2$. These  will 
    be denoted $\Ll_{2}$ and $\Ee_{2}$,
    with the convention that $\Ll_{2}$ is closer to the primary $\mathbf{p}_{2}$.
    \item For fixed $0 < \kappa < \pi^2/64$ and $\mu\in \left (0,\frac{1}{10} \right )$, there exist exactly two collinear RE on the segment $\mathcal{I}_3$. These will be denoted $\Ll_{3}$ and $\Ee_{3}$,
    with the convention that $\Ll_{3}$ is closer to the primary $\mathbf{p}_{1}$.
    \item For fixed $0 < \kappa < \pi^2/16$ and $\mu\in \left (0,\frac{1}{10} \right )$, there exists exactly one collinear RE on  the segment   $\mathcal{I}_4$. It will be denoted $\Aa_{1}$.
\end{enumerate}
\end{theorem}

The collinear RE in the theorem are illustrated in Fig. \ref{fig:circlecollinearRE}. The motivation for our  terminology comes from the analysis in Sec.
 \ref{sec:asymp} where we will consider their behavior as $\kappa\to 0^+$. We will 
 prove that $\Ll_{1}$, $\Ll_{2}$, $\Ll_{3}$ approach the classical collinear Lagrange points. On
 the other hand, due our assumption that the center of rotation $C$ is at the south pole, we will show that $\Ee_{2}$, and $\Ee_{3}$ approach the equator and thus our choice
 of the letter $\Ee$. The point $\Aa_1$ lies
 between the antipodal points of the primaries
 and hence our choice of the letter $\Aa$.
 In particular, our analysis shows that
 $\Aa_{1}$, $\Ee_{2}$, and $\Ee_{3}$ tend to infinity 
in the limit $\kappa\to 0^+$ explaining why
the planar R3BP only has 3 collinear RE.

Other classification results, valid in a broader region of the parameter space 
 $\kappa$, $\mu$ are established in 
Sec.~\ref{sss:CAP-existence-collinear} below using CAPs.

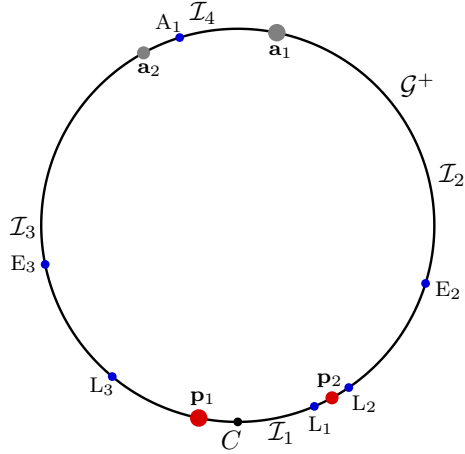
\begin{figure}[ht]
\centering
\captionsetup{width=.75\linewidth}
\begin{overpic}[width=.35\textwidth]{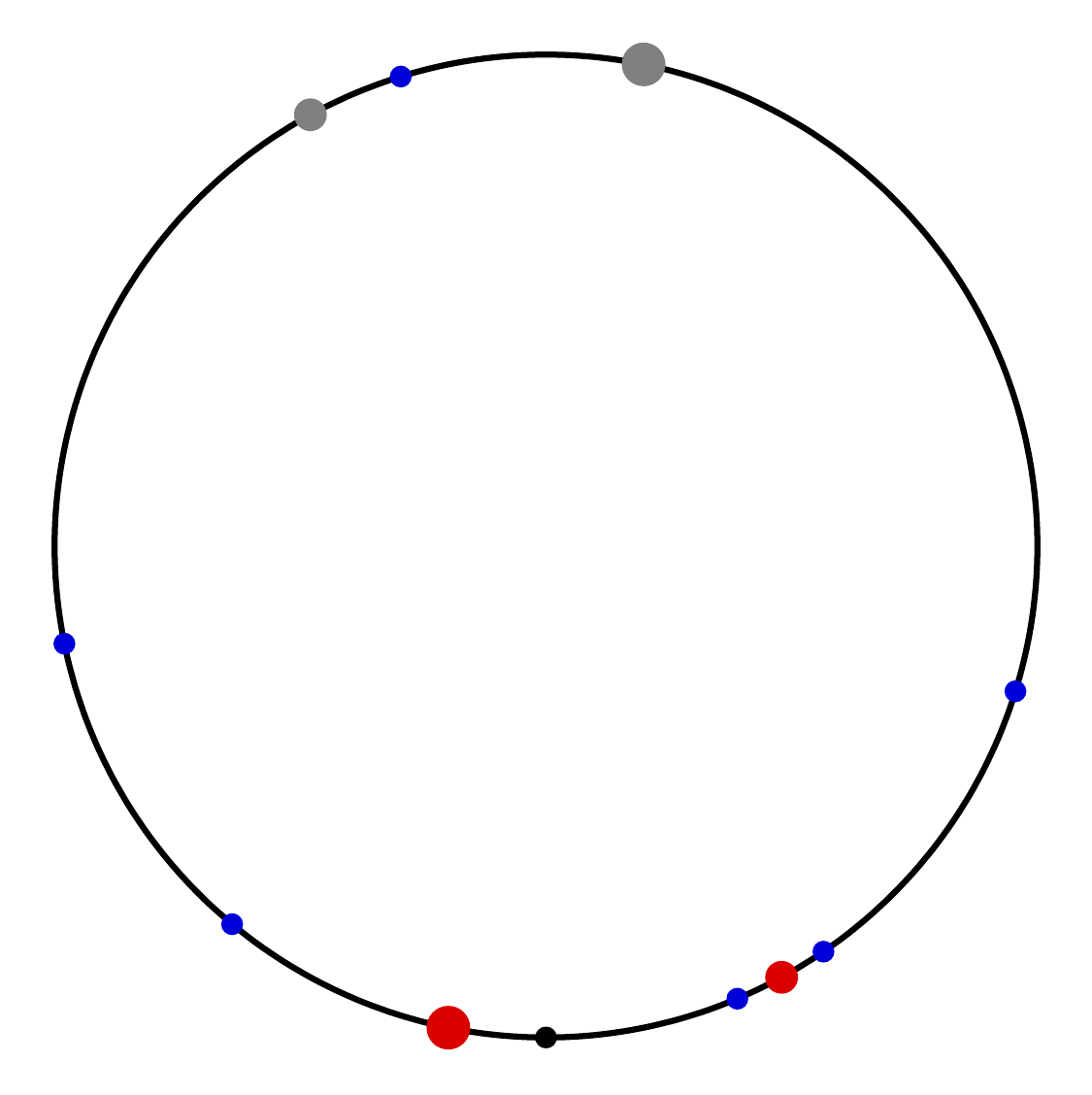}
    \put(39,10){\footnotesize $\mathbf{p}_{1}$}
    \put(68,14){\footnotesize $\mathbf{p}_{2}$}
    \put(57,89){\footnotesize $\mathbf{a}_{1}$}
    \put(27,85){\footnotesize $\mathbf{a}_{2}$}
    \put(57,1){$\mathcal{I}_{1}$}
    \put(46,-1){$C$}
    \put(96,60){$\mathcal{I}_{2}$}
    \put(-2,48){$\mathcal{I}_{3}$}
    \put(39,97){$\mathcal{I}_{4}$}
     \put(85,80){ $\mathcal{G}^+$}
    \put(66,3){\footnotesize $\Ll_{1}$}
    \put(76,8){\footnotesize $\Ll_{2}$}
    \put(95,34){\footnotesize $\Ee_{2}$}
    \put(16,12){\footnotesize $\Ll_{3}$}
    \put(-2,40){\footnotesize $\Ee_{3}$}
    \put(31,95){\footnotesize $\Aa_{1}$}
    
\end{overpic}
\caption{The collinear RE in Theorem \ref{thm:numbercollinearRE}.
}
\label{fig:circlecollinearRE}
\end{figure}

The proof of Theorem~\ref{thm:numbercollinearRE} consists of estimating the 
number of roots of the functions $\theta\mapsto f_j(\theta;\kappa,\mu)$, $j=1,\dots, 4$,  as functions 
of the parameters $(\kappa, \mu)$, using techniques from elementary calculus. The details are  provided in
Sec.~\ref{app:ProofExistenceCollinear} of Appendix~\ref{app:analytic-proofs}. A useful auxiliary result is the following
proposition, which imposes a strong restriction  the number of collinear RE that may exist in  the segments $\mathcal{I}_{2}$ and $\mathcal{I}_{3}$  for all parameter values $(\kappa,\mu)\in \Ps{1}$. This result will also be used in the implementation of CAPs 
 in Sec.~\ref{sss:CAP-existence-collinear}.

\begin{proposition}
\label{prop:concavity}
    The number of collinear RE in each of the segments $\mathcal{I}_{2}$ and $\mathcal{I}_{3}$ is either  zero, one or two for any $(\kappa,\mu)\in\Ps{1}$. Furthermore, the following statements hold for  fixed $(\kappa,\mu)\in \Ps{1}$:
    \begin{enumerate}
     \item Suppose there exists $\hat \theta\in \mathcal{I}_{2} $ such that  $f_2(\hat \theta;\kappa,\mu)>0$. Then, there exist precisely two collinear RE on $\mathcal{I}_2$ for the parameter values $(\kappa,\mu)$. 
    \item 
    Suppose there exists $\hat \theta\in \mathcal{I}_{3}$ such that  $f_3(\hat \theta;\kappa,\mu)<0$. Then, there exist precisely two collinear RE on $\mathcal{I}_3$ for the parameter values of $(\kappa,\mu)$. 
\end{enumerate}
\end{proposition}

\begin{proof}
Recall that collinear RE on $\mathcal{I}_{2}$ are in one-to-one correspondence with zeros of  the function $\theta\mapsto f_2(\theta;\kappa,\mu)$, whose explicit form is
\begin{equation}
\label{eq:f2}
    f_{2}(\theta;\kappa,\mu) = \frac{1}{2}\sin(2\theta) - \Gamma \left[\frac{1}{\sin^{2}(\theta+\sqrt{\kappa}q_{1})} + \frac{\mu}{\sin^{2}(\theta-\sqrt{\kappa}q_{2})} \right], \qquad \theta\in\mathcal{I}_{2}=(\sqrt{\kappa} q_{2}, \pi - \sqrt{\kappa} q_{1}).
\end{equation}
Considering that $\Gamma>0$ and that $\sin(2\theta)\leq 0$ for $\theta\in [\pi/2,\pi]$, it is clear that $f_2$ is strictly negative on   $(\pi/2, \pi - \sqrt{\kappa} q_{1})\subset \mathcal{I}_{2}$. In particular,  the roots
of $f_2$ can only occur on the complementary interval $(\sqrt{\kappa} q_{2}, \pi/2 ]\subset \mathcal{I}_{2}$. The second derivative of $f_{2}$ is given by
\begin{equation*}
    f_{2}''(\theta;\kappa,\mu) = -2\sin(2\theta) - 2\Gamma \left[ \frac{1+2\cos^{2}(\theta+\sqrt{\kappa}q_{1})}{\sin^{4}(\theta+\sqrt{\kappa}q_{1})} + \mu\,\frac{1+2\cos^{2}(\theta-\sqrt{\kappa}q_{2})}{\sin^{4}(\theta-\sqrt{\kappa}q_{2})}\right],
\end{equation*}
and it is clear that  $f_{2}''(\theta;\kappa,\mu)<0$ for all $\theta\in(\sqrt{\kappa} q_{2}, \pi/2)$ so $f_2$ is convex down in this
interval. Therefore, given that a continuous function which does not change concavity can have at most $2$ roots, we conclude that there
exist at most two collinear RE on $\mathcal{I}_2$.  Finally, considering that  $f_{2}(\theta;\kappa,\mu) \to -\infty$ as $\theta \to \sqrt{\kappa} q_{2}^{+}$ and $f_2(\theta;\kappa,\mu)$ is strictly
negative for all $\theta\in(\sqrt{\kappa} q_{2}, \pi/2)$, the existence
of $\hat \theta\in \mathcal{I}_2$ such that $f_2(\hat\theta;\kappa,\mu)>0$ implies that $\theta\mapsto f_2(\theta;\kappa,\mu)$ has exactly two roots by the intermediate value
theorem.

The conclusions for $\mathcal{I}_{3}$ are obtained analogously and
follow from the expressions
\begin{equation}
\label{eq:f3}
    f_{3}(\theta;\kappa,\mu) = \frac{1}{2}\sin(2\theta) + \Gamma \left[ \frac{1}{\sin^{2}(\theta+\sqrt{\kappa}q_{1})} + \frac{\mu}{\sin^{2}(\theta-\sqrt{\kappa}q_{2})} \right], \qquad \theta\in  \mathcal{I}_{3} = (-\pi + \sqrt{\kappa} q_{2},\, - \sqrt{\kappa} q_{1}),
\end{equation}
and
\begin{equation*}
    f_{3}''(\theta;\kappa,\mu) = -2\sin(2\theta) + 2\Gamma \left[ \frac{1+2\cos^{2}(\theta+\sqrt{\kappa}q_{1})}{\sin^{4}(\theta+\sqrt{\kappa}q_{1})} + \mu\,\frac{1+2\cos^{2}(\theta-\sqrt{\kappa}q_{2})}{\sin^{4}(\theta-\sqrt{\kappa}q_{2})}\right],
\end{equation*}
that allow us to conclude that $f_3$ is strictly positive  for $\theta\in (-\pi+\sqrt{\kappa} q_2,-\pi/2)\subset \mathcal{I}_3$ and is concave up in the complementary interval  $\theta\in[-\pi/2, -\sqrt{\kappa}q_{1})$. The arguments are the same as above, using this time  that $f_{3}(\theta;\kappa,\mu)\to \infty$ as $\theta \to -\sqrt{\kappa}q_1^-$.
\end{proof}

\subsection{Computer-Assisted Proofs}
\label{sss:CAP-existence-collinear}

We now extend the classification results of  Theorem~\ref{thm:numbercollinearRE} to a larger parameter region by proving, using CAPs, that equations  \eqref{eq:CollinearRECondbyIntervals} admit a precise number
 of roots using the techniques described in Appendix~\ref{App:CAPs}. For the intervals $\mathcal{I}_2$ and
 $\mathcal{I}_3$, we also take advantage of Proposition~\ref{prop:concavity}.
  For example, proving
 using interval arithmetic that $f_2(\hat \theta;\hat \kappa,\hat \mu)>0$ for
 some $\hat \theta\in (\sqrt{\hat \kappa}q_1(\hat\kappa,\hat \mu),\pi/2)$ implies that there
 exist exactly 2 roots of $f_2$ on an enclosure of the parameter values 
$(\hat\kappa,\hat \mu)$.

The results of our CAPs are illustrated in Fig.~\ref{fig:CAPexistencecollinear} in the parameter range 
$\Ps{\mu_{\Tt}}$. The different panels indicate the  collinear RE that  exist on each of the segments $\mathcal{I}_j$
for given parameter values. 
The vertical strip shaded in white on the left of each panel corresponds
to the parameter region for which Theorem \ref{thm:numbercollinearRE} applies. This information
is essential to describe the continuation of the RE as functions of $\kappa$ 
since the  CAPs are inconclusive near $\kappa=0$.

\renewcommand{\arraystretch}{1.2}
\setlength{\tabcolsep}{6pt}

The CAPs illustrated in Fig. \ref{fig:CAPexistencecollinear}
allow us to draw the following
conclusions about the RE in Theorem~\ref{thm:numbercollinearRE}:
\begin{itemize}
    \item The points $\Ll_1\in \mathcal{I}_1$ and $\Aa_1\in \mathcal{I}_4$ continue from the parameter region
    specified in Theorem~\ref{thm:numbercollinearRE} to the one
    in blue in panels (a) and (d) respectively.
    \item The pairs of points $\Ll_2, \Ee_2\in \mathcal{I}_2$ and  
    $\Ll_3, \Ee_3\in \mathcal{I}_3$ continue from the parameter
    region specified in Theorem~\ref{thm:numbercollinearRE} into the 
    the orange region in panel (b) and the left-hand orange region  in panel (c), respectively.
    As $\kappa$ increases, each pair undergoes a saddle-node bifurcation and disappears
    upon reaching  the
    parabola-like curve in the $(\kappa,\mu)$ plane,
    that separates the orange and gray regions.  
    \item Interestingly, two RE emerge
    again on $\mathcal{I}_3$ for larger values of $\kappa$, corresponding
    to the orange region on the right-hand side of panel (c). We denote 
    these RE as  $\widetilde{\Ll}_3, \widetilde{\Ee}_3\in \mathcal{I}_3$ with the 
    analogous convention that $\widetilde{\Ll}_3$ is closer to the primaries.
\end{itemize}

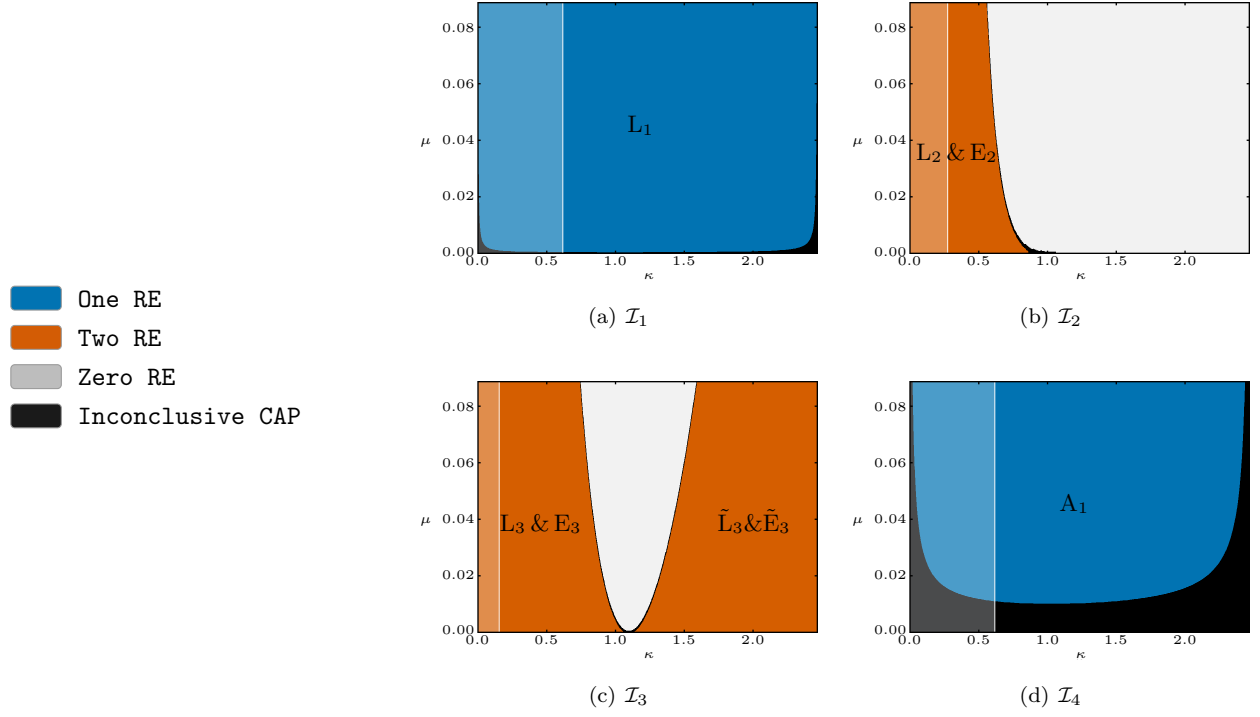
\begin{figure}[H]
\centering
\makebox[\textwidth][c]{

\begin{minipage}[c]{0.28\textwidth}
\centering
\begin{tabular}{@{}l@{\hspace{0.6em}}l@{}}
\legenditem{0173B2}{One RE} \\
\legenditem{D35C05}{Two RE} \\
\legenditem{BDBDBD}{Zero RE} \\
\legenditem{1A1A1A}{Inconclusive CAP} \\
\end{tabular}
\end{minipage}
\hspace{0.05\textwidth}

\begin{minipage}[c]{0.68\textwidth}
\centering
\subfigure[$\mathcal{I}_{1}$]{
\begin{overpic}[width=0.47\textwidth]{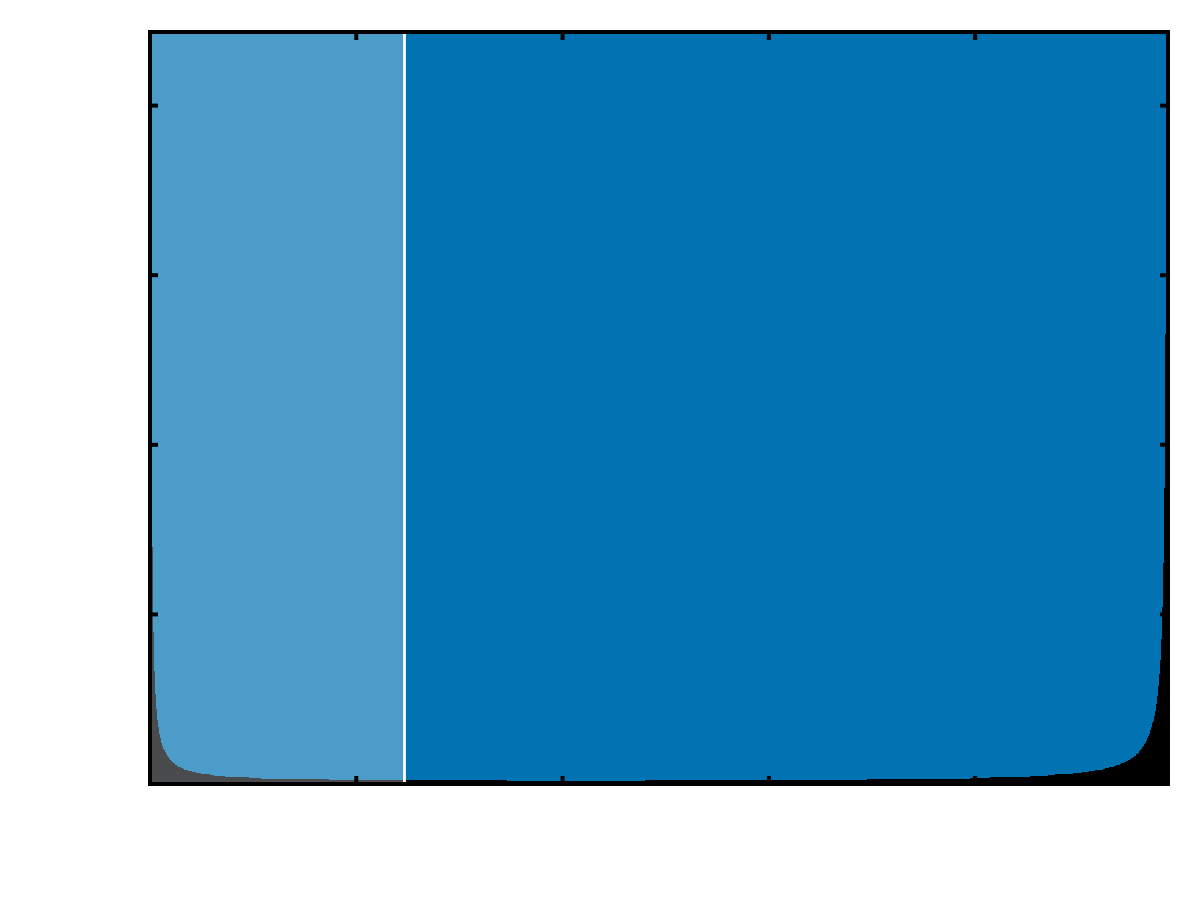}
    \put(9.4,6.5){\tiny 0.0}
    \put(26.7,6.5){\tiny 0.5}
    \put(44,6.5){\tiny 1.0}
    \put(61.2,6.5){\tiny 1.5}
    \put(78.5,6.5){\tiny 2.0}
    \put(54,3){\tiny $\kappa$}
    \put(3.5,9){\tiny 0.00}
    \put(3.5,23){\tiny 0.02}
    \put(3.5,37){\tiny 0.04}
    \put(3.5,51){\tiny 0.06}
    \put(3.5,65){\tiny 0.08}
    \put(-2,37){\tiny $\mu$}
     \put(50,40){\small $\Ll_{1}$}
\end{overpic}}\hspace{0.1cm}
\subfigure[$\mathcal{I}_{2}$]{
\begin{overpic}[width=0.47\textwidth]{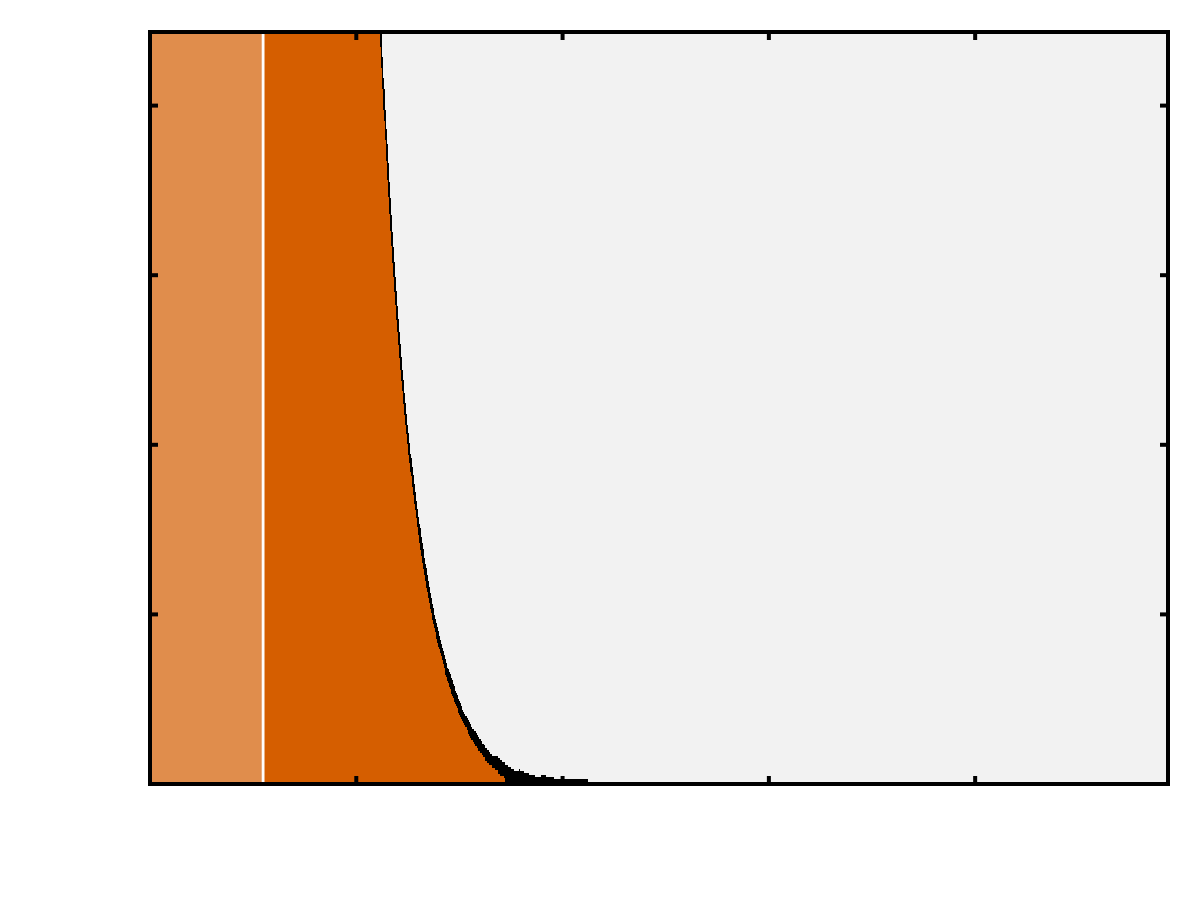}
    \put(9.4,6.5){\tiny 0.0}
    \put(26.7,6.5){\tiny 0.5}
    \put(44,6.5){\tiny 1.0}
    \put(61.2,6.5){\tiny 1.5}
    \put(78.5,6.5){\tiny 2.0}
     \put(54,3){\tiny $\kappa$}
    \put(3.5,9){\tiny 0.00}
    \put(3.5,23){\tiny 0.02}
    \put(3.5,37){\tiny 0.04}
    \put(3.5,51){\tiny 0.06}
    \put(3.5,65){\tiny 0.08}
    \put(-2,37){\tiny $\mu$}
    \put(14,33){\small $\Ll_{2} \& \Ee_{2}$}
\end{overpic}}
\\[0.5em]
\subfigure[$\mathcal{I}_{3}$]{
\begin{overpic}[width=0.47\textwidth]{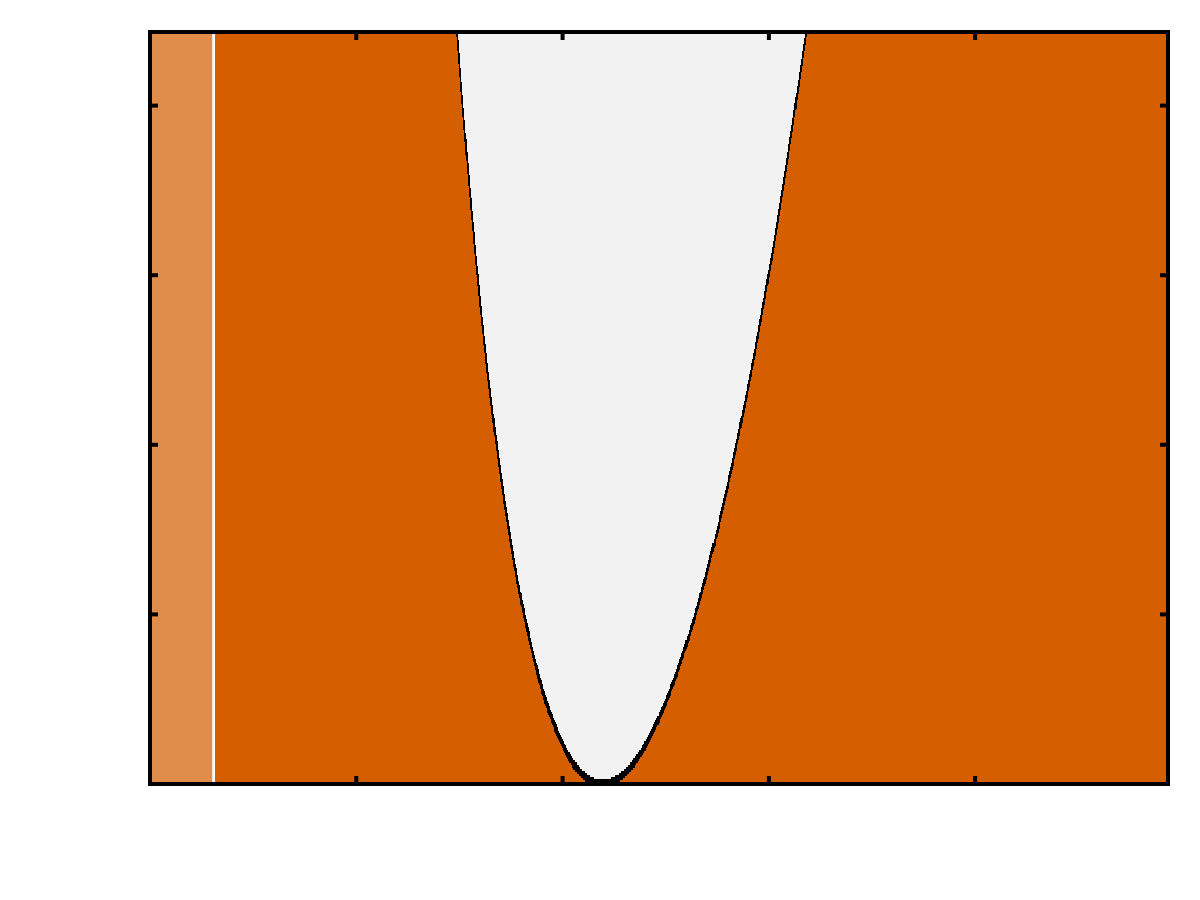}
    \put(9.4,6.5){\tiny 0.0}
    \put(26.7,6.5){\tiny 0.5}
    \put(44,6.5){\tiny 1.0}
    \put(61.2,6.5){\tiny 1.5}
    \put(78.5,6.5){\tiny 2.0}
     \put(54,3){\tiny $\kappa$}
    \put(3.5,9){\tiny 0.00}
    \put(3.5,23){\tiny 0.02}
    \put(3.5,37){\tiny 0.04}
    \put(3.5,51){\tiny 0.06}
    \put(3.5,65){\tiny 0.08}
    \put(-2,37){\tiny $\mu$}
    \put(18,35){\small $\Ll_{3}  \& \Ee_{3}$}
    \put(72.5,35){\small $\tilde{\Ll}_{3}  \& \tilde{\Ee}_{3}$}
\end{overpic}}\hspace{0.1cm}
\subfigure[$\mathcal{I}_{4}$]{
\begin{overpic}[width=0.47\textwidth]{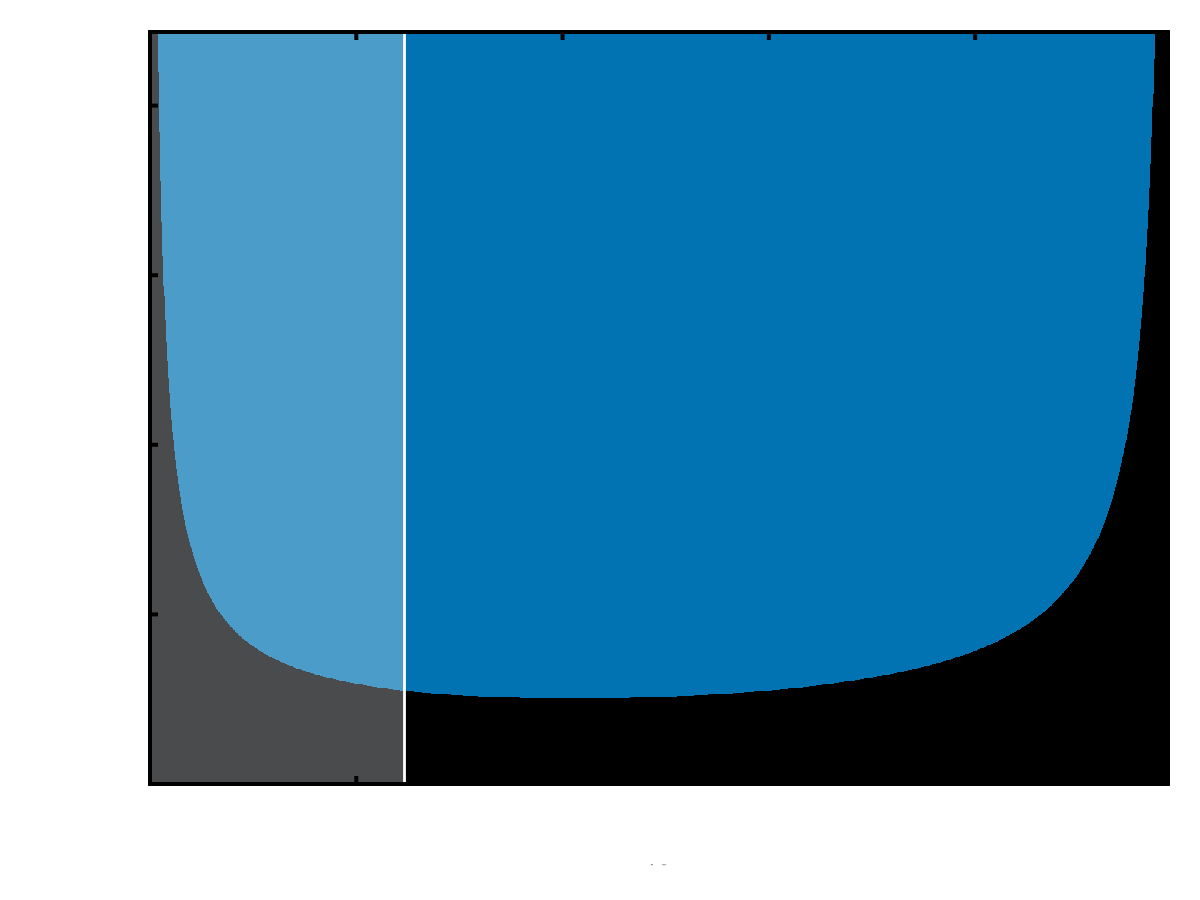}
    \put(9.4,6.5){\tiny 0.0}
    \put(26.7,6.5){\tiny 0.5}
    \put(44,6.5){\tiny 1.0}
    \put(61.2,6.5){\tiny 1.5}
    \put(78.5,6.5){\tiny 2.0}
     \put(54,3){\tiny $\kappa$}
    \put(3.5,9){\tiny 0.00}
    \put(3.5,23){\tiny 0.02}
    \put(3.5,37){\tiny 0.04}
    \put(3.5,51){\tiny 0.06}
    \put(3.5,65){\tiny 0.08}
    \put(-2,37){\tiny $\mu$}
    \put(50,40){\small $\Aa_{1}$}
\end{overpic}}
\end{minipage}
}
\captionsetup{width=0.75\textwidth}
\caption{CAPs for the existence of collinear RE in the parameter region $\Ps{\mu_{\Tt}}$. The different panels prove the existence of the collinear RE discussed in the text for each interval $\mathcal{I}_j$, $j=1,\dots,4$. The labelling of the corresponding RE is indicated in
the body of the figure. The shaded vertical stripe on the left side of each graph indicates the range of parameter values for which the existence results were proved analytically in Theorem~\ref{thm:numbercollinearRE}.}
\label{fig:CAPexistencecollinear}
\end{figure}

\subsection{Numerical investigations}
\label{sss:collinearexistencenumerics}

The different panels in Fig.~\ref{fig:CAPexistencecollinear}  show that the CAPs we implemented  
are inconclusive for $(\kappa,\mu)$ close to the boundaries of the parameter region $\Ps{\mu_{\Tt}}$, and  close to bifurcations. In principle, such problematic regions can be diminished by taking a finer grid. We did not pursue the most  optimal refinement. For completeness, the left panel in Fig.~\ref{fig:numexistencecollinear} presents numerical results for the existence of collinear relative equilibria in which we also illustrate the results for the negative curvature case (proved analytically  in \cite{MS17}). This reproduces 
some results of Kilin for $\kappa>0$ (see \cite[Figure 6]{Ki99}).  The right panel 
in Fig.~\ref{fig:numexistencecollinear} shows, in white, the region of  parameter region  $\Ps{\mu_{\Tt}}$ where the numerical results have been rigorously
proved above, either analytically or with a CAP.

\begin{figure}[H]
\centering
\makebox[\textwidth][c]{%
\begin{minipage}[c]{0.2\textwidth}
\centering
\begin{tabular}{@{}l@{\hspace{0.8em}}l@{}}
\multicolumn{2}{c}{}\\
\legenditem{E8DF34}{Six RE} \\
\legenditem{50B7E9}{Four RE} \\
\legenditem{80CEB8}{Three RE} \\
\legenditem{D35C05}{Two RE} \\
\end{tabular}
\end{minipage}
\hspace{0.01\textwidth}
\begin{minipage}[c]{0.35\textwidth}
\centering
\begin{overpic}[width=\textwidth]{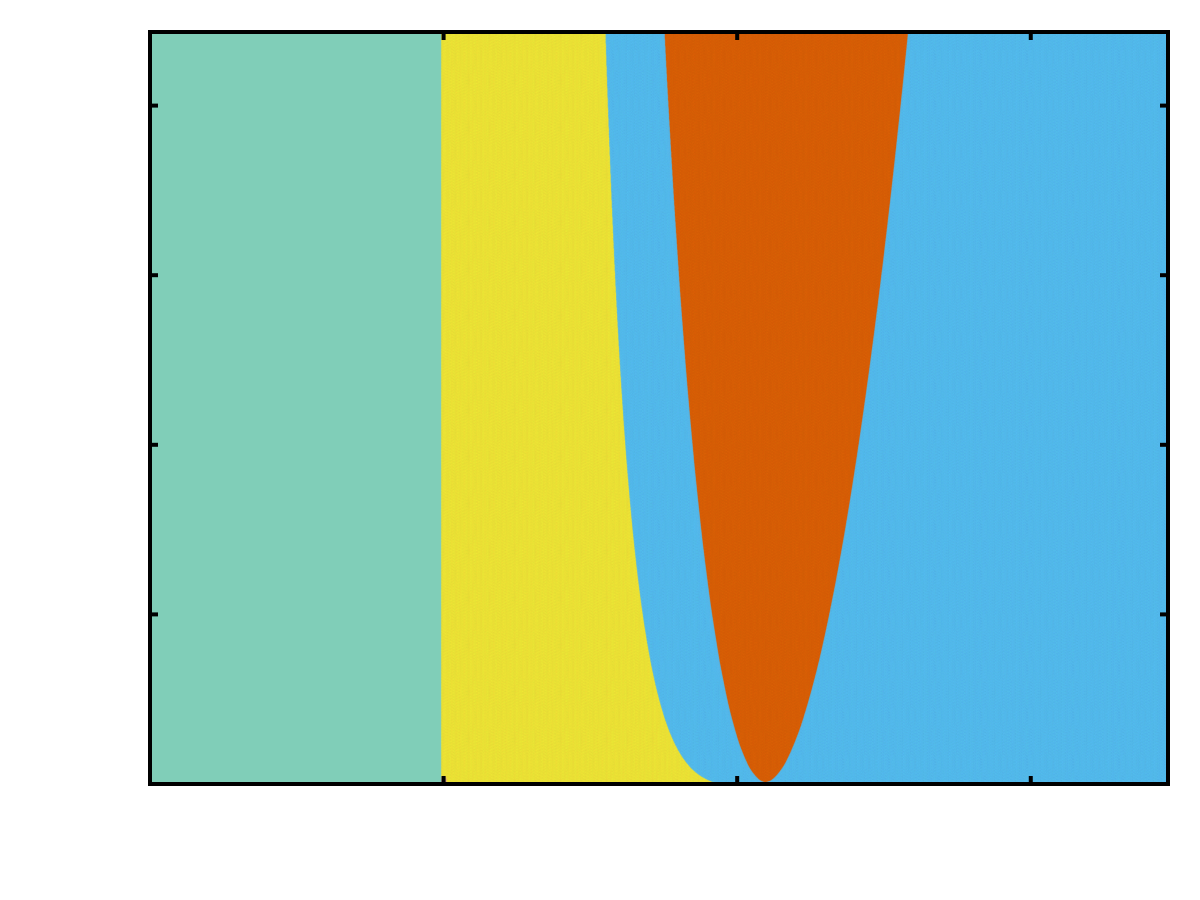}
    \put(49,40){\footnotesize{\makebox(0,0){$b_{1}$}}}
    \put(61,40){\footnotesize{\makebox(0,0){$b_{2}$}}}
    \put(75,40){\footnotesize{\makebox(0,0){$b_{3}$}}}
    \put(11,6.5){\tiny -1}
    \put(36,6.5){\tiny 0}
    \put(60.5,6.5){\tiny 1}
    \put(85,6.5){\tiny 2}
     \put(54,3){\tiny $\kappa$}
    \put(3.5,9){\tiny 0.00}
    \put(3.5,23){\tiny 0.02}
    \put(3.5,37){\tiny 0.04}
    \put(3.5,51){\tiny 0.06}
    \put(3.5,65){\tiny 0.08}
    \put(-2,37){\tiny $\mu$}
\end{overpic}
\end{minipage}
\hspace{0.25cm}
\begin{minipage}[c]{0.35\textwidth}
\centering
\begin{overpic}[width=\textwidth]{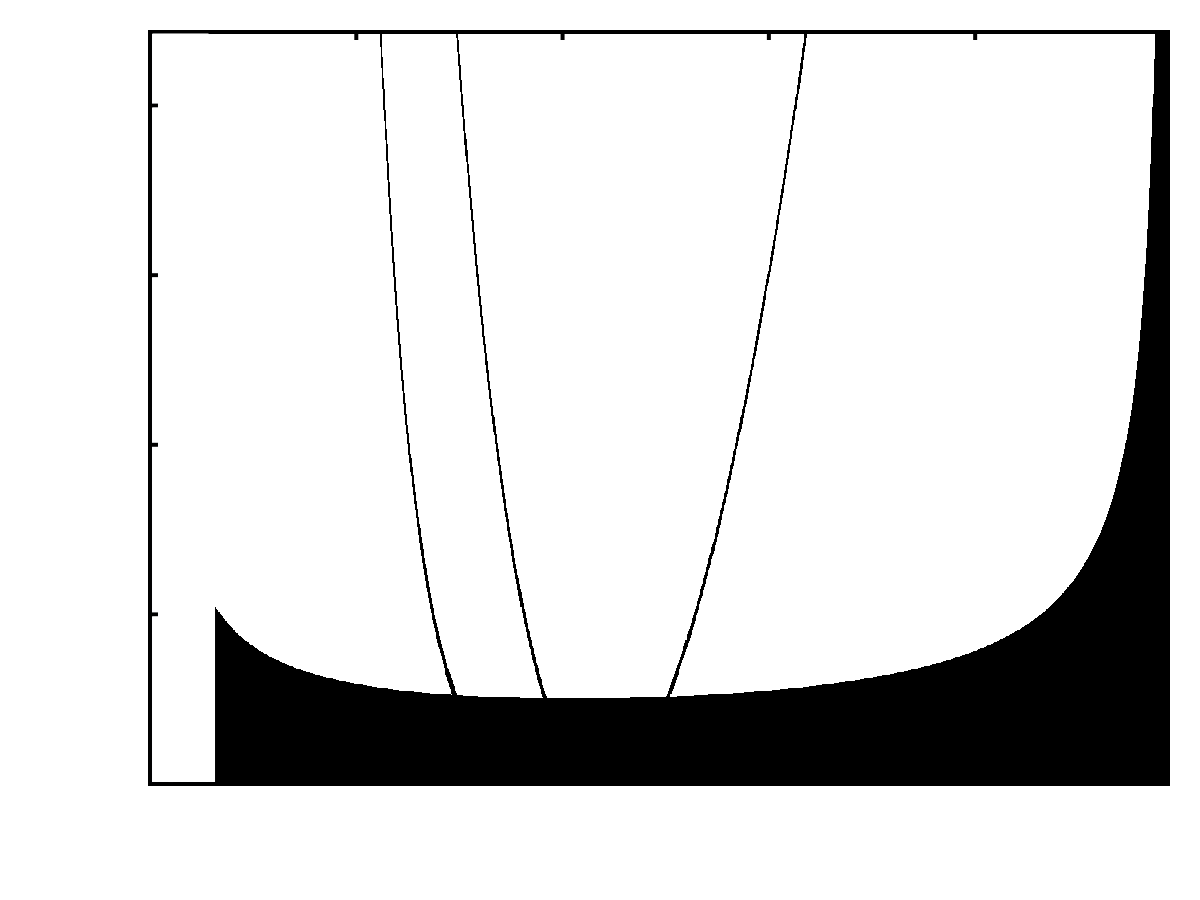}
    \put(11,6.5){\tiny -1}
    \put(36,6.5){\tiny 0}
    \put(60.5,6.5){\tiny 1}
    \put(85,6.5){\tiny 2}
     \put(54,3){\tiny $\kappa$}
    \put(3.5,9){\tiny 0.00}
    \put(3.5,23){\tiny 0.02}
    \put(3.5,37){\tiny 0.04}
    \put(3.5,51){\tiny 0.06}
    \put(3.5,65){\tiny 0.08}
    \put(-2,37){\tiny $\mu$}
\end{overpic}
\end{minipage}%
}
\captionsetup{width=0.75\textwidth}
\caption{Left panel: Numerical results for the existence of collinear relative equilibria in the $(\kappa,\mu)$ parameter space. 
Right panel: values of $(\kappa,\mu)\in \Ps{\mu_{\Tt}}$ (shown in white) for which the numerical results in the left panel have been rigorously proved, either analytically or using a CAP.
}
\label{fig:numexistencecollinear}
\end{figure}
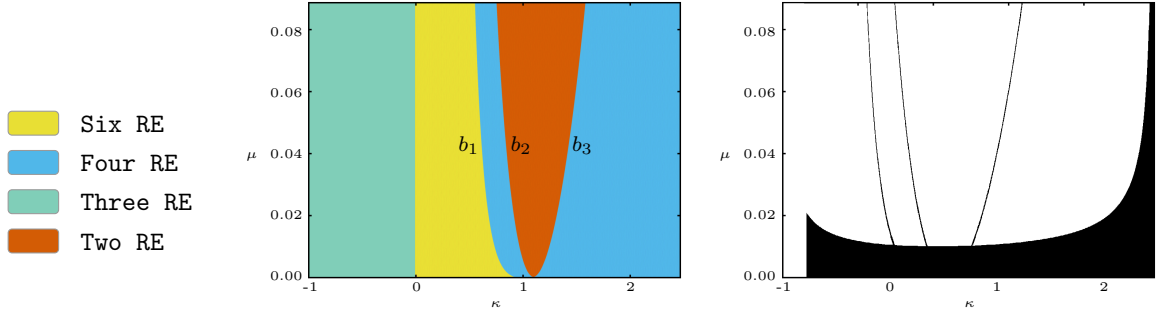

Fig.~\ref{fig:numexistencecollinear} suggests the existence of three bifurcation curves, that will be denoted
 $b_j$, $j=1,2,3$, which may be parametrized as
 $$\mu\mapsto b_j(\mu)= (\kappa_j(\mu),\mu), \qquad  j=1,2,3,$$
separating the   different colored regions for $\kappa>0$ shown in the left panel. These curves are  clearly
defined in the right panel for values of $\mu$ which are not too close to zero. We label the curves from left to right, so that 
 $$\kappa_1(\mu)<\kappa_2(\mu)<\kappa_3(\mu).$$ Accordingly, and as indicated in the left panel of Fig.~\ref{fig:numexistencecollinear}, the curve $b_1$
 separates the yellow and blue regions; the curve
  $b_2$ separates the blue and orange regions, with the blue region lying to its left;  
   and the curve $b_3$ separates the orange and blue regions,  with the orange region lying to its left.

Fig.~\ref{fig:asympcollinear} shows a bifurcation diagram, obtained numerically,
 of the collinear RE
for a fixed value of $\mu\in (0,\mu_{\Tt})$ as a function of $\kappa$.
 In the figure, the vertical axis represents the signed  Riemannian distance $D$ to the center of rotation $C$ measured in trigonometric
 sense according to Fig.~\ref{fig:circlecollinearRE}. 
 In particular, the diagram  illustrates how the Riemannian 
distance from each of $\Ee_2$,  $\Ee_3$, and $\Aa_1$ to the center of 
rotation tends to infinity as $\kappa\to 0$, while  $\Ll_1$,  $\Ll_2$ and $\Ll_3$
converge to finite values which agree with the collinear Lagrange points of the planar problem. This behavior will be proved 
by asymptotic expansions in Sec.~\ref{sec:asymp}.
The diagram also describes stability properties which are considered in detail in Sec.~\ref{s:stability-collinear}.

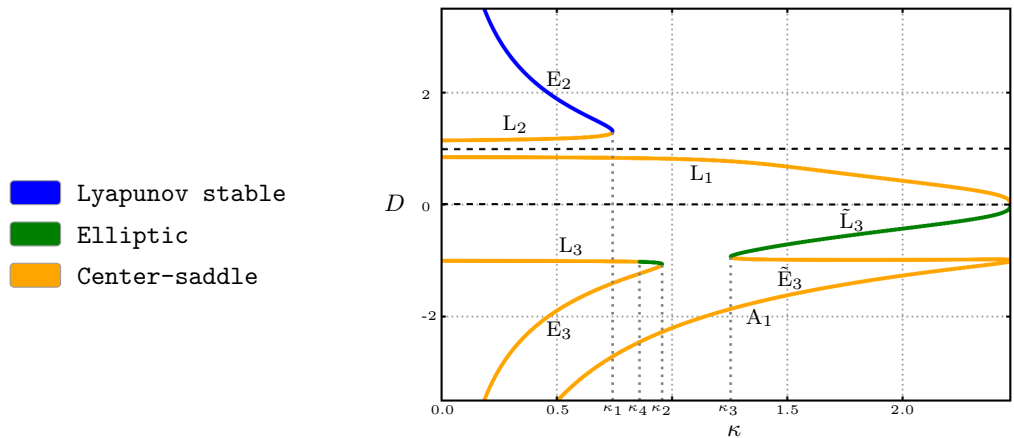
\begin{figure}[H]
\centering
\makebox[\textwidth][c]{
\begin{minipage}[c]{0.28\textwidth}
\centering
\begin{tabular}{@{}l@{\hspace{0.6em}}l@{}}
\legenditem{0000FF}{Lyapunov stable} \\
\legenditem{008000}{Elliptic} \\
\legenditem{FFA500}{Center-saddle} \\
\end{tabular}
\end{minipage}
\hspace{0.05\textwidth}
\begin{minipage}[c]{0.5\textwidth}
\centering
\begin{overpic}[width=\textwidth]{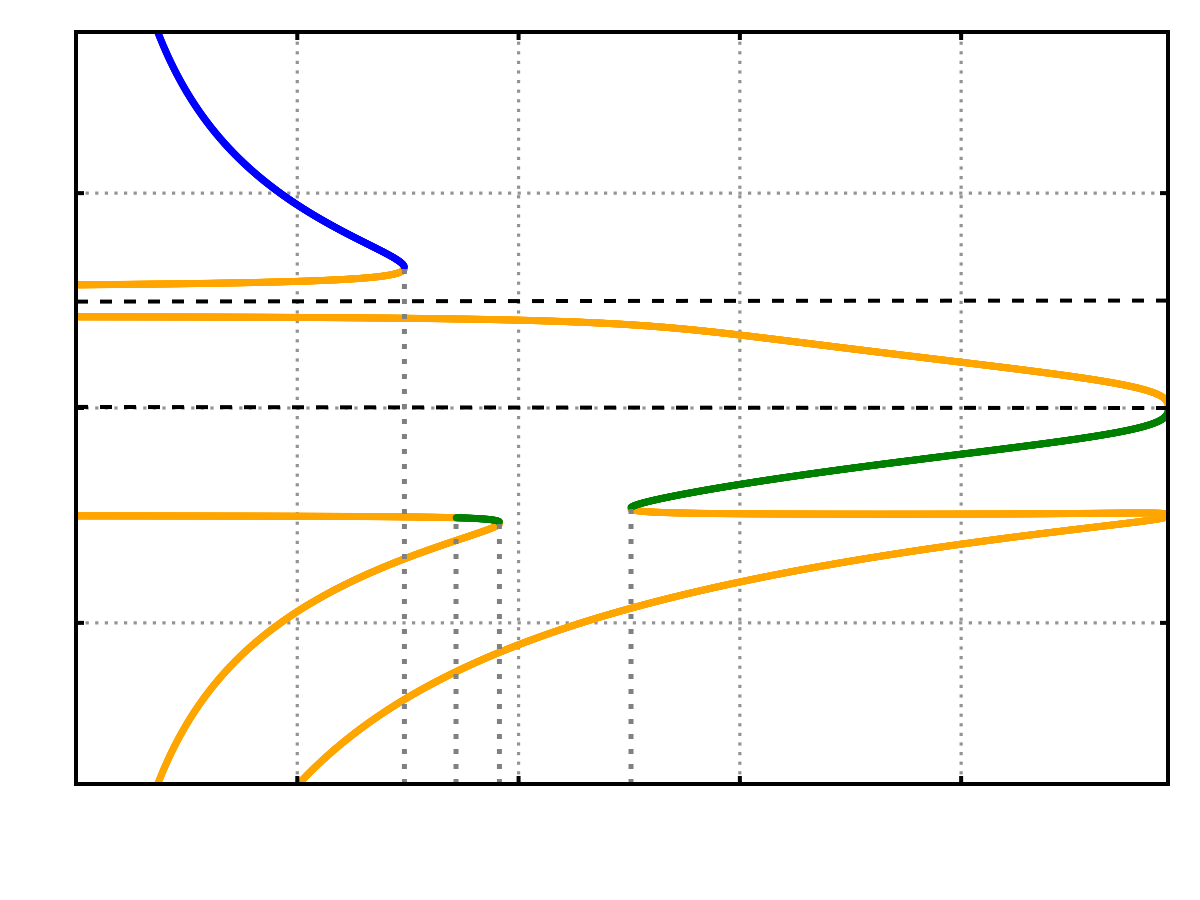}
    \put(32,8){\tiny $\kappa_{1}$}
    \put(36,8){\tiny $\kappa_{4}$}
    \put(39.75,8){\tiny $\kappa_{2}$}
    \put(50.5,8){\tiny $\kappa_{3}$}
     \put(4.5,7.5){\tiny 0.0}
    \put(23,7.5){\tiny 0.5}
    \put(59.8,7.5){\tiny 1.5}
    \put(78.2,7.5){\tiny 2.0}
     \put(52,4){\small $\kappa$}  
    \put(3,22.5){\tiny -2}
    \put(3,40){\tiny 0}
    \put(3,58){\tiny 2}
    \put(-3,40){\small $D$}
    \put(46,45){\footnotesize $\Ll_{1}$}
    \put(16,53){\footnotesize $\Ll_{2}$}
    \put(23,60){\footnotesize $\Ee_{2}$}
    \put(25,33.5){\footnotesize $\Ll_{3}$}
    \put(23,20){\footnotesize $\Ee_{3}$}
    \put(70,37.25){\footnotesize $\tilde{\Ll}_{3}$}
    \put(60,28){\footnotesize $\tilde{\Ee}_{3}$}
    \put(55,22){\footnotesize $\Aa_{1}$}
\end{overpic}
\end{minipage}
}
\captionsetup{width=0.75\textwidth}
\caption{Bifurcation diagram of collinear RE for $\mu\approx 0.01$. 
The black dotted lines indicate the signed distances from the primaries to the center of rotation.}
\label{fig:asympcollinear}
\end{figure}

\begin{remark}
The bifurcation 
curves $b_1$, $b_2$, $b_3$,   in the left panel of Fig.~\ref{fig:numexistencecollinear}, meet at the limiting point $\kappa =(\pi/3 )^2, \; \mu=0$, where an intricate bifurcation takes place. That this is actually the  point where the three curves meet can be shown as follows.  The bifurcation curve $b_1$ is determined by the condition that the function  $\theta\mapsto f_2(\theta; \kappa_1(\mu),\mu)$ admits a double root, which we denote $\theta=\theta_1^*(\mu)$. This yields the system
\begin{equation}
\label{eq:bif-f2}
    f_{2}\big(\theta_1^*(\mu);\kappa_1(\mu),\mu\big)=0, 
    \qquad 
    f_{2}'\big(\theta_1^*(\mu);\kappa_1(\mu),\mu\big)=0.
\end{equation}
In the limit $\mu\to 0$, this system can be solved explicitly, yielding $\kappa_1(0)=(\pi/3)^2$.  The same limiting value may be obtained working with the bifurcation curves  $b_2$ and $b_3$, which involve $f_3$. In the standard
spherical model for the space of constant positive curvature, this corresponds to  the primaries subtending an angle of $\pi/3$. This is consistent with figures 6 and 13 in \cite{Ki99}.
\end{remark}

\begin{remark}
\label{rmk:isoscelesRE}
Fujiwara and P\'erez-Chavela \cite{FuPC24} recently discovered a remarkable family of collinear isosceles relative equilibria in the full
 three--body problem on the sphere, for which the angle between the central body and each of the other two bodies in the spherical model is independent of the masses
and equals
\begin{equation}
\label{eq:tau0}
\tau_{0}=\frac{1}{2}\arccos\!\left(\frac{\sqrt{2}-1}{2}\right).
\end{equation}
  Our numerical investigations indicate that, when $\kappa=\tau_{0}^{2}$, the distance between $\Ll_{3}$ and $\mathbf{p}_{1}$ is $\tau_{0}$, independently of $\mu$. This suggests that the family of relative equilibria found in \cite{FuPC24} converges to $\Ll_{3}$ in the limit case where one of the masses tends to zero. 
 \end{remark}

\subsection{\texorpdfstring
  {Existence of collinear RE as a function of $\kappa>0$ for fixed $\mu\in(0,\mu_{\Tt})$.}
  {Existence of collinear RE as a function of kappa > 0 for fixed mu in (0, mu_Tt).}}

Below we summarize the rigorous results 
and numerical observations given in the previous sections,  in a single statement describing  
  the existence classification  of collinear relative equilibria as a function of $\kappa>0$ for fixed $\mu\in (0,\mu_{\Tt})$.
We formulate it as conjecture, since we do not have a rigorous proof that accounts for all the details (see Remark~\ref{rmk:collineardetails}). It may 
be useful to refer to Fig.~\ref{fig:circlecollinearRE} to recall the position of the RE with respect to the primaries,
their antipodal points, and to each other.

\begin{conjecture}[Existence classification of collinear RE as a function of $\kappa>0$  for  $0<\mu<\mu_{\Tt}$.]
\label{thm:classcollinear}
Let $\mu \in (0,\mu_{\Tt})$. Then, there exist three values $\kappa_{i}= \kappa_{i}(\mu) \in (0,\pi^{2}/4)$, with $\kappa_{1} < \kappa_{2} < \kappa_{3}$, such that the R3BP with mass ratio $\mu$ admits exactly the following number of collinear relative equilibria depending on the value of $\kappa>0$.
\begin{description}
    \item[For $\bm{0 < \kappa < \kappa_{1}}$:] six collinear RE, $\Ll_1$, $\Ll_2$, $\Ll_3$, $\Ee_2$, $\Ee_3$, $\Aa_1$;
    \item[For $\bm{\kappa_{1} < \kappa < \kappa_{2}}$] four collinear RE,  $\Ll_1$,  $\Ll_3$,  $\Ee_3$, $\Aa_1$;
    \item[For $\bm{\kappa_{2} < \kappa < \kappa_{3}}$:] two collinear RE,  $\Ll_1$, $\Aa_1$;
     \item[For $\bm{\kappa_{3} < \kappa < \pi^{2}/4}$:] four collinear RE, $\Ll_1$,  $\widetilde{\Ll}_3$,  $\widetilde{\Ee}_3$, $\Aa_1$;
\end{description}
Each value of $\kappa_i$, corresponds to a saddle-node bifurcation 
which leads to disappearance or appearance of a pair of RE. More precisely, 
\begin{description}
    \item[$\bullet$] $\Ll_2$ and  $\Ee_2$ coalesce and disappear at $\kappa=\kappa_1$;
    \item[$\bullet$] $\Ll_3$ and  $\Ee_3$ coalesce and disappear at $\kappa=\kappa_2$;
    \item[$\bullet$] $\widetilde{\Ll}_3$ and  $\widetilde{\Ee}_3$ emerge at $\kappa=\kappa_3$.     
\end{description}
Furthermore, the distance from each of $\Ee_2$, $\Ee_3$ and $\Aa_1$ to the center of rotation tends to infinity as $\kappa\to 0^+$,
whereas $\Ll_1$, $\Ll_2$, $\Ll_3$ converge to the standard Lagrange points of the planar problem in this limit.
\end{conjecture}

\begin{remark}
\label{rmk:collineardetails}
The details that remain to be proved in Conjecture~\ref{thm:classcollinear} are: 
\begin{enumerate}
\item[1.] Complete the existence proofs
for  parameter values $(\kappa,\mu)\in \Ps{\mu_{\Tt}}$ to which Theorem~\ref{thm:numbercollinearRE} does not apply and for which the CAPs in Fig.~\ref{fig:CAPexistencecollinear} are inconclusive;
\item[2.] Prove the existence of the 
bifurcation curves $b_1, b_2, b_3$.
\end{enumerate}
To address point 1, it is natural to increase the precision of the CAPs and formulate an extension of Theorem~\ref{thm:numbercollinearRE} that applies to the
the bottom and right  boundaries of  $\Ps{\mu_{\Tt}}$. Point 2 could be approached by applying CAPs to validate branches of solutions of
bifurcation systems such as system \eqref{eq:bif-f2}. We did not pursue either of these directions, as doing so 
would have substantially increased the length of the paper.
\end{remark}

\section{Classification of Triangular RE}
\label{sec:triangular}

As in the collinear case, the classification of triangular  RE for negative curvature was established by Mart\'inez and Sim\'o~\cite{MS17}.

\begin{proposition}\cite[Proposition 4.5]{MS17}[Classification of Triangular RE for Negative Curvature.]
\label{th:existence-neg-triangular}
For any mass ratio $\mu \in  (0,1)$ and curvature $\kappa < 0$, the R3BP with negative curvature admits exactly two triangular relative equilibria, that will be denoted by $\Ll_{4}$ and $\Ll_{5}$.
\end{proposition}

This section  focuses on the positive curvature case 
$\kappa>0$. The numerical investigations of \cite{Ki99,MS17} indicate that
the number and type of RE depends in a subtle manner on the parameters $\kappa$, $\mu$.
To the best of our knowledge, no rigorous results are known away
from the limit cases $\mu=1$ and $\mu=0$. We shall establish Theorem \ref{thm:numbertriangularRE}, valid for all $\mu\in (0,1)$, which guarantees that,
for sufficiently small  $\kappa>0$, the only triangular RE are continuations of $\Ll_4$ and
$\Ll_5$. This result is complemented with CAPs in Sec.~\ref{ss:CAPExistenceTriangular}
implemented for $(\kappa,\mu)\in \Ps{\mu_{\Tt}}$.
Our analysis is based on the discussion in Sec.~\ref{sec:triangularbalanced} 
in which necessary and sufficient conditions for existence of triangular RE
were given in Proposition~\ref{prop:triangularconditionpositive} in terms of the distance coordinates $(d_1,d_2)$.

\subsection{Analytical Results}
\label{ss:AnalyticalExistenceTriangular}

By their non-degeneracy, the classical Lagrange points $\Ll_4$, $\Ll_5$ for $\kappa=0$, persist
 as RE of the curved R3BP for small $\kappa>0$. The following result asserts
that, independently of the mass ratio $\mu$, for sufficiently small $\kappa>0$, there do not exist any other triangular RE
apart from $\Ll_4$, $\Ll_5$. 
 \begin{theorem}
\label{thm:numbertriangularRE}
Let $\mu\in(0,1)$ be fixed. There exists  $\kappa^{*}>0$ such that the curved R3BP admits exactly two triangular RE for $0<\kappa<\kappa^{*}$, which will be denoted by $\Ll_4$, $\Ll_5$.
\end{theorem}

 A limitation of the theorem is that it does not give an estimate on the value of $\kappa^*$
nor an indication on how it may depend on $\mu$. 

Below we describe the strategy of the proof of Theorem~\ref{thm:numbertriangularRE}, as it is also central to the 
 implementation of CAPs in Sec.~\ref{ss:CAPExistenceTriangular}. The details of the proof are presented in 
Sec.~\ref{app:existencetriangular} of Appendix~\ref{app:analytic-proofs}.

Recall, from Proposition \ref{prop:triangularconditionpositive}, that triangular relative equilibria correspond to solutions $(d_{1},d_{2})\in W_{\kappa}$ of Eqs.~\eqref{eq:partialdistancescoordspositive} which we rewrite here as
\begin{equation}
\label{eq:partialdistancescoordspositive2}
\begin{gathered}
    \sin^{3}(d_{1}) \sin(\!\sqrt{\kappa}q_{2}) \big(\sin(\!\sqrt{\kappa}q_{2})\cos(d_{1})+\sin(\!\sqrt{\kappa}q_{1})\cos(d_{2})\big) = \Gamma\sin^{2}(\sqrt{\kappa}), \\
    \sin^{3}(d_{2}) \sin(\!\sqrt{\kappa}q_{1}) \big(\sin(\!\sqrt{\kappa}q_{2})\cos(d_{1})+\sin(\!\sqrt{\kappa}q_{1})\cos(d_{2})\big) = \mu \, \Gamma\sin^{2}(\sqrt{\kappa}),
\end{gathered}
\end{equation}
where the dependence of $\Gamma$, $q_1$, $q_2$ on $\kappa$ and $\mu$ is omitted.
As mentioned in Sec. \ref{ss:existenceCollinear}, each solution $(d_1,d_2)$ of this system,
actually corresponds to two distinct triangular RE
since these arise in symmetric pairs by
the reflectional symmetry~\eqref{eq:sym} of the amended potential $\mathcal{V}_{\kappa,\mu}$.

We also recall from Proposition \ref{prop:triangularconditionpositive},
that every solution  $(d_{1},d_{2})\in W_{\kappa}$ of
Eq.~\eqref{eq:partialdistancescoordspositive2} determines a triangular-balanced-configuration. Namely, by 
Eq.~\eqref{eq:mubalancecondition}, any such  $(d_{1},d_{2})\in W_{\kappa}$  satisfies
\begin{equation}
\label{eq:mubalanced2}
    \sin(d_{1})=\Lambda\sin(d_{2}),
\end{equation}
where we have  abbreviated $\Lambda=\Lambda_{\kappa,\mu}^{+}$.

Condition \eqref{eq:mubalanced2} alone defines a curve,   denoted by $\tilde{\gamma}_{\kappa,\mu}\subset W_k$, 
consisting of  two connected components. These components can be distinguished by the conditions $d_{1}<\pi/2$ and $d_{1}>\pi/2$ (see Fig. \ref{fig:graphpos}). 
The following lemma shows that any solution of Eqs.~\eqref{eq:partialdistancescoordspositive2} necessarily 
lies on the component satisfying $d_{1}<\pi/2$.

\begin{lemma}
\label{lemma:one_component}
If $(d_1,d_2)\in W_\kappa$ is a solution of Eqs.\eqref{eq:partialdistancescoordspositive2}, then 
 $d_1<\frac{\pi}{2}$. 
In particular, all triangular RE belong to the connected component of the curve 
$\tilde{\gamma}_{\kappa,\mu}\subset W_k$ of triangular-balanced-configurations satisfying  $d_{1}<\pi/2$.
\end{lemma}

\begin{proof}
Considering that $0<\Lambda<1$ (see Eq.~\eqref{eq:inequalityLambda}) and $0<q_1<q_2<\frac{\pi}{2}$ we have
\begin{equation*}
    \frac{\cos(d_{2})}{\sqrt{\cos^{2}(d_{2})+(1-\Lambda^{2})\sin^{2}(d_{2})}}<1<\frac{\sin(\!\sqrt{\kappa}q_{2})}{\sin(\!\sqrt{\kappa}q_{1})}.
\end{equation*}

Suppose now that $(d_{1},d_{2})\in W_\kappa$ satisfies Eqs.~\eqref{eq:partialdistancescoordspositive2} 
so in particular Eq.~\eqref{eq:mubalanced2} holds. Then, 
a simple manipulation of the denominator of the left hand side of the above inequality, using Eq.~\eqref{eq:mubalanced2}, yields
\begin{equation*}
    \frac{\cos(d_{2})}{|\cos(d_{1})|}=\frac{\cos(d_{2})}{\sqrt{1-\sin^2(d_1)}}<\frac{\sin(\!\sqrt{\kappa}q_{2})}{\sin(\!\sqrt{\kappa}q_{1})}.
\end{equation*}

Suppose now that   $d_1\geq\frac{\pi}{2}$. Then  $\cos(d_{1})\leq0$, and the above inequality implies
\begin{equation*}
    \sin(\!\sqrt{\kappa}q_{2})\cos(d_{1})+\sin(\!\sqrt{\kappa}q_{1})\cos(d_{2})\leq 0.
\end{equation*}
This quantity appears as a factor in the left-hand side of each of the equations in the system~\eqref{eq:partialdistancescoordspositive2}. All other terms in these equations, both on the left- and right-hand sides, are positive. Therefore, the system  cannot admit a solution. 

This shows that  $d_1<\frac{\pi}{2}$ is a necessary condition for $(d_1,d_2)\in W_\kappa$ to be a solution of Eqs.~\eqref{eq:partialdistancescoordspositive2}.
\end{proof}

Lemma~\ref{lemma:one_component} allows us to 
solve Eq.~\eqref{eq:mubalanced2} for $d_1$ using the appropriate branch of $\arcsin$. This yields
a parametrization  of  the connected component of the curve of triangular-balanced-configurations satisfying $d_1<\frac{\pi}{2}$. Namely, these configurations may be parametrized by 
\begin{equation} 
\label{eq:auxiliardistancecoords}
    d_{1}=\arcsin\left(\Lambda\sin(x)\right), \quad d_{2}=x, 
\end{equation} 
where the real parameter $x$ belongs to the interval $\mathcal{I}_{\kappa,\mu}=  [a_{\kappa,\mu}, b_{\kappa,\mu}]$, with
\begin{equation}
\label{eq:auxintervaltriangular}
      a_{\kappa,\mu}= \arctan\!\left( \frac{\sin(\!\sqrt{\kappa})}{\Lambda + \cos(\!\sqrt{\kappa})} \right), \qquad b_{\kappa,\mu}= \pi-\arctan\!\left(\frac{\sin(\!\sqrt{\kappa})}{\Lambda -\cos(\!\sqrt{\kappa})}\right),
\end{equation}
to guarantee that $(d_1,d_2)\in W_k$.
Using this
parametrization, it is an exercise to show that 
the solutions $(d_1,d_2)\in W_k$ of Eqs.~\eqref{eq:partialdistancescoordspositive2} are in one-to-one correspondence with the roots of the real valued function $g_{\kappa,\mu} : \mathcal{I}_{\kappa,\mu} \to \mathbb{R}$, depending
parametrically on $(\kappa,\mu)$, and defined by
\begin{equation}
\label{eq:auxfunctiontriangular}
    g_{\kappa,\mu}(x)
    = \Lambda^{3}\sin(\sqrt{\kappa}q_{2})\sin^{3}(x)
    \left[\sin(\sqrt{\kappa}q_{1})\cos(x)
        + \sin(\sqrt{\kappa}q_{2})\sqrt{1-\Lambda^{2}\sin^{2}(x)}\right]
    - \Gamma\sin(\sqrt{\kappa}).
\end{equation}

The proof of Theorem~\ref{thm:numbertriangularRE}   consists of showing that 
 $g_{\kappa,\mu}$ admits exactly one root in $\mathcal{I}_{\kappa,\mu}$ when $\kappa$ is sufficiently small
 independently of the value of $\mu\in (0,1)$.
 The details are presented in  Sec.~\ref{app:existencetriangular} of Appendix~\ref{app:analytic-proofs}.

\subsection{Computer-Assisted Proofs}
\label{ss:CAPExistenceTriangular}

We complement the results of Theorem \ref{thm:numbertriangularRE} with  the CAPs
illustrated in the left panel of Fig.~\ref{fig:CAPexistencetriangular}. These  give 
the precise number of triangular RE for some specific
parameter values $(\kappa, \mu)\in  \Ps{\muT}$. These
CAPs were implemented with the approach described in Appendix~\ref{App:CAPs} to determine the 
exact number of roots of the function $g_{\kappa,\mu}:\mathcal{I}_{\kappa,\mu}\to \R$ by formula \eqref{eq:auxfunctiontriangular}.

The  left panel in Fig.~\ref{fig:CAPexistencetriangular} contains a 
a broad region of parameter values $(\kappa,\mu)\in \Ps{\muT}$, colored in orange, for which we prove the existence
of exactly  two triangular RE. In view of Theorem~\ref{thm:numbertriangularRE}, we conjecture
that these are continuations of $\Ll_4$ and  $\Ll_5$.  We do not have a rigorous proof of this because 
 the threshold value $\kappa^{*}$ in Theorem~\ref{thm:numbertriangularRE} is not explicitly known
and our CAPs have  only  been validated   for values   $\kappa\geq 5.075\times 10^{-5}$  (they are generally  inconclusive below
this value for $\mu\in (0,\mu_{\Tt})$).

We also  observe a 
thin blue stripe   in which 
we prove existence of exactly four triangular RE.
We will denote the additional two triangular RE 
$\Tt_1$ and $\Tt_2$. The figure  also shows a large
region of parameter values for which no triangular RE
exist.

The CAPs in the left panel of Fig.~\ref{fig:CAPexistencetriangular} are inconclusive for parameter values $(\kappa,\mu)$ near the boundary of $\Ps{\mu_{\Tt}}$ and near bifurcation curves where the number of triangular RE changes. The numerical investigations shown in the right panel of Fig.~\ref{fig:CAPexistencetriangular} suggest that the CAPs nevertheless capture the main qualitative features of the
classification. The numerical analysis suggests the existence of two bifurcation curves, which we denote
$b_4$ and $b_5$, and may be parametrized as
 $$\mu\mapsto b_j(\mu)= (\kappa_j(\mu),\mu), \qquad  j=4,5.$$
Our convention is that $\kappa_4(\mu)<\kappa_5(\mu)$, so $b_4$ is the curve separating the orange and blue regions,
while $b_5$ separates the orange and gray regions, as indicated in the figure.

\begin{figure}[H]
\centering
\makebox[\textwidth][c]{%
\begin{minipage}[c]{0.2\textwidth}
\centering
\begin{tabular}{@{}l@{\hspace{0.7em}}l@{}}
\multicolumn{2}{c}{}\\
\swatch{50B7E9} & \raisebox{-0.75ex}{\texttt{Four RE}} \\
\swatch{D35C05} & \raisebox{-0.75ex}{\texttt{Two RE}} \\
\swatch{BDBDBD} & \raisebox{-0.75ex}{\texttt{Zero RE}} \\
\legenditem{1A1A1A}{Inconclusive CAP}
\end{tabular}
\end{minipage}
\hspace{0.03\textwidth}
\begin{minipage}[c]{0.33\textwidth}
\centering
\begin{overpic}[width=\textwidth]{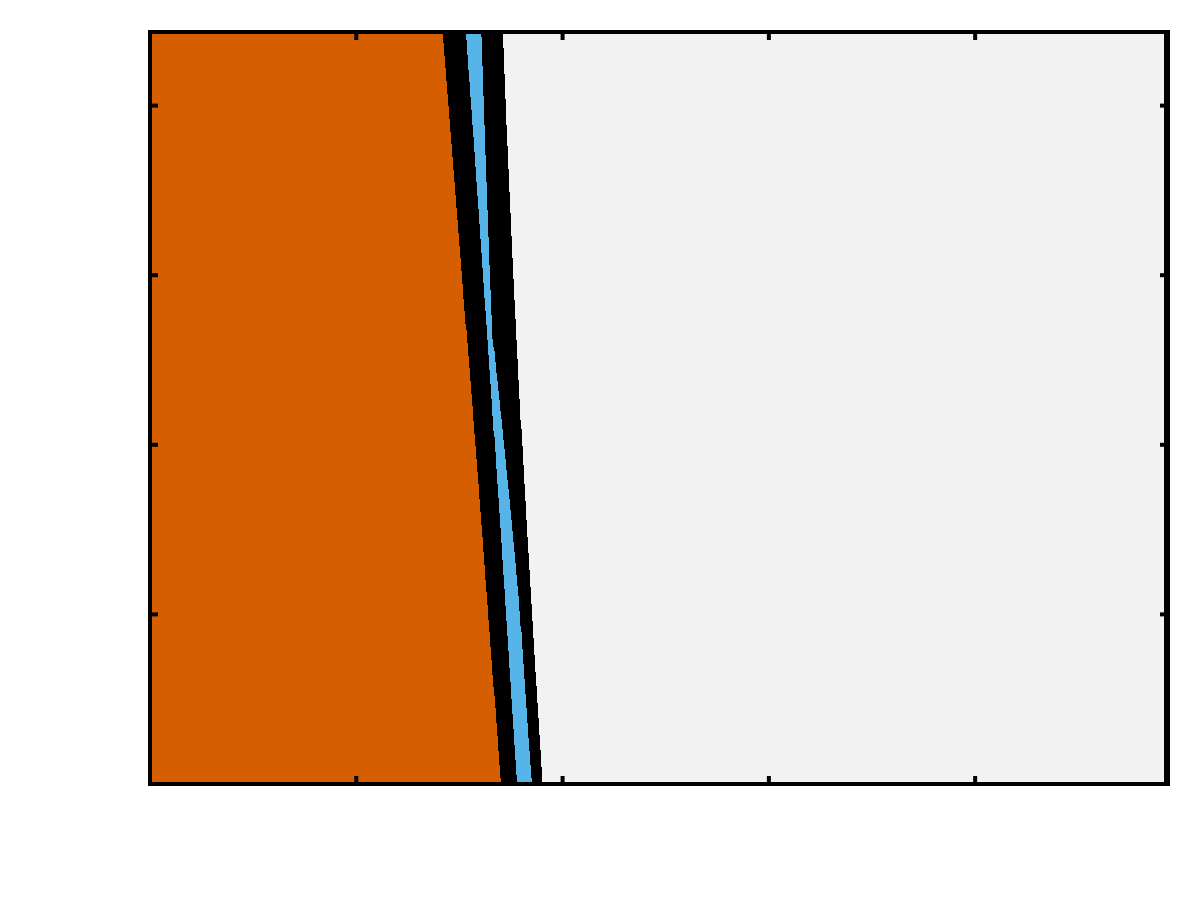}
    \put(9.4,6.5){\tiny 0.0}
    \put(26.7,6.5){\tiny 0.5}
    \put(44,6.5){\tiny 1.0}
    \put(61.2,6.5){\tiny 1.5}
    \put(78.5,6.5){\tiny 2.0}
     \put(54,3){\tiny $\kappa$}
    \put(3.5,9){\tiny 0.00}
    \put(3.5,23){\tiny 0.02}
    \put(3.5,37){\tiny 0.04}
    \put(3.5,51){\tiny 0.06}
    \put(3.5,65){\tiny 0.08}
    \put(-2,37){\tiny $\mu$}
\end{overpic}
\end{minipage}
\hspace{0.25cm}
\begin{minipage}[c]{0.33\textwidth}
\centering
\begin{overpic}[width=\textwidth]{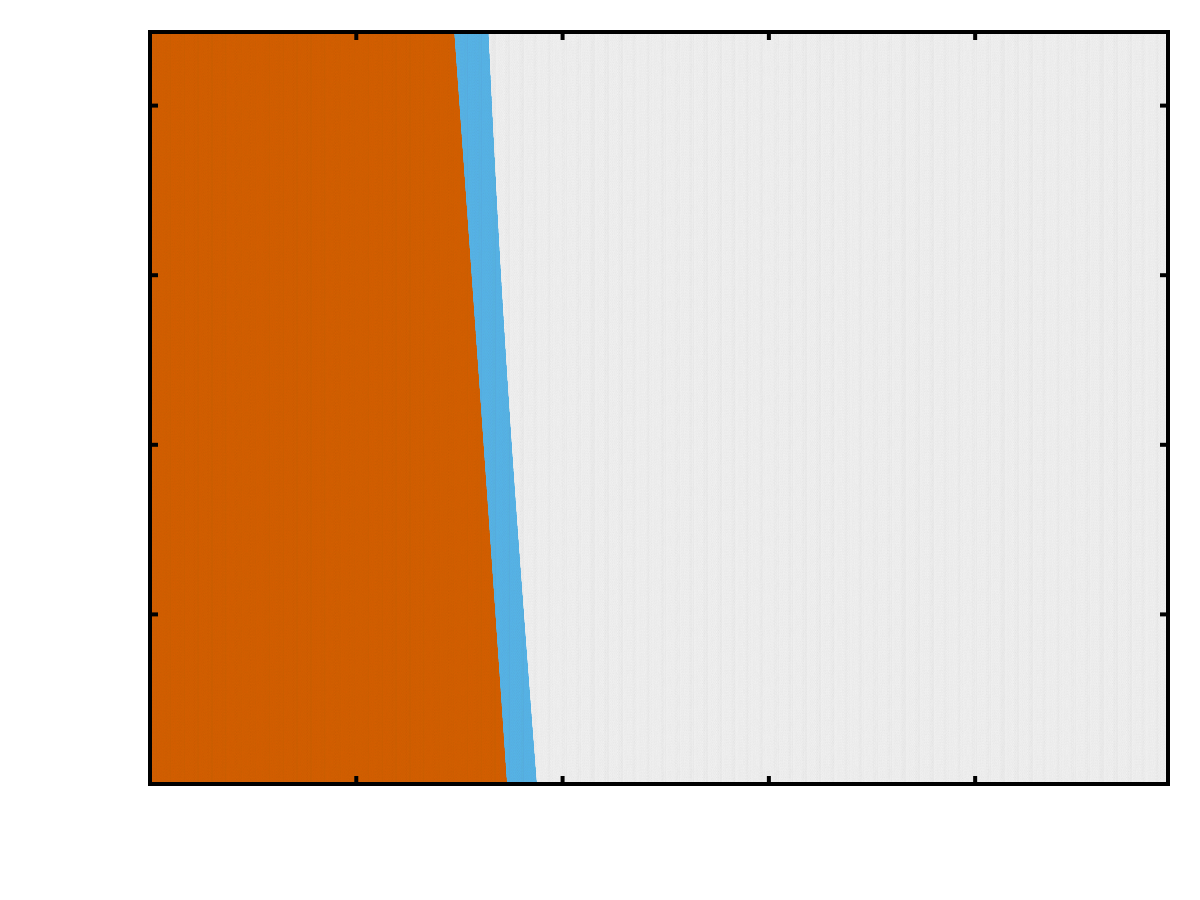}
    \put(34,35){\footnotesize{$b_4$}}
    \put(44,35){\footnotesize{$b_5$}}
    \put(9.4,6.5){\tiny 0.0}
    \put(26.7,6.5){\tiny 0.5}
    \put(44,6.5){\tiny 1.0}
    \put(61.2,6.5){\tiny 1.5}
    \put(78.5,6.5){\tiny 2.0}
     \put(54,3){\tiny $\kappa$}
    \put(3.5,9){\tiny 0.00}
    \put(3.5,23){\tiny 0.02}
    \put(3.5,37){\tiny 0.04}
    \put(3.5,51){\tiny 0.06}
    \put(3.5,65){\tiny 0.08}
    \put(-2,37){\tiny $\mu$}
\end{overpic}
\end{minipage}
}
\captionsetup{width=0.75\textwidth}
\caption{Number of triangular RE as a function of the parameters $(\kappa,\mu)\in \Ps{\muT}$. Left panel: computer-assisted proof. Right panel: numerical estimates.}
\label{fig:CAPexistencetriangular}
\end{figure}
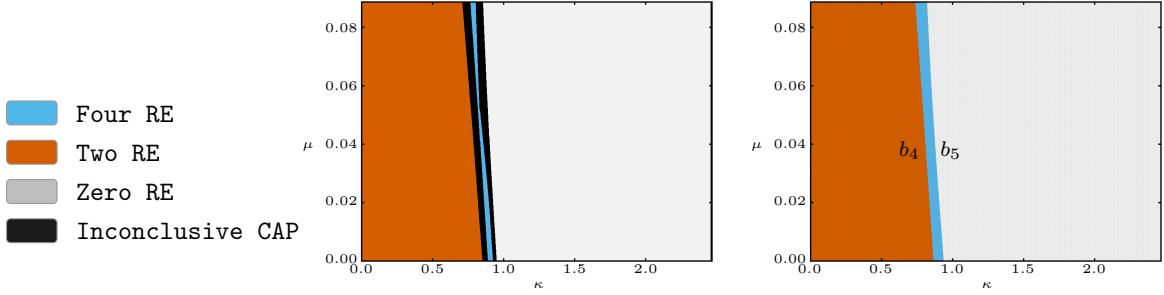

Fig.~\ref{fig:asymptriangular} shows a bifurcation diagram, obtained numerically,
 of triangular RE
for fixed values of $\mu\in (0,\mu_{\Tt})$ as a function of $\kappa$.
 In the figure, the vertical axis represents Riemannian distance, $D^{(2)}$, to the second primary located at $\mathbf{p}_{2}$.\footnote{In Sec.~\ref{sec:asymp}, the same quantity is denoted by $D_{\Ll_{4},{\bf p}_2}$.}
 The diagram  illustrates how the 
distance from each of $\Ll_4$ and $\Ll_5$ to $\mathbf{p}_{2}$ is convergent as  $\kappa\to 0$. Moreover, the limit value appears
to be $1$ suggesting convergence to the classical  equilateral Lagrange points of the planar problem. This behavior will be proved in
 Sec.~\ref{sec:asymp} using asymptotic analysis.
The stability properties indicated in the diagram will be discussed in Sec.~\ref{sec:stabilitytriangular}. The range of $\kappa$ in the diagram
extends to negative values of $\kappa$ to illustrate the destabilizing effect of negative curvature; see the discussion in Section~\ref{ss:stabilityL4L5} and Remark~\ref{rmk:kappasnegative}.

\begin{figure}[H]
\centering
\makebox[\textwidth][c]{%
\begin{minipage}[c]{0.22\textwidth}
\centering
\begin{tabular}{@{}l@{\hspace{0.8em}}l@{}}
\legenditem{008000}{Elliptic} \\
\legenditem{FFA500}{Center-saddle} \\
\legenditem{FF0000}{Complex saddle} \\
\end{tabular}
\end{minipage}
\hspace{0.03\textwidth}
\begin{minipage}[c]{0.33\textwidth}
\centering
\subfigure[$\mu\approx 0.01$]{
\begin{overpic}[width=\textwidth]{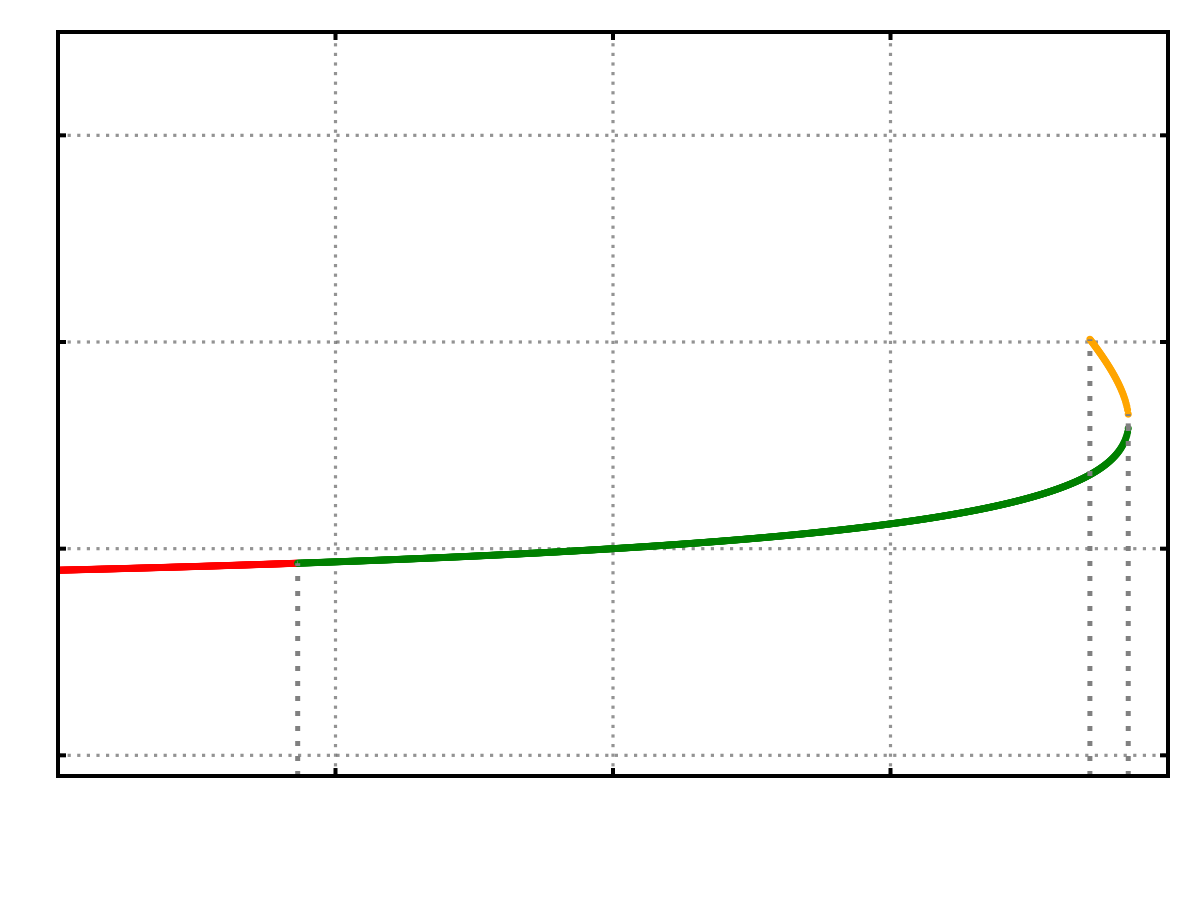}
    \put(22,8){\tiny $\kappa_{s}$}
    \put(88,8){\tiny $\kappa_{4}$}
    \put(93,8){\tiny $\kappa_{5}$}
    \put(2,11){\tiny 0}
    \put(2,28.2){\tiny 1}
    \put(2,45.5){\tiny 2}
    \put(2,62.7){\tiny 3}
    \put(-7,40){\tiny $D^{(2)}$}
    \put(48.5,7){\tiny 0.0}
     \put(71.5,7){\tiny 0.5}
     \put(50,3){\tiny $\kappa$}
     \put(48,24){\footnotesize $\Ll_{4,5}$}
     \put(87,48){\footnotesize $\Tt_{1,2}$}
\end{overpic}
}
\end{minipage}
\hspace{0.5cm}
\begin{minipage}[c]{0.33\textwidth}
\centering
\subfigure[$\mu\approx 0.06$]{
\begin{overpic}[width=\textwidth]{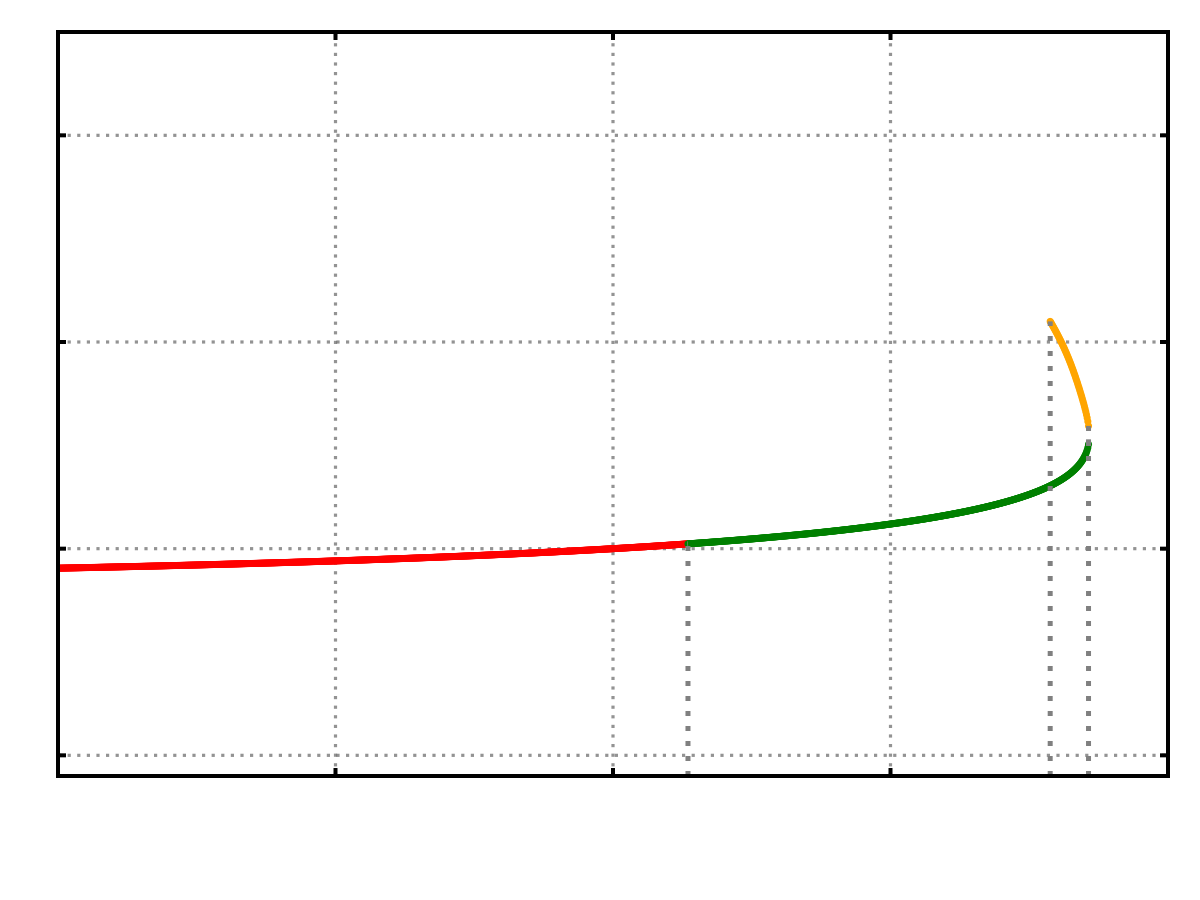}
    \put(54.5,8){\tiny $\kappa_{s}$}
    \put(84,8){\tiny $\kappa_{4}$}
    \put(89,8){\tiny $\kappa_{5}$}
     \put(2,11){\tiny 0}
    \put(2,28.2){\tiny 1}
    \put(2,45.5){\tiny 2}
    \put(2,62.7){\tiny 3}
    \put(-7,40){\tiny $D^{(2)}$}
     \put(24,7){\tiny -0.5}
    \put(48.5,7){\tiny 0.0}
     \put(71.5,7){\tiny 0.5}
     \put(50,3){\tiny $\kappa$}
     \put(48,24){\footnotesize $\Ll_{4,5}$}
     \put(83,50){\footnotesize $\Tt_{1,2}$}
\end{overpic}
}
\end{minipage}%
}
\captionsetup{width=0.75\textwidth}
\caption{Triangular relative equilibria. The value $D^{(2)}$ is the Riemannian distance to the primary $q_{2}$. In case (a), one observes how negative curvature has a destabilizing effect, whereas in case (b), positive curvature contributes to stabilization.}
\label{fig:asymptriangular}
\end{figure}
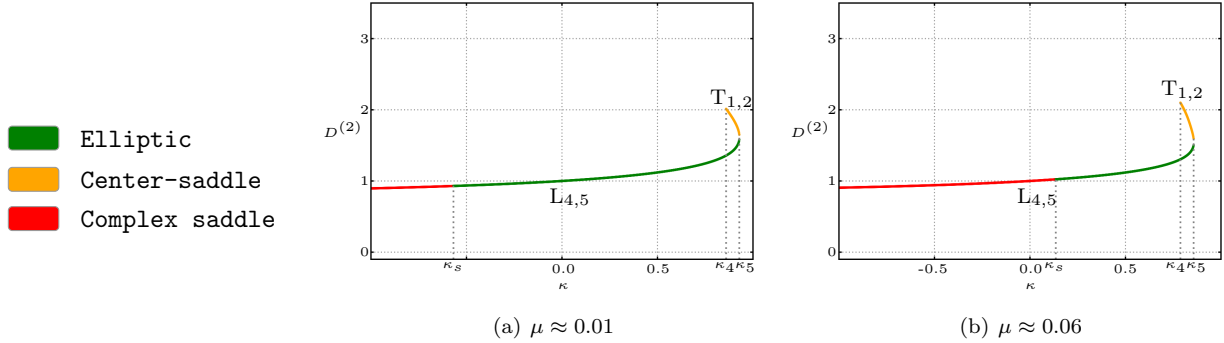

\subsection{\texorpdfstring
  {Existence of triangular RE as a function of $\kappa>0$ for fixed $\mu\in(0,\mu_{\Tt})$.}
  {Existence of triangular RE as a function of kappa > 0 for fixed mu in (0, mu_Tt).}}

The following statement summarizes our observations in the previous section regarding the classification of triangular RE for parameter values $(\kappa,\mu)\in \Ps{\muT}$.

\begin{conjecture}[Existence classification of triangular  RE as a function of $\kappa>0$  for  $0<\mu<\mu_{\Tt}$]
\label{thm:classtriangular}
Let $\mu \in (0,\mu_{\Tt})$. Then, there exist two   values of $\kappa$, denoted $\kappa_{4}(\mu)< \kappa_{5}(\mu) \in (0,\pi^{2}/4)$, such that the curved R3BP with mass ratio $\mu$ admits exactly the following number of triangular relative equilibria depending on the value of $\kappa>0$.
\begin{description}
    \item[For $\bm{0 < \kappa < \kappa_{4}}$:] two triangular RE, denoted $\Ll_4$ and $\Ll_5$, which are continuations of the classical Lagrange points. 
    \item[For $\bm{\kappa_{4} < \kappa < \kappa_{5}}$:] four triangular RE. Two of them are continuations of $\Ll_4$ and $\Ll_5$,  and the other two are  
    $\Tt_1$ and $\Tt_2$.
    \item[For $\bm{\kappa_{5} < \kappa < \pi^{2}/4}$:] zero triangular RE.
\end{description}
Moreover, $\Tt_1$ and $\Tt_2$ are created when $\Ll_3$ undergoes a subcritical pitchfork bifurcation at $\kappa=\kappa_4$. On the other hand, a simultaneous saddle-node bifurcation occurs at $\kappa=\kappa_5$, at which the pairs $(\Ll_4,\Tt_1)$ and $(\Ll_5,\Tt_2)$ collide and disappear.
\end{conjecture}

We formulated the classification above as a conjecture because we do not have a rigorous proof of the existence of the bifurcation curves 
$b_4$ and $b_5$, and because the CAPs in the left panel of Fig.~\ref{fig:CAPexistencetriangular} are inconclusive for some parameter values $(\kappa,\mu)\in \Ps{\muT}$. Our conjecture that $\Tt_1$ and $\Tt_2$ emerge from a pitchfork bifurcation of $\Ll_3$ is supported by the stability analysis of $\Ll_3$ presented in Section~\ref{s:stability-collinear}. We did not attempt to establish these details rigorously, in order to keep the paper at 
a reasonable length.

\section{Stability of collinear RE}
\label{s:stability-collinear}

We now turn to the study of the stability of  collinear RE.  We begin by recalling the results
for negative curvature established in \cite{MS17}.

\begin{proposition}[Stability of collinear RE, negative curvature {\cite[Proposition 6.1]{MS17}}]
\label{thm:stabilitycollinearneg}
For any $\kappa<0$ and any mass ratio $\mu\in (0,1)$, the collinear equilibrium points $\Ll_{1}$, $\Ll_{2}$ and $\Ll_{3}$ are center--saddles.
\end{proposition}

In this section, we focus on the case of positive curvature, where the behavior depends subtly on the parameters $\kappa$ and $\mu$. 
 To our knowledge, no rigorous results have been established for $\mu\in(0,1)$. Our analysis relies on the application of the stability criteria in 
Proposition~\ref{prop:stabilitycollpos} and allow us to determine the stability of the collinear RE whose existence was shown in Theorem~\ref{thm:numbercollinearRE}
and via the CAPs illustrated in Fig.~\ref{fig:CAPexistencecollinear}. We first present Theorem~\ref{thm:stabilitycollinear}, which is proved
analytically in Sec.~\ref{app:stability-collinear} of Appendix~\ref{app:analytic-proofs}, 
that establishes the stability properties of all collinear RE for small values of $\kappa$ and $\mu$.  These results are extended
to larger parameter values using CAPs in Subsection \ref{ss:CAPsCollinear}. Our conjectured stability classification is presented  in Sec.~\ref{conj:stability-collinear}.

\subsection{Analytical results}
\label{ss:analytical-collinear}

The following theorem  establishes
the stability properties of the collinear RE for small positive $\kappa$, $\mu$.

\begin{theorem}[Stability of collinear RE, positive curvature]
\label{thm:stabilitycollinear}
Let $\mu\in \left (0,\frac{1}{10} \right )$ be a fixed mass ratio. The following stability statements hold 
for the RE whose existence was established in Theorem~\ref{thm:numbercollinearRE}.
\begin{enumerate}
\item For $0<\kappa<\pi^{2}/16$,  $\Ll_{1}$ and $\Aa_{1}$ are center--saddles.
\item For $0<\kappa<\pi^{2}/36$,  $\Ll_{2}$ is a center--saddle and
 $\Ee_{2}$ is Lyapunov stable.
\item For $0<\kappa<\pi^{2}/64$,  $\Ll_{3}$ and $\Ee_{3}$ center--saddles. 
\end{enumerate}
In fact, interpreted as critical points of $\mathcal{V}_{\kappa,\mu}$, and in the parameter ranges specified above, 
the points $\Ll_{1}$, $\Ll_{2}$, $\Ll_{3}$, $\Ee_{3}$, and $\Aa_{1}$ are saddle-points, whereas
 $\Ee_{2}$ is a minimum.
\end{theorem}

The proof consists of estimating the quantities $\lambda_1$ and $\lambda_2$ in Proposition~\ref{prop:stabilitycollpos} at each RE. 
We show that, in the appropriate parameter range, they always have opposite signs, except in the case of $\Ee_2$, for which   both are positive. The conclusion
regarding their stability  then
follows directly from  Proposition~\ref{prop:stabilitycollpos}. Furthermore, in the proof of that proposition, it is shown that $\lambda_1$ and $\lambda_2$
are precisely 
the eigenvalues of the Hessian matrix of $\mathcal{V}_{\kappa,\mu}$ evaluated at the corresponding critical point. Hence, the conclusions
about their nature as critical points of $\mathcal{V}_{\kappa,\mu}$ also follow from  the estimates on $\lambda_1$ and $\lambda_2$.
These estimates  are presented in Sec.~\ref{app:stability-collinear} of Appendix~\ref{app:analytic-proofs}.

\subsection{Computer-Assisted Proofs}
\label{ss:CAPsCollinear}

We extended the  stability results of Theorem~\ref{thm:stabilitycollinear}  to a larger region of the parameter set  $\Ps{\mu_{\Tt}}$
using the CAPs illustrated in Fig.~\ref{fig:CAPstabilitycollinear}. The vertical stripe on the left of each diagram corresponds to  the 
region of the parameter space covered by Theorem~\ref{thm:stabilitycollinear}.  This information is
essential, since the CAPs become inconclusive near the boundary of $\Ps{\mu_{\Tt}}$, and in particular
near $\kappa=0$.

For each of the collinear RE, these CAPs were implemented as follows. First, we determined an enclosure of the corresponding root $(\theta_0, \kappa_0, \mu_0)$ of the equations \eqref{eq:CollinearRECondbyIntervals}, whose existence is guaranteed by the CAPs illustrated in Fig.~\ref{fig:CAPexistencecollinear}. We then estimated, using interval arithmetic, the signs of the quantities $\lambda_1$, $\lambda_2$, and $b^2 - 4\lambda_1\lambda_2$ in Proposition~\ref{prop:stabilitycollpos} to derive the corresponding stability conclusion from the proposition.

Notice that these stability CAPs cannot be implemented for parameter values $(\kappa,\mu)$ 
for which the existence CAPs of Section~\ref{sss:CAP-existence-collinear} are inconclusive. 
Furthermore, even when the existence CAPs are successful, the stability CAPs may fail if the enclosures of the quantities $\lambda_1$, $\lambda_2$, and $b^2 - 4\lambda_1\lambda_2$ are not sufficiently small. These limitations typically  lead to the failure of the stability CAPs near the boundaries of the parameter space $\Ps{\muT}$ and close to bifurcation curves where the stability changes. Nevertheless, the numerical investigations illustrated in Fig.~\ref{fig:numstabilitycollinear} suggest that the CAPs shown in Fig.~\ref{fig:CAPexistencecollinear} capture the essential stability behavior of all collinear RE.

It is particularly interesting to observe the change in the stability of $\Ll_3$ as the curvature increases, 
transitioning from a center--saddle to an elliptic relative equilibrium (corresponding to the change from the yellow to the green region). Our numerical investigations suggest that this transition precisely 
occurs at the bifurcation curve $b_4$, where the triangular equilibria $\Tt_1$ and $\Tt_2$ emerge. This observation is supported by
the following topological consideration: the change of stability occurs as  $\Ll_3$ transitions from a saddle-point to 
a local maximum of $\mathcal{V}_{\kappa,\mu}$. Therefore, in order to balance Eq.~\ref{eq:formulaindex}, two additional saddle-points of $\mathcal{V}_{\kappa,\mu}$ 
should emerge after the transition takes place. And in
fact, as will be shown in Sec.~\ref{sec:stabilitytriangular},  $\Tt_1$ and $\Tt_2$ are center--saddles. These considerations 
motivate our assertion in Conjecture~\ref{thm:classtriangular} that $\Tt_1$ and $\Tt_2$ are created through a pitchfork bifurcation of $\Ll_3$. The bifurcation is subcritical, since both $\Tt_1$ and $\Tt_2$ are unstable.

 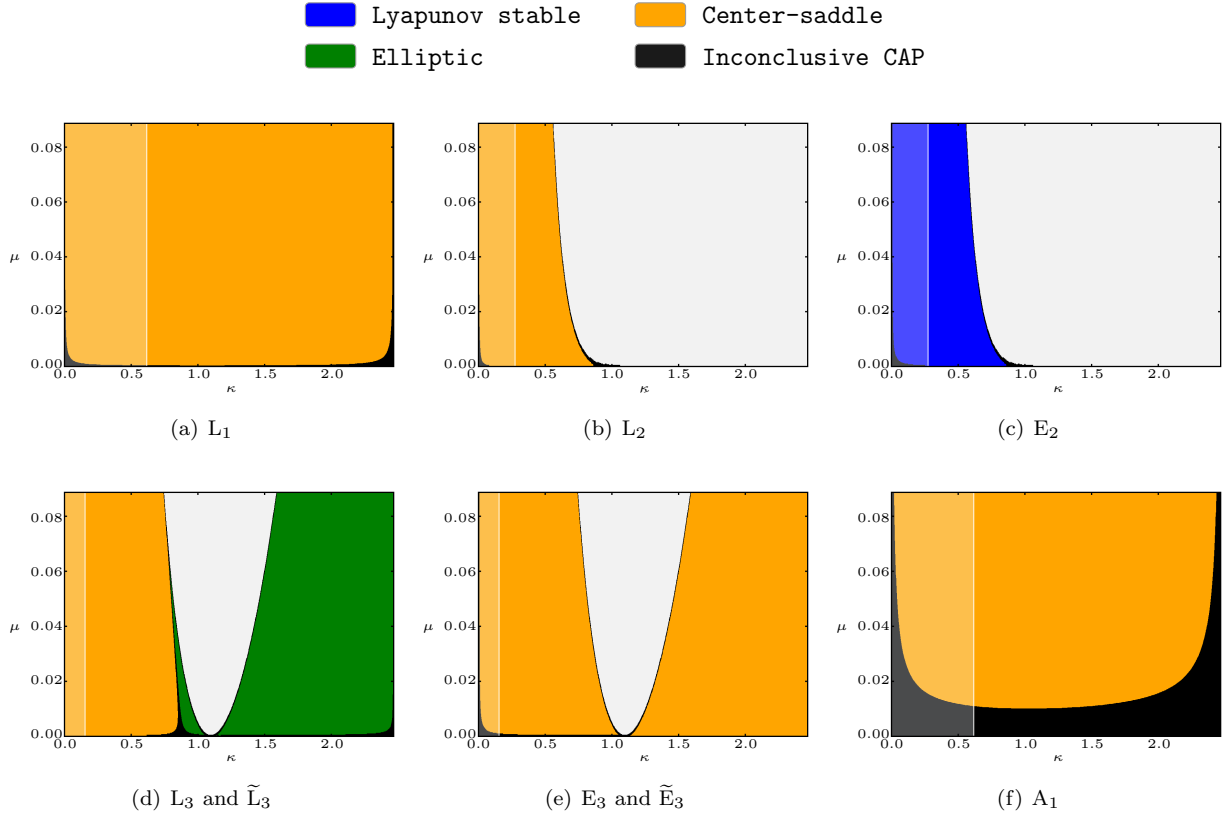
\begin{figure}[H]
 \centering

\begin{minipage}{0.6\textwidth}
\centering
\renewcommand{\arraystretch}{1.2}
\setlength{\tabcolsep}{6pt}
\begin{tabular}{@{}l@{\hspace{0.6em}}l@{\hspace{2em}}l@{\hspace{0.6em}}l@{}}
\legenditem{0000FF}{Lyapunov stable} &
\legenditem{FFA500}{Center-saddle} \\
\legenditem{008000}{Elliptic} &
\legenditem{1A1A1A}{Inconclusive CAP} \\
\end{tabular}
\end{minipage}

\vspace{1em}

\begin{minipage}{\textwidth}
\centering
 
\subfigure[$\Ll_{1}$]{
\begin{overpic}[width=0.31\textwidth]{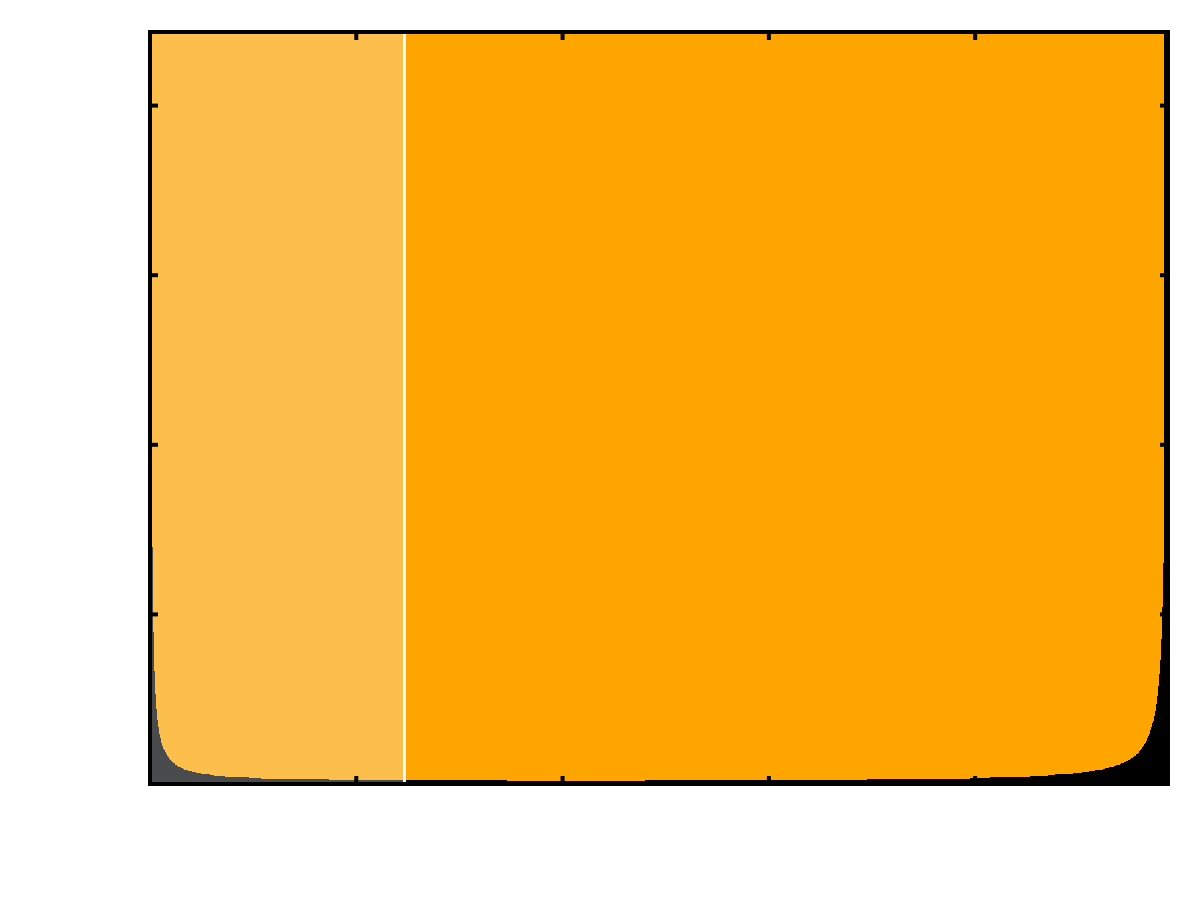}
    \put(9.4,6.5){\tiny 0.0}
    \put(26.7,6.5){\tiny 0.5}
    \put(44,6.5){\tiny 1.0}
    \put(61.2,6.5){\tiny 1.5}
    \put(78.5,6.5){\tiny 2.0}
     \put(54,3){\tiny $\kappa$}
    \put(3.5,9){\tiny 0.00}
    \put(3.5,23){\tiny 0.02}
    \put(3.5,37){\tiny 0.04}
    \put(3.5,51){\tiny 0.06}
    \put(3.5,65){\tiny 0.08}
    \put(-2,37){\tiny $\mu$}
\end{overpic}}
\subfigure[$\Ll_{2}$]{
\begin{overpic}[width=0.31\textwidth]{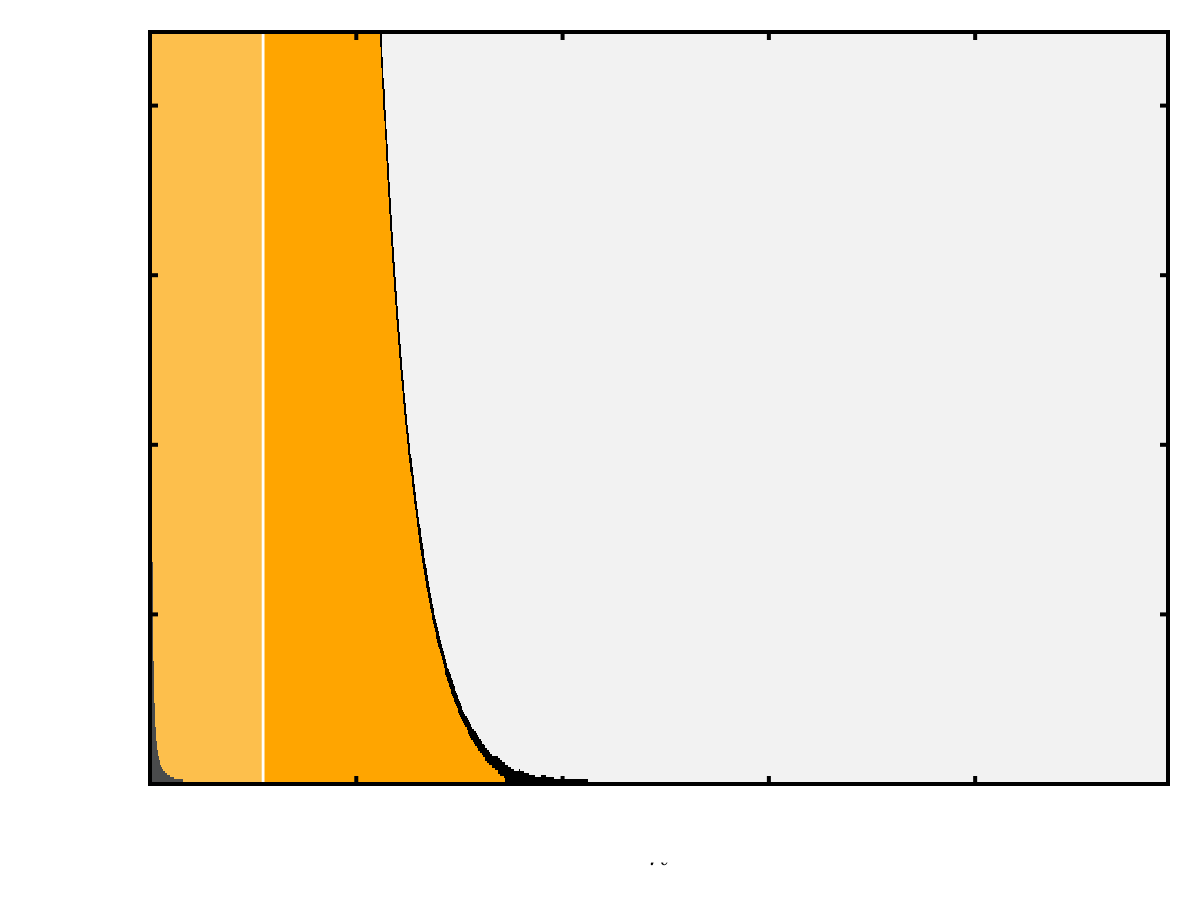}
    \put(9.4,6.5){\tiny 0.0}
    \put(26.7,6.5){\tiny 0.5}
    \put(44,6.5){\tiny 1.0}
    \put(61.2,6.5){\tiny 1.5}
    \put(78.5,6.5){\tiny 2.0}
     \put(54,3){\tiny $\kappa$}
    \put(3.5,9){\tiny 0.00}
    \put(3.5,23){\tiny 0.02}
    \put(3.5,37){\tiny 0.04}
    \put(3.5,51){\tiny 0.06}
    \put(3.5,65){\tiny 0.08}
    \put(-2,37){\tiny $\mu$}
\end{overpic}}
 \subfigure[$\Ee_{2}$]{
\begin{overpic}[width=0.31\textwidth]{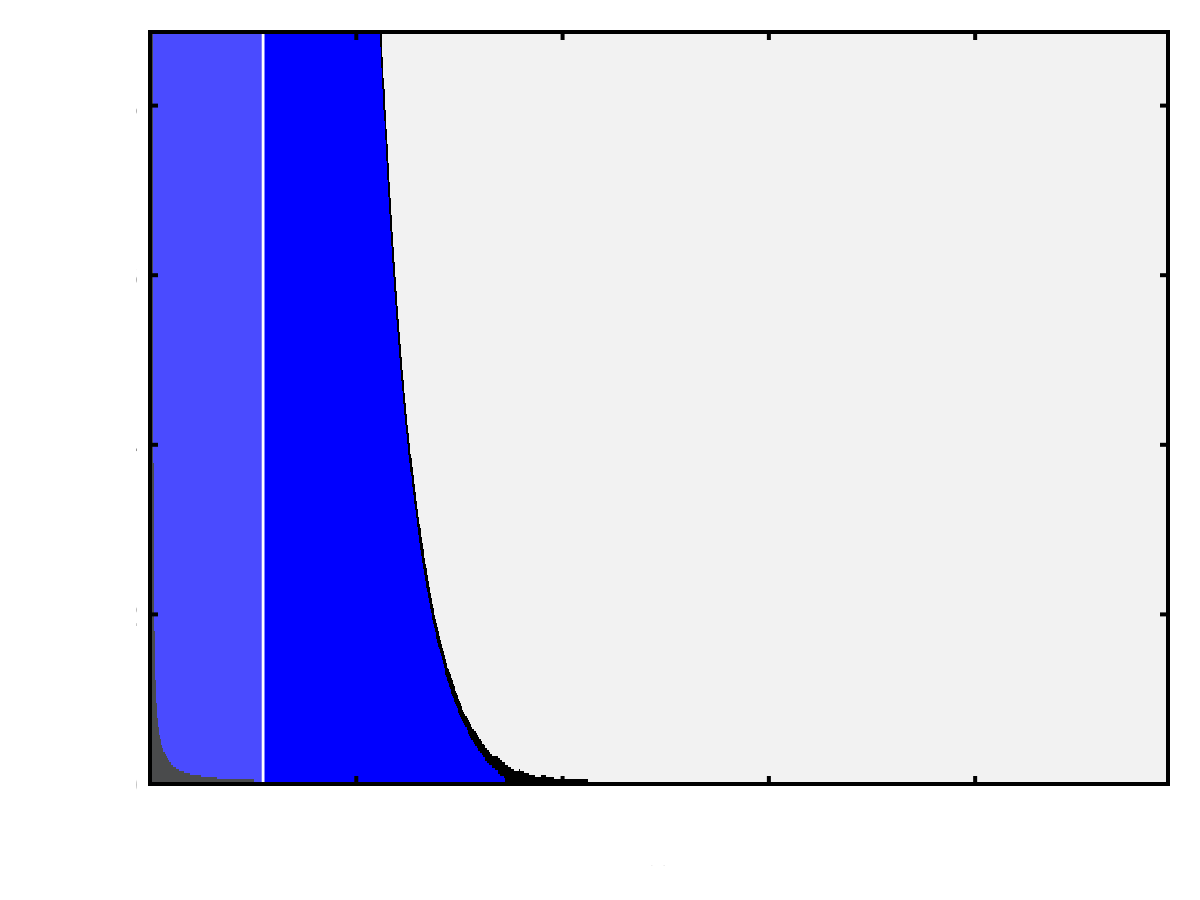}
    \put(9.4,6.5){\tiny 0.0}
    \put(26.7,6.5){\tiny 0.5}
    \put(44,6.5){\tiny 1.0}
    \put(61.2,6.5){\tiny 1.5}
    \put(78.5,6.5){\tiny 2.0}
     \put(54,3){\tiny $\kappa$}
    \put(3.5,9){\tiny 0.00}
    \put(3.5,23){\tiny 0.02}
    \put(3.5,37){\tiny 0.04}
    \put(3.5,51){\tiny 0.06}
    \put(3.5,65){\tiny 0.08}
    \put(-2,37){\tiny $\mu$}
\end{overpic}}\\[0.5em]
 \subfigure[$\Ll_{3}$ and $\widetilde{\Ll}_{3}$]{
\begin{overpic}[width=0.31\textwidth]{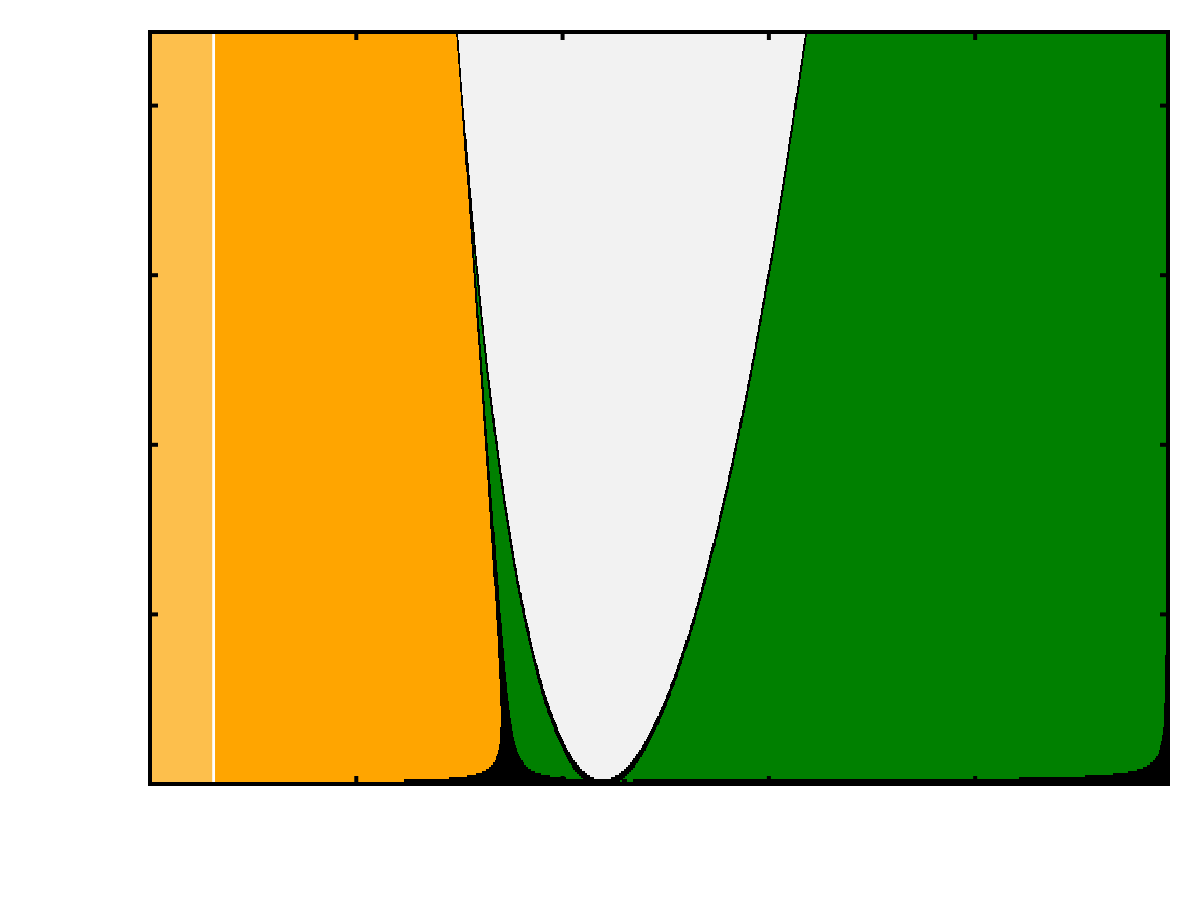}
    \put(9.4,6.5){\tiny 0.0}
    \put(26.7,6.5){\tiny 0.5}
    \put(44,6.5){\tiny 1.0}
    \put(61.2,6.5){\tiny 1.5}
    \put(78.5,6.5){\tiny 2.0}
     \put(54,3){\tiny $\kappa$}
    \put(3.5,9){\tiny 0.00}
    \put(3.5,23){\tiny 0.02}
    \put(3.5,37){\tiny 0.04}
    \put(3.5,51){\tiny 0.06}
    \put(3.5,65){\tiny 0.08}
    \put(-2,37){\tiny $\mu$}
\end{overpic}}
\subfigure[$\Ee_{3}$ and $\widetilde{\Ee}_{3}$]{
\begin{overpic}[width=0.31\textwidth]{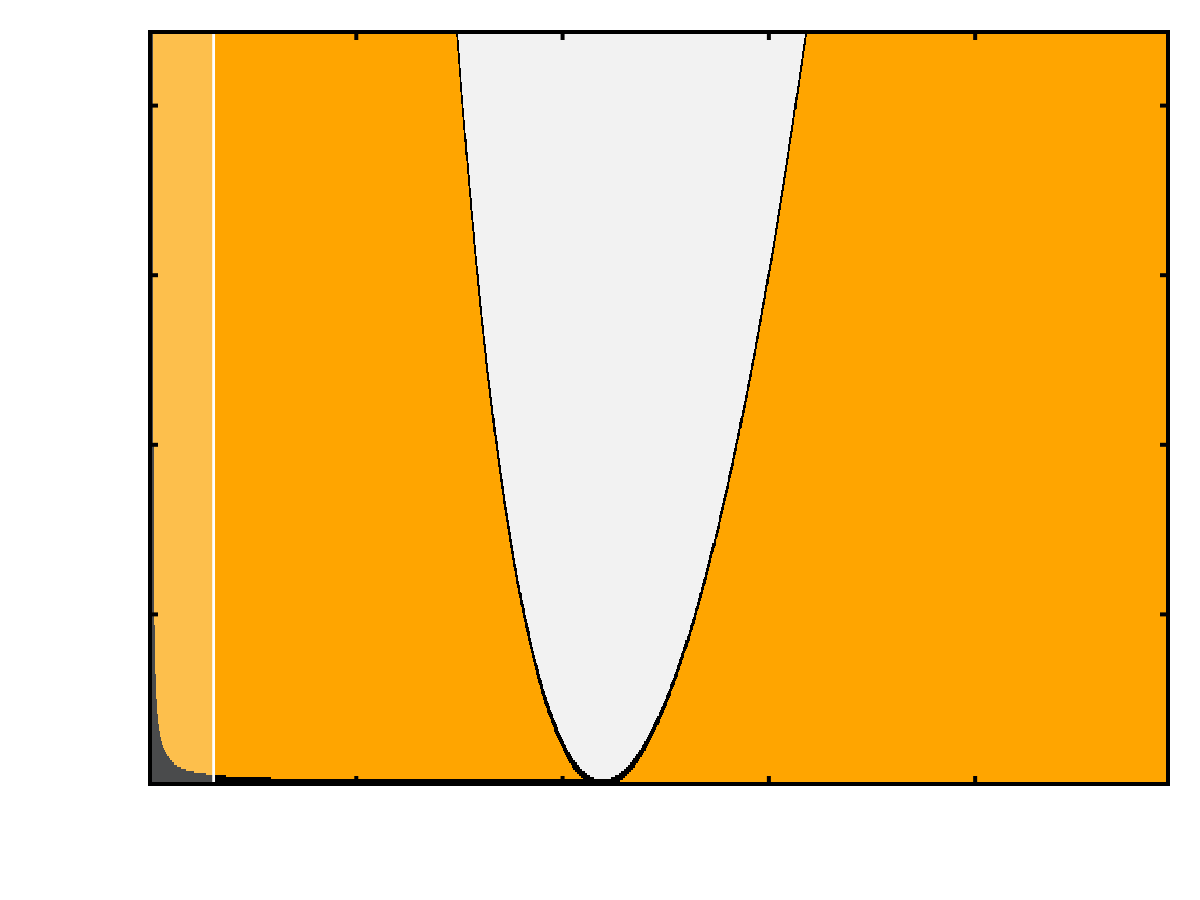}
    \put(9.4,6.5){\tiny 0.0}
    \put(26.7,6.5){\tiny 0.5}
    \put(44,6.5){\tiny 1.0}
    \put(61.2,6.5){\tiny 1.5}
    \put(78.5,6.5){\tiny 2.0}
     \put(54,3){\tiny $\kappa$}
    \put(3.5,9){\tiny 0.00}
    \put(3.5,23){\tiny 0.02}
    \put(3.5,37){\tiny 0.04}
    \put(3.5,51){\tiny 0.06}
    \put(3.5,65){\tiny 0.08}
    \put(-2,37){\tiny $\mu$}
\end{overpic}}
\subfigure[$\Aa_{1}$]{
\begin{overpic}[width=0.31\textwidth]{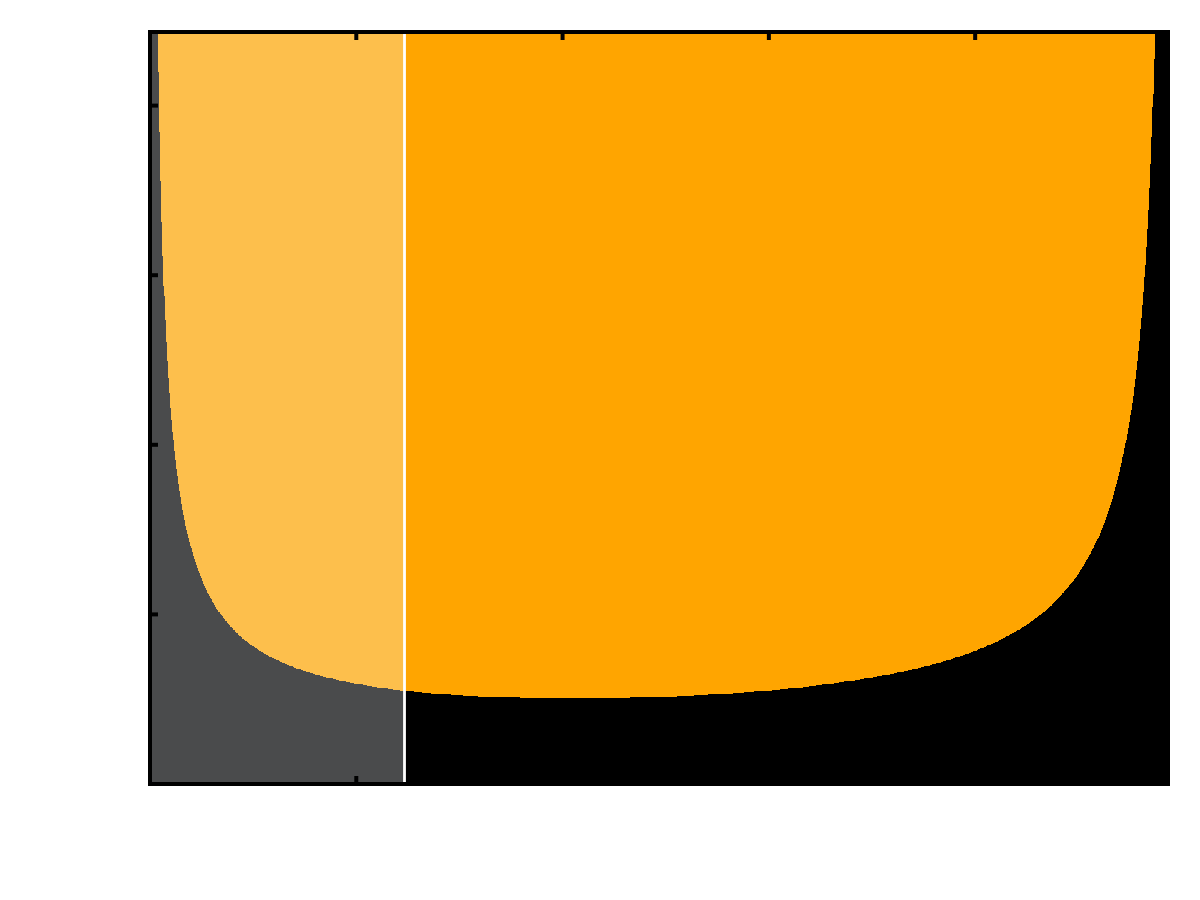}
    \put(9.4,6.5){\tiny 0.0}
    \put(26.7,6.5){\tiny 0.5}
    \put(44,6.5){\tiny 1.0}
    \put(61.2,6.5){\tiny 1.5}
    \put(78.5,6.5){\tiny 2.0}
     \put(54,3){\tiny $\kappa$}
    \put(3.5,9){\tiny 0.00}
    \put(3.5,23){\tiny 0.02}
    \put(3.5,37){\tiny 0.04}
    \put(3.5,51){\tiny 0.06}
    \put(3.5,65){\tiny 0.08}
    \put(-2,37){\tiny $\mu$}
\end{overpic}}
 \end{minipage}
 \captionsetup{width=0.75\textwidth}
 \caption{CAPs of stability of collinear relative equilibria in terms of the parameters  $(\kappa,\mu)\in \Ps{\mu_{\Tt}}$. 
 The vertical stripe on the left of each diagram covers 
values of $(\kappa,\mu)$ where the stability results were established analytically in Theorem~\ref{thm:stabilitycollinear}.}
 \label{fig:CAPstabilitycollinear}
 \end{figure}

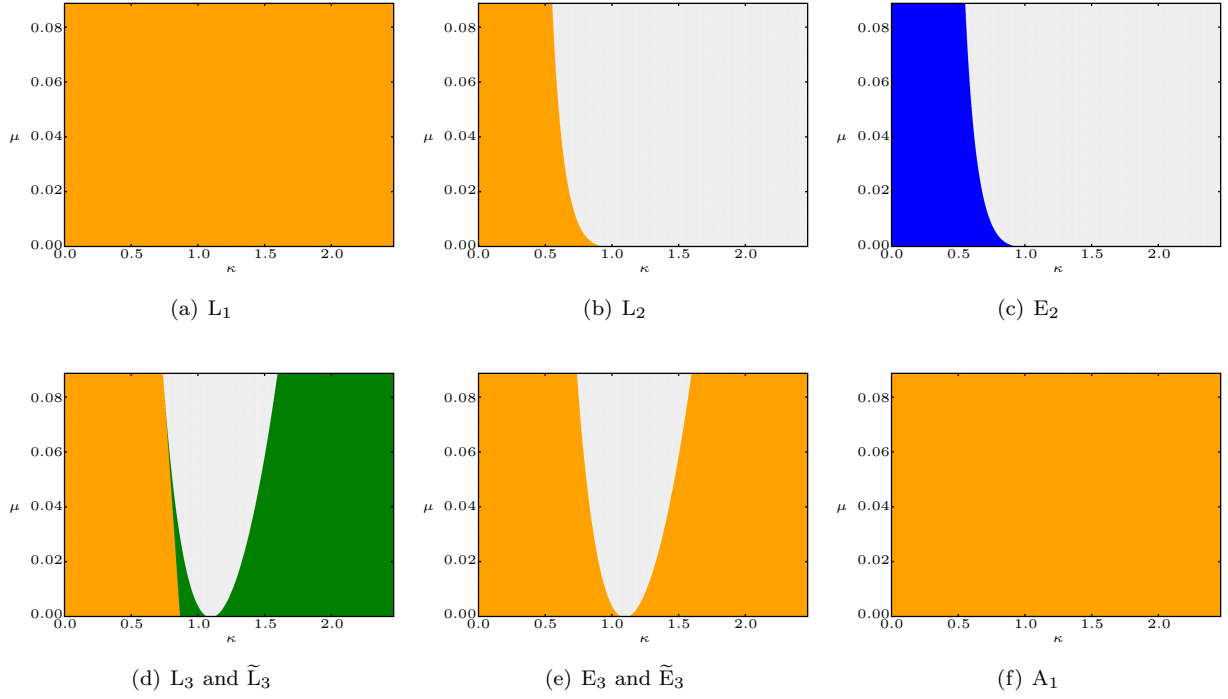
\begin{figure}[H]
\centering
\begin{minipage}{\textwidth}
\centering
\subfigure[$\Ll_{1}$]{
\begin{overpic}[width=0.31\textwidth]{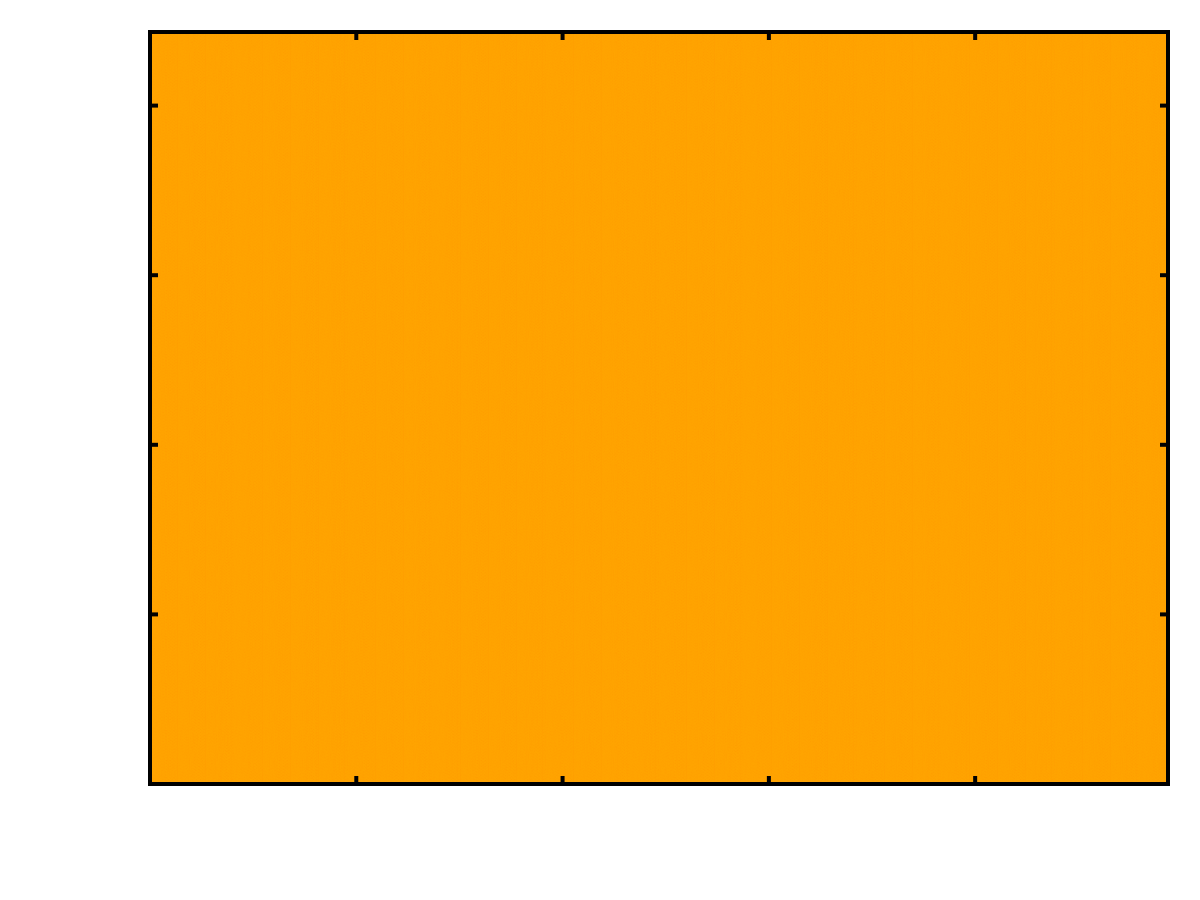}
    \put(9.4,6.5){\tiny 0.0}
    \put(26.7,6.5){\tiny 0.5}
    \put(44,6.5){\tiny 1.0}
    \put(61.2,6.5){\tiny 1.5}
    \put(78.5,6.5){\tiny 2.0}
     \put(54,3){\tiny $\kappa$}
    \put(3.5,9){\tiny 0.00}
    \put(3.5,23){\tiny 0.02}
    \put(3.5,37){\tiny 0.04}
    \put(3.5,51){\tiny 0.06}
    \put(3.5,65){\tiny 0.08}
    \put(-2,37){\tiny $\mu$}
\end{overpic}}
\subfigure[$\Ll_{2}$]{
\begin{overpic}[width=0.31\textwidth]{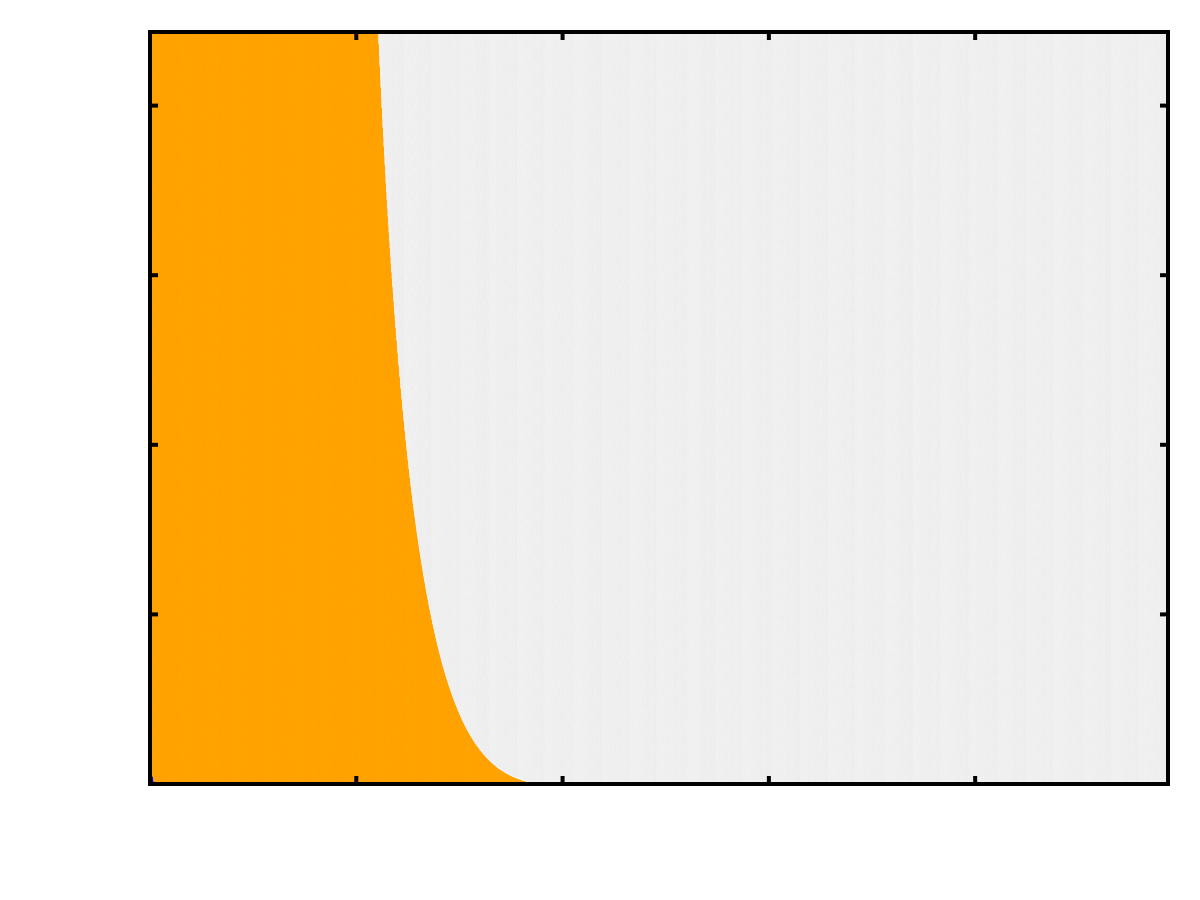}
    \put(9.4,6.5){\tiny 0.0}
    \put(26.7,6.5){\tiny 0.5}
    \put(44,6.5){\tiny 1.0}
    \put(61.2,6.5){\tiny 1.5}
    \put(78.5,6.5){\tiny 2.0}
     \put(54,3){\tiny $\kappa$}
    \put(3.5,9){\tiny 0.00}
    \put(3.5,23){\tiny 0.02}
    \put(3.5,37){\tiny 0.04}
    \put(3.5,51){\tiny 0.06}
    \put(3.5,65){\tiny 0.08}
    \put(-2,37){\tiny $\mu$}
\end{overpic}}
 \subfigure[$\Ee_{2}$]{
\begin{overpic}[width=0.31\textwidth]{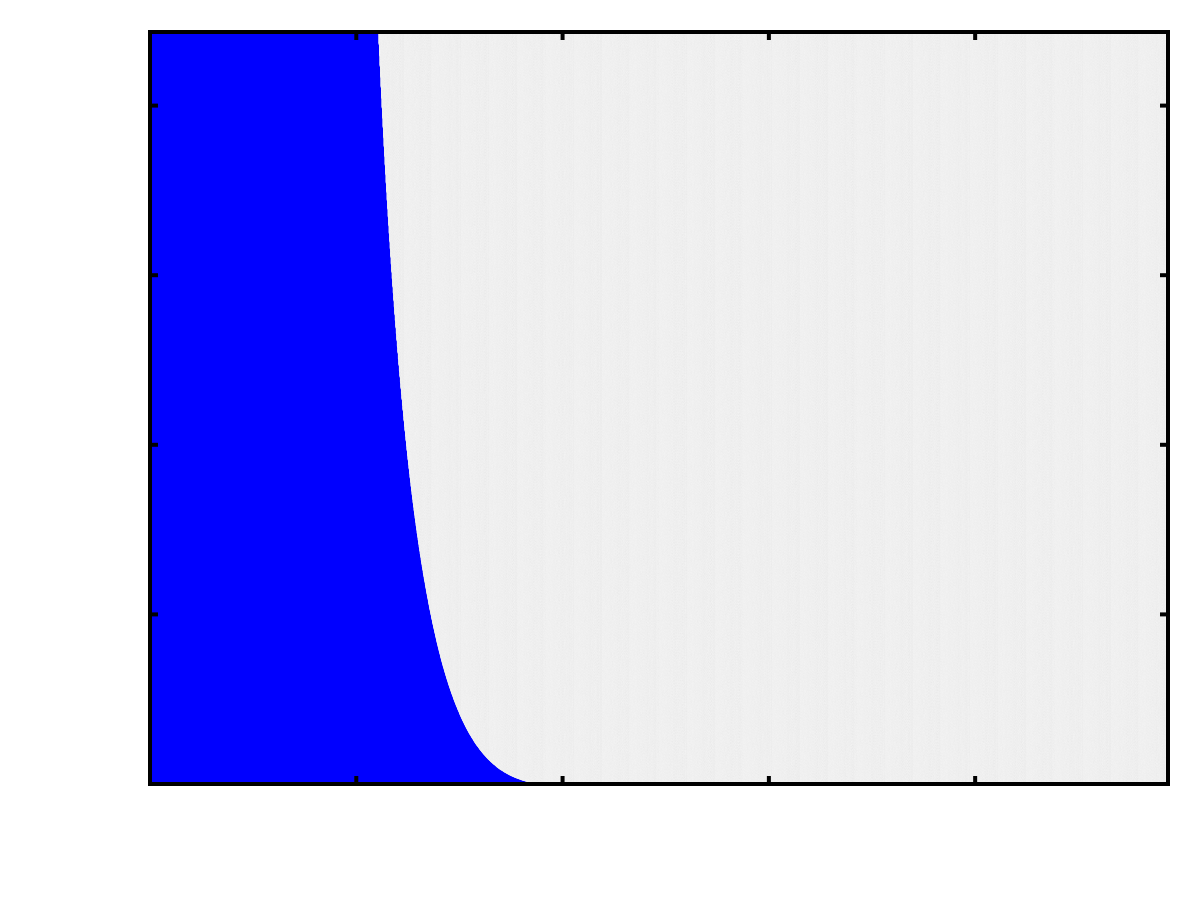}
    \put(9.4,6.5){\tiny 0.0}
    \put(26.7,6.5){\tiny 0.5}
    \put(44,6.5){\tiny 1.0}
    \put(61.2,6.5){\tiny 1.5}
    \put(78.5,6.5){\tiny 2.0}
     \put(54,3){\tiny $\kappa$}
    \put(3.5,9){\tiny 0.00}
    \put(3.5,23){\tiny 0.02}
    \put(3.5,37){\tiny 0.04}
    \put(3.5,51){\tiny 0.06}
    \put(3.5,65){\tiny 0.08}
    \put(-2,37){\tiny $\mu$}
\end{overpic}}\\[0.5em]
 \subfigure[$\Ll_{3}$ and $\widetilde{\Ll}_{3}$]{
\begin{overpic}[width=0.31\textwidth]{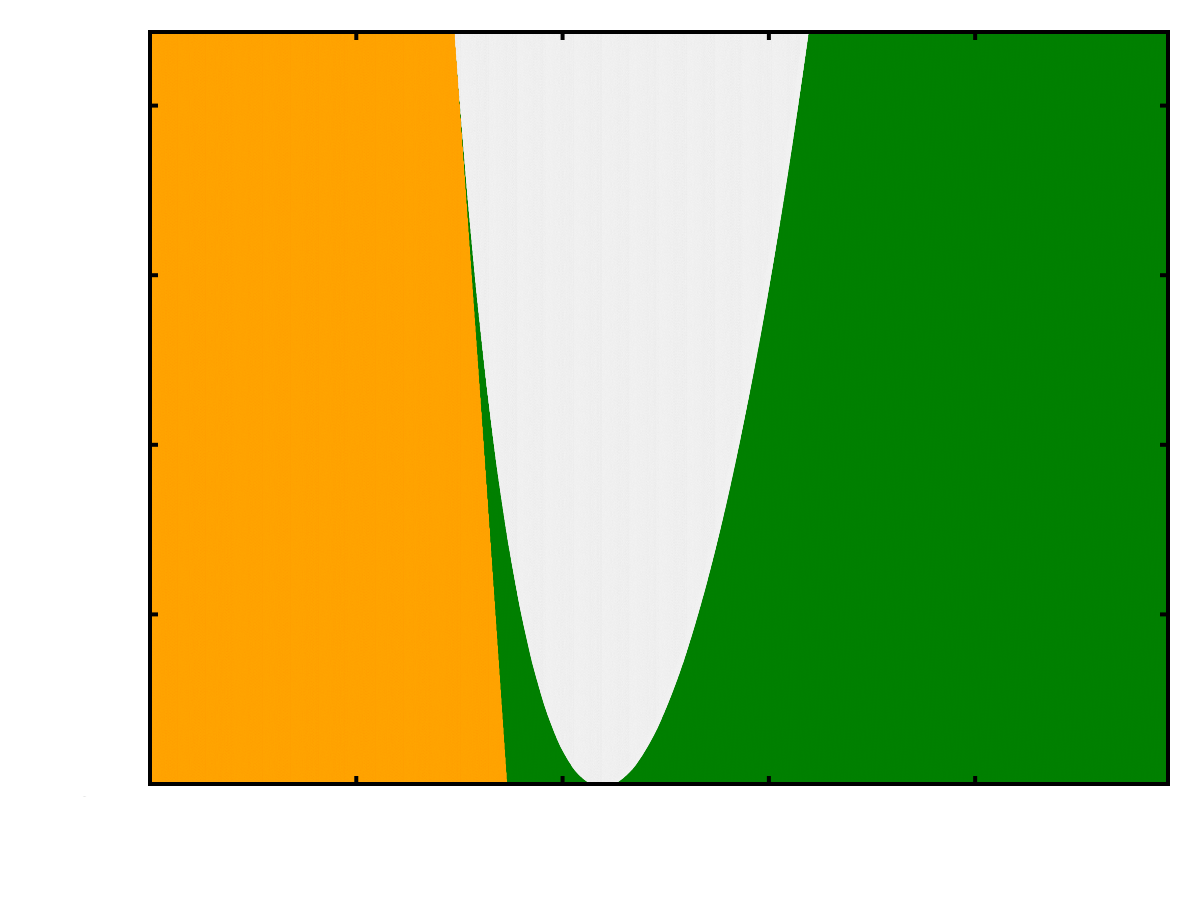}
    \put(9.4,6.5){\tiny 0.0}
    \put(26.7,6.5){\tiny 0.5}
    \put(44,6.5){\tiny 1.0}
    \put(61.2,6.5){\tiny 1.5}
    \put(78.5,6.5){\tiny 2.0}
     \put(54,3){\tiny $\kappa$}
    \put(3.5,9){\tiny 0.00}
    \put(3.5,23){\tiny 0.02}
    \put(3.5,37){\tiny 0.04}
    \put(3.5,51){\tiny 0.06}
    \put(3.5,65){\tiny 0.08}
    \put(-2,37){\tiny $\mu$}
\end{overpic}}
\subfigure[$\Ee_{3}$ and $\widetilde{\Ee}_{3}$]{
\begin{overpic}[width=0.31\textwidth]{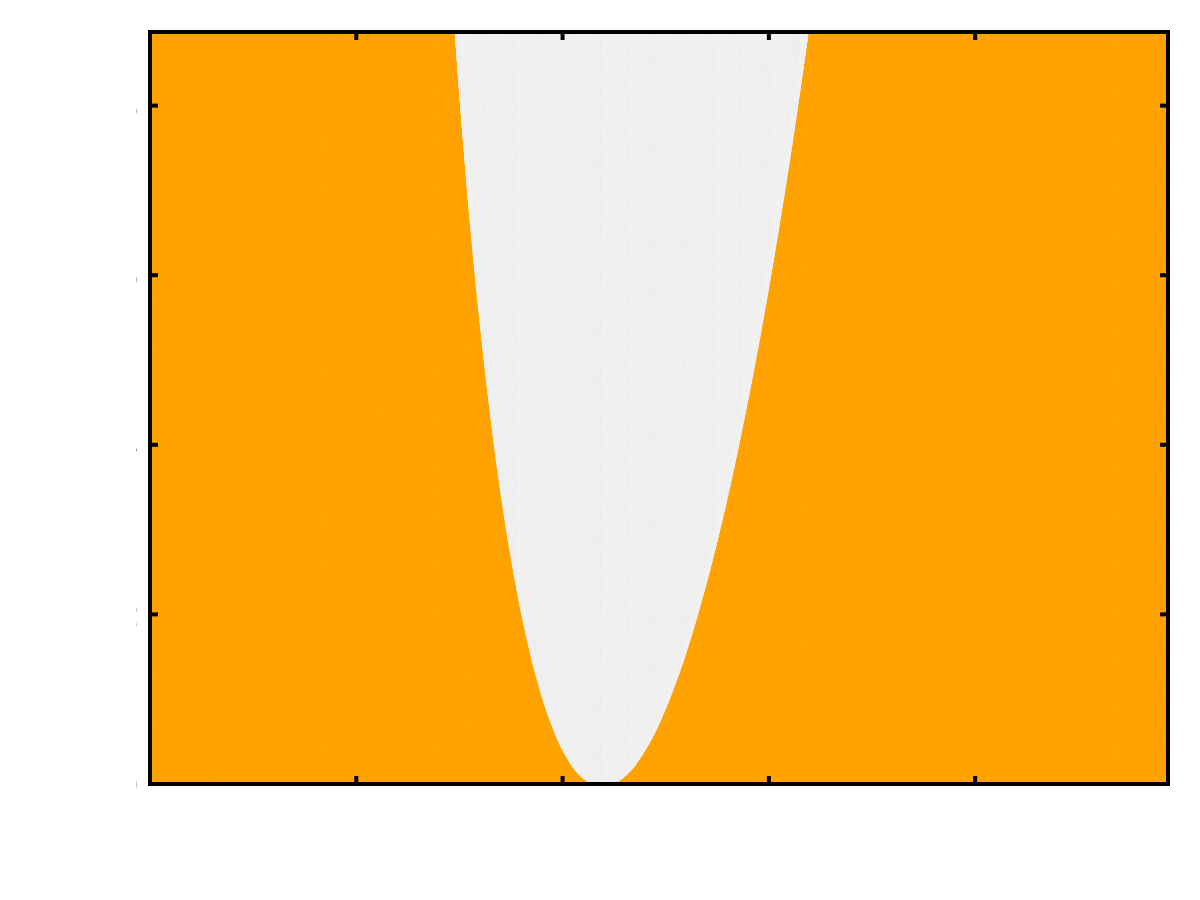}
    \put(9.4,6.5){\tiny 0.0}
    \put(26.7,6.5){\tiny 0.5}
    \put(44,6.5){\tiny 1.0}
    \put(61.2,6.5){\tiny 1.5}
    \put(78.5,6.5){\tiny 2.0}
     \put(54,3){\tiny $\kappa$}
    \put(3.5,9){\tiny 0.00}
    \put(3.5,23){\tiny 0.02}
    \put(3.5,37){\tiny 0.04}
    \put(3.5,51){\tiny 0.06}
    \put(3.5,65){\tiny 0.08}
    \put(-2,37){\tiny $\mu$}
\end{overpic}}
\subfigure[$\Aa_{1}$]{
\begin{overpic}[width=0.31\textwidth]{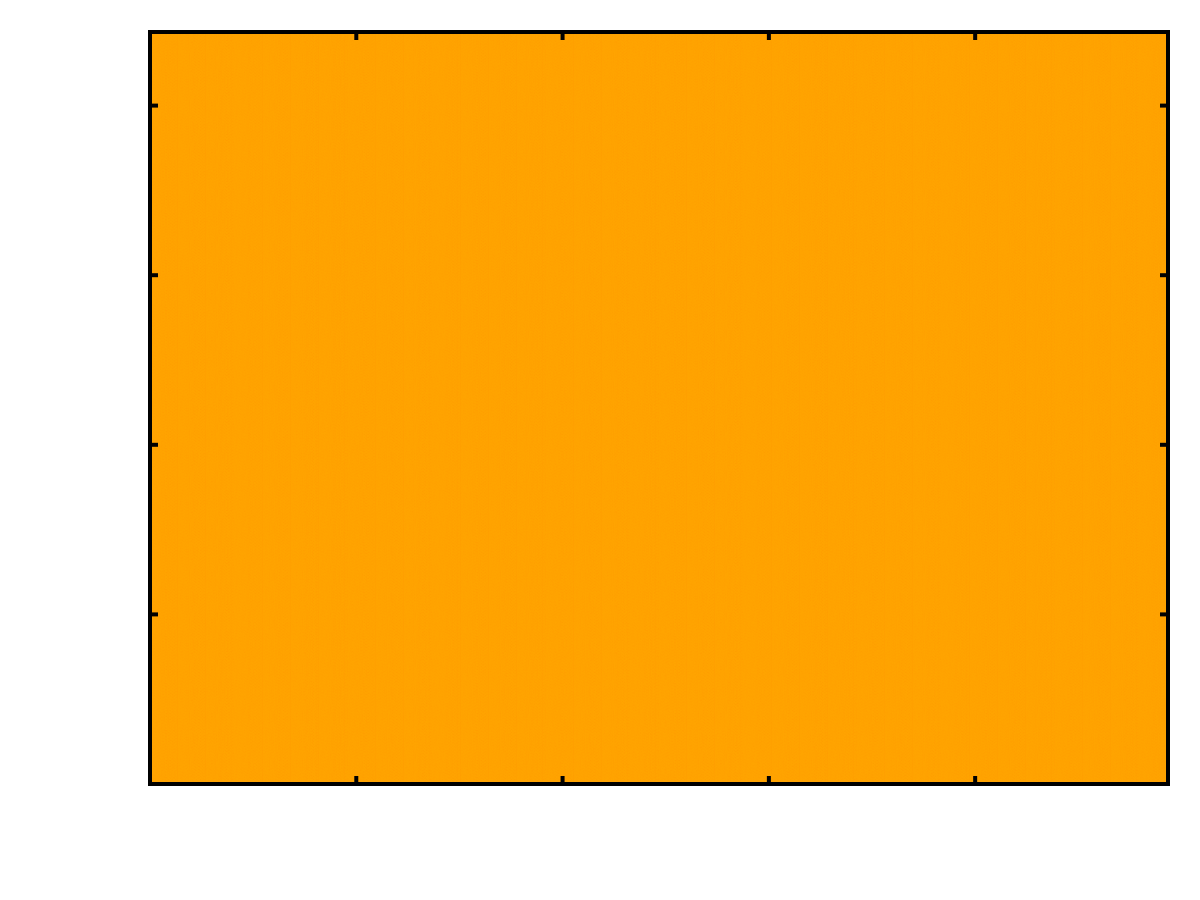}
    \put(9.4,6.5){\tiny 0.0}
    \put(26.7,6.5){\tiny 0.5}
    \put(44,6.5){\tiny 1.0}
    \put(61.2,6.5){\tiny 1.5}
    \put(78.5,6.5){\tiny 2.0}
     \put(54,3){\tiny $\kappa$}
    \put(3.5,9){\tiny 0.00}
    \put(3.5,23){\tiny 0.02}
    \put(3.5,37){\tiny 0.04}
    \put(3.5,51){\tiny 0.06}
    \put(3.5,65){\tiny 0.08}
    \put(-2,37){\tiny $\mu$}
\end{overpic}}
\end{minipage}
\captionsetup{width=0.75\textwidth}
\caption{Numerical evidence for the stability and instability regions in parameter region  $\Ps{\mu_{\Tt}}$.
See Fig.~\ref{fig:CAPstabilitycollinear} for the color code.}
\label{fig:numstabilitycollinear}
\end{figure}

\begin{remark}\label{rmk:isoscelesREstability}
Panel (d) of Fig.~\ref{fig:CAPstabilitycollinear} proves that  $\Ll_{3}$ is unstable when
$\kappa=\tau_{0}^{2}\approx 0.4639$, where $\tau_0$ is given by formula \eqref{eq:tau0}, for values of $\mu\in (0,\mu_{\Tt})$ 
which are not too close to zero
(the boundary). In view of Remark~\ref{rmk:isoscelesRE},
it is natural to expect that, at least in certain ranges of the masses, the isosceles RE 
of the positively curved 3-body problem discovered in  \cite{FuPC24} are unstable.
\end{remark}

\subsection{\texorpdfstring
  {Stability of collinear RE as a function of $\kappa>0$ for fixed $\mu\in(0,\mu_{\Tt})$.}
  {Stability of collinear RE as a function of kappa > 0 for fixed mu in (0, mu_Tt).}}
\label{conj:stability-collinear}

The following conjecture  summarizes the discussion in the previous sections. The value of $\kappa_4(\mu)$ in the statement 
is taken from Conjecture~\ref{thm:classtriangular}.

\begin{conjecture}[Stability classification of collinear RE as a function of $\kappa>0$  for  $0<\mu<\mu_{\Tt}$.]
\label{thm:class-stability-collinear}
Let $\mu \in (0,\mu_{\Tt})$. The stability of collinear RE is as follows. 
\begin{description}
    \item[$\bullet$]  $\Ll_1$, $\Ll_2$,  $\Ee_3$, and  $\Aa_1$ are center--saddles 
     for all $\kappa>0$ for which each of them exists.
    \item[$\bullet$]  $\Ee_2$ is Lyapunov stable, and  $\tilde{\Ll}_3$ is elliptic for all $\kappa>0$ for which each of them exists.
    \item[$\bullet$] $\Ll_3$ is a center--saddle  for $0<\kappa<\kappa_4(\mu)$. It undergoes a subcritical pitchfork bifurcation at $\kappa=\kappa_4(\mu)$
    and is elliptic for all values $\kappa>\kappa_4(\mu)$ for which it exists. 
    \end{description}
\end{conjecture}

Note that, despite our strong evidence, 
 we do  not present a rigorous  proof that the change of stability of $\Ll_3$ occurs at the value of $\kappa_4$ in 
 Conjecture~\ref{thm:classtriangular}. We also do not have rigorous proofs of stability for those parameter values $(\kappa,\mu)\in \Ps{\muT}$ for
which Theorem~\ref{thm:stabilitycollinear} does not apply and the CAPs in Fig.~\ref{fig:CAPstabilitycollinear} are inconclusive. For this reason, we have presented the 
above stability classification as a conjecture. 

\subsection{\texorpdfstring
  {Qualitative change of the bifurcation structure at $(\kappa_{\Tt},\mu_{\Tt})$}
  {Qualitative change of the bifurcation structure at (kappa_Tt, mu_T)}}
\label{sec:muT}

As mentioned above, our conjecture that, for $\mu\in (0,\muT)$, the triangular relative equilibria 
 $\Tt_{1}$ and $\Tt_{2}$ bifurcate from $\Ll_{3}$ when $\kappa=\kappa_4(\mu)$ 
is based on topological considerations supported by numerical investigations. Our conjecture would imply 
 that, at the parameter values $(\kappa_4(\mu),\mu)$, 
the equilibrium $\Ll_{3}$ lies at the intersection of the curve $\gamma_{\kappa,\mu}$ of 
triangular balanced configurations and the geodesic $\mathcal{G}^+$, which contains all collinear RE.
 Such intersection point belongs to the segment $\mathcal{I}_3$ and 
  is  labeled $\tilde \Xx$ in Fig.~\ref{fig:graphpos}(a) and $\Xx$ in Fig.~\ref{fig:graphpos}(b).

Our numerical investigations show that as $\mu$ increases and approaches $\muT$, the value of $\kappa_4(\mu)$ 
 approaches the bifurcation value $\kappa_2(\mu)$ at which $\Ll_3$ coalesces with $\Ee_3$ 
and disappears in a saddle-node bifurcation. This behavior corresponds to the shrinking of the green portion of the parameter space between the yellow region and the parabolic-shaped curve near the top of the graph, as can be observed in Figs.~\ref{fig:CAPstabilitycollinear}(d) and \ref{fig:numstabilitycollinear}(d).

At $\mu=\muT$, we have $\kappa_2(\muT)=\kappa_4(\muT)=:\kappa_{\Tt}$. This corresponds to the collision of the saddle-node bifurcation involving $\Ll_3$ and $\Ee_3$ with the pitchfork bifurcation at which 
the triangular equilibria $\Tt_1$ and $\Tt_2$ are created.
The parameter values $(\kappa_{\Tt}, \muT)$ belong to the intersection, $b_2\cap b_4$, of the bifurcation curves $b_2$ and $b_4$, and
can be determined  imposing the condition that $\theta\mapsto f_3(\theta;\kappa,
\mu)$ has a double root (saddle-node bifurcation) at the point $\Xx$. Since the position
of $\Xx$ is determined by Eq.~\eqref{eq:Xcoord}, this yields the
equations 
\begin{equation*}
f_{3}(\theta;\kappa,\mu)=0, \qquad 
    f_{3}'(\theta;\kappa,\mu)=0, \qquad
    \sin \left ( \theta + \sqrt{\kappa}	q_1\right )- \Lambda \sin \left ( \theta -  \sqrt{\kappa}	q_2 \right ) =0,
\end{equation*}
which, upon rearrangement,  reduce to system \eqref{eq:mubalancedsystem}. 
Numerical approximations of  $(\kappa_{\Tt}, \muT)$ are given in Eq.~\eqref{eq:numapprox}.

The bifurcation structure of the RE for $\mu>\muT$ is qualitatively different from that described in the conjectures above, and a more detailed study  of the bifurcation curves in the parameter space $(\kappa,\mu)$ is required to properly understand the behavior  at $(\kappa_{\Tt},\muT)$. The complexity of the bifurcation structure 
near $(\kappa_{\Tt},\muT)$ is reflected in the multiple crossings of the bifurcation curves computed by Kilin~\cite[Fig.~13]{Ki99} near the point $(\alpha,m_2/m_1)\approx(\pi/4,0.1)$.

\section{Stability of triangular RE}
\label{sec:stabilitytriangular}

We now consider the stability properties of triangular RE for both positive and negative
curvature. Because of  the intricate nature of the formulas given in Propositions~\ref{prop:stabilitytripos} and \ref{prop:stabilitytrineg}, we were unable to obtain analytical  results so our discussion
only relies on numerical investigations and CAPs.

We begin by recalling the classical stability result about the
equilateral equilibria $\Ll_{4}$ and $\Ll_{5}$  in the planar R3BP. As is well-known,  
 they are  elliptic\footnote{The nonlinear stability of $\Ll_{4}$ and $\Ll_{5}$ has been considered in e.g. \cite{Leon62, DeDe67, Mar69}, and also
 in the spatial version of the R3BP in \cite{BenFaGu98}.} if $0<\mu<\mu_R$ and unstable (complex--saddles) otherwise  (see e.g. \cite[Chapter 5]{Sz67}),
where $\mu_R$ is Routh's critical mass ratio
\begin{equation}
\label{eq:routh}
    \mu_{R} = \frac{9 - \sqrt{69}}{9 + \sqrt{69}} \approx 0.04006.\footnote{The threshold value 
     for stability is  
    more commonly written in terms of the reduced mass as $\frac{\mu_2}{\mu_1 + \mu_2}=\frac{9 - \sqrt{69}}{18} \approx 0.03852$.} 
\end{equation}
It is common to say that a gyroscopic stabilization of  $\Ll_{4}$ and $\Ll_{5}$
takes place for $0<\mu<\mu_R$, and  we will use this same terminology for non-zero curvature below.

We are particularly interested in the  threshold for gyroscopic stabilization of 
$\Ll_{4}$ and $\Ll_{5}$ as a function of the curvature $\kappa$, as it transitions from negative, through $0$,  to positive values. Our analysis shows that increasing the curvature enlarges the range of mass ratios $\mu$ for which gyroscopic stabilization takes place. In particular, 
positive curvature has a stabilizing effect, whereas negative curvature has a destabilizing effect.

\subsection{Negative curvature}

As recalled in Proposition~\ref{th:existence-neg-triangular},  the only triangular RE for $\kappa<0$ are $\Ll_{4}$ and $\Ll_{5}$, which have identical  stability properties due to symmetry.
Our numerical investigations, illustrated
in Fig.~\ref{fig:numstabilityL45triangularneg}, indicate their stability depending on the parameter values
$(\kappa,\mu)$ with $0<\mu<\muT$. These investigations were performed by applying the stability criteria given
in Proposition~\ref{prop:stabilitytrineg}. The figure shows that the linear stability region, colored  in 
green, becomes small for large negative $\kappa$. Numerical evidence for this phenomenon 
was already given in \cite{MS17} (see bottom right panel of Fig.~21), but our choice of parameters
$(\kappa,\mu)$ indicates this in a more transparent manner.

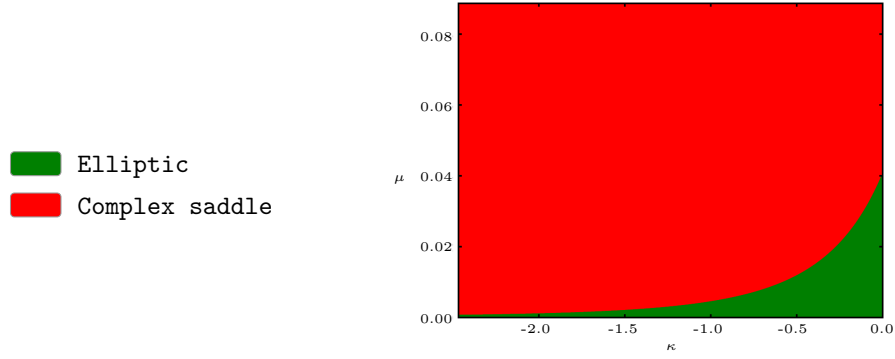
\begin{figure}[H]
\centering
\makebox[\textwidth][c]{
\begin{minipage}[c]{0.28\textwidth}
\centering
\begin{tabular}{@{}l@{\hspace{0.6em}}l@{}}
\legenditem{008000}{Elliptic} \\
\legenditem{FF0000}{Complex saddle} \\
\end{tabular}
\end{minipage}
\hspace{0.05\textwidth}
\begin{minipage}[c]{0.4\textwidth}
\centering
\begin{overpic}[width=\textwidth]{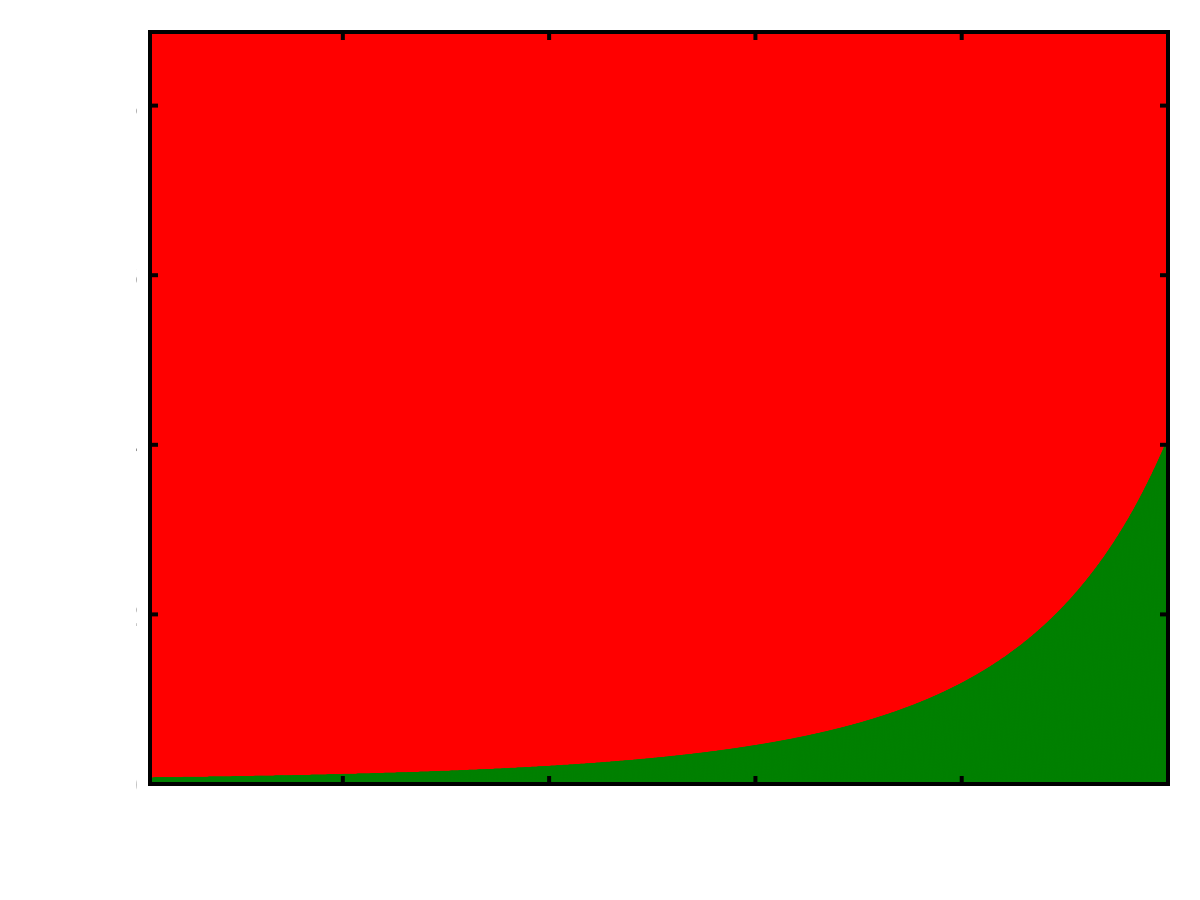}
	\put(25,6.5){\tiny -2.0}
	\put(42.5,6.5){\tiny -1.5}
	\put(59.5,6.5){\tiny -1.0}
	\put(76.5,6.5){\tiny -0.5}
	\put(95,6.5){\tiny 0.0}
	\put(54,3){\tiny $\kappa$}
	\put(4.8,8.5){\tiny 0.00}
	\put(4.8,23){\tiny 0.02}
	\put(4.8,37){\tiny 0.04}
	\put(4.8,51){\tiny 0.06}
	\put(4.8,65){\tiny 0.08}
	\put(-0.5,37){\tiny $\mu$}
\end{overpic}
\end{minipage}
}
\captionsetup{width=0.75\textwidth}
\caption{Numerical investigation of stability of the triangular RE $\Ll_{4}$ and $\Ll_{5}$ for negative curvature as a function of $(\kappa,\mu)$.}
\label{fig:numstabilityL45triangularneg}
\end{figure}

\subsection{Positive curvature}
\label{ss:stability-triangular-CAPs}

Our investigations of stability of triangular RE for $\kappa>0$ are shown in Fig.~\ref{fig:CAPstabilitytriangular}
which illustrates  both CAPs and numerical investigations in the parameter region $\Ps{\muT}$. These
refer to the triangular points $\Ll_{4,5}$ and $\Tt_{1,2}$ described in the classification of triangular RE given
in Conjecture \ref{thm:classtriangular}.

The stability CAPs in panels (a) and (c) of Fig.~\ref{fig:CAPstabilitytriangular}
 were implemented in analogy with our analysis of collinear RE in 
Sec.~\ref{ss:CAPsCollinear}: we first computed an enclosure of the  roots $(x_0, \kappa_0, \mu_0)$ of the function
$g_{\kappa,\mu}:\mathcal{I}_{\kappa,\mu}\to \R$ defined in Eq.~\eqref{eq:auxfunctiontriangular}, 
whose existence was proved using a CAP illustrated in Fig.~\ref{fig:CAPexistencetriangular}. We then 
set $d_2=x_0$, $d_1=\arcsin\left(\Lambda\sin(x_{0})\right)$, and 
estimated, using interval arithmetic, the signs of the quantities $\Delta$ and $M-4\Delta$ 
 in Proposition~\ref{prop:stabilitytripos} in order 
 to derive the corresponding stability conclusion from the proposition.

This procedure has the same 
limitations as the stability CAPs for collinear RE in Sec.~\ref{ss:CAPsCollinear}. In particular, it 
cannot be implemented for parameter values $(\kappa,\mu)$ 
for which the existence CAPs in Fig.~\ref{fig:CAPexistencetriangular} fail, or when the
signs of $\Delta$ and $M-4\Delta$ cannot be determined via interval arithmetic.

The numerical investigations in panels (b) and (d) of Fig.~\ref{fig:CAPstabilitytriangular} were also
performed using the stability criteria of Proposition~\ref{prop:stabilitytripos}. These panels suggest that, although the CAPs are inconclusive near the bifurcation curves where the stability or the number of triangular RE changes, they nevertheless capture the essential features of the behavior.

The CAPs in panel (a) of Fig.~\ref{fig:CAPstabilitytriangular} provide a rigorous proof
that the range of mass ratios $\mu$ for which $\Ll_{4}$ and $\Ll_{5}$  are  gyroscopically  stabilized 
 increases as the curvature increases, as was mentioned before. On the other hand, 
panels (c) and (d) provide very strong evidence that the triangular points $\Tt_{1}$ and $\Tt_{2}$
are unstable for all parameter values $(\kappa,\mu)\in \Ps{\muT}$ for which they exist.

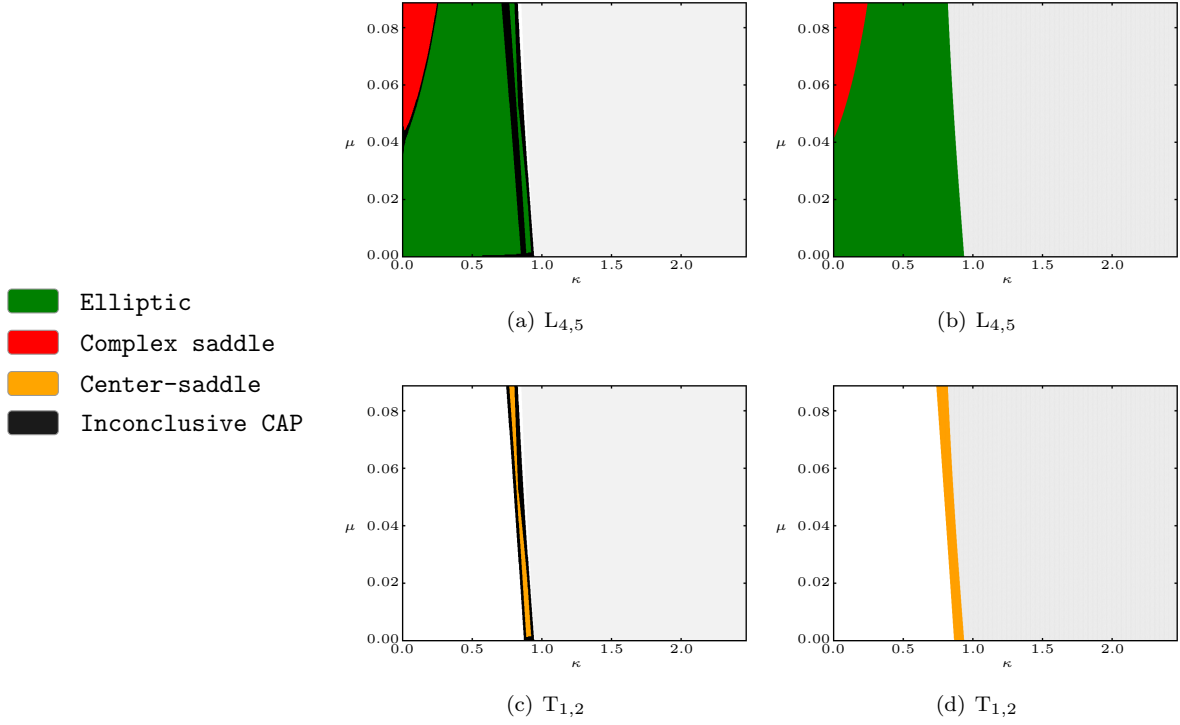
\begin{figure}[H]
\centering
\makebox[\textwidth][c]{%
\begin{minipage}[c]{0.20\textwidth}
\centering
\begin{tabular}{@{}l@{\hspace{0.8em}}l@{}}
\legenditem{008000}{Elliptic} \\
\legenditem{FF0000}{Complex saddle} \\
\legenditem{FFA500}{Center-saddle} \\
\legenditem{1A1A1A}{Inconclusive CAP}
\end{tabular}
\end{minipage}
\hspace{0.03\textwidth}
\begin{minipage}[c]{0.72\textwidth}
\centering
\subfigure[$\Ll_{4,5}$]{
\begin{overpic}[width=0.45\textwidth]{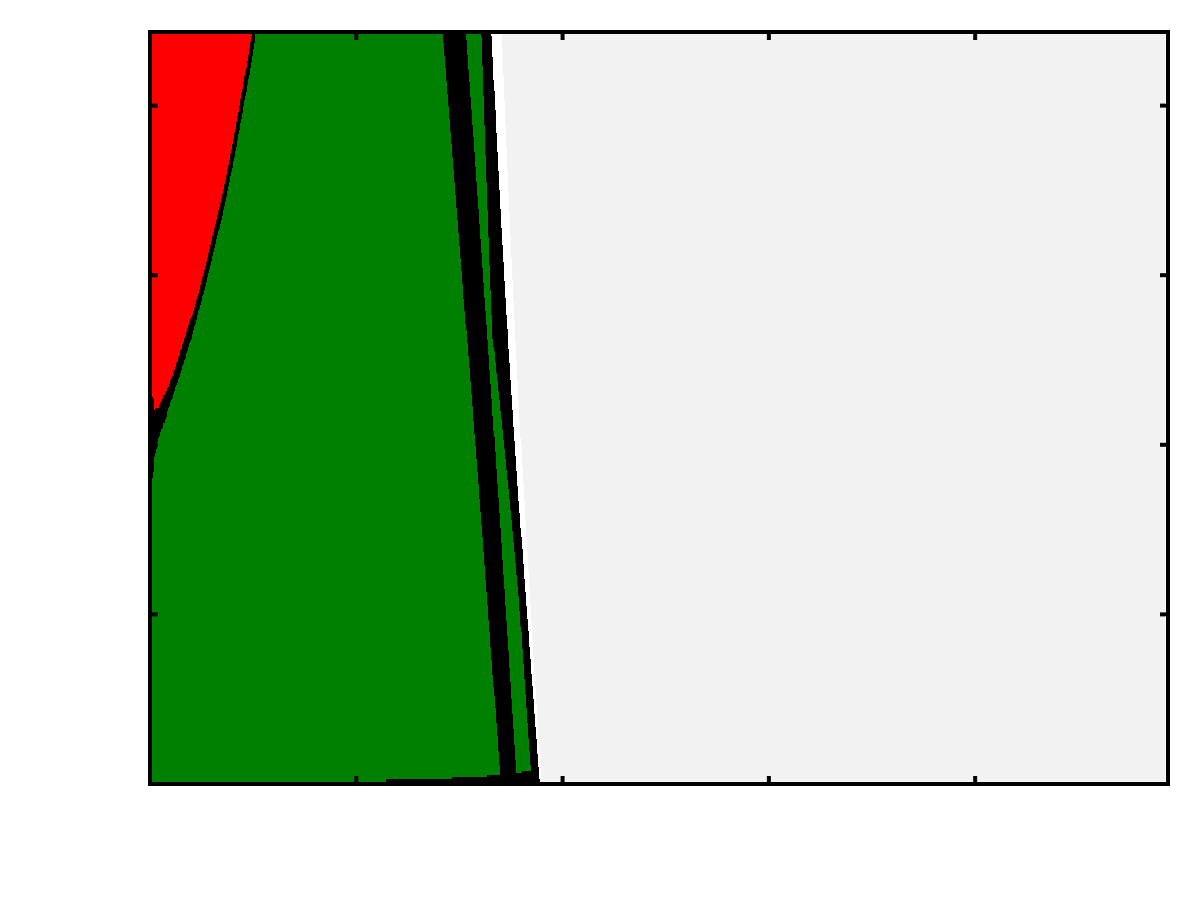}
    \put(9.4,6.5){\tiny 0.0}
    \put(26.7,6.5){\tiny 0.5}
    \put(44,6.5){\tiny 1.0}
    \put(61.2,6.5){\tiny 1.5}
    \put(78.5,6.5){\tiny 2.0}
     \put(54,3){\tiny $\kappa$}
    \put(3.5,9){\tiny 0.00}
    \put(3.5,23){\tiny 0.02}
    \put(3.5,37){\tiny 0.04}
    \put(3.5,51){\tiny 0.06}
    \put(3.5,65){\tiny 0.08}
    \put(-2,37){\tiny $\mu$}
\end{overpic}}
\subfigure[$\Ll_{4,5}$]{
\begin{overpic}[width=0.45\textwidth]{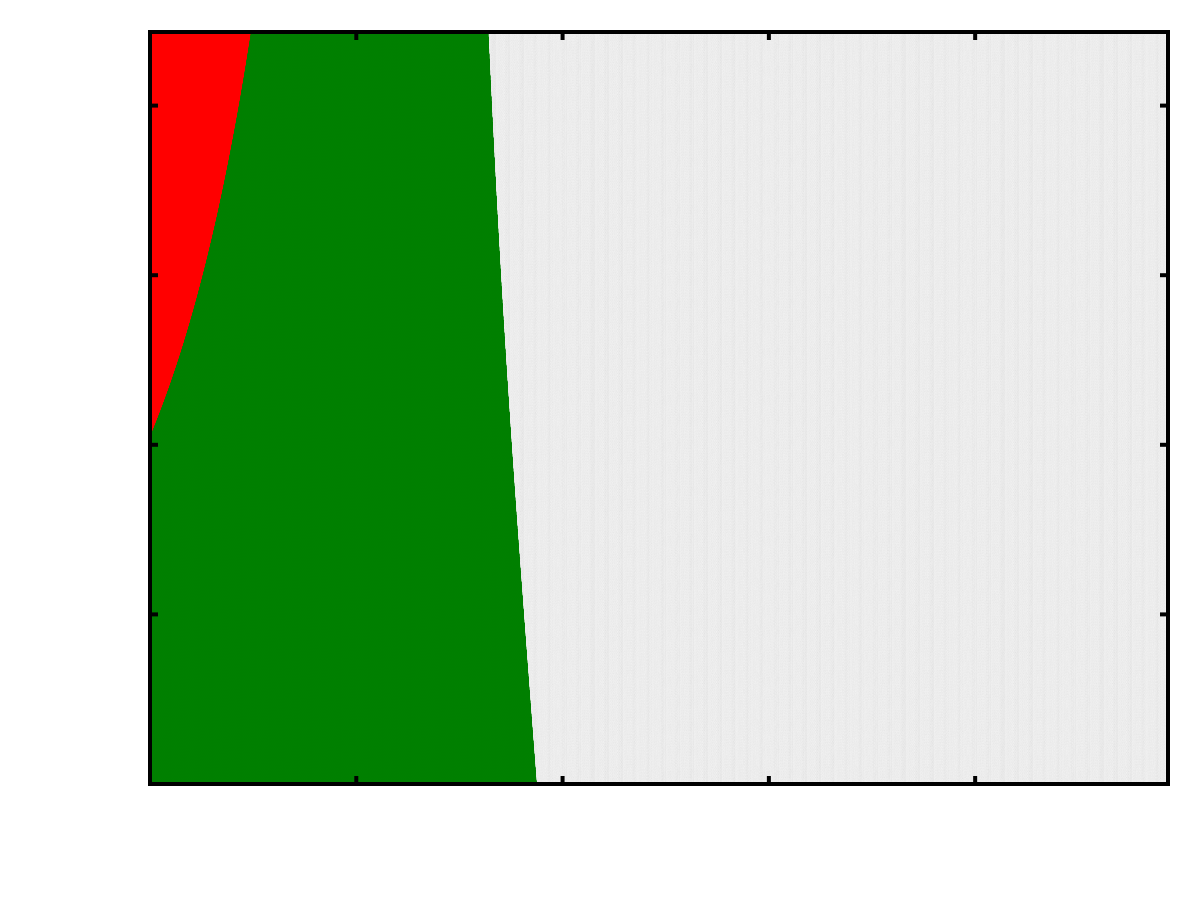}
    \put(9.4,6.5){\tiny 0.0}
    \put(26.7,6.5){\tiny 0.5}
    \put(44,6.5){\tiny 1.0}
    \put(61.2,6.5){\tiny 1.5}
    \put(78.5,6.5){\tiny 2.0}
     \put(54,3){\tiny $\kappa$}
    \put(3.5,9){\tiny 0.00}
    \put(3.5,23){\tiny 0.02}
    \put(3.5,37){\tiny 0.04}
    \put(3.5,51){\tiny 0.06}
    \put(3.5,65){\tiny 0.08}
    \put(-2,37){\tiny $\mu$}
\end{overpic}}
\\[0.5em]
\subfigure[$\Tt_{1,2}$]{
\begin{overpic}[width=0.45\textwidth]{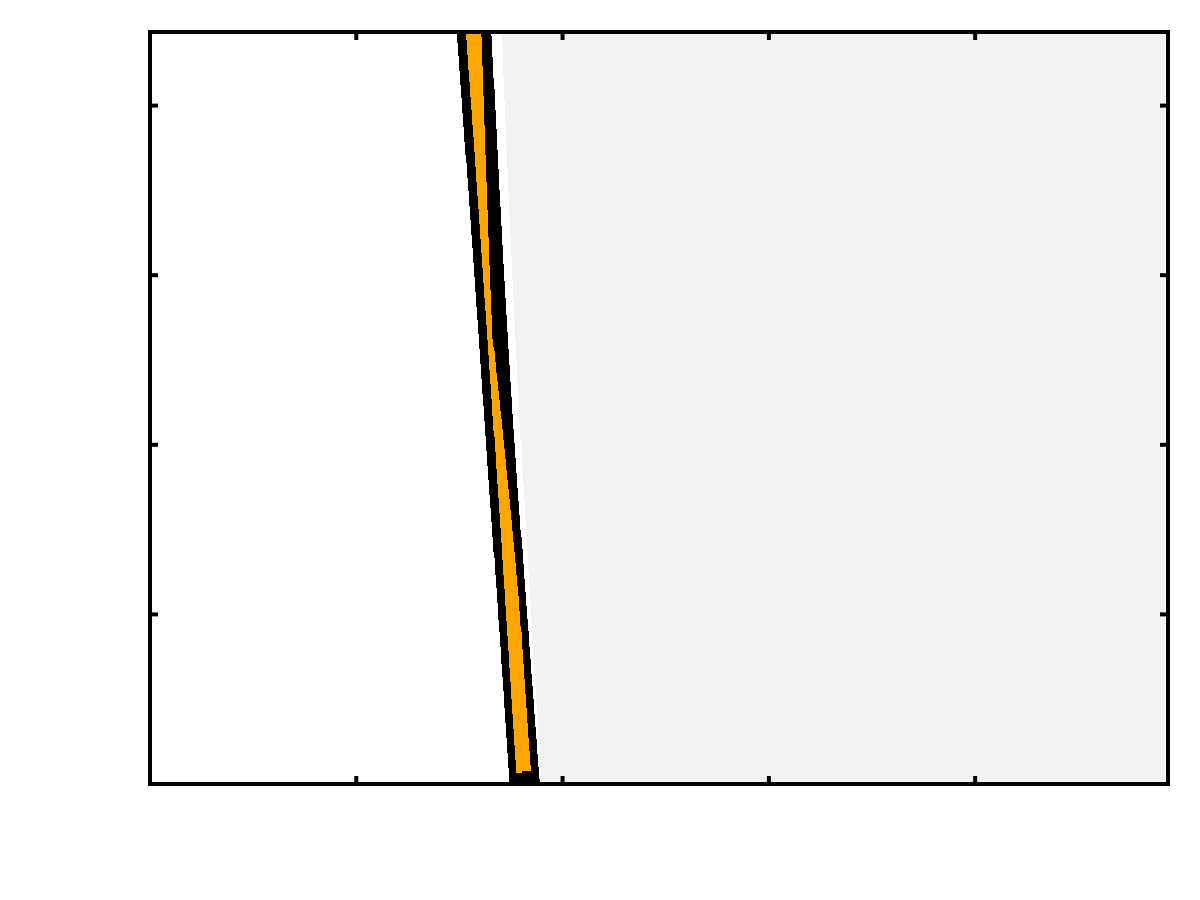}
    \put(9.4,6.5){\tiny 0.0}
    \put(26.7,6.5){\tiny 0.5}
    \put(44,6.5){\tiny 1.0}
    \put(61.2,6.5){\tiny 1.5}
    \put(78.5,6.5){\tiny 2.0}
     \put(54,3){\tiny $\kappa$}
    \put(3.5,9){\tiny 0.00}
    \put(3.5,23){\tiny 0.02}
    \put(3.5,37){\tiny 0.04}
    \put(3.5,51){\tiny 0.06}
    \put(3.5,65){\tiny 0.08}
    \put(-2,37){\tiny $\mu$}
\end{overpic}}
\subfigure[$\Tt_{1,2}$]{
\begin{overpic}[width=0.45\textwidth]{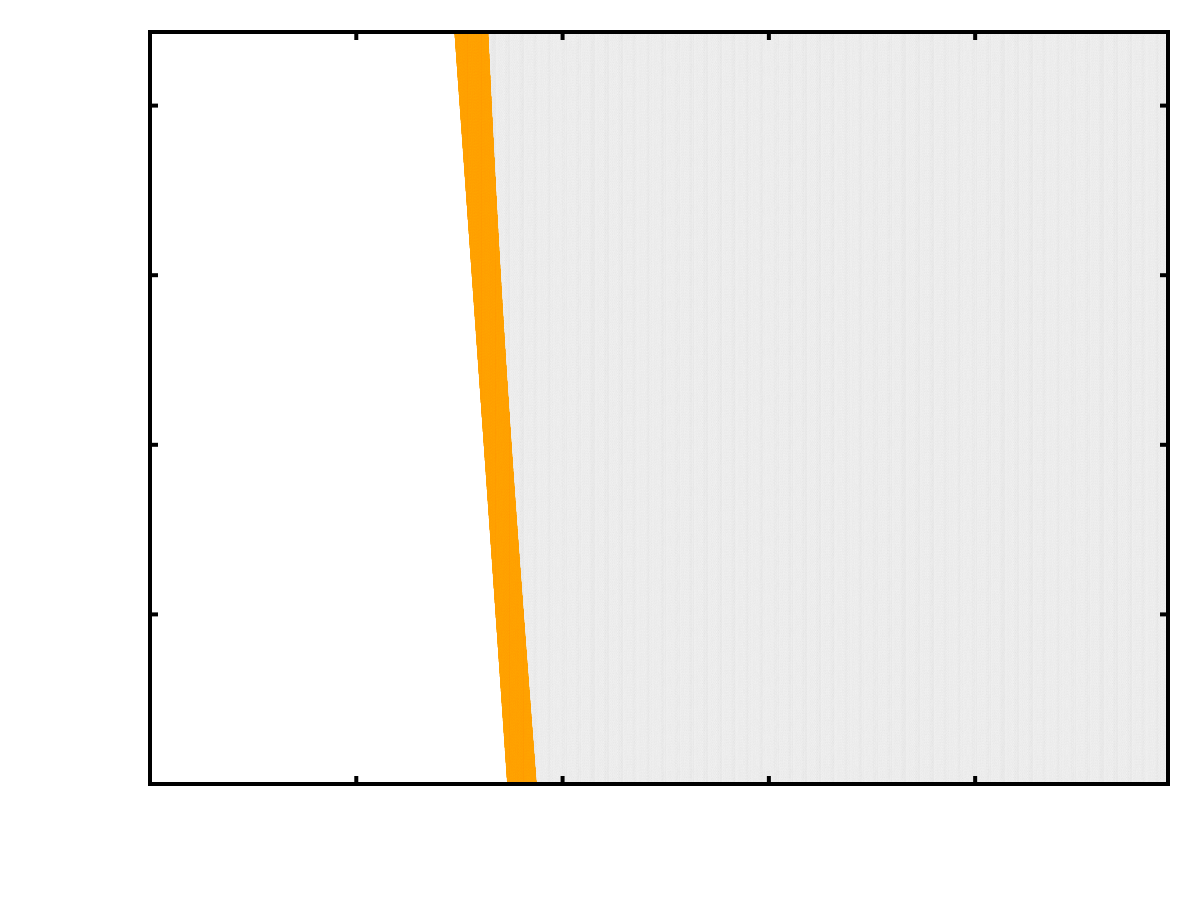}
    \put(9.4,6.5){\tiny 0.0}
    \put(26.7,6.5){\tiny 0.5}
    \put(44,6.5){\tiny 1.0}
    \put(61.2,6.5){\tiny 1.5}
    \put(78.5,6.5){\tiny 2.0}
     \put(54,3){\tiny $\kappa$}
    \put(3.5,9){\tiny 0.00}
    \put(3.5,23){\tiny 0.02}
    \put(3.5,37){\tiny 0.04}
    \put(3.5,51){\tiny 0.06}
    \put(3.5,65){\tiny 0.08}
    \put(-2,37){\tiny $\mu$}
\end{overpic}}
\end{minipage}%
}
\captionsetup{width=0.75\textwidth}
\caption{Stability CAPs (panels (a), (c)) and numerical investigations (panels (b), (d)) of triangular RE in the case of positive curvature, as a function of the parameters $(\kappa,\mu)\in \Ps{\muT}$.}
\label{fig:CAPstabilitytriangular}
\end{figure}

We summarize  the results in Fig.~\ref{fig:CAPstabilitytriangular} in the following. 

\begin{conjecture}[Stability classification of triangular RE as a function of $\kappa>0$  for  $0<\mu<\mu_{\Tt}$]
\label{thm:class-stability-triangular}
Let $\mu \in (0,\mu_{\Tt})$. The stability of triangular RE is as follows. 
\begin{description}
    \item[$\bullet$] If $0<\mu < \mu_R$, then    $\Ll_4$ and $\Ll_5$ 
    are elliptic for all values $\kappa>0$ for which they exist.
     \item[$\bullet$] If $\mu_R<\mu< \muT$, there exists $\kappa_s(\mu)>0$ such that
       $\Ll_4$ and $\Ll_5$ 
    are complex--saddles if $0<\kappa<  \kappa_s(\mu)$ and instead elliptic 
    for all other $\kappa>  \kappa_s(\mu)$ for which they exist. Furthermore, $\kappa_s(\mu)$ is
    an increasing function of $\mu$ that converges to $0$ as $\mu\to \mu_R$.
 \item[$\bullet$]      $\Tt_1$ and $\Tt_2$ are center--saddles for all $\kappa>0$ for which they exist.
     \end{description}
\end{conjecture}

We stress that we do not have a rigorous  proof of the existence of the bifurcation curve $\mu\mapsto (\kappa_s(\mu),\mu)$. We also do not have rigorous proofs of stability 
for the parameter values $(\kappa,\mu)\in \Ps{\muT}$ for
which the CAPs in panels (a) and (c) of Fig.~\ref{fig:CAPstabilitytriangular} are inconclusive. 

\begin{remark}
\label{rmk:kappasnegative}
The threshold value $\kappa_{s}(\mu)$ appearing in the second item of Conjecture~\ref{thm:class-stability-triangular} is shown explicitly in Fig.~\ref{fig:asymptriangular}(b) for $\mu\approx0.06$. For $0<\mu<\mu_{R}$, numerical evidence suggests the existence of a negative value of the curvature, which we also denote by $\kappa_{s}$, at which $\Ll_{4}$ and $\Ll_{5}$ transition from ellipticity to instability as $\kappa$ decreases; Figure~\ref{fig:asymptriangular}(a) illustrates this transition for $\mu\approx0.01$. The curve $\mu\mapsto (\mu,\kappa_{s}(\mu))$
in the parameter space separates the red and green regions in Fig.~\ref{fig:stabilityL4}.
\end{remark}

\subsection{Gyroscopic stabilization of \texorpdfstring{$\Ll_{4}$}{L4} and \texorpdfstring{$\Ll_{5}$}{L5}}
\label{ss:stabilityL4L5}

Finally, we emphasize  the effect of curvature in the gyroscopic stabilization of the
triangular points $\Ll_{4}$ and $\Ll_{5}$ by presenting Fig.~\ref{fig:stabilityL4}, obtained
numerically, in which the stability properties of these points is illustrated as a function 
of parameter values $\kappa, \mu$, with $\kappa$ passing from negative to positive. Since the critical mass ratio $\mu_{\Tt}$ is unrelated to the triangular points $\Ll_{4,5}$, we extend the range of $\mu$ including values greater than $\mu_{\Tt}$ in this figure.} The vertical black line in the figure corresponds to $\kappa=0$, and the horizontal one to
Routh's critical mass
ratio $\mu_R$ in Eq.~\eqref{eq:routh}. These lines intersect at the boundary between the red and green regions,
since $\mu_R$ is the threshold value for linear stability of the planar problem.

The figure clearly illustrates that positive curvature promotes gyroscopic stabilization, whereas negative curvature diminishes the parameter region corresponding to stability.

\begin{figure}[H]
\centering
\makebox[\textwidth][c]{
\begin{minipage}[c]{0.28\textwidth}
\centering
\begin{tabular}{@{}l@{\hspace{0.6em}}l@{}}
\legenditem{008000}{Elliptic} \\
\legenditem{FF0000}{Complex saddle} \\
\end{tabular}
\end{minipage}
\hspace{0.05\textwidth}
\begin{minipage}[c]{0.4\textwidth}
\centering
\begin{overpic}[width=\textwidth]{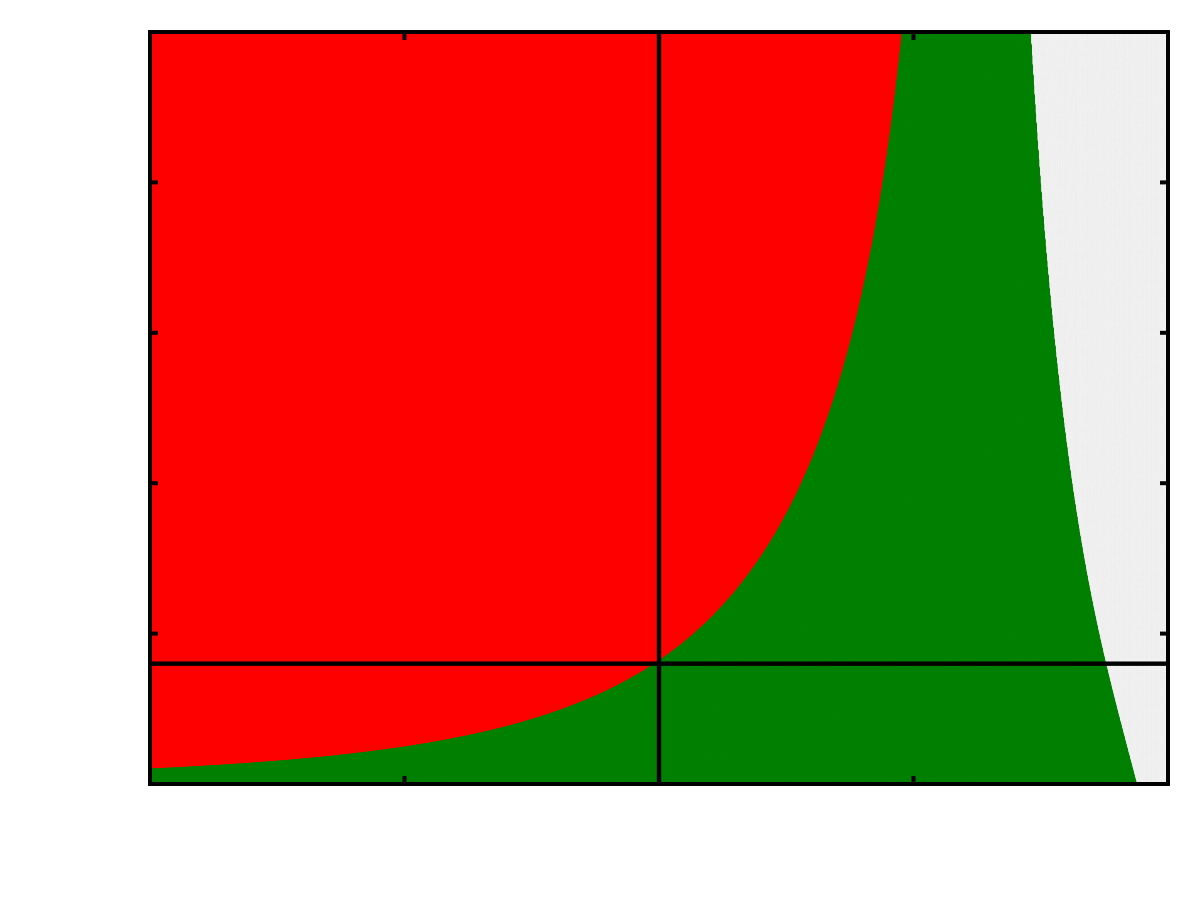}
	\put(9,6){\tiny -1.0}
	\put(30,6){\tiny -0.5}
	\put(52.5,6){\tiny 0.0}
	\put(73.75,6){\tiny 0.5}
	\put(95,6){\tiny 1.0}
	\put(54,3){\tiny $\kappa$}
	\put(5,8.5){\tiny 0.00}
	\put(5,21.5){\tiny 0.05}
	\put(8,19){\scalebox{0.5}{$\mu_{R}$}}
	\put(5,34){\tiny 0.10}
	\put(5,46.5){\tiny 0.15}
	\put(5,59){\tiny 0.20}
	\put(-0.2,37){\tiny $\mu$}
\end{overpic}
\end{minipage}
}
\captionsetup{width=0.75\textwidth}
\caption{Linear stability of $\Ll_{4}$ and $\Ll_{5}$ as a function of the parameters
$\kappa$ and $\mu$, for positive and negative curvature.}
\label{fig:stabilityL4}
\end{figure}
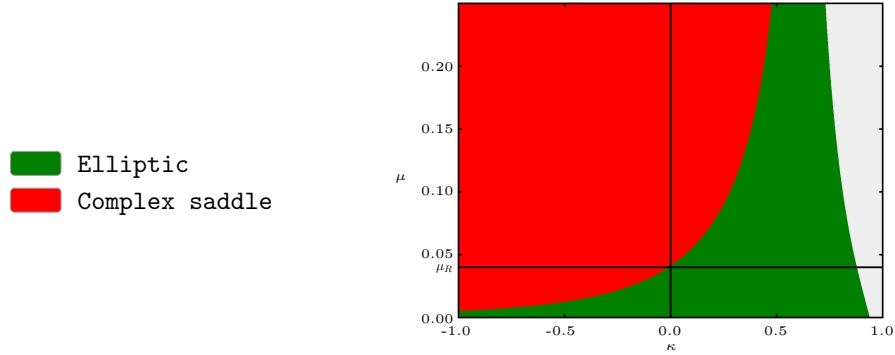

\section{\texorpdfstring
  {Asymptotic behavior of relative equilibria for small $\kappa$ and $\mu$.}
  {Asymptotic behavior of relative equilibria for small kappa and mu.}}
\label{sec:asymp}

We now provide asymptotic expansions that yield information about the locations of the RE
of the curved R3BP in the limit  $\kappa, \mu \to 0$. 
In our discussion we allow for both positive and negative curvature. 
The details behind the derivation of these asymptotic expansions are given  Sec.~\ref{app:appendix}
 of Appendix~\ref{app:analytic-proofs}.

In the statement of the expansions in this section, for a function $f(\kappa,\mu)$, which need not be defined at $\kappa=0$,  and
$n\in \mathbb{N}$,    the notation
\begin{equation*}
f(\kappa,\mu)=\mathcal{O}_{\mu}\, ( |\kappa|^{n/2}), 
\end{equation*}
means that there exists a continuous function  $C:(0,1)\to \R$ such that 
\begin{equation*}
|f(\kappa,\mu)|\leq  C(\mu)|\kappa|^{n/2},
\end{equation*}
for sufficiently small non-zero $|\kappa|$. In some cases, the same notation will be 
used for values of $\kappa$ that are either strictly positive, with obvious adaptations of our convention.

\subsection{Collinear RE}

Let  $(\kappa,\mu)$ be parameter values for which the collinear point $\Zz$ exists with $\kappa\neq 0$. Here,   $\Zz$ denotes any of the points 
for $\Ll_1$, $\Ll_2$, or $\Ll_3$ in Proposition~\ref{prop:ClassCollinearNeg} if $\kappa< 0$,  and any of the points  $\Ll_1$, $\Ll_2$, $\Ll_3$, $\Ee_2$, $\Ee_3$, or $\Aa_1$  in Theorem~\ref{thm:numbercollinearRE} if $\kappa>0$. We  denote by 
\begin{equation*}
D_{\Zz,C}(\kappa,\mu),
\end{equation*}
the corresponding signed Riemannian distance from  $\Zz$ to the center of rotation $C$. The sign is determined  by the position of 
$\Zz$ relative to  $C$ along the geodesic $\mathcal{G}^{\pm}$, according to the following convention: the distance is positive if
 $\Zz$ lies in the segment containing the light primary ${\bf p}_2$ and negative if $\Zz$ lies in the segment containing the heavy primary ${\bf p}_1$. In the case of positive $\kappa$, where the geodesic is a circle, the appropriate segment is determined by choosing the
 smallest arc connecting  $\Zz$ and $C$. 
 
\begin{theorem}
\label{prop:asymLpositive}
The following 
asymptotic expansions hold as $\kappa\to 0$ and  $\mu \to 0^+$:
\begin{equation}
\label{eq:asymp-L1L2L3}
\begin{split}
    D_{\Ll_{1},C}(\kappa,\mu) &= 1 - \left(\frac{1}{3}\right)^{1/3}\mu^{1/3} + \mathcal{O}(\mu^{2/3}) + \mathcal{O}_\mu(|\kappa|^{1/2}), \\
    D_{\Ll_{2},C}(\kappa,\mu) &= 1 + \left(\frac{1}{3}\right)^{1/3}\mu^{1/3} + \mathcal{O}(\mu^{2/3}) + \mathcal{O}_\mu(|\kappa|^{1/2}), \\
    D_{\Ll_{3},C}(\kappa,\mu) &= -1 - \frac{5}{12}\mu + \mathcal{O}(\mu^2) + \mathcal{O}_\mu(|\kappa|^{1/2}).
\end{split}
\end{equation}
Additionally, restricting $\kappa>0$, the following asymptotic expansions hold as $\kappa\to 0^+$ and  $\mu \to 0^+$
\begin{equation}
\label{eq:asymp-E2E3A1}
\begin{split}
    D_{\Aa_{1},C}(\kappa,\mu) &= -\frac{\pi}{\sqrt{\kappa}} + \left(1 - \frac{1}{\sqrt{2}}\mu^{1/2} - \frac{7}{8}\mu + \mathcal{O}(\mu^{3/2}) \right) + \mathcal{O}_{\mu}(\kappa),\\
    D_{\Ee_{2},C}(\kappa,\mu) &= \frac{\pi}{2\sqrt{\kappa}} - \kappa + \mathcal{O}_{\mu}(\kappa^{3/2}),\\
    D_{\Ee_{3},C}(\kappa,\mu) &= -\frac{\pi}{2\sqrt{\kappa}} + \kappa + \mathcal{O}_{\mu}(\kappa^{3/2}).
\end{split}
\end{equation}
\end{theorem}

 Comparing the expansions in Eqs.~\eqref{eq:asymp-L1L2L3} with formulas
 \eqref{eq:planarcollinear} shows that the collinear points $\Ll_1$, $\Ll_2$ and $\Ll_3$ of the curved R3BP converge 
 to the classical collinear points of the planar problem. On the other hand, the expansions in Eqs.~\eqref{eq:asymp-E2E3A1},
 show that the distances from the collinear points  $\Aa_{1}$, $\Ee_{2}$, and $\Ee_{3}$ of Theorem~\ref{thm:numbercollinearRE} to the center of rotation $C$ becomes unbounded as $\kappa\to 0^+$. This explains
 why these RE do not exist when $\kappa=0$. 
 
 The behavior predicted by these asymptotic expansions can be appreciated in the bifurcation diagram in Fig.~\ref{fig:asympcollinear},
 obtained numerically,  where the quantities $D_{\Zz,C}(\kappa,\mu)$ are plotted against $\kappa>0$ for a fixed value $\mu\approx 0.01$.
 
 \subsection{Triangular RE}
 
 We now consider the triangular points $\Ll_4$, $\Ll_5$ in Proposition~\ref{th:existence-neg-triangular} for $\kappa<0$ and in Theorem~\ref{thm:numbertriangularRE} for $\kappa>0$. We denote by
 \begin{equation*}
D_{\Ll_{4},{\bf p}_1}(\kappa,\mu)= D_{\Ll_{5},{\bf p}_1}(\kappa,\mu), \qquad D_{\Ll_{4},{\bf p}_2}(\kappa,\mu)= D_{\Ll_{5},{\bf p}_2}(\kappa,\mu),
\end{equation*}
their Riemannian distances from $\Ll_4$ and $\Ll_5$ to each of the primaries. Note that the above pair of 
equalities  hold by the reflectional symmetry of
the problem with respect to the geodesic $\mathcal{G}^+$.

\begin{theorem}
\label{th:asympexpL4L5}
For any $\mu\in (0,1)$, the following asymptotic expansions hold as $\kappa\to 0$.
\begin{equation*}
\begin{split}
  D_{\Ll_{4},{\bf p}_1}(\kappa,\mu)&= D_{\Ll_{5},{\bf p}_1}(\kappa,\mu) = 1 + \frac{\mu}{6(1+\mu)}\vert\kappa\vert + \mathcal{O}_{\mu}(\kappa^{2}), \\
  D_{\Ll_{4},{\bf p}_2}(\kappa,\mu)&= D_{\Ll_{5},{\bf p}_2}(\kappa,\mu) = 1 + \frac{1}{6(1+\mu)}\vert\kappa\vert + \mathcal{O}_{\mu}(\kappa^{2}).
  \end{split}
\end{equation*}
\end{theorem}

In particular, note that the Riemannian distances from $\Ll_4$ and $\Ll_5$ to both  primaries converge to $1$ in the limit
as $\kappa\to 0$. This agrees with the equilateral nature of the triangular configurations 
formed by  $\Ll_4$, $\Ll_5$, and the two primaries in the 
planar case.

The behavior described above is evident  in the bifurcation diagrams shown in Fig.~\ref{fig:asymptriangular},
which were  obtained numerically. These diagrams plot  the quantity $D_{\Ll_4,{\bf p}_2}(\kappa,\mu)=D_{\Ll_5,{\bf p}_2}(\kappa,\mu)$ 
against $\kappa>0$ for the  fixed mass ratios $\mu\approx 0.01$ and $\mu\approx 0.06$.

\section{Future work}
\label{sec:future-work}

Below we indicate some natural directions of future research motivated by our work. These involve the 
R3BP and the full-three body problem on spaces of constant positive curvature. 
\begin{enumerate}
\item Complete the proofs of Conjectures \ref{thm:classcollinear}, \ref{thm:classtriangular}, \ref{thm:class-stability-collinear}, and
\ref{thm:class-stability-triangular}.
\item Investigate the  bifurcation in the parameter space occurring at $(\muT,\kappa_{\Tt})$ (see Sec.~\ref{sec:muT}).
\item Provide a detailed description of the  bifurcation of the RE for mass ratios $\muT<\mu<1$.
\item Investigate the nonlinear stability of the elliptic RE that we found 
using KAM techniques. In the appropriate parameter
ranges, these are $\Ll_3$,  $\tilde{\Ll}_3$, $\Ll_4$, $\Ll_5$.
\item Investigate the continuation of the elliptic/Lyapunov-stable RE that we found, to RE of the 
full-3 body problem on spaces of positive constant curvature 
as the mass of the satellite becomes positive. It is reasonable to expect that the RE 
found in this way will be elliptic/Lyapunov stable. 
\end{enumerate}

\subsection*{Funding} CBA and LGN acknowledge support from the project MIUR-PRIN 2022FPZEES Stability in Hamiltonian
dynamics and beyond.

\subsection*{Author Contributions} CBA and LGN developed the main theoretical framework and established the analytical results. 
MA implemented the computer-assisted proofs and carried out the numerical investigations. All authors contributed to the interpretation of the results, the writing of the manuscript, and the approval of the final version.

\subsection*{Data Availability} The code to reproduce the CAPs presented in our paper is available in \cite{github_codes}.

\subsection*{AI use statement} The authors used a generative AI tool to assist with language editing and stylistic improvements during manuscript preparation. The AI was not used to generate mathematical results, proofs, or scientific conclusions. All technical content was written, checked, and approved by the authors, who take full responsibility for the final manuscript.

\subsection*{Conflict of Interest Declaration} The authors declare that they have no conflicts of interest.

\begin{appendices}

\section{Stability of equilibria of  Lagrangian Systems with two degrees of freedom.}
\label{ap:linear}

This appendix reviews the stability criteria for equilibrium points  of Lagrangian systems of the type arising in the  curved R3BP. Namely, the Lagrangian 
has 2 degrees of freedom, consists of   the  kinetic energy term, derived from a Riemannian metric, a gyroscopic term which is linear in the velocities, and a potential. Since
the linearization around an equilibrium is a local feature of the problem, upon a choice of local coordinates,  we may assume that our configuration space is an open set  $U\subseteq\mathbb{R}^2$. 
Consider then the  Lagrangian system defined by $L : TU=U \times \mathbb{R}^2 \to \mathbb{R}$ given as
\begin{equation}
\label{eq:laggeneral}
    L(\mathbf{x},\dot{\mathbf{x}}) = \frac{1}{2} \langle \dot{\mathbf{x}}, \mathbb{M}(\mathbf{x})\dot{\mathbf{x}}\rangle + \langle\Sigma(\mathbf{x}), \dot{\mathbf{x}}\rangle - V(\mathbf{x}),
\end{equation}
where $\mathbb{M}(\mathbf{x})$ is a symmetric, positive-definite $2 \times 2$ matrix representing the kinetic energy metric, $\Sigma(\mathbf{x})\in\R^{2}$ encodes gyroscopic effects, and $V : U \to \mathbb{R}$ is a scalar potential.
The Euler--Lagrange equation associated with $L$ is given by
\begin{equation}
\label{eq:ELeqs}
    \frac{\mathrm{d}}{\mathrm{d}t} \frac{\partial L}{\partial \dot{\mathbf{x}}} - \frac{\partial L}{\partial \mathbf{x}} = 0.
\end{equation}
It is well-known that all equilibrium points of these equations are of the form $(\mathbf{x}_0, \mathbf{0})\in U\times \R^2$, where $\mathbf{x}_0$ is a critical point of $V$, i.e. 
\begin{equation*}
    \nabla V(\mathbf{x}_0) = 0.
\end{equation*}

The linearization around  $(\mathbf{x}_0, \mathbf{0})$ is obtained by first rewriting Eqs.~\eqref{eq:ELeqs} as a first-order system. 
Letting $\mathbf{u} = (\mathbf{x}, \dot{\mathbf{x}})$ an using that $\mathbb{M}(\mathbf{x})$ is invertible, we can write
\begin{equation*}
    \dot{\mathbf{u}} = F(\mathbf{u}),
\end{equation*}
for an appropriate vector function $F$. The linearized system at $\mathbf{u}_0 \coloneqq (\mathbf{x}_0, 0)$ becomes
\begin{equation*}
    \dot{\delta \mathbf{u}} = F'(\mathbf{u}_0) \delta \mathbf{u},
\end{equation*}
where $F'(\mathbf{u}_0)$ denotes the Jacobian matrix of the vector function $F$ at $\mathbf{u}_0$. The explicit form of this matrix is given next, where we denote by  $\mathbb{J}$  the $2\times 2$ symplectic matrix
\begin{equation}
\label{eq:symplectic-matrix}
\mathbb{J}=\begin{pmatrix} 0 & -1 \\ 1 &0 \end{pmatrix}.
\end{equation}

\begin{proposition}
\label{prop:linearization}
Let  $\mathbf{u}_0 = (\mathbf{x}_0, 0)$ be an equilibrium point of the Euler--Lagrange equations \eqref{eq:ELeqs} for the Lagrangian function $L$ given by formula \eqref{eq:laggeneral}. Then, the linearization at $\mathbf{u}_0$ is given by the $4\times 4$ matrix
of block form
\begin{equation}
\label{eq:blockgeneral}
    F'(\mathbf{u}_0)
    =
    \begin{pmatrix}
        \mathbb{O} & \mathbb{I} \\
        -\mathbb{M}(\mathbf{x}_0)^{-1} V''(\mathbf{x}_0)
        &
        - \sigma(\mathbf{x}_0) \mathbb{M}(\mathbf{x}_0)^{-1} \,  \mathbb{J}
    \end{pmatrix},
\end{equation}
where $\mathbb{O}$ and $\mathbb{I}$ denote the $2 \times 2$ zero and identity matrices,   $V''(\mathbf{x}_0)$ is the Hessian matrix of the potential $V$
 evaluated at $\mathbf{x}_0$, and  $\sigma(\mathbf{x}_0)\in \R$ is determined by the condition 
 \begin{equation}
\label{eq:formskew}
\Sigma'(\mathbf{x}_0) - \Sigma'(\mathbf{x}_0)^{T} = \sigma(\mathbf{x}_0)
\mathbb{J},
\end{equation}
 where $\Sigma'(\mathbf{x}_0)$ is the Jacobian matrix of $\Sigma$ evaluated at $\mathbf{x}_0$.
\end{proposition}

A proof of the previous result is given in \cite[Sec. 3--3]{G97}. The precise form of the characteristic polynomial
of the above matrix is given next.
\begin{lemma}
\label{l:charpolmaxima}
The characteristic polynomial, $p(\lambda)$, of the block matrix $F'(\mathbf{u}_{0})$ given in expression \eqref{eq:blockgeneral} is
bi-quadratic
\begin{equation}
\label{eq:biquadpol}
    p(\lambda) = \lambda^{4}
    + b \lambda^{2}
    + c,
\end{equation}
with coefficients $b$ and $c$ given by
\begin{equation}
\label{eq:b-c-pol}
    b = \operatorname{tr}\big(\mathbb{M}(\mathbf{x}_{0})^{-1} V''(\mathbf{x}_{0})\big)
      +  \sigma(\mathbf{x}_{0})^{2} \det\big(\mathbb{M}(\mathbf{x}_{0})\big)^{-1}, \qquad c=\det\big(\mathbb{M}(\mathbf{x}_{0})^{-1} V''(\mathbf{x}_{0})\big).
\end{equation}
Furthermore,  $p(\lambda)$ is invariant under changes of Lagrangian coordinates.
\end{lemma}

\begin{proof}
The characteristic polynomial of a general $4 \times 4$ matrix $M$ can be written as
\begin{equation*}
    p_{M}(\lambda) = \lambda^{4} + a_{3}\lambda^{3} + a_{2}\lambda^{2} + a_{1}\lambda + a_{0},
\end{equation*}
where the coefficients are expressed in terms of traces of powers of $M$ and its determinant:
\begin{equation*}
\begin{aligned}
    a_{0} &= \det(M),\\
    a_{1} &= \tfrac{1}{6} \Big( \mathrm{tr}(M)^{3} - 3\,\mathrm{tr}(M)\,\mathrm{tr}(M^{2}) + 2\,\mathrm{tr}(M^{3}) \Big),\\
    a_{2} &= \tfrac{1}{2} \Big( \mathrm{tr}(M)^{2} - \mathrm{tr}(M^{2}) \Big),\\
    a_{3} &= -\,\mathrm{tr}(M).
\end{aligned}
\end{equation*}
These coefficients can be deduced by applying Leverrier's method (see \cite[Sec.~25]{Fa59}). 
Assume now, as in expression \eqref{eq:blockgeneral}, that $M$ has the block structure
\begin{equation*}
    M=\begin{pmatrix} 
         \mathbb{O} & \mathbb{I} \\ 
        RS & RA 
    \end{pmatrix},
\end{equation*}
where
\begin{equation*}
    R=-\mathbb{M}(\mathbf{x}_{0})^{-1}, 
    \qquad 
    S=V''(\mathbf{x}_{0}), 
    \qquad 
    A=\sigma(\mathbf{x}_{0})\mathbb{J}.
\end{equation*}
With this notation, we have
\begin{equation*}
    M^{2}=\begin{pmatrix} 
        RS & RA \\ 
        RARS & RS + (RA)^{2} 
    \end{pmatrix}, 
    \qquad 
    M^{3}= \begin{pmatrix} 
        RARS & RS + (RA)^{2} \\ 
        (RS)^{2} + (RA)^{2}RS  & RSRA + RARS + (RA)^{3} 
    \end{pmatrix}.
\end{equation*}
Since $R$ is symmetric and $A$ is skew-symmetric, we have,
\begin{equation*}
    \mbox{tr}(RA) = \mbox{tr}((RA)^{T}) =  \mbox{tr}(A^{T}R^{T}) =  -\mbox{tr}(AR) = - \mbox{tr}(RA),   
\end{equation*}
which implies that $\mbox{tr}(RA)=0$. Proceeding analogously, exploiting also the symmetry of $S$,  it is straightforward to show that
\begin{equation*}
   \mbox{tr}(RARS)  = \mbox{tr}(RSRA)  = 0, \qquad  \mbox{tr}( (RA)^{3})=0,
\end{equation*}
and therefore,
\begin{equation*}
     \mbox{tr}(M) =  \mbox{tr}(M^{3}) = 0,
\end{equation*}
which immediately implies $a_{1} = a_{3} = 0$. 
It remains to compute $a_{2}$, namely
\begin{equation*}
    a_{2} = -\tfrac{1}{2}\,\mathrm{tr}(M^{2})
          = -\,\mathrm{tr}(RS) - \tfrac{1}{2} \mbox{tr} \big[(RA)^{2}\big].
\end{equation*}
This reduces to the desired form 
\begin{equation*}
    a_{2} =\mbox{tr} \big(\mathbb{M}(\mathbf{x}_{0})^{-1} V''(\mathbf{x}_{0}) \big) 
            + \sigma(\mathbf{x}_{0})^{2}\det\!\big(\mathbb{M}(\mathbf{x}_{0})\big)^{-1},
\end{equation*}
by substituting the matrix values of $R, S$ and $A$, and noting that, in virtue of the symmetry of $R$, we have
\begin{equation*}
\mbox{tr}((R\mathbb{J})^2)=-2\det(R).
\end{equation*}

Finally, a smooth change of Lagrangian coordinates 
\begin{equation*}
\mathbf{y}=\psi(\mathbf{x}), \qquad \dot{\mathbf{y}}=\psi'(\mathbf{x}) \dot{\mathbf{x}},
\end{equation*}
defined by a diffeomorphism $\psi:U\to W:=\psi(U)\subseteq \R^2$, maps the equilibrium point $\mathbf{u}_{0}=(\mathbf{x}_{0},\mathbf{0})$ to $\mathbf{w}_{0}=(\mathbf{y}_{0},\mathbf{0})\in W\times \R^2$ where $ \mathbf{y}_{0}=\psi(\mathbf{x}_{0})$. Then  $\mathbf{w}_{0}$ is an equilibrium
of the Euler--Lagrange equations written in terms of $(\mathbf{y}, \dot{\mathbf{y}})$.
The linearization matrix of the corresponding  first order system, $\dot{\mathbf{w}}=G(\mathbf{w})$, can be checked to be given
by
 \begin{equation*}
G'(\mathbf{w}_0)=\begin{pmatrix} \psi'(\mathbf{x}_0) &  \mathbb{O}\\   \mathbb{O} & \psi'(\mathbf{x}_0)\end{pmatrix} F'(\mathbf{u}_0) \begin{pmatrix} \psi'(\mathbf{x}_0) &  \mathbb{O}\\   \mathbb{O} & \psi'(\mathbf{x}_0)\end{pmatrix}^{-1}.
\end{equation*}
Hence, the matrices $F'(\mathbf{u}_0)$ and $G'(\mathbf{w}_0)$ are similar, and their
characteristic polynomials coincide.
\end{proof}

The main stability classification result that we need is given below. We employ the terminology introduced in Sec.~\ref{subsec:stability}.

\begin{theorem} 
\label{thm:criticalpointsV}
Let $\mathbf{u}_{0}=(\mathbf{x}_{0},0)$ be an equilibrium point of the Lagrangian system \eqref{eq:laggeneral}. Then, according to the values of $b$, $c$ in formulas \eqref{eq:biquadpol}
and \eqref{eq:b-c-pol}, we have:
\begin{enumerate}
	\item[(i)] $\mathbf{u}_{0}$ is elliptic if and only if $b^2-4c\geq 0$, $b>0$ and $c>0$.
	\item[(ii)]  $\mathbf{u}_{0}$ is a center--saddle if and only if  $c<0$.
	\item[(iii)] $\mathbf{u}_{0}$ is a complex--saddle if and only if  $b^2-4c< 0$.
	\item[(iv)]  $\mathbf{u}_{0}$ is a saddle--saddle if and only if  $b^2-4c\geq 0$, $b<0$  and $c>0$.
\end{enumerate}
Furthermore, 
\begin{enumerate}
	\item[(v)] If $\mathbf{x}_{0}$ is a local minimum of $V$, then $\mathbf{u}_{0}$ is Lyapunov--stable.
	\item[(vi)] $\mathbf{u}_{0}$ is a center--saddle if and only if $\mathbf{x}_{0}$ is a saddle--point of $V$.
\end{enumerate}

\end{theorem}
\begin{proof}
The first four items follow from an elementary analysis of the roots of the bi-quadratic polynomial \eqref{eq:biquadpol}
as a function of the coefficients $b$ and $c$. For item (vi) note that the coefficient $c$ in formula \eqref{eq:b-c-pol} has the same
sign as the determinant of the Hessian matrix $V''(\mathbf{x}_{0})$. In particular, $c$ is negative if and only if 
$\mathbf{x}_{0}$ is a saddle--point of $V$. Finally, the proof of  (v) does not require the linearization: 
since $L$ is time-independent, the system admits the Jacobi (or energy) integral 
\begin{equation}
\label{eq:energy}
    E(\mathbf{x},\dot{\mathbf{x}})=\frac{1}{2} \langle \dot{\mathbf{x}}, \mathbb{M}(\mathbf{x})\dot{\mathbf{x}}\rangle + V(\mathbf{x}),    
\end{equation}
The assumption that $V$ has a local minimum at $\mathbf{x}_{0}$ implies that $E$ has a local minimum at $\mathbf{u}_{0}$ and,
therefore, $E$ can be used  as a Lyapunov function. 
\end{proof}

\section{Analytical Proofs}
\label{app:analytic-proofs}

This appendix contains the proofs of Theorems~\ref{thm:numbercollinearRE}, \ref{thm:numbertriangularRE}, and~\ref{thm:stabilitycollinear}, as well as the derivation of the asymptotic expansions in Theorems~\ref{prop:asymLpositive} and~\ref{th:asympexpL4L5}.

\subsection{Proof of Theorem~\ref{thm:numbercollinearRE} on existence of collinear RE.}
\label{app:ProofExistenceCollinear}
We first prove items 1 and 4, and then establish items 2 and  3 using Proposition~\ref{prop:concavity}.

\medskip
1. We must show that the function $\theta\mapsto f_1(\theta;\mu,\kappa)$ given by
\begin{equation}
\label{eq:f1}
    f_{1}(\theta;\kappa,\mu) = \frac{1}{2}\sin(2\theta) - \Gamma \left[ \frac{1}{\sin^{2}(\theta+\sqrt{\kappa}q_{1})} - \frac{\mu}{\sin^{2}(\theta-\sqrt{\kappa}q_{2})} \right], \qquad \theta \in \mathcal{I}_{1} =(-\sqrt{\kappa} q_{1}, \sqrt{\kappa} q_{2}),
\end{equation}
has exactly one root assuming that $0 < \kappa < \pi^{2}/16$. It is easily seen that $f_1$ tends to $-\infty$ as
$\theta$ approaches the left endpoint of $\mathcal{I}_{1}$ and to $+\infty$ at the right endpoint. Hence, by the intermediate value theorem, $f_{1}$ has at least one root in this interval. To prove that the root is unique,
we show that $f_1$ is monotone by showing that its derivative satisfies
\begin{equation}
\label{eq:auxf1prime}
    f_{1}'(\theta;\kappa,\mu) = \cos(2\theta) + 2\Gamma \left[ \frac{\cos(\theta+\sqrt{\kappa}q_{1})}{\sin^{3}(\theta+\sqrt{\kappa}q_{1})} - \mu\,\frac{\cos(\theta-\sqrt{\kappa}q_{2})}{\sin^{3}(\theta-\sqrt{\kappa}q_{2})} \right]>0 \qquad \mbox{for all $\theta \in \mathcal{I}_{1}$}. 
\end{equation}

Using that $q_{1}+q_{2}=1$ and $\theta \in \mathcal{I}_{1} =(-\sqrt{\kappa} q_{1}, \sqrt{\kappa} q_{2})$, we deduce
\begin{equation*}
     0 < \theta+\sqrt{\kappa} q_{1} < \sqrt{\kappa}, \qquad 
      -\sqrt{\kappa} < \theta - \sqrt{\kappa} q_{2} < 0,
\end{equation*}
which, since $\kappa<\pi^{2}/16<\pi^{2}/4$, imply
\begin{equation*}
    0 < \theta+\sqrt{\kappa}q_{1} < \pi/2, \qquad \mbox{and} \qquad -\pi/2 < \theta-\sqrt{\kappa}q_{2} < 0.
\end{equation*}
Therefore, under our assumptions on $\kappa$, the following inequalities hold for all $\theta \in \mathcal{I}_1$
\begin{equation}
\label{eq:estimatesdomainf1}
\begin{gathered}
    \cos(\theta+\sqrt{\kappa}q_{1})>0, \qquad \cos(\theta-\sqrt{\kappa}q_{2})>0, \\
    \sin(\theta+\sqrt{\kappa}q_{1})>0, \qquad \sin(\theta-\sqrt{\kappa}q_{2})<0.
\end{gathered}
\end{equation}

Moreover, using that $\kappa < \pi^{2}/16$ and $q_{1}, q_{2} < 1$, we have 
\begin{equation*}
\begin{gathered}
     -\pi/4<-\sqrt{\kappa}q_{1},\qquad  \sqrt{\kappa}q_{2}<\pi/4,\\
\end{gathered}
\end{equation*}
which implies that $-\pi/2 < 2\theta < \pi/2$ if $\theta\in\mathcal{I}_{1}$.
Therefore,
\begin{equation*}
    \cos(2\theta) > 0 \qquad \mbox{for all $\theta\in \mathcal{I}_1$}.
\end{equation*}
This inequality, together with inequalities \eqref{eq:estimatesdomainf1}, imply that every term in the formula of $f_{1}'$ is positive for all $\theta\in\mathcal{I}_{1}$. 

\medskip

4. The analysis on the interval $\mathcal{I}_{4}$ is analogous. 
This time we consider the function $\theta\mapsto f_4(\theta;\kappa,\mu)$
given by 
\begin{equation}
\label{eq:f4}
    f_{4}(\theta;\kappa,\mu) = \frac{1}{2}\sin(2\theta) + \Gamma \left[ \frac{1}{\sin^{2}(\theta+\sqrt{\kappa}q_{1})} - \frac{\mu}{\sin^{2}(\theta-\sqrt{\kappa}q_{2})} \right], \qquad \theta\in \mathcal{I}_{4}=(-\pi - \sqrt{\kappa} q_{1}, -\pi + \sqrt{\kappa} q_{2}),
\end{equation}
  assuming $\kappa<\pi^{2}/16$, $0<\mu<1/10$. The function $f_{4}$ tends to $+\infty$ near the left endpoint of $\mathcal{I}_{4}$ and to $+\infty$ near the right endpoint, so it has at least one root. To show uniqueness of this root, we prove that 
\begin{equation}
\label{eq:f4pneg}
    f_{4}'(\theta;\kappa,\mu) = \cos(2\theta) - 2\Gamma \left[ \frac{\cos(\theta+\sqrt{\kappa}q_{1})}{\sin^{3}(\theta+\sqrt{\kappa}q_{1})} - \mu\,\frac{\cos(\theta-\sqrt{\kappa}q_{2})}{\sin^{3}(\theta-\sqrt{\kappa}q_{2})} \right] <0, \qquad \mbox{for
    all $\theta \in \mathcal{I}_{4}$,}
\end{equation}
using our
assumptions on $\kappa$ and $\mu$.

Proceeding as above, using that $q_{1}+q_{2}=1$, and $\kappa<\pi^{2}/16$, one can conclude that
\begin{equation}
\label{eq:estimatedomainf4}
    -\pi < \theta+\sqrt{\kappa}q_{1} <-3\pi/4, \qquad \mbox{and} \qquad -5\pi/4 < \theta-\sqrt{\kappa}q_{2} < -\pi \qquad \mbox{for all $\theta\in \mathcal{I}_4$}
\end{equation}
Therefore, by the monotonicity of $\sin$ and $\cos$ on the corresponding intervals, and the condition $\mu<\frac{1}{10}$, the following inequalities hold 
\begin{equation*}
\begin{gathered}
    \frac{\cos(\theta+\sqrt{\kappa}q_{1})}{\sin^{3}(\theta+\sqrt{\kappa}q_{1})}\geq\frac{\cos(-\pi+\sqrt{\kappa})}{\sin^{3}(-\pi+\sqrt{\kappa})}=\frac{\cos(\!\sqrt{\kappa})}{\sin^{3}(\!\sqrt{\kappa})},\\
    \mu\frac{\cos(\theta-\sqrt{\kappa}q_{2})}{\sin^{3}(\theta-\sqrt{\kappa}q_{2})}\leq\frac{1}{10}\frac{\cos(-\pi-\sqrt{\kappa})}{\sin^{3}(-\pi-\sqrt{\kappa})}=-\frac{1}{10}\frac{\cos(\sqrt{\kappa})}{\sin^{3}(\sqrt{\kappa})}.
\end{gathered}  
\end{equation*}
Moreover, from the expression for $\Gamma$ in Eq.~\eqref{eq:auxGammaq} and the fact that $1+\mu<11/10$, we have
\begin{equation}
\label{eq:boundGamma}
    \frac{1}{1+\mu}\leq \frac{\Gamma}{\sin^{3}(\sqrt{\kappa})\cos(\sqrt{\kappa})} \leq \frac{1}{1-\mu} \quad\Longrightarrow\quad \frac{10}{11}\leq \frac{\Gamma}{\sin^{3}(\sqrt{\kappa})\cos(\sqrt{\kappa})} \leq \frac{10}{9}.
\end{equation}
Combining the previous inequalities and using that $\kappa < \pi^{2}/16$, we obtain
\begin{equation*}
\begin{aligned}
     2\Gamma\left[ \frac{\cos(\theta+\sqrt{\kappa}q_{1})}{\sin^{3}(\theta+\sqrt{\kappa}q_{1})} - \mu\,\frac{\cos(\theta-\sqrt{\kappa}q_{2})}{\sin^{3}(\theta-\sqrt{\kappa}q_{2})} \right] &\geq \frac{20}{11}\sin^{3}(\sqrt{\kappa})\cos(\sqrt{\kappa})\left[ \frac{\cos(\sqrt{\kappa})}{\sin^{3}(\sqrt{\kappa})} + \frac{1}{10}\frac{\cos(\sqrt{\kappa})}{\sin^{3}(\sqrt{\kappa})} \right]\\
     & = 2 \cos^{2}(\sqrt{\kappa}) \geq 2 \cos^{2}(\pi/4) = 1 \\
     & \geq \cos(2\theta).
\end{aligned}
\end{equation*}
Therefore, $f_{4}'(\theta)<0$, which implies that $f_{4}$ is strictly decreasing on $\mathcal{I}_{4}$ and its root is unique.

\medskip

2. Let  $0 < \kappa < \pi^{2}/36$, $0<\mu<1/10$ be fixed.
 Using the monotonicity of $\sin$ on $(0,\pi/2)$ and our assumption that $0<\mu <1/10$, we obtain
\begin{equation*}
    \frac{1}{\sin^{2}(\pi/4 + \sqrt{\kappa}q_{1})} < \frac{1}{\sin^{2}(\pi/4)} = 2, 
    \qquad  
    \frac{\mu}{\sin^{2}(\pi/4 - \sqrt{\kappa}q_{2})} < \frac{\mu}{\sin^{2}(\pi/4 - \pi/6)} < \frac{8}{5(\sqrt{6}-\sqrt{2})^{2}}.
\end{equation*}
Considering that  $\sqrt{\kappa} < \pi/6$, we have $\cos(2\sqrt{\kappa}) \geq 1/2$, and hence we can bound $\Gamma$, given by Eq.~\eqref{eq:auxGammaq} as follows:
\begin{equation*}
    \Gamma = \frac{\sin^{3}(\sqrt{\kappa})\cos(\sqrt{\kappa})}{\sqrt{\mu^{2} + 2\mu\cos(2\sqrt{\kappa}) + 1}} \leq \frac{\sin^{3}(\sqrt{\kappa})}{\sqrt{\mu^{2} + \mu + 1}} \leq \sin^{3}(\pi/6) = \frac{1}{8}.
\end{equation*}
Therefore, combining our estimates above we get
\begin{equation*}
    \Gamma \left[ \frac{1}{\sin^{2}(\pi/4 + \sqrt{\kappa}q_{1})} + \frac{\mu}{\sin^{2}(\pi/4 - \sqrt{\kappa}q_{2})} \right] 
    < \frac{1}{8} \left( 2 + \frac{8}{5(\sqrt{6}-\sqrt{2})^{2}} \right) < \frac{1}{2} \sin(2 \cdot \pi/4).
\end{equation*}
In view of the explicit form of $f_2$ given in the formula \eqref{eq:f2}, this
inequality is equivalent to $f_{2}(\pi/4;\kappa,\mu) > 0$.
Applying item 1 of Proposition~\ref{prop:concavity}  with $\hat \theta=\pi/4$,
we conclude that there exist exactly $2$ collinear RE on $\mathcal{I}_2$
(note that $\pi/4 \in (\sqrt{\kappa} q_{2}, \pi/2)\subset \mathcal{I}_2$
 since $0<q_2<1$ and $0<\sqrt{\kappa} < \pi/6$).

\medskip

3. Let now $0 < \kappa < \pi^{2}/64$, $0<\mu<1/10$ be fixed. We proceed
in analogy with the proof of item 2 above. Using the monotonicity of $\sin$ on $(-\pi/2,0)$ and our assumptions on 
$\kappa$ and $\mu$ we obtain
\begin{equation*}
    \frac{1}{\sin^{2}(-\pi/4 + \sqrt{\kappa}q_{1})} < \frac{1}{\sin^{2}(-\pi/4+\pi/8)} = \frac{4}{2-\sqrt{2}}, 
    \qquad  
    \frac{\mu}{\sin^{2}(-\pi/4 - \sqrt{\kappa}q_{2})} < \frac{\mu}{\sin^{2}(-\pi/4)} < \frac{1}{5}.
\end{equation*}
On the other hand,  since $\sqrt{\kappa} < \pi/8$, we have $\cos(2\sqrt{\kappa}) \geq \sqrt{2}/2$, and using our assumption $0<\mu<1/10$
we can bound  $\Gamma$ by
\begin{equation*}
    \Gamma = \frac{\sin^{3}(\sqrt{\kappa})\cos(\sqrt{\kappa})}{\sqrt{\mu^{2} + 2\mu\cos(2\sqrt{\kappa}) + 1}} \leq \frac{\sin^{3}(\sqrt{\kappa})}{\sqrt{\mu^{2} + \sqrt{2}\mu + 1}} \leq \sin^{3}(\pi/8) = \frac{(2-\sqrt{2})^{3/2}}{8}.
\end{equation*}
Combining the above inequalities we get:
\begin{equation*}
\begin{split}
    \Gamma \left[ \frac{1}{\sin^{2}(-\pi/4 + \sqrt{\kappa}q_{1})} + \frac{\mu}{\sin^{2}(-\pi/4 - \sqrt{\kappa}q_{2})} \right] 
   & < \frac{(2-\sqrt{2})^{3/2}}{8} \left( \frac{4}{2-\sqrt{2}} + \frac{1}{5} \right) \\ 
   & \approx  0.394< \frac{1}{2}= -\frac{1}{2} \sin(2 \cdot (-\pi/4)),
    \end{split}
\end{equation*}
which is equivalent to $f_{3}(-\pi/4;\kappa,\mu) < 0$ in view of
the expression for $f_3$ in the formula \eqref{eq:f3}. Therefore,
by item 2 of Proposition \ref{prop:concavity}, there exist 
exactly 2 collinear RE on $\mathcal{I}_3$ (note that  $-\pi/4 \in (-\pi/2,-\sqrt{\kappa} q_{1})\subset \mathcal{I}_3$  since $0<\sqrt{\kappa}<\pi/8$ and $0<q_1<1$).
\qed

\subsection{Proof of Theorem~\ref{thm:numbertriangularRE} on  existence of triangular RE.}
\label{app:existencetriangular}

As mentioned in Sec.~\ref{ss:AnalyticalExistenceTriangular}, the proof consists of showing that the function  $g_{\kappa,\mu} : \mathcal{I}_{\kappa,\mu} \to \mathbb{R}$ defined by the formula \eqref{eq:auxfunctiontriangular} admits exactly one root in $\mathcal{I}_{\kappa,\mu}$ when
$\kappa$ is sufficiently small. Recall that the interval $\mathcal{I}_{\kappa,\mu}=[a_{\kappa,\mu},b_{\kappa,\mu}]$ where
$a_{\kappa,\mu}$ and $b_{\kappa,\mu}$ are given by formulas \eqref{eq:auxintervaltriangular}. For the reader's convenience,
 we reproduce below  the expressions for $g_{\kappa,\mu}$, $a_{\kappa,\mu}$, and $b_{\kappa,\mu}$.
\begin{equation*}
\begin{split}
    g_{\kappa,\mu}(x)
    &= \Lambda^{3}\sin(\sqrt{\kappa}q_{2})\sin^{3}(x)
    \left[\sin(\sqrt{\kappa}q_{1})\cos(x)
        + \sin(\sqrt{\kappa}q_{2})\sqrt{1-\Lambda^{2}\sin^{2}(x)}\right]
    - \Gamma\sin(\sqrt{\kappa}), \\
      a_{\kappa,\mu}&= \arctan\!\left( \frac{\sin(\sqrt{\kappa})}{\Lambda + \cos(\sqrt{\kappa})} \right), \\ b_{\kappa,\mu}&= \pi-\arctan\!\left(\frac{\sin(\sqrt{\kappa})}{\Lambda -\cos(\sqrt{\kappa})}\right).
    \end{split}
\end{equation*}

We will find it useful to write 
 $t:=\sqrt{\kappa}$ at some points, and employ  the notation introduced in Sec.~\ref{sec:asymp} where
 for a function $f(\kappa,\mu)$, which is probably not defined at $\kappa=0$,  and
$n\in \mathbb{N}$,    the notation $f(\kappa,\mu)=\mathcal{O}_{\mu}\!\left(t^n\right)$, means that 
there exists a continuous 
 function  $C:(0,1)\to \R$ such that 
\begin{equation*}
|f(\kappa,\mu)|\leq C(\mu)t^n=C(\mu)\kappa^{n/2},
\end{equation*}
for sufficiently small $\kappa>0$.

Using formulas \eqref{eq:positionprimaries},  \eqref{eq:constantgamma}, and \eqref{eq:constantLambda}, it is seen that
the quantities $q_{1}$, $q_{2}$, $\Gamma$, and $\Lambda$ are analytic in $t$. Expanding them 
in a Taylor series around $t=0$ yields the following expansions, which will be used below:\footnote{Note that the leading terms in the expansions of $q_{1}$ and $q_{2}$ are consistent with the positions of the primaries in the planar case given in expressions \eqref{eq:primariespositions}.}
\begin{equation}
\label{eq:expansionsparameters}
\begin{aligned}
    q_{1} &= \frac{\mu}{1+\mu} - \frac{2\mu(1-\mu)}{3(1+\mu)^{3}}\,t^{2} + \mathcal{O}_{\mu}(t^{4}), \\
    q_{2} &= \frac{1}{1+\mu} + \frac{2\mu(1-\mu)}{3(1+\mu)^{3}}\,t^{2} + \mathcal{O}_{\mu}(t^{4}), \\
    \Gamma &= \frac{1}{1+\mu}\,t^{3} + \mathcal{O}_{\mu}(t^{4}), \\
    \Lambda &= 1 - \frac{1-\mu}{6(1+\mu)}\,t^{2} + \mathcal{O}_{\mu}(t^{4}).
\end{aligned}
\end{equation}

 We claim that if $\kappa>0$ is sufficiently small, then $g_{\kappa,\mu}$ has a unique root in $\mathcal{I}_{\kappa,\mu}$ 
 that actually belongs to 
$[\sqrt{\kappa},2\sqrt{\kappa}]$. We begin by showing that  if $\kappa>0$ is small enough, then
\begin{equation}
\label{eq:intervalscont}
\left[\sqrt{\kappa},\,2\sqrt{\kappa}\,\right] \subset \mathcal{I}_{\kappa,\mu}.
\end{equation}
First, observe that for any $\kappa$ and $\mu$,
\begin{equation*}
    a_{\kappa,\mu}
    \leq \arctan\!\left(\frac{\sin(\sqrt{\kappa})}{\cos(\sqrt{\kappa})}\right)
    = \sqrt{\kappa}.
\end{equation*}
On the other hand, we may expand $b_{\kappa,\mu}$ using formulas \eqref{eq:expansionsparameters}
to obtain
\begin{equation*}
    b_{\kappa,\mu} = \frac{\pi}{2} + \frac{1+2\mu}{3+3\mu}\, t + \mathcal{O}_{\mu}(t^{2}).
\end{equation*}
Since the coefficient of $t$ is positive,  there exists a value $\kappa_{I}>0$ such that $b_{\kappa,\mu}>\tfrac{\pi}{2}$ for all $0<\kappa<\kappa_{I}$. In particular,
\begin{equation*}
    2\sqrt{\kappa} < b_{\kappa,\mu}
    \quad \text{for } \quad 0<\kappa < \min\!\left\{\,\kappa_{I},\,\frac{\pi^{2}}{16}\right\},
\end{equation*}
which establishes condition \eqref{eq:intervalscont}.

We now examine the sign of the quantities $g_{\kappa,\mu}(a_{\kappa,\mu})$, $g_{\kappa,\mu}(\sqrt{\kappa})$, and
$g_{\kappa,\mu}(2\sqrt{\kappa})$, which depend analytically on $t$. Computing their Taylor series expansions 
around $t=0$  yields
\begin{equation*}
\begin{aligned}
    g_{\kappa,\mu}(a_{\kappa,\mu})
        &= -\frac{7}{8(1+\mu)}\, t^{5} + \mathcal{O}_{\mu}(t^{8}), \\[4pt]
    g_{\kappa,\mu}(\sqrt{\kappa})
        &= -\frac{1}{2(1+\mu)^{2}}\, t^{7} + \mathcal{O}_{\mu}(t^{8}), \\[4pt]
    g_{\kappa,\mu}(2\sqrt{\kappa})
        &= \frac{7}{1+\mu}\, t^{5} + \mathcal{O}_{\mu}(t^{6}).
\end{aligned}
\end{equation*}
The signs of the leading terms show that, for $t=\sqrt{\kappa}>0$ sufficiently small,
\begin{equation*}
    g_{\kappa,\mu}(a_{\kappa,\mu})<0,  
    \qquad
    g_{\kappa,\mu}(\sqrt{\kappa})<0,
    \qquad
    g_{\kappa,\mu}(2\sqrt{\kappa})>0.
\end{equation*}
More precisely, there exists a positive constant $\kappa_{II}$, which we choose to be smaller than  $\kappa_{I}$, such that all of the inequalities above hold whenever $0<\kappa<\kappa_{II}$. Therefore, the function $g_{\kappa,\mu}$ has at least one root in the interval
$[\sqrt{\kappa},\,2\sqrt{\kappa}]$.

We now  prove that $g_{\kappa,\mu}$ has no other roots on $[a_{\kappa,\mu},2\sqrt{\kappa}]$ by showing  that its  derivative, 
$g'_{\kappa,\mu}$, is strictly positive there. After simplification, we have
\begin{equation*}
    g_{\kappa,\mu}'(x)
    = \Lambda^{3}\sin^{2}(x)\,\sin(\sqrt{\kappa}q_{2})\,
      \Bigg[
        \sin(\sqrt{\kappa}q_{1})\bigl(3\cos^{2}(x)-\sin^{2}(x)\bigr)
        + \frac{\cos(x)\sin(\sqrt{\kappa}q_{2})\bigl(3 - 4\Lambda^{2}\sin^{2}(x)\bigr)}
               {\sqrt{1 - \Lambda^{2}\sin^{2}(x)}}
      \Bigg].
\end{equation*}
All factors outside the square brackets are positive for $x\in [0,2\sqrt{\kappa}]$ when
$\kappa$ is sufficiently small, and we claim that the same is true for those inside. For the first one note that
if $\kappa<\pi^{2}/36$, then $x\le 2\sqrt{\kappa}<\pi/3$, and hence
$\cos(x)\ge 1/2$. Therefore,
\begin{equation*}
    3\cos^{2}(x)-\sin^{2}(x)
    = 4\cos^{2}(x)-1 > 0.
\end{equation*}
To show that the second term inside the square brackets is positive, given that  $\Lambda<1$,  it suffices to show that
$\sin^{2}(x)<3/4$. This holds whenever $x<\arcsin(\sqrt{3}/2)=\pi/3$.
Thus, for $\kappa<\pi^{2}/36$ we have $x\le 2\sqrt{\kappa}<\pi/3$ and hence
\begin{equation*}
     3 - 4\Lambda^{2}\sin^{2}(x) > 3 - 4\sin^{2}(x) > 0.
\end{equation*}
Thus $g_{\kappa,\mu}'(x)>0$ for $x\in [0,2\sqrt{\kappa}]$, and
$g_{\kappa,\mu}$ is strictly increasing on this interval. Summarizing, we have shown
that $g_{\kappa,\mu}$ has exactly one root in $[a_{\kappa,\mu},2\sqrt{\kappa}]$
whenever $0<\kappa<\kappa_{III}$, where
$\kappa_{III}=\min\{\kappa_{II},\pi^{2}/36\}$.

We finally verify that $g_{\kappa,\mu}(x)>0$ for all $x\in [2\sqrt{\kappa},b_{\kappa,\mu}]$ if $\kappa$ is sufficient small.  A term by term
comparison using formula \eqref{eq:auxfunctiontriangular} allows us to bound $g_{\kappa,\mu}(x)$ from below for $x\in[2\sqrt{\kappa},\,b_{\kappa,\mu}]$ as follows:
\begin{equation*}
\begin{aligned}
g_{\kappa,\mu}(x) \geq & \Lambda^{3}
\sin\!\big(\sqrt{\kappa}\,q_{2}\big)\left[\,
\sin\!\big(\sqrt{\kappa}\,q_{1}\big)\,
\cos\!\big(b_{\kappa,\mu}\big)+
\sin^{3}\!\big(2\sqrt{\kappa}\big)\,
\sin\!\big(\sqrt{\kappa}\,q_{2}\big)\,
\sqrt{\,1-\Lambda^{2}\sin^{2}\!\big(2\sqrt{\kappa}\,\big)}\right]
\\
&-\Gamma\,\sin^{2}\!\big(\sqrt{\kappa}\big).
\end{aligned}
\end{equation*}
The right-hand side of the inequality is analytic in  $t=\sqrt{\kappa}$ and may be expanded using formulas \eqref{eq:expansionsparameters} to yield
\begin{equation*}
g_{\kappa,\mu}(x) \geq \frac{7-\mu}{(1+\mu)^{2}}t^{5} + \mathcal{O}_{\mu}(t^{6}).
\end{equation*}
Since the leading coefficient is positive we conclude the existence of 
$\kappa_{IV}>0$ such that $g_{\kappa,\mu}(x)>0$ for all $x\in[2\sqrt{\kappa},b_{\kappa,\mu}]$ whenever $0<\kappa<\kappa_{IV}$.
\qed

\subsection{Proof of Theorem~\ref{thm:stabilitycollinear} on stability of collinear RE}
\label{app:stability-collinear}

As mentioned in Sec.~\ref{ss:analytical-collinear}, we estimate the quantities  
 $\lambda_1$ and $\lambda_2$ in Proposition~\ref{prop:stabilitycollpos} at each RE, and show that 
 they have opposite signs for all of $\Ll_1$, $\Ll_1$, $\Ll_3$, $\Ee_3$, and $\Aa_1$, and are instead
 both positive for  $\Ee_2$. 

We begin by establishing the following auxiliary estimates.
\begin{lemma}
\label{lemma:estimatesparameters}
Let $(\kappa,\mu)\in\Ps{\frac{1}{10}}$.  The following estimates hold:
\begin{equation*}
    \frac{10}{11}\leq \frac{\Gamma}{\sin^{3}(\sqrt{\kappa})\cos(\sqrt{\kappa})} \leq \frac{10}{9}, \qquad 0 < q_{1} \leq \frac{1}{9},  \qquad \frac{8}{9} \leq q_{2} < 1, \qquad \cos^{1/3}(\sqrt{\kappa}) \leq \Lambda < 1.
\end{equation*}
\end{lemma}

\begin{proof}
The first estimate was already established in the proof of Theorem \ref{thm:numbercollinearRE} (inequality~\eqref{eq:boundGamma}). For the estimate on $q_{1}$, one can start from the formulas \eqref{eq:positionprimaries} and the fact that $\arctan(x)\leq x$ for $x\geq0$, to obtain \begin{equation*}
q_{1} = \frac{1}{2\sqrt{\kappa}}\arctan\left(\frac{\mu\sin(2\sqrt{\kappa})}{1+\mu\cos(2\sqrt{\kappa})}\right) \leq \frac{\mu}{1-\mu}\frac{\sin(2\sqrt{\kappa})}{2\sqrt{\kappa}} \leq \frac{\mu}{1-\mu} \leq \frac{1}{9}.
\end{equation*}
The estimate for $q_{2}$ follows immediately from the identity $q_{2}=1-q_{1}$. Finally, combining the previous estimate with the alternative expression \eqref{eq:alternativeLambda}, and using the fact that $\cos(x)$ is decreasing on $(0,\pi)$, we obtain
\begin{equation*}
	\Lambda^{3}=\frac{\cos(\sqrt{\kappa}q_{2})}{\cos(\sqrt{\kappa}q_{1})}\geq \frac{\cos(\sqrt{\kappa})}{\cos(0)}=\cos(\sqrt{\kappa}).
\end{equation*}
Finally, we recall that the  upper bound for $\Lambda$ was given in inequalities~\eqref{eq:inequalityLambda}.
\end{proof}

Recall from Eq.~\eqref{eq:CollinearRECondbyIntervals} that the RE correspond to roots of the functions $f_j:\mathcal{I}_i\to \R$, $j=1,2,3,4$, where the index $j$ associated with each RE is specified in the statement of Theorem \ref{thm:numbercollinearRE}. 
The explicit forms of these functions are given in formulas \eqref{eq:f1}, \eqref{eq:f2}, \eqref{eq:f3} and \eqref{eq:f4}. In all cases, it will be useful to notice that the quantity $\lambda_1$ in Proposition \ref{prop:stabilitycollpos}
satisfies
\begin{equation}
\label{eq:auxlambda1}
\lambda_{1} = -f_{j}'(\theta_{0}).
\end{equation}
On the other hand, it will be convenient to
use Eq.~\eqref{eq:lambda2alt} for  $\lambda_2$, and note that it may be rewritten as
 \begin{equation}
 \label{eq:auxlambda2}
\lambda_{2} = -\Gamma\csc\theta_{0}\,h_j(\theta_{0}),
\end{equation}
where $h_j:\mathcal{I}_j\to \R$ is the restriction to $\mathcal{I}_j$ of the function
\begin{equation*}
h(\theta) = \frac{\sin(\!\sqrt{\kappa}q_{1})}{\vert\sin(\theta+\sqrt{\kappa}q_{1})\vert^{3}} - \mu\frac{\sin(\!\sqrt{\kappa}q_{2})}{\vert\sin(\theta-\sqrt{\kappa}q_{2})\vert^{3}}.
\end{equation*}

We divide our analysis on the different intervals $\mathcal{I}_j$ specified in 
Subsection~\ref{sss:defIj}.

\paragraph{Analysis on $\mathcal{I}_1$.} Denote by $\theta_{\Ll_{1}}\in \mathcal{I}_1$ the value of $\theta_0$ corresponding
to  $\Ll_1$, and by $\lambda_{1}^{(\Ll_{1})}$ and $\lambda_{2}^{(\Ll_{1})}$ the 
corresponding values of $\lambda_1$ and $\lambda_2$. Similar  notation will be used for the  other  RE. We will prove that $\lambda_{1}^{(\Ll_{1})}<0$ and $\lambda_{2}^{(\Ll_{1})}>0$.

We have already shown that  $f_{1}'$ is strictly positive on $\mathcal{I}_{1}$ (see  condition \eqref{eq:auxf1prime}). Therefore, by condition \eqref{eq:auxlambda1},  $\lambda_{1}^{(\Ll_{1})}<0$,
as claimed. Next,   using $q_{1}<q_{2}$ and $\mu<1$, we obtain
\begin{equation*}
f_{1}(0)=-\Gamma\left(\frac{1}{\sin^{2}(\sqrt{\kappa}q_{1})}-\frac{\mu}{\sin^{2}(\sqrt{\kappa}q_{2})}\right)<0,
\end{equation*}
which, since $f_1$ is increasing, implies that $\theta_{\Ll_{1}}\in(0,\sqrt{\kappa}q_{2})$. Consequently, in view of formula \eqref{eq:auxlambda2},  $\lambda_{2}^{(\Ll_{1})}$  has the
opposite  sign of $h_{1}(\theta_{\Ll_{1}})$, and the proof that $\lambda_{2}^{(\Ll_{1})}>0$ reduces to showing that 
$h_{1}(\theta_{\Ll_{1}})$ is negative. 

For $\theta \in \mathcal{I}_{1}$ we have,
\begin{equation*}
h_{1}(\theta)=\frac{\sin(\!\sqrt{\kappa}q_{1})}{\sin^{3}(\theta+\sqrt{\kappa}q_{1})}+\mu\frac{\sin(\!\sqrt{\kappa}q_{2})}{\sin^{3}(\theta-\sqrt{\kappa}q_{2})},
\end{equation*}
and 
\begin{equation*}
h_{1}'(\theta)=-3\sin(\!\sqrt{\kappa}q_{1})\frac{\cos(\theta+\sqrt{\kappa}q_{1})}{\sin^{4}(\theta+\sqrt{\kappa}q_{1})}-3\mu\sin(\!\sqrt{\kappa}q_{2})\frac{\cos(\theta-\sqrt{\kappa}q_{2})}{\sin^{4}(\theta-\sqrt{\kappa}q_{2})}.
\end{equation*}
Each term in the previous expression is negative (see estimates~\eqref{eq:estimatesdomainf1}), and therefore $h_{1}$ is strictly decreasing. Due to its asymptotic behavior at the endpoints of $\mathcal{I}_{1}$, the function $h_{1}$ has a unique zero $\vartheta_{1}\in\mathcal{I}_{1}$. 
We claim that 
\begin{equation}
\label{eq:f1varthetaaux}
f_{1}(\vartheta_{1})<0.
\end{equation}
Accept   that this  inequality holds. Since $f_{1}$ is increasing and
$f_1(\theta_{\Ll_{1}})=0$, we have  $\vartheta_{1}<\theta_{\Ll_{1}}$. Given that 
 $h_{1}$ is decreasing, we conclude $h_{1}(\theta_{\Ll_{1}})<0$, which is what we needed
 to prove. Therefore, it only remains to show that inequality \eqref{eq:f1varthetaaux} holds.

First note   that $h_1(0)=-\frac{1}{\Gamma}f_1(0)>0$.
Therefore, since $h_1$ is decreasing, we must have $$\vartheta_{1}\in(0,\sqrt{\kappa}q_{2}).$$ 
Next, using the expressions   for $h_1(\theta)$ and $\Lambda$ (see Eq.~\eqref{eq:constantLambda}),    the 
condition  $h_{1}(\vartheta_{1})=0$ reduces to:
\begin{equation*}
\sin(\vartheta_{1}+\sqrt{\kappa}q_{1})=-\Lambda\sin(\vartheta_{1}-\sqrt{\kappa}q_{2}).
\end{equation*}
Substituting the above expression into the definition of $f_{1}$ (Eq.~\eqref{eq:f1}), yields
\begin{equation}
\label{eq:f1vartheta1}
f_{1}(\vartheta_{1})=\frac{1}{\sin^{2}(\vartheta_{1}-\sqrt{\kappa}q_{2})}\left[\sin\vartheta_{1}\cos\vartheta_{1}\sin^{2}(\vartheta_{1}-\sqrt{\kappa}q_{2})-\Gamma\left(\frac{1}{\Lambda^{2}}-\mu\right)\right],
\end{equation}
which is the quantity that we will estimate to establish inequality \eqref{eq:f1varthetaaux}.

On the one hand, using the upper bound for $\Gamma$  in Lemma~\ref{lemma:estimatesparameters}, together with the fact that $\Lambda<1$ and $\mu<1/10$, we obtain
\begin{equation}
\label{eq:estimateGammavartheta}
\Gamma\left(\frac{1}{\Lambda^{2}}-\mu\right)\ge\frac{10}{11}\sin^{3}(\sqrt{\kappa})\cos(\!\sqrt{\kappa})(1-1/10)=\frac{9}{11}\sin^{3}(\sqrt{\kappa})\cos(\!\sqrt{\kappa}).
\end{equation}

On the other hand, consider the auxiliary function $l_{1}:[0,\sqrt{\kappa}q_{2}]\to\mathbb{R}$ defined by
\begin{equation*}
l_{1}(\theta)=\sin\theta\cos\theta\sin^{2}(\theta-\sqrt{\kappa}q_{2}).
\end{equation*}
Then  $l_{1}$ is nonnegative on $[0,\sqrt{\kappa}q_{2}]$, and satisfies $l_{1}(0)=l_{1}(\sqrt{\kappa}q_{2})=0$. Its derivative is
\begin{equation*}
l_{1}'(\theta)=\sin(\theta-\sqrt{\kappa}q_{2})\sin(3\theta-\sqrt{\kappa}q_{2}).
\end{equation*}
The unique critical point of $l_1$ in $(0,\sqrt{\kappa}q_{2})$ is $\theta^{*}=\sqrt{\kappa}q_{2}/3$, which is necessarily a maximum. Therefore,
\begin{equation*}
\max_{\theta\in\mathcal{I}_{1}}\left[\sin\theta\cos\theta\sin^{2}(\theta-\sqrt{\kappa}q_{2})\right]
=l_{1}(\theta^{*})
=\frac{1}{2}\sin^{3}\left(\frac{2\sqrt{\kappa}q_{2}}{3}\right).
\end{equation*}
Since $\vartheta_{1}\in(0,\sqrt{\kappa}q_{2})$, it follows that
\begin{equation}
\label{eq:estimateconnectionvartheta}
\sin\vartheta_{1}\cos\vartheta_{1}\sin^{2}(\vartheta_{1}-\sqrt{\kappa}q_{2})
\le\frac{1}{2}\sin^{3}\left(\frac{2\sqrt{\kappa}q_{2}}{3}\right).
\end{equation}
Now  recall from the statement of the theorem that our analysis for  $\Ll_1$ should be made under the assumption 
that   $0<\kappa\le\pi^{2}/16$.
Given that  $0<q_{2}<1$ and  $\sin$ is increasing on $[0,\pi/2]$, we can estimate
\begin{equation*}
    \sin^3\left(\frac{2\sqrt{\kappa}q_{2}}{3}\right)\le\sin^3\left(\frac{2\sqrt{\kappa}}{3}\right).
\end{equation*}
Now set $x=\sqrt{\kappa}\in\left (0,\pi/4\right]$, and write
\begin{equation*}
 \sin^3\left(\frac{2x}{3}\right) = \sin^3(x) \left(\frac{\sin\left(\frac{2x}{3}\right)}{\sin(x)}\right)^{3} =\sin^3(x)\left(\frac{\sin\left(\frac{x}{3}\right)}{\sin\left(\frac{x}{2}\right)} \frac{\cos\left(\frac{x}{3}\right)}{\cos\left(\frac{x}{2}\right)}\right)^{3}.
\end{equation*}
The elementary bounds, 
\begin{equation*}
    \frac{\sin\left(\frac{x}{3}\right)}{\sin\left(\frac{x}{2}\right)}\le\frac{2}{3}, \qquad \frac{\cos\left(\frac{x}{3}\right)}{\cos\left(\frac{x}{2}\right)}\le 1
\end{equation*}
valid for $x\in (0,\pi/4]$,  imply
\begin{equation*}
  \sin^3\left(\frac{2x}{3} \right ) \leq  \frac{8}{27}\sin^3(x).
\end{equation*}
Moreover, for $x\in (0,\pi/4]$ we also have 
\begin{equation*}
    \cos(x)\ge\frac{\sqrt{2}}{2} \quad\Longrightarrow\quad \frac{8}{27}\le\frac{3}{2}\cos(x).
\end{equation*}
Therefore, 
\begin{equation*}
   \sin^{3}\left(\frac{2\sqrt{\kappa}}{3}\right)=   \sin^{3}\left(\frac{2x}{3}\right)\le\frac{3}{2}\sin^{3}(x)\cos(x)=\frac{3}{2}\sin^{3}(\sqrt{\kappa})\cos(\sqrt{\kappa}),
\end{equation*}
which allows us to conclude, from inequality \eqref{eq:estimateconnectionvartheta},
that
\begin{equation}
\label{eq:estimateconnectionvartheta2}
\sin\vartheta_{1}\cos\vartheta_{1}\sin^{2}(\vartheta_{1}-\sqrt{\kappa}q_{2})
\le \frac{3}{4}\sin^{3}(\sqrt{\kappa})\cos(\!\sqrt{\kappa}).
\end{equation}

Using estimates \eqref{eq:estimateGammavartheta} and \eqref{eq:estimateconnectionvartheta2} in formula \eqref{eq:f1vartheta1} yields
\begin{equation*}
f_{1}(\vartheta_{1})\leq \frac{-\frac{3}{44} \sin^{3}(\sqrt{\kappa})\cos(\!\sqrt{\kappa}) }{\sin^{2}(\vartheta_{1}-\sqrt{\kappa}q_{2})}<0,
\end{equation*}
proving inequality \eqref{eq:f1varthetaaux}.

\paragraph{Analysis on $ \mathcal{I}_4$.}
The analysis for   $\theta_{\Aa_{1}}\in \mathcal{I}_4\in (-\pi - \sqrt{\kappa} q_{1}, -\pi + \sqrt{\kappa} q_{2})$ is analogous to the above one for $\theta_{\Ll_1}\in \mathcal{I}_1$. Since $f_{4}$ is strictly decreasing (see Eq.~\eqref{eq:f4pneg}), we have $f_{4}'(\theta_{\Aa_{1}})<0$, and therefore, by relation \eqref{eq:auxlambda1},  $\lambda_{1}^{(\Aa_{1})}>0$. 
We now show that $\lambda_{2}^{(\Aa_{1})}<0$. 
Using $q_{1}<q_{2}$ and $\mu<1$, we obtain
\begin{equation*}
    f_{4}(-\pi) =\Gamma\left[\frac{1}{\sin^{2}(\pi+\sqrt{\kappa}q_{1})}-\frac{\mu}{\sin^{2}(\pi-\sqrt{\kappa}q_{2})}\right] =\Gamma\left[\frac{1}{\sin^{2}(\sqrt{\kappa}q_{1})}-\frac{\mu}{\sin^{2}(\sqrt{\kappa}q_{2})}\right]>0.
\end{equation*}
Since $f_4$ is decreasing, it follows that $\theta_{\Aa_{1}}\in(-\pi,-\pi+\sqrt{\kappa}q_{2})$. Therefore, by formula \eqref{eq:auxlambda2}, the sign of $\lambda_{2}^{(\Aa_{1})}$ coincides with the sign of $h_{4}(\theta_{\Aa_{1}})$, which we prove below to be negative.

 For   $\theta \in \mathcal{I}_{4}$ we have
\begin{equation*}
    h_{4}(\theta)=-\frac{\sin(\!\sqrt{\kappa}q_{1})}{\sin^{3}(\theta+\sqrt{\kappa}q_{1})}-\mu\frac{\sin(\!\sqrt{\kappa}q_{2})}{\sin^{3}(\theta-\sqrt{\kappa}q_{2})},
\end{equation*}
and 
\begin{equation*}
    h_{4}'(\theta) =3\sin(\!\sqrt{\kappa}q_{1})\frac{\cos(\theta+\sqrt{\kappa}q_{1})}{\sin^{4}(\theta+\sqrt{\kappa}q_{1})} +3\mu\sin(\!\sqrt{\kappa}q_{2})\frac{\cos(\theta-\sqrt{\kappa}q_{2})}{\sin^{4}(\theta-\sqrt{\kappa}q_{2})}.
\end{equation*}
By estimates \eqref{eq:estimatedomainf4}, each term in the previous expression is negative, and therefore $h_{4}$ is strictly decreasing. The asymptotic 
behavior of $h_4$ near the end-points of $\mathcal{I}_4$ implies that
$h_{4}$ admits a unique root,  $\vartheta_{4}\in\mathcal{I}_{4}$.
We will prove below that $f_4(\vartheta_{4})>0$ which by monotonicity of
$f_4$ implies $\vartheta_4<\theta_{\Aa_1}$. This in turn  implies  $h_{4}(\theta_{\Aa_{1}})<0$ by monotonicity of $h_4$, completing the proof
that  $\lambda_{2}^{(\Aa_{1})}<0$.

We have,
\begin{equation*}
    h_{4}(-\pi)=-\frac{\sin(\!\sqrt{\kappa}q_{1})}{\sin^{3}(-\pi+\sqrt{\kappa}q_{1})} -\mu\frac{\sin(\!\sqrt{\kappa}q_{2})}{\sin^{3}(-\pi-\sqrt{\kappa}q_{2})} =\frac{1}{\sin^{2}(\sqrt{\kappa}q_{1})}-\frac{\mu}{\sin^{2}(\sqrt{\kappa}q_{2})}>0,
\end{equation*}
so, since $h_4$ is decreasing, $\vartheta_{4}\in(-\pi,-\pi+\sqrt{\kappa}q_{2})$.
Considering that $\vartheta_{4}$ satisfies 
\begin{equation*}
    \sin(\vartheta_{4}+\sqrt{\kappa}q_{1}) =-\Lambda\sin(\vartheta_{4}-\sqrt{\kappa}q_{2}),
\end{equation*}
we may write, using formula \eqref{eq:f4},
\begin{equation*}
    f_{4}(\vartheta_{4}) =\frac{1}{\sin^{2}(\vartheta_{4}-\sqrt{\kappa}q_{2})} \left[\sin\vartheta_{4}\cos\vartheta_{4}\sin^{2}(\vartheta_{4}-\sqrt{\kappa}q_{2}) + \Gamma\left(\frac{1}{\Lambda^{2}}-\mu\right)\right].
\end{equation*}
The above expression is positive because
all terms inside the square brackets are positive, given that $\vartheta_{4}\in(-\pi,-\pi+\sqrt{\kappa}q_{2})$.

\paragraph{Analysis on $\mathcal{I}_2$.}
From the proof of Proposition \ref{prop:concavity}, we know that the roots
of $f_2$ necessarily belong to the subinterval $ (\sqrt{\kappa}q_{2},\pi/2)\subset \mathcal{I}_2$.
Thus,  we have $\theta_{\Ll_{2}},\theta_{\Ee_{2}}\in(\sqrt{\kappa}q_{2},\pi/2)$ and,
according to our labeling convention, $\theta_{\Ll_{2}}<\theta_{\Ee_{3}}$. Since $f_{2}''(\theta)<0$ for all $\theta\in (\sqrt{\kappa}q_{2},\pi/2)$ (again from the proof of Proposition \ref{prop:concavity}),
it follows that  $f_{2}'(\theta_{\Ll_{2}})>0$ and $f_{2}'(\theta_{\Ee_{2}})<0$, which 
by relation \eqref{eq:auxlambda1}, implies
\begin{equation*}
    \lambda_{1}^{(\Ll_{2})}<0,\qquad \lambda_{1}^{(\Ee_{2})}>0.
\end{equation*}

Below we will show that 
\begin{equation}
\label{eq:auxh4lambda2}
   h_2(\theta)<0 \qquad \mbox{for all $\theta\in(\sqrt{\kappa}q_{2},\pi/2)$}.
\end{equation}
In view of formula \eqref{eq:auxlambda2}, this immediately implies,
\begin{equation*}
   \lambda_{2}^{(\Ll_{2})}>0,\qquad \lambda_{2}^{(\Ee_{2})}>0,
\end{equation*}
since $\theta_{\Ll_{2}},\theta_{\Ee_{2}}\in(\sqrt{\kappa}q_{2},\pi/2)$.

In order to prove inequality \eqref{eq:auxh4lambda2}, note that
for $\theta\in \mathcal{I}_2$ we have
\begin{equation*}
    h_{2}(\theta)=\frac{\sin(\!\sqrt{\kappa}q_{1})}{\sin^{3}(\theta+\sqrt{\kappa}q_{1})}-\mu\frac{\sin(\!\sqrt{\kappa}q_{2})}{\sin^{3}(\theta-\sqrt{\kappa}q_{2})},
\end{equation*}
and 
\begin{equation*}
    h_{2}'(\theta)=-3\sin(\!\sqrt{\kappa}q_{1})\frac{\cos(\theta+\sqrt{\kappa}q_{1})}{\sin^{4}(\theta+\sqrt{\kappa}q_{1})}+3\mu\sin(\!\sqrt{\kappa}q_{2})\frac{\cos(\theta-\sqrt{\kappa}q_{2})}{\sin^{4}(\theta-\sqrt{\kappa}q_{2})}.
\end{equation*}
We claim that $h_2'(\theta)>0$ for $\theta\in (\sqrt{\kappa}q_{2},\pi/2)$. Indeed, for
 $\theta\in[\pi/2-\sqrt{\kappa}q_{1},\pi/2)$ we have
\begin{equation*}
    \frac{\pi}{2}\leq \theta+\sqrt{\kappa}q_{1}<\frac{\pi}{2}+\sqrt{\kappa}q_{1}<\pi \quad \mbox{and} \quad  0<\frac{\pi}{2}-\sqrt{\kappa}\leq\theta-\sqrt{\kappa}q_{2}<\frac{\pi}{2},
\end{equation*}
so the first terms in $h_{2}'(\theta)$ is non-negative and the second one is  positive, and hence $h_{2}'(\theta)>0$. If instead $\theta\in(\sqrt{\kappa}q_{2},\pi/2-\sqrt{\kappa}q_{1})$, we have
\begin{equation*}
  0<  \sqrt{\kappa}<\theta+\sqrt{\kappa}q_{1}<\frac{\pi}{2} \quad \mbox{and} \quad 0<\theta -\sqrt{\kappa}q_2<\frac{\pi}{2}
  -\sqrt{\kappa}<\frac{\pi}{2},
\end{equation*}
and thus,
\begin{equation*}
    \frac{\sin^{4}(\theta+\sqrt{\kappa}q_{1})}{\sin^{4}(\theta-\sqrt{\kappa}q_{2})}\ge \left (\frac{\sin(\!\sqrt{\kappa}q_{1})}{\mu\sin(\!\sqrt{\kappa}q_{2})} \right )\left (\frac{\cos(\theta+\sqrt{\kappa}q_{1})}{\cos(\theta-\sqrt{\kappa}q_{2})} \right ),
\end{equation*}
since the left-hand side is greater than one, while both factors on the right-hand side are less than one. Rearranging this inequality yields again $h_{2}'(\theta)>0$. Therefore, $h_{2}$ is monotone increasing on $(\sqrt{\kappa}q_{2},\pi/2)$.
Considering that 
\begin{equation*}
    h_{2}(\pi/2)=\frac{\sin(\!\sqrt{\kappa}q_{1})}{\cos^{3}(\sqrt{\kappa}q_{1})}-\mu\frac{\sin(\!\sqrt{\kappa}q_{2})}{\cos^{3}(\sqrt{\kappa}q_{2})}<0,
\end{equation*}
we conclude that inequality \eqref{eq:auxh4lambda2} holds.

\paragraph{Analysis on $\mathcal{I}_3$.} We proceed  analogously to the analysis on $\mathcal{I}_2$ above. From the proof of Proposition \ref{prop:concavity}, we know that the roots
of $f_3$ necessarily belong to the subinterval $(-\pi/2,-\sqrt{\kappa}q_{1})\subset \mathcal{I}_3$, where  $f_3''>0$.
Due to our labeling convention, we have $\theta_{\Ee_{3}}<\theta_{\Ll_{3}}$, and therefore we conclude 
\begin{equation*}
\theta_{\Ll_{3}},\theta_{\Ee_{3}}\in(-\pi/2,-\sqrt{\kappa}q_{1}), \qquad f_3'(\theta_{\Ee_{3}})<0, \qquad f_3'(\theta_{\Ll_{3}})>0.
\end{equation*}
 Hence, in view of relation \eqref{eq:auxlambda1}, we have  
\begin{equation*}
\lambda_{1}^{(\Ee_{3})}>0,\qquad \lambda_{1}^{(\Ll_{3})}<0.
\end{equation*}

Since $\csc\theta_0<0$ for $\theta_0\in(-\pi/2,-\sqrt{\kappa}q_{1})$,
 it follows from formula \eqref{eq:auxlambda2} that the signs of $\lambda_{2}^{(\Ee_{3})}$ and $\lambda_{2}^{(\Ll_{3})}$ coincide with the signs of $h_3(\theta_{\Ee_{3}})$ and $h_3(\theta_{\Ll_{3}})$, respectively. Therefore, it remains to prove that
 \begin{equation}
 \label{eq:toproveh3}
h_3(\theta_{\Ee_{3}})<0 \qquad \mbox{and} \qquad h_3(\theta_{\Ll_{3}})>0,
\end{equation}
 to conclude that $\lambda_1$ and $\lambda_2$ have opposite signs for both $\Ee_{3}$ and $\Ll_{3}$.

For $\theta\in \mathcal{I}_3$ we have 
\begin{equation*}
h_{3}(\theta)=-\frac{\sin(\!\sqrt{\kappa}q_{1})}{\sin^{3}(\theta+\sqrt{\kappa}q_{1})}
+\mu\frac{\sin(\!\sqrt{\kappa}q_{2})}{\sin^{3}(\theta-\sqrt{\kappa}q_{2})},
\end{equation*}
and
\begin{equation*}
h_{3}'(\theta)
=3\sin(\!\sqrt{\kappa}q_{1})\frac{\cos(\theta+\sqrt{\kappa}q_{1})}{\sin^{4}(\theta+\sqrt{\kappa}q_{1})}
-3\mu\sin(\!\sqrt{\kappa}q_{2})\frac{\cos(\theta-\sqrt{\kappa}q_{2})}{\sin^{4}(\theta-\sqrt{\kappa}q_{2})}.
\end{equation*}
Similar estimates to those given above for  $h_2'$ show that 
$$h_{3}'(\theta)>0\qquad \mbox{for all $\theta \in \left(-\pi/2,-\sqrt{\kappa}q_{1}\right)$},$$ 
and hence
$h_{3}$ is monotone increasing on  this interval. Furthermore, 
considering that $h_{3}(-\pi/2)=h_2(\pi/2)<0$, and that $h_3(\theta)\to \infty$ as $\theta \to  -\sqrt{\kappa}q_{1}^-$,
we conclude the existence of a unique $\vartheta_{3}\in(-\pi/2,-\sqrt{\kappa}q_{1})$ such that $h_{3}(\vartheta_{3})=0$. 
We will prove below that 
\begin{equation}
\label{eq:toproveh3-2}
f_3(\vartheta_{3})<0.
\end{equation}
Since $f_3''(\theta)<0$ for $\theta\in (-\pi/2,-\sqrt{\kappa}q_{1})$, this implies that $\theta_{\Ee_{3}}< \vartheta_{3}<
\theta_{\Ll_{3}}$, and by monotonicity of $h_3$ we conclude that inequality\eqref{eq:toproveh3} holds.

Notice that, by  the definition of $\Lambda$ in formula \eqref{eq:constantLambda},  $\vartheta_{3}$ satisfies
\begin{equation}
\label{eq:varthetah3}
\sin(\vartheta_{3}+\sqrt{\kappa}q_{1})
=\Lambda\sin(\vartheta_{3}-\sqrt{\kappa}q_{2}).
\end{equation} 
Hence, using formula \eqref{eq:f3}, we have
\begin{equation*}
f_{3}(\vartheta_{3})
=\frac{1}{\sin^{2}(\vartheta_{3}+\sqrt{\kappa}q_{1})}
\left[
\sin(\vartheta_{3})\cos(\vartheta_{3})\sin^{2}(\vartheta_{3}+\sqrt{\kappa}q_{1})
+\Gamma(1+\Lambda^{2}\mu)
\right].
\end{equation*}
Therefore, proving inequality \eqref{eq:toproveh3-2}
 is equivalent to showing that
\begin{equation}
\label{eq:inequalityf3}
\Gamma(1+\Lambda^{2}\mu)
<
\sin(-\vartheta_{3})\cos(-\vartheta_{3})\sin^{2}(\vartheta_{3}+\sqrt{\kappa}q_{1}).
\end{equation}

On the one hand, 
using the upper bound for $\Gamma$ from Lemma \ref{lemma:estimatesparameters}, together with $\Lambda<1$, $\mu<\frac{1}{10}$, and $\kappa<\pi^{2}/64$, we obtain
\begin{equation}
\label{eq:firstpartf3}
\Gamma(1+\Lambda^{2}\mu)
<
\frac{10}{9}\sin^{3}(\sqrt{\kappa})\cos(\!\sqrt{\kappa})\left (1+\frac{1}{10}\right )
\le
\frac{11}{9}\sin^{2}\left(\frac{\pi}{8}\right)\sin(\!\sqrt{\kappa}) \approx 0.17899 \sin(\!\sqrt{\kappa}).
\end{equation}
On the other hand, we claim that the following estimates hold:
\begin{subequations}
\label{eq:evar3}
\begin{align}
	\sin(-\vartheta_{3}) & \ge \cos^{1/3}\!\left(\frac{\pi}{8}\right)\cos^{3}\!\left(\frac{\pi}{16}\right), \label{eq:evara} \\
	\cos(-\vartheta_{3}) &\ge \frac{\cos^{1/3}\!\left(\frac{\pi}{8}\right)\sin\left(\frac{7\pi}{144}\right)\cos^{2}\!\left(\frac{\pi}{16}\right)}{\sin\left(\frac{\pi}{8}\right)}\sin(\sqrt{\kappa}), \label{eq:evarb} \\
    	\sin^{2}(\vartheta_{3}+\sqrt{\kappa}q_{1}) &\ge 
	\cos^{2/3}\!\left( \sqrt{\kappa} \right)\cos^{4}\!\left( \frac{\pi}{16} \right). \label{eq:evarc}
\end{align}
\end{subequations}
Combining the three estimates \eqref{eq:evar3} yields,
\begin{equation}
\label{eq:secondpartf3}
\sin(-\vartheta_{3})\cos(-\vartheta_{3})\sin^{2}(\vartheta_{3}+\sqrt{\kappa}q_{1})
\ge
\frac{\cos^{4/3}\left(\frac{\pi}{8}\right)\cos^{9}\left(\frac{\pi}{16}\right)\sin\left(\frac{7\pi}{144}\right)}
{\sin\left(\frac{\pi}{8}\right)}\sin(\sqrt{\kappa})\approx 0.3004\sin(\sqrt{\kappa}),
\end{equation}
and inequality \eqref{eq:inequalityf3} follows from estimates \eqref{eq:firstpartf3} and \eqref{eq:secondpartf3}.

Therefore, it only remains to establish the estimates \eqref{eq:evar3} to complete the proof. We will use the following elementary trigonometric inequalities, which can be easily verified. For $x\in\left[0,\frac{\pi}{8}\right]$, one has
\begin{subequations}
\label{eq:trigineq}
\begin{align}
    \sin\left(\frac{x}{2}\right) &\leq \frac{1}{2\cos\left(\frac{\pi}{16}\right)}\sin(x), \label{eq:trigineqa} \\
    \sin^{2}\left(\frac{x}{2}\right) &\leq \frac{1}{2}\tan\left(\frac{\pi}{16}\right)\sin(x), \label{eq:trigineqb} \\
    \sin\left(\frac{x}{2}\right) &\geq \frac{1}{2}\sin(x), \label{eq:trigineqc} \\
    \sin\left(\frac{7x}{18}\right) &\geq \frac{\sin\left(\frac{7\pi}{144}\right)}{\sin\left(\frac{\pi}{8}\right)}\sin(x). \label{eq:trigineqd}
\end{align}
\end{subequations}
We now express the quantities to be estimated in terms of the parameters $\Gamma$, $\Lambda$, $q_{1}$ and $q_{2}$. First, we use the identity
\begin{equation*}
    \sin(\vartheta_{3}-\sqrt{\kappa}q_{2})=\sin(\vartheta_{3}+\sqrt{\kappa}q_{1}-\sqrt{\kappa})=\sin(\vartheta_{3}+\sqrt{\kappa}q_{1})\cos(\!\sqrt{\kappa})-\cos(\vartheta_{3}+\sqrt{\kappa}q_{1})\sin(\!\sqrt{\kappa}).
\end{equation*}
Substituting this into formula \eqref{eq:varthetah3} and solving for $\tan(\vartheta_{3}+\sqrt{\kappa}q_{1})$, we obtain
\begin{equation*}
    \tan(\vartheta_{3}+\sqrt{\kappa}q_{1})=-\frac{\Lambda\sin(\!\sqrt{\kappa})}{1-\Lambda\cos(\!\sqrt{\kappa})}.
\end{equation*}
Now we use the identity $\sin y=\tfrac{\tan y}{\sqrt{1+\tan^{2}y}}$ which is valid for $y\in (-\frac{\pi}{2},\frac{\pi}{2})$.
Considering that  $\vartheta_{3}\in(-\pi/2,-\sqrt{\kappa}q_{1})$, the previous expression for $\tan(\vartheta_{3}+\sqrt{\kappa}q_{1})$, implies 
\begin{equation}
\label{eq:sinsquarevartheta3}
    \sin^{2}(\vartheta_{3}+\sqrt{\kappa}q_{1})=\frac{\Lambda^{2}\sin^{2}(\sqrt{\kappa})}{1-2\Lambda\cos(\!\sqrt{\kappa})+\Lambda^{2}}.
\end{equation}

We now return to formula \eqref{eq:varthetah3} to derive an explicit expression for $\tan(-\vartheta_{3})$. Expanding the sine of the sum on both sides and solving for $\tan(-\vartheta_{3})$, yields
\begin{equation*}
    \tan(-\vartheta_{3}) = \frac{\sin(\!\sqrt{\kappa}q_{1})+\Lambda\sin(\!\sqrt{\kappa}q_{2})}{\cos(\!\sqrt{\kappa}q_{1})-\Lambda\cos(\!\sqrt{\kappa}q_{2})}.
\end{equation*}
Using the identities $\sin y=\frac{\tan y}{\sqrt{1+\tan^{2}y}}$ and $\cos y=\frac{1}{\sqrt{1+\tan^{2}y}}$, 
which are valid for $y\in (-\frac{\pi}{2},\frac{\pi}{2})$,
 together with the previous expression for $\tan(-\vartheta_{3})$, we deduce
\begin{subequations}
\label{eq:sincosvartheta3}
\begin{align}
    \sin(-\vartheta_{3}) & = \frac{\sin(\!\sqrt{\kappa}q_{1})+\Lambda\sin(\!\sqrt{\kappa}q_{2})}{\sqrt{1-2\Lambda\cos(\!\sqrt{\kappa})+\Lambda^{2}}},\label{eq:sincosvartheta3a}\\
    \cos(-\vartheta_{3}) & =\frac{\cos(\!\sqrt{\kappa}q_{1})-\Lambda\cos(\!\sqrt{\kappa}q_{2})}{\sqrt{1-2\Lambda\cos(\!\sqrt{\kappa})+\Lambda^{2}}}.\label{eq:sincosvartheta3b}
\end{align}
\end{subequations}

Next, we obtain an upper bound for the expression $1-2\Lambda\cos(\!\sqrt{\kappa})+\Lambda^{2}$ appearing in the
denominators of formulas \eqref{eq:sinsquarevartheta3} and \eqref{eq:sincosvartheta3}. Using that $\Lambda<1$ and the trigonometric inequality \eqref{eq:trigineqa}, we have:
\begin{equation}
\label{eq:denominatorLambda}
\begin{aligned}
    1-2\Lambda\cos(\!\sqrt{\kappa})+\Lambda^{2} 
    &= 1-2\Lambda+\Lambda^{2}+2\Lambda\big(1-\cos(\!\sqrt{\kappa})\big) = (1-\Lambda)^{2} + 4\Lambda\sin^{2}\left(\frac{\sqrt{\kappa}}{2}\right) \\
    &\le (1-\Lambda)^{2} + \frac{1}{\cos^{2}\left(\frac{\pi}{16}\right)}\sin^{2}(\sqrt{\kappa}).
\end{aligned}
\end{equation}
On the other hand, using that $\Lambda<1$ and formula \eqref{eq:alternativeLambda}, we have
\begin{equation*}
    1-\Lambda \leq 1-\Lambda^{3} = 1-\frac{\cos(\sqrt{\kappa}q_{2})}{\cos(\sqrt{\kappa}q_{1})} = \frac{\cos(\sqrt{\kappa}q_{1})-\cos(\sqrt{\kappa}q_{2})}{\cos(\sqrt{\kappa}q_{1})}.
\end{equation*}
We may further bound the numerator of this expression using the identity for the difference of cosines,  
together with the fact that $q_1+q_2=1$, and inequality~\eqref{eq:trigineqb}, to obtain
\begin{equation*}
    \cos(\sqrt{\kappa}q_{1})-\cos(\sqrt{\kappa}q_{2}) = 2\sin\left(\frac{\sqrt{\kappa}}{2}\right)\sin\left(\frac{\sqrt{\kappa}}{2}(q_{2}-q_{1})\right) \leq 2\sin^{2}\left(\frac{\sqrt{\kappa}}{2}\right) \leq \tan\left(\frac{\pi}{16}\right)\sin(\!\sqrt{\kappa}).
\end{equation*}
Moreover, using $q_{1}<\tfrac{1}{2}$ and $0<\kappa<\tfrac{\pi^{2}}{64}$ together with the monotonicity of the cosine function, we may bound
\begin{equation*}
    \frac{1}{\cos(\!\sqrt{\kappa}q_{1})} \leq \frac{1}{\cos\left(\frac{\pi}{16}\right)}.
\end{equation*}
Combining the last inequalities, we conclude that
\begin{equation*}
    1-\Lambda \leq \frac{\tan\left(\frac{\pi}{16}\right)}{\cos\left(\frac{\pi}{16}\right)}\sin(\!\sqrt{\kappa}).
\end{equation*}
Coming back to the computation \eqref{eq:denominatorLambda} and using the previous bound for $1-\Lambda$, we obtain
the following  upper bound, which will be used below:
\begin{equation}
\label{eq:bound-denominator}
    1-2\Lambda\cos(\!\sqrt{\kappa})+\Lambda^{2} \leq \left[ \frac{\tan^{2}\left( \frac{\pi}{16} \right)}{\cos^{2}\left(\frac{\pi}{16}\right)} + \frac{1}{\cos^{2}\left( \frac{\pi}{16} \right)} \right]\sin^{2}(\sqrt{\kappa}) = \frac{\sin^{2}(\sqrt{\kappa})}{\cos^{4}\left( \frac{\pi}{16} \right)}.
\end{equation}

Using the lower bound for $\Lambda$ from Lemma~\ref{lemma:estimatesparameters}, together with 
the estimate \eqref{eq:bound-denominator} in relation \eqref{eq:sinsquarevartheta3}, yields
\begin{equation*}
	\sin^{2}(\vartheta_{3}+\sqrt{\kappa}q_{1}) \geq \frac{\cos^{2/3}\!\left( \sqrt{\kappa} \right)\sin^{2}(\sqrt{\kappa})}{\frac{\sin^{2}(\sqrt{\kappa})}{\cos^{4}\left( \frac{\pi}{16} \right)}}=\cos^{2/3}\!\left( \sqrt{\kappa} \right)\cos^{4}\!\left( \frac{\pi}{16} \right),\end{equation*}
establishing inequality \eqref{eq:evarc}.

We now estimate the numerator of the right-hand side of Eq.~\eqref{eq:sincosvartheta3a}  Using the inequalities $0<\Lambda<1$,
 $0<q_{2}-q_{1}<1$, and the inequality \eqref{eq:trigineqc} we can bound
 \begin{equation*}
\begin{aligned}
    \sin(\!\sqrt{\kappa}q_{1})+\Lambda\sin(\!\sqrt{\kappa}q_{2}) 
    & \geq \Lambda\big[\sin(\!\sqrt{\kappa}q_{1})+\sin(\!\sqrt{\kappa}q_{2})\big]  = 2\Lambda\sin\left(\frac{\sqrt{\kappa}}{2}\right)\cos\left(\frac{\sqrt{\kappa}}{2}(q_{2}-q_{1})\right) \\
    & \geq \Lambda\sin(\!\sqrt{\kappa})\cos\left(\frac{\sqrt{\kappa}}{2}\right).
\end{aligned}
\end{equation*}
Using  the  lower bound for $\Lambda$ from Lemma~\ref{lemma:estimatesparameters}
and the restriction $\sqrt{\kappa}<\frac{\pi}{8}$, we may further bound the above expression 
to obtain
\begin{equation*}
    \sin(\!\sqrt{\kappa}q_{1})+\Lambda\sin(\!\sqrt{\kappa}q_{2}) 
     \geq  \cos^{1/3}\left(\frac{\pi}{8}\right)\cos\left(\frac{\pi}{16}\right)\sin(\!\sqrt{\kappa}).
\end{equation*}
The above inequality together with the estimate \eqref{eq:bound-denominator} allows us to bound $\sin(-\vartheta_{3})$
in Eq.\eqref{eq:sincosvartheta3a} by 
\begin{equation*}
    \sin(-\vartheta_{3}) \geq \frac{\cos^{1/3}\left(\frac{\pi}{8}\right)\cos\left(\frac{\pi}{16}\right)\sin(\!\sqrt{\kappa})}{\frac{\sin(\!\sqrt{\kappa})}{\cos^{2}\left(\frac{\pi}{16}\right)}}=\cos^{1/3}\left(\frac{\pi}{8}\right)\cos^{3}\left(\frac{\pi}{16}\right),
\end{equation*}
which proves inequality \eqref{eq:evara}.

Finally, in order to prove inequality \eqref{eq:evarb}, we estimate the numerator on the right-hand side of Eq.~\eqref{eq:sincosvartheta3b}. 
Using $\Lambda<1$, together with the estimate \eqref{eq:trigineqc}, $q_1+q_2=1$, and the inequality $q_2-q_1\geq \frac{7}{9}$, 
which follows from Lemma~\ref{lemma:estimatesparameters}, yields
\begin{equation*}
\begin{aligned}
    \cos(\!\sqrt{\kappa}q_{1})-\Lambda\cos(\!\sqrt{\kappa}q_{2})
    & \geq \Lambda\big[\cos(\!\sqrt{\kappa}q_{1})-\cos(\!\sqrt{\kappa}q_{2})\big] = 2\Lambda\sin\left(\frac{\sqrt{\kappa}}{2}\right)\sin\left(\frac{\sqrt{\kappa}}{2}(q_{2}-q_{1})\right) \\
    & \geq \Lambda \sin(\sqrt{\kappa}) \sin\left(\frac{7\sqrt{\kappa}}{18}\right) .
\end{aligned}
\end{equation*}
Using the estimate \eqref{eq:trigineqd}, the lower bound for $\Lambda$ from Lemma~\ref{lemma:estimatesparameters},
and $0<\sqrt{\kappa}<\frac{\pi}{8}$, we may further bound the above expression to obtain,
\begin{equation*}
    \cos(\!\sqrt{\kappa}q_{1})-\Lambda\cos(\!\sqrt{\kappa}q_{2})\geq \cos^{1/3}\left(\frac{\pi}{8}\right)\frac{\sin\left(\frac{7\pi}{144}\right)}{\sin\left(\frac{\pi}{8}\right)}\sin^{2}(\sqrt{\kappa}).
\end{equation*}
Combining the above inequality with the estimate \eqref{eq:bound-denominator}, we obtain the following bound for $\cos(-\vartheta_{3})$ in formula \eqref{eq:sincosvartheta3b}:
\begin{equation*}
    \cos(-\vartheta_{3})
     \geq \frac{\cos^{1/3}\left(\frac{\pi}{8}\right)\frac{\sin\left(\frac{7\pi}{144}\right)}{\sin\left(\frac{\pi}{8}\right)}\sin^{2}(\sqrt{\kappa})}{\frac{\sin(\!\sqrt{\kappa})}{\cos^{2}\left(\frac{\pi}{16}\right)}}  = \frac{\cos^{1/3}\left(\frac{\pi}{8}\right)\sin\left(\frac{7\pi}{144}\right)\cos^{2}\left(\frac{\pi}{16}\right)}{\sin\left(\frac{\pi}{8}\right)}\sin(\!\sqrt{\kappa}),
\end{equation*}
which proves inequality \eqref{eq:evarb}.
\qed

\subsection{Derivation of the asymptotic expansions in Sec.~\ref{sec:asymp}}
\label{app:appendix}

\subsubsection{Expansions in Theorem~\ref{prop:asymLpositive}}

We will employ the notation from the proof of Theorem~\ref{thm:stabilitycollinear} in Sec.~\ref{app:stability-collinear},  and denote by $\theta_{\Ll_{1}}$  the value of the coordinate $\theta$ of the point   $\Ll_{1}\in \mathcal{G}^\pm$. This value is determined as the appropriate root of  Eq.~\eqref{eq:collinearpositive} if $\kappa>0$ and of Eq.~\eqref{eq:collinearnegative} if $\kappa<0$.  Furthermore, we will explicitly indicate its dependence on the parameters $(\kappa, \mu)$ and write $\theta_{\Ll_{1}}(\kappa,\mu)$. This notation applies to both positive and negative curvature, and extends to the other collinear RE in an obvious manner. We present the computations for the case $\kappa>0$; the expansions for $D_{\Ll_{i},C}$ when $\kappa<0$ follow by analogous arguments. 

We will repeatedly employ the expansion 
\begin{equation}
\label{eq:q2exp}
    q_{2} = 1 + \mathcal{O}(\mu) + \mathcal{O}_\mu(\kappa),
\end{equation}
which follows from Eq.~\eqref{eq:expansionsparameters}, where $\sqrt{\kappa}$ is denoted by $t$.

In view of Proposition~\ref{prop:scaling}, for $\kappa>0$, the signed Riemannian distance 
from $\Ll_1$ to the center of rotation $C$ in the statement of the theorem
 is given by
\begin{equation}
\label{eq:rescalingcolpos}
    D_{\Ll_{1},C}(\kappa,\mu) = \frac{\theta_{\Ll_{1}}(\kappa,\mu)}{\sqrt{\kappa}},
\end{equation}
and analogous formulas hold for all other collinear RE.

We derive an asymptotic expansion for  $\theta_{\Ll_{1}}(\kappa,\mu)$ as $\kappa\to 0^+$ based on the 
ansatz:
\begin{equation*}
    \theta_{\Ll_{1}}(\kappa,\mu) = \sqrt{\kappa}q_{2} - \sum_{n=1}^{\infty} c_n(\mu) \kappa^{n/2}.
\end{equation*}
Substituting the above expansion into the equation
\begin{equation*}
    f_1\big(\theta_{\Ll_{1}}(\kappa,\mu);\kappa,\mu\big) = 0,
\end{equation*}
 that determines $\theta_{\Ll_{1}}(\kappa,\mu)$,
and expanding in powers of $\sqrt{\kappa}$, we obtain equations 
for the determination of the coefficients $c_n(\mu)$, $n=1,2,\dots$. At first order, we get
\begin{equation*}
    (1 + \mu)c_{1}^5 + (3 + 2\mu)c_{1}^4 + (3 + \mu)c_{1}^3 - \mu c_{1}^2 - 2\mu c_{1} - \mu = 0.
\end{equation*}
Introducing the scaling $c_{1} = \mu^{1/3} y$ and dividing by $\mu$ yields
\begin{equation*}
    (1 + \mu)\mu^{2/3}y^5 - (3 + 2\mu)\mu^{1/3}y^4 + (3 + \mu)y^3 - \mu^{2/3}y^2 + 2\mu^{1/3}y - 1 = 0,
\end{equation*}
which reduces in the limit $\mu \to 0$ to
\begin{equation*}
    3y^3 - 1 = 0.
\end{equation*}
The unique real solution is $y_0 = (1/3)^{1/3}$, and by the Implicit Function Theorem there exists a unique real solution $y(\mu)$ for $\mu$ sufficiently small such that
\begin{equation*}
    y(\mu) = (1/3)^{1/3} + \mathcal{O}(\mu^{1/3}),
\end{equation*}
giving
\begin{equation*}
    c_{1}(\mu) = (1/3)^{1/3} \mu^{1/3} + \mathcal{O}(\mu^{2/3}).
\end{equation*}

Therefore, in view of formula \eqref{eq:q2exp}, we obtain
\begin{equation*}
    \theta_{\Ll_{1}}(\kappa,\mu) = (1 - (1/3)^{1/3} \mu^{1/3} + \mathcal{O}(\mu^{2/3}))\sqrt{\kappa}+\mathcal{O}_\mu(\kappa),
\end{equation*}
and relation \eqref{eq:rescalingcolpos} yields
\begin{equation*}
    D_{\Ll_{1},C}(\kappa,\mu) = 1 - (1/3)^{1/3} \mu^{1/3} + \mathcal{O}(\mu^{2/3}) + \mathcal{O}_\mu(\kappa^{1/2}),
\end{equation*}
as required.

The determination of the expansion of $D_{\Ll_{2},C}(\kappa,\mu)$ is analogous. We substitute the ansatz
\begin{equation*}
    \theta_{\Ll_{2}}(\kappa,\mu) = \sqrt{\kappa}q_{2} + \sum_{n=1}^{\infty} c_n(\mu) \kappa^{n/2},
\end{equation*}
into the equation
\begin{equation*}
    f_2\big(\theta_{\Ll_{2}}(\kappa,\mu);\kappa,\mu\big) = 0,
\end{equation*}
that determines $\theta_{\Ll_{2}}(\kappa,\mu)$, and expand in powers of $\sqrt{\kappa}$. This leads to the condition that $c_{1}(\mu)$ satisfies the quintic polynomial
\begin{equation*}
    c_{1}^5 + (3 + 2\mu)c_{1}^4 + (3 + \mu)c_{1}^3 - \mu c_{1}^2 - 2\mu c_{1} - \mu = 0.
\end{equation*}
Proceeding as above one obtains
\begin{equation*}
    c_{1}(\mu) = (1/3)^{1/3} \mu^{1/3} + \mathcal{O}(\mu^{2/3}),
\end{equation*}
which, in view of formula \eqref{eq:q2exp}, and the condition $D_{\Ll_{2},C}(\kappa,\mu)= \theta_{\Ll_{2}}(\kappa,\mu)/\sqrt{\kappa}$
analogous to relation \eqref{eq:rescalingcolpos}, yields the desired expansion
\begin{equation*}
    D_{\Ll_{2},C}(\kappa,\mu) = 1 + (1/3)^{1/3} \mu^{1/3} + \mathcal{O}(\mu^{2/3}) + \mathcal{O}_\mu(\kappa^{1/2}).
\end{equation*}

To study $\Ll_{3}$, we substitute the ansatz
\begin{equation*}
    \theta_{\Ll_{3}}(\kappa,\mu) = -\sqrt{\kappa} + \sum_{n=1}^{\infty} c_n(\mu) \kappa^{n/2},
\end{equation*}
into the equation
\begin{equation*}
    f_3\big(\theta_{\Ll_{3}}(\kappa,\mu);\kappa,\mu\big) = 0,
\end{equation*}
that determines $\theta_{\Ll_{3}}(\kappa,\mu)$. Expanding in powers of $\sqrt{\kappa}$ yields
\begin{equation*}
    (1 + 4 \mu) c_{1}^5 - (7 + 24 \mu) c_{1}^4 + (19 + 54 \mu) c_{1}^3 - (24 + 52 \mu) c_{1}^2 + (12 + 14 \mu) c_{1} + 5 \mu + \mathcal{O}(\mu^2) = 0.
\end{equation*}
where we have kept all terms up to order $\mu$. Introducing the rescaling $c_{1} = \mu y$ and dividing the polynomial by $\mu$, the leading-order polynomial becomes
\begin{equation*}
    12 y + 5 = 0,
\end{equation*}
which yields,
\begin{equation*}
    c_{1}(\mu) = -\frac{5}{12} \mu + \mathcal{O}(\mu^2).
\end{equation*}
Substituting back into the ansatz for $\theta_{\Ll_{3}}(\kappa,\mu)$ and using the condition $D_{\Ll_{3}}(\kappa,\mu)= \theta_{\Ll_{3}}(\kappa,\mu) /\sqrt{\kappa}$ analogous to relation \eqref{eq:rescalingcolpos}, we obtain
\begin{equation*}
    D_{\Ll_{3},C}(\kappa,\mu) = -1 - \frac{5}{12} \mu + \mathcal{O}(\mu^2) + \mathcal{O}_\mu(\kappa^{1/2}),
\end{equation*}
as required.

For $\theta_{\Aa_{1}}(\kappa,\mu)$ we propose the ansatz
\begin{equation*}
    \theta_{\Aa_{1}}(\kappa,\mu) = -\pi + \kappa^{1/2}\, q_{2} - \sum_{n=1}^{\infty} c_n(\mu)\, \kappa^{n/2}.
\end{equation*}
Substituting this expansion into the equation 
\begin{equation*}
    f_{4}\big(\theta_{\Aa_{1}}(\kappa,\mu);\kappa,\mu\big) = 0,
\end{equation*}
that determines $\theta_{\Aa_{1}}(\kappa,\mu)$ and 
 expanding in powers of $\kappa^{1/2}$, we find that 
 the leading order coefficients $c_{1}(\mu)$ and $c_{2}(\mu)$ satisfy the polynomial equations
\begin{equation*}
\begin{gathered}
    (1 + \mu) c_{1}^{5} - (3 + 2 \mu) c_{1}^{4} + (3 + \mu) c_{1}^{3} - (2 - \mu) c_{1}^{2} - 2 \mu c_{1} + \mu = 0, \\
    c_{2} \Big[ (1 + \mu) c_{1}^{6} - (3 + 3 \mu) c_{1}^{5} + (3 + 3 \mu) c_{1}^{4} + (1 - 3 \mu) c_{1}^{3} + 6 \mu c_{1}^{2} - 6 \mu c_{1} + 2 \mu \Big] = 0.
\end{gathered}
\end{equation*}
Solving these equations up to order $\mu^{3/2}$ yields three possible solutions, but only one ensures that $\theta_{\Aa_{1}}$ remains in $\mathcal{I}_4$ for small $\mu$. This solution is
\begin{equation*}
    c_{1}(\mu) = \frac{\sqrt{\mu}}{\sqrt{2}} - \frac{\mu}{8}, \qquad c_{2}(\mu) = 0.
\end{equation*}
Substituting back into the ansatz for $\theta_{\Aa_{1}}(\kappa,\mu)$, we get
\begin{equation*}
    \theta_{\Aa_{1}}(\kappa,\mu) = -\pi + \left( q_{2} - \frac{1}{\sqrt{2}}\mu^{1/2} + \frac{\mu}{8}  \right)\kappa^{1/2} + \mathcal{O}_\mu(\kappa).
\end{equation*}
 Using formula \eqref{eq:q2exp}, and the condition $D_{\Aa_{1},C}(\kappa,\mu)= \theta_{\Ll_{2}}(\kappa,\mu)/\sqrt{\kappa}$
analogous to relation \eqref{eq:rescalingcolpos}, yields the desired expansion
\begin{equation*}
    D_{\Aa_{1},C}(\kappa,\mu) = -\frac{\pi}{\sqrt{\kappa}} + \left(1 - \frac{1}{\sqrt{2}}\mu^{1/2} - \frac{7}{8} \mu \right) + \mathcal{O}_\mu(\kappa).
\end{equation*}

We now work with $\Ee_{2}$ proposing the ansatz
\begin{equation*}
    \theta_{\Ee_{2}}(\kappa,\mu) = \sum_{n=0}^{\infty} c_n(\mu) \kappa^{n/2}.
\end{equation*}
Substituting into the equation
\begin{equation*}
    f_2\big(\theta_{\Ee_{2}}(\kappa,\mu);\kappa,\mu\big) = 0,
\end{equation*}
and expanding in powers of $\kappa^{1/2}$, we find that the first few coefficients satisfy the system
\begin{equation*}
\begin{aligned}
& \text{order } \kappa^0: && \cos(c_0) \, \sin(c_0) = 0, \\
& \text{order } \kappa^{1/2}: && c_{1} \, \cos(2 c_0) = 0, \\
& \text{order } \kappa^1: && c_{2} \, \cos(2 c_0) - c_{1}^2 \, \sin(2 c_0) = 0, \\
& \text{order } \kappa^{3/2}: && \big(-\tfrac{2}{3} c_{1}^3 + c_{3}\big) \cos(2 c_0) - \csc^2(c_0) - 2 c_{1} c_{2} \, \sin(2 c_0) = 0.
\end{aligned}
\end{equation*}
Solving this system, we find a unique solution in $\mathcal{I}_2$ given by
\begin{equation*}
    c_0 = \frac{\pi}{2}, \quad c_{1} = 0, \quad c_{2} = 0, \quad c_{3} = -1.
\end{equation*}
Substituting back into the ansatz, we have
\begin{equation*}
    \theta_{\Ee_{2}}(\kappa,\mu) = \frac{\pi}{2} - \kappa^{3/2} + \mathcal{O}(\kappa^2),
\end{equation*}
and using an analog of relation \eqref{eq:rescalingcolpos} for $\Ee_2$ we obtain
\begin{equation*}
    D_{\Ee_{2},C}(\kappa,\mu) = \frac{\pi}{\sqrt{\kappa}} - \kappa + \mathcal{O}(\kappa^2),
\end{equation*}
as required.

The procedure for $\Ee_3$ is completely analogous. We propose the ansatz
\begin{equation*}
    \theta_{\Ee_{3}}(\kappa,\mu) = \sum_{n=0}^{\infty} c_n(\mu) \kappa^{n/2}.
\end{equation*}
Substituting  into 
\begin{equation*}
    f_3\big(\theta_{\Ee_{3}}(\kappa,\mu);\kappa,\mu\big) = 0,
\end{equation*}
and expanding in powers of $\kappa^{1/2}$, yields
\begin{equation*}
\begin{aligned}
& \text{order } \kappa^0: && \cos(c_0) \, \sin(c_0) = 0, \\
& \text{order } \kappa^{1/2}: && c_{1} \, \cos(2 c_0) = 0, \\
& \text{order } \kappa^1: && c_{2} \, \cos(2 c_0) - c_{1}^2 \, \sin(2 c_0) = 0, \\
& \text{order } \kappa^{3/2}: && \big(-\tfrac{2}{3} c_{1}^3 + c_{3}\big) \cos(2 c_0) + \csc^2(c_0) - 2 c_{1} c_{2} \, \sin(2 c_0) = 0.
\end{aligned}
\end{equation*}
The unique solution of interest is 
\begin{equation*}
    c_0 = -\frac{\pi}{2}, \quad c_{1} = 0, \quad c_{2} = 0, \quad c_{3} = 1.
\end{equation*}
which leads to the desired expansion
\begin{equation*}
    D_{\Ee_{3},C}(\kappa,\mu) = -\frac{\pi}{\sqrt{\kappa}} + \kappa + \mathcal{O}(\kappa^2).
\end{equation*}

\subsubsection{Expansions in Theorem~\ref{th:asympexpL4L5}}

We already know that, for sufficiently small $\kappa>0$ and for any $\kappa<0$, there exist only two triangular relative equilibria, which we associate with $\Ll_{4}$ and $\Ll_{5}$. This identification will be justified below. 

Instead of considering the distance to the center of rotation, we study the asymptotic expansion of the Riemannian distance to the primary $\mathbf{p}_{2}$. This  distance, measured
 from either $\Ll_{4}$ or $\Ll_{5}$ to  $\mathbf{p}_{2}$, follows from formulas \eqref{eq:auxiliardistancecoords} and Remark \ref{rmk:distances}, and is given by
\begin{equation}
\label{eq:rescalingtripos}
	D_{\Ll_{4},\mathbf{p}_{2}}(\kappa,\mu) = D_{\Ll_{5},\mathbf{p}_{2}}(\kappa,\mu) = \frac{x_{0}(\kappa,\mu)}{\sqrt{\vert\kappa\vert}}.
\end{equation}
Here $x_{0}(\kappa,\mu)$ denotes the root of the function $g_{\kappa,\mu}$ defined by the formula \eqref{eq:auxfunctiontriangular} when $\kappa>0$, and of the function
\begin{equation}
\label{eq:functiontrineg}
\begin{aligned}
    g_{\kappa,\mu}^{-}(x)
    &=\left(\Lambda_{\kappa,\mu}^{-}\right)^{3}\sinh(\sqrt{-\kappa}q_{2})\sinh^{3}(x)
    \left[\sinh(\sqrt{-\kappa}q_{1})\cosh(x)
    +\sinh(\sqrt{-\kappa}q_{2})\sqrt{1+\left(\Lambda_{\kappa,\mu}^{-}\right)^{2}\sinh^{2}(x)}\right] \\
    &\quad-\Gamma_{\kappa,\mu}\sinh(\sqrt{-\kappa}),
\end{aligned}
\end{equation}
when $\kappa<0$. The distance to the primary $\mathbf{p}_{1}$ can be obtained from the conditions \eqref{eq:mubalancecondition} for both $\kappa>0$ and $\kappa<0$. This choice is natural since the corresponding distances are already known in the planar case, allowing for a direct comparison.

We present the computations for the case $\kappa>0$. The case $\kappa<0$ follows by analogous arguments. Let $x_{0}$ denote the root of the function $g_{\kappa,\mu}(x)$ defined in Eq.~\eqref{eq:auxfunctiontriangular}, corresponding to either $\Ll_{4}$ or $\Ll_{5}$. To derive its asymptotic expansion for small $\kappa$, we consider the ansatz
\begin{equation*}
    x_{0}(\kappa,\mu) = \sum_{n=1}^{\infty} c_n(\mu) \, \kappa^{n/2}.
\end{equation*}
Substituting this expansion into $g_{\kappa,\mu}(x_{0})$ and expanding, we obtain a system of equations for the first coefficients. It is worth noting that it was necessary to expand $g_{\kappa,\mu}$ up to fifth order in $\kappa^{1/2}$ to find the first nontrivial equations. The initial equations read
\begin{equation*}
\begin{aligned}
& \text{order } \kappa^{5/2}: && \frac{-1 + c_{1}^3}{1+\mu} = 0, \\
& \text{order } \kappa^3: && \frac{3 c_{1}^2 c_{2}}{1+\mu} = 0.
\end{aligned}
\end{equation*}
Solving the first two equations gives $c_{1}=1$ and $c_{2}=0$. Substituting these values into the ansatz and proceeding to the next orders, we obtain
\begin{equation*}
 \text{order } \kappa^{7/2}:  \frac{-1 + 6(1+\mu)c_{3}}{2(1+\mu)^{2}} = 0,
\end{equation*}
which yields, $c_{3} = 1 / 6(1+\mu)$. Finally, substituting these coefficients back into the ansatz and recalling that the Riemannian distance from the triangular equilibria to $\mathbf{p}_{2}$ is given by formula \eqref{eq:rescalingtripos}, we obtain the desired result.

For the formula for the Riemannian distance to $\mathbf{p}_{1}$, we can use formula \eqref{eq:mubalancecondition} for $\kappa>0$ to obtain that
\begin{equation*}
	D_{\Ll_{4},\mathbf{p}_{1}}(\kappa,\mu) = D_{\Ll_{5},\mathbf{p}_{1}}(\kappa,\mu) = \frac{\arcsin\big(\Lambda \sin(x_{0}(\kappa,\mu))\big)}{\sqrt{\kappa}}.
\end{equation*}
The stated expansion then follows by substituting the asymptotic expansion of $x_{0}$, together with the expansion of $\Lambda$ given in formulas \eqref{eq:expansionsparameters}, and expanding the resulting expression in powers of $\kappa$.

\section{Computer-assisted Proofs (CAPs)}
\label{App:CAPs}

The CAPs developed in this paper rely on rigorous bounds of real-valued  functions over closed intervals.
By \emph{rigorous}, we mean that the computed bounds are mathematically 
guaranteed to enclose the true values of the function on the domain under consideration. These bounds  are obtained
 using interval arithmetic \cite{MR0231516,Tucker2011}, which provides  rigorous implementations of the basic arithmetic operations and elementary functions on closed intervals, while accounting for computer round-off errors.

The code for our CAPs \cite{github_codes} is written in the Julia programming language \cite{Julia-2017} using
 the IntervalArithmetic.jl library \cite{IntervalArithmetic.jl}. For any function \(F \colon [a,b] \to \mathbb{R}\),
  built from the operations and functions supported by the library, the corresponding interval implementation, \(F_{\mathrm{Int}}\),
   satisfies  
\[
F(X):=\{ F(x) : x \in X \} \subset F_{\mathrm{Int}}(X),
\]
for every closed subinterval \(X \subset [a,b]\). In particular, the computed interval \(F_{\mathrm{Int}}(X)\) is guaranteed 
to contain the true range of \(F\) on \(X\).

\subsection{Existence CAPs}
\label{app:existenceCAPs}

The existence CAPs in Secs.~\ref{sss:CAP-existence-collinear} and \ref{ss:CAPExistenceTriangular} rely on
the rigorous determination of the number of roots of the functions $f_1,\dots, f_4$, and $g_{\kappa,\mu}$ in their respective 
domains. As we explain below,  this procedure 
can be reduced to the rigorous evaluation of real-valued functions on closed intervals. 

Consider a continuous function \(F \colon [a,b] \to \mathbb{R}\)  which is differentiable on $(a,b)$. Suppose that 
there is numerical evidence that $F$ has a unique root in \([a,b]\).   To rigorously 
show that \(F\) has indeed a unique root in \([a,b]\), we proceed as follows. First, we verify that the intervals \(F_{\mathrm{Int}}(a)\) and \(F_{\mathrm{Int}}(b)\) have opposite signs, which implies that there is at least one root in \([a,b]\) by the intermediate 
value theorem. Second, we find a subinterval \([c,d] \subset [a,b]\) on which we can show that \(F\) is monotone by checking that the derivative \(F'\) is strictly positive or strictly negative on \([c,d]\). Finally, we verify that \(F\) does not vanish on the complement \([a,b] \setminus [c,d]\). These three steps prove the existence of a unique root in \([a,b]\), lying on \([c,d]\), and are performed using the  
interval arithmetic implementations \(F_{\mathrm{Int}}\) and \(F'_{\mathrm{Int}}\). In practice, we look for a sharp enclosure
of the root, meaning that we prove the root lies in a very small subinterval \([c,d]\), since
this is useful for the  stability CAPs. 
 
The case of nonexistence of roots is treated similarly. Once we have numerical evidence that \(F\) has no roots on \([a,b]\), we compute  \(F_{\mathrm{Int}}([a,b])\) and verify that this interval does not contain zero.

The procedures described above may fail when the values of $F$ or $F'$ remains close to zero over a substantial portion  
of $[a,b]$.
 In such cases,  the interval arithmetic implementations \(F_{\mathrm{Int}}\) and \(F'_{\mathrm{Int}}\) may be unable
 to determine the sign of $F$ or $F'$ on the relevant subintervals.

Suppose now that the function $F:[a,b]\to \R$, like the functions $f_1, \dots, f_4$ and $g_{\kappa,\mu}$, 
depends on the parameters $(\kappa,\mu)$. Our approach treats these parameters as intervals. 
 Specifically, for given values of \(\kappa\) and \(\mu\), we consider the intervals  
$$
[\kappa - \delta_\kappa,\;\kappa + \delta_\kappa]
\quad\text{and}\quad
[\mu - \delta_\mu,\;\mu + \delta_\mu],
$$
where \(\delta_\kappa,\delta_\mu >0 \). Each choice of \((\kappa,\mu)\) 
together with $\delta_\kappa$ and $\delta_\mu$ determines a rectangle in the parameter plane. If the interval-arithmetic 
procedure described above successfully determines the existence and location of roots of $F$ on $[a,b]$, then the conclusion is valid 
 on the whole  parameter rectangle. Similar considerations allow us to treat the dependence
 of the endpoints of the interval $[a,b]$ on the parameters $(\kappa,\mu)$. To obtain results on larger regions of parameter space, we cover the region of interest 
 by a finite collection of overlapping parameter rectangles and perform the same verification on each rectangle. 
  
The procedure described above is not directly  applicable to the functions $f_j$, $j=1,\dots, 4$, because their
domains  $\mathcal{I}_j$ are open intervals and the functions become unbounded near the endpoints. 
 To address this issue, we first define appropriate  closed subintervals $\mathcal{J}_j\subset \mathcal{I}_j$, and prove analytically
 that all roots of $f_j$ in $\mathcal{I}_j$ are contained in   $\mathcal{J}_j$.
The details are presented in Sec.~\ref{app:compactification} below.

\subsection{Stability CAPs}
\label{app:stabilityCAPs}

The implementation of the stability CAPs illustrated in Secs.~\ref{ss:CAPsCollinear} and \ref{ss:stability-triangular-CAPs} 
is straightforward. After determining an enclosure for a root corresponding to a RE, as described in 
the previous section, we evaluate the relevant quantities in Propositions~\ref{sss:stability-conds-collinear} or
\ref{prop:stabilitytrineg} using interval arithmetic, according to the type of RE under consideration (collinear or triangular).
 If the interval evaluation establishes that these quantities have a definite
sign, the corresponding stability conclusion follows from the proposition. 
The parameters $(\kappa,\mu)$ are incorporated into the implementation in the same manner as in the existence CAPs.

\subsection{\texorpdfstring
  {The closed subintervals $\mathcal{J}_j\subset \mathcal{I}_j$.}
  {The closed subintervals J_j subset I_j.}}
\label{app:compactification}

Let $(\kappa,\mu)\in \Ps{1}$. We define,
\begin{equation*}
\begin{split}
\mathcal{J}_1&=\left[ 0,\, \arctan\!\left( \frac{\tan(\sqrt{\kappa} q_{2})}{1+\mu} \right) \right], \\
 \mathcal{J}_2&= \left[
            \left( 1 + \frac{3^{1/4}}{\pi^{3/2}} 
            \sqrt{\frac{\mu}{1+\mu}}\, \epsilon_{\kappa} \right) 
            \sqrt{\kappa} q_{2}, \, 
            \frac{\pi}{2}
        \right], \\
  \mathcal{J}_3&= \left[ -\frac{\pi}{2},
            -\left( 1 + \frac{3^{1/4}}{\pi^{3/2}} 
            \frac{1}{\sqrt{1+\mu}}\, \epsilon_{\kappa} \right) 
            \sqrt{\kappa} q_{1}
        \right],\\
        \mathcal{J}_4&=\left[  -\pi - \left(1 - \sqrt{\frac{4 \Gamma}{2 \kappa q_{2}^{2} + \pi^{2} \Gamma \mu}} \right)\sqrt{\kappa} q_{1}  \; , \;  -\pi + \left( 1 - \frac{q_{1}}{q_{2}}\sqrt{\frac{4\Gamma \mu}{2\kappa q_{1}^{2}+\pi^{2}\Gamma}} \right) \sqrt{\kappa}q_{2} \right], \\
\end{split}
\end{equation*}
where the value of $\epsilon_\kappa$ in $\mathcal{J}_2$ and $\mathcal{J}_3$ is given by
\begin{equation}
\label{eq:defepsilonk}
	\epsilon_{\kappa} =
        		\begin{cases}
            		2 \kappa^{1/4}, & \text{if } 0 < \kappa < \left( \frac{\pi}{3} \right)^{2}, \\[6pt]
            		\dfrac{3}{4} \left( \pi - 2\sqrt{\kappa} \right), & 
            		\text{if } \left( \frac{\pi}{3} \right)^{2} < \kappa < \left( \frac{\pi}{2} \right)^{2}.
        		\end{cases}
\end{equation}

The following lemma shows that, for all $j=1,\dots, 4$, all  roots of  $f_j$ are contained in the closed interval 
interval $\mathcal{J}_j$. Consequently,  the CAPs described in Sec.~\ref{app:existenceCAPs} can be applied on $\mathcal{J}_j$.

\begin{lemma}
Fix $(\kappa,\mu)\in \Ps{1}$. For all $j=1,\dots, 4$, the intervals $\mathcal{J}_j$
are contained in $\mathcal{I}_j$. Furthermore, if $\theta_0\in \mathcal{I}_j$ is a root of the function
$\theta\mapsto f_j(\theta;\kappa,\mu)$ then $\theta_0\in \mathcal{J}_j$.
\end{lemma}
\begin{proof}
We recall that the explicit expressions for  $f_j$ and their domains $\mathcal{I}_j$ are given in formulas \eqref{eq:f1}, \eqref{eq:f2}, \eqref{eq:f3} and \eqref{eq:f4}. We require different estimates for the different  values of $j$, and therefore divide the  presentation accordingly. In each case, the fact that $\mathcal{J}_{j}$ is non-empty and contained in $\mathcal{I}_{j}$ follows from the corresponding estimates. Throughout the proof we often abbreviate and write  $f_j(\theta)$ instead of $f_j(\theta;\kappa,\mu)$.

\paragraph{Case $j=1$.} Since $0<q_{1}<q_{2}<1$ and the function $1/\sin^{2}$ is decreasing in $(0,\pi/2)$, we have
\begin{equation*}
	f_{1}(0) = -\Gamma \left[ \frac{1}{\sin^{2}(\sqrt{\kappa} q_{1})} - \frac{\mu}{\sin^{2}(\sqrt{\kappa} q_{2})} \right] < 0.
\end{equation*}
On the other hand, using that $q_{1}<1/2$, we obtain the following inequality valid for $\theta\in(-\sqrt{\kappa}q_{1},0)$,
\begin{equation*}
	-\pi/2 < 2\theta < 0.
\end{equation*}
Using the formula \eqref{eq:auxf1prime} for $f_{1}'$, together with the previous inequality and inequalities \eqref{eq:estimatesdomainf1}, we can conclude that $f_{1}$ is increasing on $(-\sqrt{\kappa}q_{1},0)$. Therefore,
\begin{equation}
\label{eq:f1negative}
 	f_{1}(\theta) < 0 \quad \text{for all } \theta \in (-\sqrt{\kappa} q_{1}, 0).
\end{equation}

Consider now the auxiliary function
\begin{equation*}
	g_{1}(\theta) = -\Gamma \left[ \frac{1}{\sin^{2}(\theta+\sqrt{\kappa} q_{1})} - \frac{\mu}{\sin^{2}( \theta -\sqrt{\kappa} q_{2})} \right].
\end{equation*}
We have that $f_{1}(\theta) \ge g_{1}(\theta)$ for $\theta \in [0, \sqrt{\kappa} q_{2})$. Differentiating gives
\begin{equation*}
	g_{1}'(\theta) = 2\Gamma \left[ \frac{\cos(\theta+\sqrt{\kappa} q_{1})}{\sin^{3}(\theta+\sqrt{\kappa} q_{1})} - \mu\, \frac{\cos(\theta-\sqrt{\kappa} q_{2})}{\sin^{3}(\theta-\sqrt{\kappa} q_{2})} \right].
\end{equation*}
Formulas \eqref{eq:estimatesdomainf1} imply that $g_{1}'(\theta)>0$, so $g_1$ is increasing  for $\theta \in [0, \sqrt{\kappa} q_{2})$. 
Furthermore, a direct computation shows that $g_{1}$ has a unique zero $\theta^{*} \in (0, \sqrt{\kappa} q_{2})$ given by
\begin{equation*}
	\theta^{*} = \arctan\!\left( \frac{\sin(\sqrt{\kappa} q_{2}) - \sqrt{\mu}\, \sin(\sqrt{\kappa} q_{1})} {\cos(\sqrt{\kappa} q_{2}) + \sqrt{\mu}\, \cos(\sqrt{\kappa} q_{1})}\right).
\end{equation*}
Let
\begin{equation*}
	\theta_{R} = \arctan\!\left( \frac{\tan(\sqrt{\kappa} q_{2})}{1+\mu} \right).
\end{equation*}
Then clearly $0 < \theta_{R} < \sqrt{\kappa} q_{2}$. Moreover, using $\sqrt{\mu} > \mu$ and $\cos(\sqrt{\kappa} q_{1}) > \cos(\sqrt{\kappa} q_{2})$, we obtain
\begin{equation*}
\begin{aligned}
        \tan(\theta_{R}) & = \frac{\tan(\sqrt{\kappa} q_{2})}{1+\mu} > \frac{\tan(\sqrt{\kappa} q_{2})}{1 + \sqrt{\mu}\, \frac{\cos(\sqrt{\kappa} q_{1})}{\cos(\sqrt{\kappa} q_{2})}} > \frac{\sin(\sqrt{\kappa} q_{2})}{\cos(\sqrt{\kappa} q_{2}) + \sqrt{\mu}\, \cos(\sqrt{\kappa} q_{1})} \\
	& > \frac{\sin(\sqrt{\kappa} q_{2}) - \sqrt{\mu}\, \sin(\sqrt{\kappa} q_{1})}{\cos(\sqrt{\kappa} q_{2}) + \sqrt{\mu}\, \cos(\sqrt{\kappa} q_{1})} = \tan(\theta^{*}),
\end{aligned}
\end{equation*}
so, by monotonicity of $g_1$, we conclude that   $\theta_{R} \in (\theta^{*}, \sqrt{\kappa} q_{2})$. Using the monotonicity of $f_1$  together with the fact that $f_{1}(\theta) \ge g_{1}(\theta)$ for  $\theta \in [0, \sqrt{\kappa} q_{2})$, we obtain the following sequence of inequalities valid for $\theta\in (\theta_R,\sqrt{\kappa}q_2)$, 
\begin{equation}
\label{eq:f1positive}
	f_1(\theta)>f_1(\theta_R)>f_1(\theta^*)\geq g_1(\theta^*)=0.
\end{equation}

Inequalities \eqref{eq:f1negative} and \eqref{eq:f1positive} show that if $\theta_0\in \mathcal{I}_1$ is a root of $f_1$, then $\theta_0$ does not belong to the union $(-\sqrt{\kappa}q_1,0) \bigcup \, (\theta_R,\sqrt{\kappa}q_2)$, which is the complement of $\mathcal{J}_1$ in $\mathcal{I}_1$.

\paragraph{Case $j=4$.} We first consider the interval $(-\pi-\sqrt{\kappa}q_{1},-\pi)$. Define the auxiliary function
\begin{equation*}
	g_{4}(\theta) = -\frac{1}{2} + \Gamma \left[ \frac{1}{\sin^{2}(\theta + \sqrt{\kappa} q_{1})} - \frac{\mu}{\sin^{2}(\theta - \sqrt{\kappa} q_{2})} \right], \qquad \theta \in (-\pi-\sqrt{\kappa}q_{1},-\pi).
\end{equation*}
Then, $g_{4}(\theta)\leq f_{4}(\theta)$. Since we assume that $0<\kappa<\pi^{2}/4$, the estimates \eqref{eq:estimatedomainf4} are not longer valid and are instead replaced by
\begin{equation}
\label{eq:estimatedomaing4}
	-\pi < \theta+\sqrt{\kappa}q_{1} < -\pi/2, \qquad \text{and} \qquad -3\pi/2 < \theta-\sqrt{\kappa}q_{2} < -\pi.
\end{equation}
These estimates imply that $g_{4}'(\theta)<0$. Define
\begin{equation*}
	\theta_{L} = -\pi - \left(1 - \delta_{L}\right)\sqrt{\kappa} q_{1}, \qquad \delta_{L} = \sqrt{\frac{4 \Gamma}{2 \kappa q_{2}^{2} + \pi^{2} \Gamma \mu}}.
\end{equation*}
We claim that $g_{4}(\theta_{L}) > 0$. Since $g_{4}$ is decreasing and $g_{4}(\theta)\leq f_{4}(\theta)$, this implies that 
\begin{equation}
\label{eq:f4positive}
	f_{4}(\theta)>0 \quad \text{for all } \theta \in (-\pi-\sqrt{\kappa}q_{1},\theta_{L})
\end{equation}
Now we prove the claim. By the definition of $g_{4}$, it suffices to show that
\begin{equation*}
	\frac{\Gamma}{\sin^{2}(\delta_{L}\sqrt{\kappa} q_{1})} \geq \frac{1}{2} + \frac{\Gamma \mu }{\sin^{2}(\sqrt{\kappa} q_{2})}.
\end{equation*}
Indeed, using that $q_{1}<q_{2}$ and the definition of $\delta_{L}$, we obtain the following inequalities,
\begin{equation*}
	\frac{\Gamma}{\sin^{2}(\delta_{L}\sqrt{\kappa} q_{1})} \geq \frac{\Gamma}{\delta_{L}^{2}\kappa q_{1}^{2}} \geq \frac{2 \kappa q_{2}^{2} + \pi^{2} \Gamma \mu}{4 \Gamma} \cdot \frac{\Gamma}{\kappa q_{1}^{2}}  \geq \left(  \frac{\kappa q_{2}^{2}}{2}  + \frac{\Gamma \mu}{\tfrac{4}{\pi^{2}}} \right) \frac{1}{\kappa q_{2}^{2}} \geq \frac{1}{2} + \frac{\Gamma \mu }{\sin^{2}(\sqrt{\kappa} q_{2})}.
\end{equation*}
In the previous computation, we have used the estimate
\begin{equation*}
	\sin^{2}(\sqrt{\kappa} q_{2}) \geq \frac{4}{\pi^{2}} \kappa q_{2}^{2},
\end{equation*}
which holds because $\sqrt{\kappa} q_{2}<\pi/2$. 

We can prove in a similar way that
\begin{equation}
\label{eq:f4negative}
	f_{4}(\theta) < 0 \quad \text{for all } \theta \in (\theta_{R},-\pi+\sqrt{\kappa}q_{2}), \qquad \theta_{R} = -\pi + \left(1 - \frac{q_{1}}{q_{2}}\sqrt{\frac{4 \Gamma \mu}{2 \kappa q_{1}^{2} + \pi^{2} \Gamma}} \right)\sqrt{\kappa} q_{2},
\end{equation}
by defining the auxiliary function
\begin{equation*}
	p_{4}(\theta) = \frac{1}{2} + \Gamma \left[ \frac{1}{\sin^{2}(\theta + \sqrt{\kappa} q_{1})} - \frac{\mu}{\sin^{2}(\theta - \sqrt{\kappa} q_{2})} \right], \qquad \theta \in (-\pi,-\pi+\sqrt{\kappa}q_{2}),
\end{equation*}
and proving that $p_{4}'(\theta) < 0$ and $p_{4}(\theta_{R})<0$. 

Inequalities \eqref{eq:f4positive} and \eqref{eq:f4negative} show that if $\theta_0\in \mathcal{I}_4$ is a root of $f_4$, then $\theta_0$ does not belong to the union $ (-\pi-\sqrt{\kappa}q_{1},\theta_{L}) \bigcup \, (\theta_{R},-\pi+\sqrt{\kappa}q_{2})$, which is the complement of $\mathcal{J}_4$ in $\mathcal{I}_4$.

\paragraph{Case $j=2$.} In the proof of Proposition \ref{prop:concavity}, we verified that the roots of $f_{2}$ can only occur in the interval $(\sqrt{\kappa}q_{2},\pi/2]$. Now, consider the auxiliary function
\begin{equation*}
    g_{2}(\theta)=\frac{1}{2}-\Gamma \left[
        \frac{1}{\sin^{2}(\theta+\sqrt{\kappa}q_{1})}
        + \frac{\mu}{\sin^{2}(\theta-\sqrt{\kappa}q_{2})}
    \right], \qquad \theta\in(\sqrt{\kappa}q_{2}, \pi/2 ].
\end{equation*}
It is clear that $f_{2}(\theta)\leq g_{2}(\theta)$. Taking the derivative of $g_{2}$, we verify that $g_{2}'(\theta)>0$ for $\theta\in(\sqrt{\kappa}q_{2},\tfrac{\pi}{2}-\sqrt{\kappa}q_{1})$. Now, we define
\begin{equation*}
 	\theta_{L} = \left( 1+  \frac{3^{1/4}}{\pi^{3/2}}\sqrt{\frac{\mu}{1+\mu}}\,\epsilon_{\kappa}\right) \sqrt{\kappa}q_{2},
\end{equation*}
where $\epsilon_{\kappa}$ is given by the formula \eqref{eq:defepsilonk}. We claim that
\begin{equation}
\label{eq:claimJ2_1}
	\theta_{L}\in(\sqrt{\kappa}q_{2},\,\tfrac{\pi}{2} - \sqrt{\kappa}q_{1}),
\end{equation}
and
\begin{equation}
\label{eq:claimJ2_2}
	g_{2}(\theta_{L})<0.
\end{equation}
In view of the properties of $g_{2}$ and $f_{2}$, the conditions \eqref{eq:claimJ2_1} and \eqref{eq:claimJ2_2} imply that 
\begin{equation*}
	f_{2}(\theta) < 0 \qquad \text{for all } \theta \in (\sqrt{\kappa}q_{2} , \theta_{L}).
\end{equation*}
Therefore, if $f_{2}(\theta_{0})=0$, then necessarily $\theta_{0}\in[\theta_{L},\pi/2]$, which establishes the desired result. Thus, it only remains to prove condition \eqref{eq:claimJ2_1} and inequality \eqref{eq:claimJ2_2}. 

Using the definition of $\theta_{L}$ we obtain
\begin{equation}
\label{eq:alternativethetaL}
	\frac{\theta_{L} - \sqrt{\kappa}q_{2}}{\sqrt{\kappa}q_{2}} = \frac{3^{1/4}}{\pi^{3/2}}\sqrt{\frac{\mu}{1+\mu}}\,\epsilon_{\kappa} \leq \frac{3^{1/4}}{\pi^{3/2}}\,\epsilon_{\kappa}.
\end{equation}
Therefore, to establish the inequality \eqref{eq:claimJ2_2}, it sufficient to prove the following
\begin{equation}
\label{eq:equivalentJ2}
	\sin^{2}\!\left(\frac{3^{1/4}}{\pi^{3/2}}\sqrt{\frac{\mu}{1+\mu}}\,\epsilon_{\kappa}\sqrt{\kappa}q_{2}\right) < 2\mu\Gamma.
\end{equation}

We prove the condition \eqref{eq:claimJ2_1} and the inequality
\eqref{eq:equivalentJ2} by dividing the argument into two cases
according to the value of $\kappa$.

\begin{enumerate}[(i)]    
\item $0<\kappa<\left(\frac{\pi}{3}\right)^{2}$. We substitute the definition of $\epsilon_{\kappa}$ for this case into the estimate \eqref{eq:alternativethetaL}, obtaining
\begin{equation}
\label{eq:estimatethetaL}
	\frac{\theta_{L} - \sqrt{\kappa}q_{2}}{\sqrt{\kappa}q_{2}} \leq \left( \frac{4\sqrt{3}}{\pi^{3}} \right)^{1/2} \kappa^{1/4}.
\end{equation}
Define the function
\begin{equation*}
	\varphi(x) = x + \left( \frac{4\sqrt{3}}{\pi^{3}} \right)^{1/2} x^{3/2}, \qquad x \in (0,\pi/3).
\end{equation*}
Taking its derivative, we verify that $\varphi$ is increasing. Letting $x=\sqrt{\kappa}$, we have
\begin{equation*}
	\sqrt{\kappa}+ \left( \frac{4\sqrt{3}}{\pi^{3}} \right)^{1/2} \kappa^{3/4} = \varphi(\sqrt{\kappa}) < \varphi\left( \frac{\pi}{3} \right) = \frac{\pi}{3} + \frac{2}{3^{5/4}} < \frac{\pi}{2}.
\end{equation*}
Using this last inequality, together with inequality \eqref{eq:estimatethetaL} and the fact that $q_{2} < 1$, we obtain
\begin{equation*}
	\frac{\theta_{L} - \sqrt{\kappa}q_{2}}{\sqrt{\kappa}q_{2}} \leq \left( \frac{4\sqrt{3}}{\pi^{3}} \right)^{1/2} \kappa^{1/4} < \frac{\tfrac{\pi}{2}-\sqrt{\kappa}}{\sqrt{\kappa}q_{2}},
\end{equation*}
which proves that condition \eqref{eq:claimJ2_1} holds. To verify inequality \eqref{eq:equivalentJ2}, we use the elementary estimate
\begin{equation*}
	\sin^{3}(x)\cos(x) \geq \frac{81\sqrt{3}}{16\pi^{3}}x^{3}, \qquad 0 < x < \pi/3.
\end{equation*}
Using the expression for $\Gamma$ given in formula \eqref{eq:auxGammaq} together with the above estimate (with $x=\sqrt{\kappa}$), we obtain the lower bound
\begin{equation*}
	\Gamma \geq \frac{81\sqrt{3}}{16\pi^{3}(1+\mu)}\kappa^{3/2}.
\end{equation*}
Using the fact that $q_{2}<1$ and the definition of $\epsilon_{\kappa}$, we obtain
\begin{equation*}
	\sin^{2}\!\left(\frac{3^{1/4}}{\pi^{3/2}}\sqrt{\frac{\mu}{1+\mu}}\,\epsilon_{\kappa}\sqrt{\kappa}q_{2}\right) \leq \frac{4\sqrt{3}}{\pi^{3}} \frac{\mu}{1+\mu} \kappa^{3/2}\, q_{2}^{2} < 2\mu\,\frac{81\sqrt{3}}{16\pi^{3}(1+\mu)}\,\kappa^{3/2} \leq 2\mu\Gamma,
\end{equation*}
establishing inequality \eqref{eq:equivalentJ2}.

\item $ \left(\frac{\pi}{3}\right)^{2} < \kappa <\left(\frac{\pi}{2}\right)^{2}$. In this case, we use the estimate for $\theta_{L}$ given in formula \eqref{eq:alternativethetaL} and the definition of $\epsilon_{\kappa}$ to obtain
\begin{equation*}
	\frac{\theta_{L} - \sqrt{\kappa}q_{2}}{\sqrt{\kappa}q_{2}} \leq  \left( \frac{9\sqrt{3}}{16\pi^{3}}\right)^{1/2} (\pi-2\sqrt{\kappa}) \leq  \, \frac{\pi-2\sqrt{\kappa}}{\pi} < \frac{\frac{\pi}{2}-\sqrt{\kappa}}{\sqrt{\kappa}q_{2}}.
\end{equation*}
The second inequality follows since the coefficient on the left-hand side is bounded above by $1/\pi$, while the last inequality follows from the fact that $\sqrt{\kappa}q_{2}<\pi/2$. This proves the validity of condition \eqref{eq:claimJ2_1}. 

In order to verify inequality \eqref{eq:equivalentJ2}, we use the estimate
\begin{equation*}
	\sin^{3}(x)\cos(x) \geq \frac{9\sqrt{3}}{16\pi}\left( \pi - 2x \right), \qquad \pi/3 < x < \pi/2.
\end{equation*}
Using the expression for $\Gamma$ given in formula \eqref{eq:auxGammaq}, together with the above estimate (with $x=\sqrt{\kappa}$), we obtain
\begin{equation*}
	\Gamma \geq \frac{9\sqrt{3}}{16\pi(1+\mu)}\left( \pi - 2\sqrt{\kappa} \right). 
\end{equation*}
Additionally, in the range of $\kappa$ under consideration, we have
\begin{equation*}
	\frac{\pi-2\sqrt{\kappa}}{2} < 4, \qquad (\pi/3)^{2} < \kappa <  (\pi/2)^{2},
\end{equation*}
Using the definition of $\epsilon_{\kappa}$, together with the above inequalities, and the conditions  $q_{2}<1$ and $\kappa < (\pi/2)^{2}$, we conclude
\begin{equation*}
	\sin^{2}\!\left(\frac{3^{1/4}}{\pi^{3/2}}\sqrt{\frac{\mu}{1+\mu}}\,\epsilon_{\kappa}\sqrt{\kappa}q_{2}\right) \leq \frac{9\sqrt{3}}{16\pi} \frac{\mu}{1+\mu} \left( \frac{\pi-2\sqrt{\kappa}}{\pi} \right)^{2} \, \kappa < \mu\,\frac{\pi-2\sqrt{\kappa}}{4} \, \frac{9\sqrt{3}}{16\pi(1+\mu)}\left( \pi - 2\sqrt{\kappa} \right) \leq 2\mu\Gamma,
\end{equation*}
establishing inequality \eqref{eq:equivalentJ2}.

\end{enumerate}

\paragraph{Case $j=3$.} The argument is analogous to the previous case. By Proposition \ref{prop:concavity}, the roots of $f_{3}$ can only occur in the interval $[-\pi/2,-\sqrt{\kappa}q_{1})$. Here, the corresponding auxiliary function is
\begin{equation*}
	g_{3}(\theta) = -\frac{1}{2} + \Gamma \left[\frac{1}{\sin^{2}(\theta+\sqrt{\kappa}q_{1})} + \frac{\mu}{\sin^{2}(\theta-\sqrt{\kappa}q_{2})} \right], \qquad \theta \in [-\pi/2,-\sqrt{\kappa}q_{1}).
\end{equation*}
Then, $f_{3}(\theta)\geq g_{3}(\theta)$ and $g_{3}$ is increasing. Define
\begin{equation*}
 	\theta_{R} = - \left( 1+  \frac{3^{1/4}}{\pi^{3/2}}\frac{1}{\sqrt{1+\mu}}\,\epsilon_{\kappa}\right) \sqrt{\kappa}q_{1},
\end{equation*}
where $\epsilon_{\kappa}$ is given in Eq.~\eqref{eq:defepsilonk}. In this case, it remains to verify that
\begin{equation}
\label{eq:claimJ3_1}
	\theta_{R} \in \left(-\tfrac{\pi}{2},-\sqrt{\kappa}q_{1}\right),
\end{equation}
and
\begin{equation}
\label{eq:claimJ3_2}
	g_{3}(\theta_{R}) > 0.
\end{equation}
These properties imply that
\begin{equation*}
	f_{3}(\theta) > 0 \qquad \text{for all } \theta \in (\theta_{R},-\sqrt{\kappa}q_{1}).
\end{equation*}
The proof of condition \eqref{eq:claimJ3_1} and inequality \eqref{eq:claimJ3_2} is divided into the same two cases as in the proof of the case $j = 2$. The analogue of the formula \eqref{eq:alternativethetaL} is
\begin{equation*}
	\frac{\theta_{R} + \sqrt{\kappa}q_{1}}{\sqrt{\kappa}q_{1}} = - \frac{3^{1/4}}{\pi^{3/2}}\frac{1}{\sqrt{1+\mu}}\,\epsilon_{\kappa} \geq -\frac{3^{1/4}}{\pi^{3/2}}\,\epsilon_{\kappa},
\end{equation*}
and inequality \eqref{eq:claimJ3_2} is equivalent to verifying
\begin{equation*}
	\sin^{2}\!\left(\frac{3^{1/4}}{\pi^{3/2}}\frac{1}{\sqrt{1+\mu}}\,\epsilon_{\kappa}\sqrt{\kappa}q_{1}\right) < 2\Gamma.
\end{equation*}
The details are omitted, since the arguments are identical.
\end{proof}

\end{appendices}

\bibliographystyle{acm}
\bibliography{ref}

\end{document}